\documentclass{article}
\usepackage{microtype}
\usepackage[utf8]{inputenc}
\usepackage{color}
\usepackage{proofnet} 
\usepackage{fullpage}
\usepackage{subfigure}
\usepackage{mathsetc}

\usepackage{prerex}
\usepackage{wrapfig}

\usepackage{xargs}
\usepackage{enumerate}
\usepackage{appendix}

\newtheorem{definition}{Definition}
\newtheorem{theorem}{Theorem}
\newtheorem{lemma}{Lemma}
\newtheorem{coro}{Corollary}
\newtheorem{conjecture}{Conjecture}

\def\smallskip{\vskip\smallskipamount}
\def\medskip{\vskip\medskipamount}
\def\bigskip{\vskip\bigskipamount}

\usepackage{tipa}
\usepackage{bm}

\def\claim{  \noindent\textit{Proof.} }
\def\endclaim{ }
\newenvironment{prf}{\claim}{\endclaim}

\title{On paths-based criteria for polynomial time complexity in proof-nets}
\author{Matthieu Perrinel}

\definecolor{darkgreen}{rgb}{0,0.4,0}
\definecolor{lightblue}{rgb}{0.8,0.8,1}
\definecolor{lightgreen}{rgb}{0.8,1,0.8}

\begin{document}
\maketitle
\tableofcontents

\begin{abstract}
  Girard's Light linear logic ($LLL$) characterized polynomial time in the proof-as-program paradigm with a bound on cut elimination. This logic relied on a stratification principle and a "one-door" principle which were generalized later respectively in the systems $L^4$ and $L^{3a}$. Each system was brought with its own complex proof of Ptime soundness. In this paper we propose a broad sufficient criterion for Ptime soundness for linear logic subsystems, based on the study of paths inside the proof-nets, which factorizes proofs of soundness of existing systems and may be used for future systems. As an additional gain, our bound stands for any reduction strategy whereas most bounds in the literature only stand for a particular strategy.
\end{abstract}

\section{Introduction}
\paragraph{Implicit computational complexity} For decades, computer scientists have tried to characterize complexity classes by restricting syntactically models of computation. This effort is known as implicit computational complexity. The main application is to achieve automated certification of a program complexity. Some consider also other goals for implicit computational complexity: to understand the root of complexity, as suggested by Dal Lago~\cite{lago2006context} or to create polytime mathematics, as suggested by Girard~\cite{girard1995light}.

There are different approaches to implicit complexity, corresponding to different models of computation: restriction of recursion~\cite{cobham1965intrinsic,bellantoni1992new,leivant1994ramified}, interpretation methods~\cite{bonfante2001lexicographic}, type systems~\cite{aehlig2002syntactical,hofmann2003linear}, restriction of linear logic~\cite{lafont2004soft,girard1992bounded,girard1995light,danos2003linear}. One of the interests of the linear logic approach is the possibility to quantify over types (second order quantifiers). This allows, for example, to write a sorting algorithm which can be applied to any type of data, as soon as a comparison function is given. A milestone in that field was the creation of Light Linear Logic ($LLL$) by Girard~\cite{girard1995light} which characterizes polynomial time (Ptime). In the following study, we were mainly interested by systems based on $LLL$ characterizing Ptime, with the perspective of automated inference of complexity bounds as a final goal. 

\paragraph{Proofs as programs} In the linear logic approach, the programs are proofs of linear logic formulae and the execution of programs is done by the elimination of the cut rule in the proof~\cite{gentzen1935untersuchungen,girard1987linear}. Programming in such a system seems quite unnatural for most people. Fortunately, the proofs-as-programs correspondence~\cite{howard1980formulae} states that a logical system corresponds to a type system for $\lambda$-calculus. Thus, we would like to transform subsystems of linear logic into type systems for $\lambda$-calculus such that any typed $\lambda$-term reduces to its normal form in a number of steps bounded by a polynomial on its size (Ptime soundness). For any Ptime function $\phi$, such a type system also has to type at least one $\lambda$-term computing $\phi$ (Ptime extensionnal completeness). $\lambda$-calculus is not used directly as a programming language but functionnal programming languages (such as Lisp, Haskell or Caml) are based on it. Thus, a type system for $\lambda$-calculus characterizing a class can be seen as a first step towards the creation of a real programming language characterizing the class. Such a transformation was done on $LLL$ by Baillot and Terui, who transformed the logical system $LLL$ into a type system $DLAL$ for $\lambda$-calculus~\cite{baillot2004light}.

\paragraph{Weak and strong bounds}A programming language comes with a reduction strategy, which determines the order of the reduction. For example: do we reduce the arguments before passing them to functions (call by value) or not (call by name)? We want this strategy to be simple enough to be understood by programmers. Thus a type system must enforce a polynomial bound for any strategy (strong bound) or at least for a simple strategy. Complexity bounds are sometimes proved for farfetched strategies, which are unlikely to be implemented in a real programming language. In this paper, we will prove bounds which do not depend on the strategy (strong bounds).
 
\paragraph{Intensional completeness} $DLAL$ is extensionnally complete, so for every Ptime function $f : \mathbb{N} \rightarrow \mathbb{N}$, there exists a $\lambda$-term $t$ computing $f$ and typable by $DLAL$. However, there may be other $\lambda$-terms computing $f$ which are not typable by $DLAL$. For example, there exists $\lambda$-terms computing multiplication but the shortest $\lambda$-term for multiplication on Church unary integers ($\lambda m.\lambda n.\lambda f.m (n f)$) is not typable. This $\lambda$-term is not a complex term created for the purpose of tricking $DLAL$. Therefore, it seems that a programming language based on $DLAL$ would not type some natural $Ptime$ programs. We would like type systems typing more polynomial programs than $DLAL$. 

\paragraph{Linear logics by levels} Several extensions of $LLL$ have been studied, like $L^4$ and its refinement $L^4_0$~\cite{baillot2010linear} and $L^{3a}$~\cite{dorman2009linear}. The main novelty in those systems is that stratification and depth are no longer related (see section \ref{section_applications}). Unfortunately, the expressivity gain seems small: no meaningful program separating $L^4$ from $LLL$ or $L^4_0$ from $L^4$ has been found yet. For example, none of those systems contain the proof-net $Mult$ corresponding to $\lambda m.\lambda n.\lambda f.m (n f)$. So, it is a small step on the path of expressivity. However these systems thought new ideas that can help futre progress. It seems it unblocked the situation, by bringing new ideas. Proving strong polynomial bounds for $L^4$,$L^4_0$ and $L^3_a$  is thus interesting because it would make type systems based on these logics possible. Nevertheless, such proofs would be more interesting if the methods used  were general enough to be used on future systems. 
\paragraph{Factorizing proofs}Indeed, many systems based on linear logic have been defined to characterize several complexity classes, each system coming with its own soundness proof. Those proofs are often similar, so it seems we could ease both the search and the understandings of such proofs by factorizing parts of those proofs. An important progress was made in this direction by Dal Lago with context semantics~\cite{lago2006context}: he provided a common method to prove complexity bounds for several systems like $ELL$, $LLL$ and $SLL$. Here we go a step further by designing higher level criteria, based on context semantics. The idea of geometry of interaction~\cite{girard1989geometry} and context semantics is to study the reduction of proof-nets (or $\lambda$-terms) by leaving the proof-net unchanged and analysing instead some paths in it.

\paragraph{Contributions} In this paper, we will define a ``stratification criterion'' and a ``dependence control criterion'' on proof-nets. Stratification alone implies a strong elementary bound on cut elimination. If both criteria are satisfied they imply a strong polynomial time bound. We will then prove that $ELL$ satisfies the stratification criterion, and that $LLL$, $L^4$ and $L^{3a}$ satisfy both criteria. This proves strong polynomial bounds $L^4$ and $L^{3a}$ (for which only weak bounds were previously known). We also prove a strong polynomial bound for $L^4_0$ thanks to the $L^4$ strong bound.

\paragraph{Related works}
In the search for an expressive system for complexity properties, Dal Lago and Gaboardi have defined the type system $dlPCF$ which can validate all $dlPCF$ programs. Type-checking in $dlPCF$ is undecidable, but one can imagine restricting $dlPCF$ to a decidable fragment. Their framework can be seen as a top-down approach. Here we follow instead a bottom-up line of work: we take inspiration from previous deciable type systems characterizing $Ptime$ and try to relax conditions without losing neither soundess nor decidability.

Our main tool will be context semantics. Context semantics is related to geometry of interaction~\cite{danos1995proof} and has first been used to study qualitative properties~\cite{gonthier1992geometry}. In~\cite{baillot2001elementary}, Baillot and Pedicini use geometry of interaction to characterize elementary time. In~\cite{lago2006context}, Dal Lago adapted context semantics to study quantitative properties of $cut$-elimination. From this point of view an advantage of context semantics compared to the syntactic study of reduction is its genericity: some common results can be proven for different variants of linear logic, which allows to factor out proofs of complexity results for these various systems. Our framework is slightly different from Dal Lago's context semantics. In particular, Dal Lago worked in intuitionnistic linear logic, and we work in classical linear logic. So the results of~\cite{lago2006context} can not be directly applied. However most theorems of~\cite{lago2006context} have correspondents in our framework, with quite similar proofs.

Many papers on complexity in linear logic work on subsystems of linear logic. Our goal was to define criteria on full linear logic so that our results would be as reusable as possible. Unfortunately, we had to get rid off digging. Some other works deal with full linear logic. For example, in~\cite{laurent2006obsessional}, Tortora De Falco and Laurent define a criterion on the elements of the relational model of full linear logic: the obsessional cliques. Then they show that the linear logic proofs mapped to obsessional cliques are exactly the $SLL$ proofs and hence $Ptime$. 

\paragraph{Outline} 
In the remaining of this introduction we will give an intuitive understanding of stratification and dependence control. In section~\cite{section_proofnets}, we present linear logic and proof-nets. In section~\cite{section_contextsemantics}, we present our version of context semantics and prove useful lemmas. The only significant results in this part which were not in~\cite{lago2006context} are those on ``underlying formulae'' (lemmas \ref{lemma_underlying_step} and \ref{lemma_underlying_finite}). Notice also the Dal Lago's weight theorem, for which we give a more detailed and formal proof than the original. In sections \ref{section_stratification} and \ref{section_dependence}, we give our stratification and dependence control conditions and prove that they imply polynomial soundness. In section \ref{section_applications}, we prove strong polynomial bounds for $LLL$, $L^4$ and $L^{3a}$ (only the first being previously proved). In section \ref{section_l40}, we prove a strong polynomial bound for $L^4_0$ for which no bound was previously known.

\subsection{The roots of complexity}
\label{complex}
To investigate the roots of complexity, we begin by looking at a famous non-normalizing $\lambda$ term: $t= (\lambda x. x x) (\lambda y. y y)$. It is often said that the cause of this divergence is the fact that $x$ is applied to itself~\cite{pierce2002types}. However the problem seems to be more subtle, because the $\lambda$ term $u= (\lambda x. x x) (\lambda y. y)$ normalizes. The difference seems to be that in $t$, $\lambda y. y y$ will duplicate itself during reduction, whereas in $u$, $\lambda y. y$ is applied to itself without duplication. If we want to control complexity precisely, we will need to make a difference between self-application and self-duplication. The usual type systems for $\lambda$ calculus, based on intuitionnistic logic (Simple typing, System F) can not do this. This is why we will use a type system based on a logic which controls duplication: linear logic. 

While highlighting the difference between self-application and self-duplication, this example was not really convincing. Indeed System F already accepts $(\lambda x. x x) (\lambda y. y)$ while refusing $(\lambda x. x x) (\lambda y. y y)$. Let's look at other examples to see the limits of system $F$.

Let us use the Church encoding for natural numbers. Integers are typed by type $\underline{\mathbb{N}}=\forall X. \oc (X \multimap X) \multimap \oc (X \multimap X)$. The integer $n$ is represented by
\begin{equation*}
  {\underline  n} = \lambda f. \lambda x. \underbrace{ f(f(...(f ~ x)))}_{n \text{ applications of } f} : \underline{\mathbb{N}}
\end{equation*}
Then the following terms represent the functions $n \mapsto 2.n$, $n \mapsto 2^n$ and $n \mapsto Ackermann(n,n)$. In fact, System $F$ makes no difference between those terms, typing each term with $\mathbb{N} \Rightarrow \mathbb{N}$. The complexity of those functions being quite different (linear, exponential and over-elementary), we would like to understand what makes their complexity so different and tell them apart. 

\begin{align*}
  mult & := \lambda n. n \left( \lambda a. \lambda f. \lambda x. f (f (a f x)) \right) 0 \\
  exp & := \lambda n. n \left( \lambda a. \lambda f. \lambda x. a f (a f x) \right) \\
  ack & := \lambda n. n \left( \lambda a. \lambda k. (S k) a 1 \right) S ~ n 
\end{align*}

We can see that those three functions are built similarly: they are all based on the iteration of some auxiliary function. They all have a subterm of shape $n (\lambda a. t) b$ as a key component. 

When we reduce those terms, $n (\lambda a. t) b$ becomes $(\lambda a. t) (\lambda a. t) \cdots (\lambda a. t) b$. So $\lambda a. t$ is self-applicating. Depending on the use of $a$ in $t$, it might also be self-duplicating. This will tell apart the complex functions from the simple ones. Indeed, in $mult$, $a$ is used only once. In $exp$, $a$ is used twice, so $(\lambda a. t)$ is self-duplicating. But this duplication is bounded. In $ack$, $a$ is itself iterated, and during the $i$-th iteration of $\lambda a. \lambda k. (S k) a 1$, the number of times $a$ is used is equal to $k+1$ so depends on the next iteration. So, in $ack$, $(\lambda a. t)$ is self-duplicating in an unbounded way.

This analysis is quite unformal. In order to forbid this notion of ``self-duplication'', we will use linear logic. More precisely we will define a subset of linear logic in which self-duplication is impossible.

\subsection{Proof-nets, intuitively}

In the same way that $\lambda$-calculus computes functions with the $\beta$-reduction, linear logic computes functions with the $cut$-elimination. It is the elimination of the $cut$-rules in proofs of linear logic (the {\it Hauptsatz} of Gentzen). This elimination is more natural on {\it proof nets}, which are graphical representation of proofs. The proof nets for $L^4$ will be defined in section \ref{section_proofnets}. Here, we will only give an intuitive understanding of proof-nets, necessary for \ref{strat_dep}.

Intuitively, the edges of a proof net correspond either to programs or to requests. For example the function $x \mapsto x +2$ is a program. $x+2$ alone is a program, but it comes with a request $x^\perp$ for a parameter $x$. The program represented by the proof net, is the program labelling its pending edge. $\otimes$ can be thought of as an application and $\parr$ as an abstraction. 

Some programs ($x \mapsto x + x $ for example) need to duplicate a subprogram ($x$ here). We want to control this operation. When we duplicate a program, we need to duplicate the requests associated to it (in $x \mapsto (x*y) + (x*y)$ we observe that to duplicate $x*y$ we need to request twice $x$ and $y$). This is the purpose of boxes: a box is a part of the proof net with a principal door ($!P$) which is a program we want to duplicate, and optional auxiliary doors ($?P$) which are requests associated to the program. Only programs preceded by a $!$ can be duplicated. So, contracting various requests for a same program $x$ in a single requests can only be done if $x$ is preceded by a $!$. 
 The right part of figure \ref{deltadelta} is labelled corresponding to this intuitive understanding of proof nets.

\begin{figure}[ht]
\centering
\begin{tikzpicture}
\begin{scope}[scale=0.9]
\draw (0,0) node (par) {$\parr$};
\draw (par) ++(140 : 1) node (cont) {$?C$};
\draw (cont)++( 60 : 1) node (der)  {$?D$};
\draw (der) ++( 60 : 1) node (tens) {$\otimes$};
\draw (tens)++( 60 : 1) node (ax)   {$ax$};
\draw (tens)++(120 : 1) node [draw, circle, inner sep=0.02cm] (bprom) {$!P$};
\draw (bprom)++(-0.5, 0.7) node (ax2) {$ax$};
\draw (bprom)++(-1,0) node [draw, circle, inner sep=0.02cm] (wprom) {$?P$};

\draw [<-, thin] (par)--(cont);
\draw [<-, thin] (cont)--(der);
\draw [<-, thin](cont) to [bend left=10] (wprom);
\draw [<-, thick](der)--(tens) node [right, midway] {$f$};
\draw [<-, thin]  (tens) to [bend left] (ax.180);
\draw [<-, thin] (tens)--(bprom) node [left,midway] {$e$};
\draw [<-, thick] (bprom) to [bend right] (ax2);
\draw [->, thick] (ax2) to [bend right] (wprom);
\draw [->, thin] (ax.0) to [out=-55, in = 70] (par);
\draw [thin] (bprom) -| ++(0.3,1) -| ++(-1.6,-1) -- (wprom) -- (bprom);

\draw (par) ++(2,-1) node (cut) {$cut$};
\draw (par) ++(4,-0.5) node (tens) {$\otimes$};
\draw (tens)++(30:1) node (ax3) {$ax$};
\draw (ax3) ++(-30:1) node (sort) {};
\draw (tens) ++(120:1.5) node [draw, circle, inner sep=0.02cm](bang) {$!$};

\draw [->, thin] (par) to [bend right] (cut);
\draw [<-, thin] (cut) to [out=0, in=-120] (tens);
\draw ($(cut)!0.5!(tens)$) node [right] {$\Delta^\perp$};
\draw [<-, thin] (tens)--(bang) node [pos=0.8, right] {$\Delta := !(!x \mapsto x (!x))$};
\draw [<-, thin] (tens) to [bend left] (ax3);
\draw [->, thin] (ax3) to [bend left] (sort);
\draw ($(ax3)!0.5!(sort)$) node [below] {$\Delta (\Delta)$};
\draw [thin] (bang) -| ++(2,5.5) -| ++(-4,-5.5) -- (bang);

\draw (bang) ++(0,1) node (par) {$\parr$};
\draw [<-, thick] (bang)--(par) node [midway, right] {$!x \mapsto x (!x)$};
\draw (par) ++(140 : 1) node (cont) {$?C$};
\draw (cont)++( 60 : 1) node (der)  {$?D$};
\draw (der) ++( 60 : 1) node (tens) {$\otimes$};
\draw (tens)++( 60 : 1) node (ax)   {$ax$};
\draw (tens)++(120 : 1) node [draw, circle, inner sep=0.02cm] (bprom) {$!P$};
\draw (bprom)++(-0.5, 0.7) node (ax2) {$ax$};
\draw (bprom)++(-1,0) node [draw, circle, inner sep=0.02cm] (wprom) {$?P$};

\draw [<-, thick] (par)--(cont) node [midway, below left] {$(!x)^\perp$};
\draw [<-, thick] (cont)--(der) node [midway, right] {$(!x)^\perp$};
\draw ($(cont)!0.5!(wprom)$) node [left] {$(!x)^\perp$};
\draw [<-, thick]  (cont)to [bend left=10] (wprom);
\draw [<-, ultra thick](der)--(tens) node [midway, right] {$x^\perp$};
\draw [<-, thick] (tens)--(bprom) node [midway, left] {$!x$};
\draw [<-, thick] (tens) to [bend left] (ax);
\draw ($(tens)!0.5!(ax) + (0.4,0)$) node  {$(x (!x))^\perp$};
\draw [<-, ultra thick] (bprom) to [bend right] (ax2);
\draw ($(bprom)!0.5!(ax2) +(0.3,0.3)$) node {$x$};
\draw [->, ultra thick] (ax2) to [bend right] (wprom);
\draw ($(ax2)!0.5!(wprom) +(-0.3,0.3)$) node {$x^\perp$};
\draw [->, thick] (ax) to [out=-10, in = 50] (par);
\draw ($(ax)!0.6!(par) +(1.2,0)$) node {$x (!x)$};
\draw [thick] (bprom) -| ++(0.3,1) -| ++(-1.6,-1) -- (wprom) -- (bprom);
\end{scope}
\end{tikzpicture}
\caption{This proof net reduces to itself, it represents $(\lambda x. x x) (\lambda x. x x)$}
\label{deltadelta}
\end{figure}
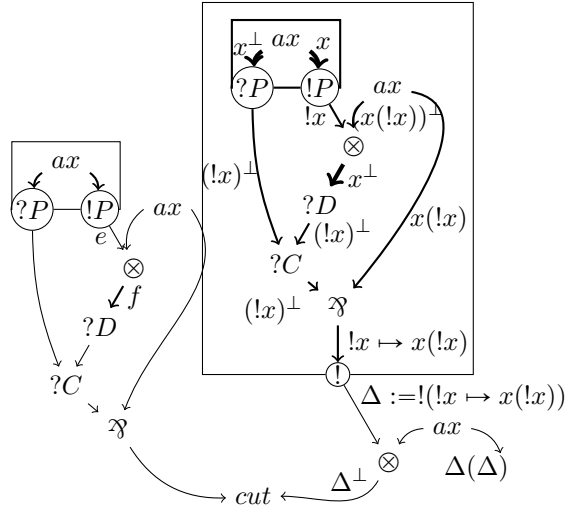

\subsection{Stratification and dependence control}
\label{strat_dep}
To avoid the complexity explosion of $ack$, we want to forbid iteration of functions $\lambda x.t$ where the number of times $x$ is used in $(t)u$ depends on $u$. The first such restriction was created by Girard. In $LLL$, he forbids the dereliction and digging principles ($!A \multimap A$ and $!A \multimap !!A$). Such a restriction corresponds to elementary-time (tower of exponentials of fixed height) functions. For example, you can see in figure \ref{deltadelta} that the proof-net corresponding to $(\lambda x. x x)(\lambda y. y y)$ is ruled out because it uses a dereliction (the $?D$ link). The proof-net corresponding to Ackermann is ruled out because it would use digging. We will call similar restrictions {\em stratification conditions} (more details in section \ref{section_stratification}).

Though stratification gives us a bound on the length of the reduction, elementary time is not considered as a reasonable bound. Figure \ref{exp} explains us how the complexity arises, despite stratification. On this proof net, the box A duplicates the box B. Each copy of B duplicates C, each copy of C,... To avoid it, Girard~\cite{girard1995light} limited the number of $?P$-doors of each $!$-boxes to 1. To keep some expressivity, he introduced a new modality $\S$ with $\S$-boxes which can have an arbitrary number of $?P$-doors. We will call {\em dependence control condition} any restriction preventing this kind of sequences (more details in section \ref{section_dependence}.

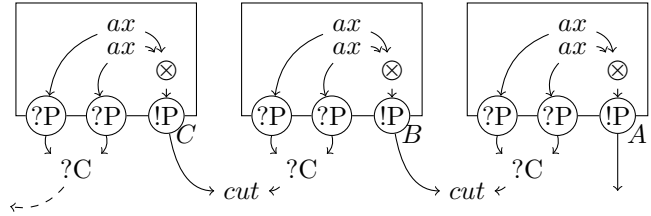
\begin{figure}[ht]
  \centering
  \begin{tikzpicture}
    \tikzstyle{door}=[draw, circle, inner sep=0.02cm]
    \draw (0,0) node [door] (bang1) {!P};
\draw (bang1) node [below right] {$C$}; 
\draw (bang1)++(-1.2,-0.7) node (why1) {?C};
\draw (why1)++(-0.4,0.7) node [door] (ancl1) {?P};
\draw (why1)++(0.4,0.7) node [door] (ancr1) {?P};
\draw (bang1) ++(0,0.6) node (tens1) {$\otimes$};
\draw (tens1) ++(-0.6,0.3) node (ax1b) {$ax$};
\draw (tens1) ++(-0.6,0.6) node (ax1h) {$ax$};
\draw (bang1) ++ (1,-1) node (cut1) {$cut$};
\draw [<-] (bang1)--(tens1);
\draw [<-] (tens1) to [bend right] (ax1h);
\draw [<-] (tens1) to [bend right] (ax1b);
\draw [->] (ax1h) to [bend right] (ancl1);
\draw [->] (ax1b) to [bend right] (ancr1);
\draw (3,0) node (bang2) [door] {!P};
\draw (bang2) node [below right] {$B$}; 
\draw (bang2)++(-1.2,-0.7) node (why2) {?C};
\draw (why2)++(-0.4,0.7) node [door]  (ancl2) {?P};
\draw (why2)++(0.4,0.7) node  [door] (ancr2) {?P};
\draw (bang2) ++(0,0.6) node (tens2) {$\otimes$};
\draw (tens2) ++(-0.6,0.3) node (ax2b) {$ax$};
\draw (tens2) ++(-0.6,0.6) node (ax2h) {$ax$};
\draw (bang2) ++ (1,-1) node (cut2) {$cut$};
\draw [<-] (bang2)--(tens2);
\draw [<-] (tens2) to [bend right] (ax2h);
\draw [<-] (tens2) to [bend right] (ax2b);
\draw [->] (ax2h) to [bend right] (ancl2);
\draw [->] (ax2b) to [bend right] (ancr2);
\draw (6,0) node (bang3)  [door] {!P};
\draw (bang3) node [below right] {$A$}; 
\draw (bang3)++(-1.2,-0.7) node (why3) {?C};
\draw (why3)++(-0.4,0.7) node [door] (ancl3) {?P};
\draw (why3)++(0.4,0.7) node [door] (ancr3) {?P};
\draw [->](bang1) to [out=-80, in=180] (cut1);
\draw [<-](cut1) to [out=0, in=-140] (why2);
\draw [->](bang2) to [out=-80, in=180] (cut2);
\draw [<-] (cut2) to [out=0, in=-140] (why3);
\draw (bang3) ++(0,0.6) node (tens3) {$\otimes$};
\draw (tens3) ++(-0.6,0.3) node (ax3b) {$ax$};
\draw (tens3) ++(-0.6,0.6) node (ax3h) {$ax$};
\draw [<-] (bang3)--(tens3);
\draw [<-] (tens3) to [bend right] (ax3h);
\draw [<-] (tens3) to [bend right] (ax3b);
\draw [->] (ax3h) to [bend right] (ancl3);
\draw [->] (ax3b) to [bend right] (ancr3);
\draw [->] (bang3) -- ++ (0,-1);
\draw (bang1) --++(0.4,0) --++(0,1.5) --++(-2.4,0) |- (ancl1)--(ancr1)--(bang1);
\draw (bang2) --++(0.4,0) --++(0,1.5) --++(-2.4,0) |- (ancl2)--(ancr2)--(bang2);
\draw (bang3) --++(0.4,0) --++(0,1.5) --++(-2.4,0) |- (ancl3)--(ancr3)--(bang3);
\draw [<-](why1) to [bend left] (ancl1);
\draw [<-](why1) to [bend right](ancr1);
\draw [<-](why2) to [bend left] (ancl2);
\draw [<-](why2) to [bend right] (ancr2);
\draw [<-](why3) to [bend left] (ancl3);
\draw [<-](why3) to [bend right](ancr3);
\draw (why1) ++(-1,-0.5) node (suite) {};
\draw [->,dashed] (why1) to [bend left] (suite);

\end{tikzpicture}
\caption{This proof net (if extended to n boxes) reduces in $2^n$ steps}
\label{exp}
\end{figure}

\section{Linear logic}
\label{section_proofnets}

\paragraph{}Linear logic ($LL$)~\cite{girard1987linear} can be considered as a refinement of System F~\cite{girard1971extension} where we focus especially on how the duplication of formulae is managed. In System F, $A \Rightarrow B$ means ``with many proofs of $A$, I can create a proof of $B$''. Linear logic decompose it into two connectives: $\oc A$ means ``infinitely many proofs of $A$'', $A \multimap B$ means ``using exactly one proof of $A$, I can create a proof of $B$''. We can notice that we can represent $A \Rightarrow B$ with $(\oc A) \multimap B$. In fact, $A \multimap B$ is a notation of $A^\perp \parr B$. $(\_)^\perp$ can be considered as a negation and $\parr$ as a disjunction. In fact the disjunctions $\vee$ and conjunction $\wedge$ are separated into two disjunctions ($\parr$ and $\oplus$) and two conjunctions ($\otimes$ and $\with$). In this paper, we will only use the ``multiplicative'' ones: $\parr$ and $\otimes$. 

Finally $\forall$ and $\exists$ allow us, as in System F, to quantify over the sets of formulae. As examples, let us notice that $\forall X. X \multimap X$ is true (for any formula $X$, using exactly one proof of $X$, we can create a proof of $X$). On the contrary, $\forall X. X \multimap (X \otimes X)$ is false because, in the general case, we need two proofs of $X$ to prove $X \otimes X$. The set $\mathcal{F}_{LL}$, defined as follows, designs the set of formulae of linear logic.
\begin{equation*}\label{def_fll}
  \mathcal{F}_{LL}  = X \mid X^\perp  \mid \mathcal{F}_{LL} \otimes \mathcal{F}_{LL} \mid \mathcal{F}_{LL} \parr \mathcal{F}_{LL} \mid \forall X \mathcal{F}_{LL} \mid \exists X \mathcal{F}_{LL} \mid !\mathcal{F}_{LL} \mid ?\mathcal{F}_{LL} 
\end{equation*}

In the following paper, we will study variations of Linear Logic. These variations are not really subsystems of $LL$ because we extend the language of formula with indexes on atomic formulae and a connective $\S$. However, those variations all are subsystems of a same system, which we will call $LL_0$. We will work inside $LL_0$ for the main results, keeping in mind that we are not interested in $LL_0$ itself but in its subsystems. The system $LL_0$ uses the following formulas, where $X$ ranges over a denumerable set of variables ($X, X^\perp$ are atomic formulae) and $p$ ranges over $\mathbb{Z}$.
\begin{equation*} \label{def_fll0}
  \mathcal{F}_{LL_0}  = p.X \mid p.X^\perp  \mid \mathcal{F}_{LL_0} \otimes \mathcal{F}_{LL_0} \mid \mathcal{F}_{LL_0} \parr \mathcal{F}_{LL_0} \mid \forall X \mathcal{F}_{LL_0} \mid \exists X \mathcal{F}_{LL_0} \mid !\mathcal{F}_{LL_0} \mid ?\mathcal{F}_{LL_0} \mid \S \mathcal{F}_{LL_0}
\end{equation*}

You can notice that $(\_)^\perp$ is only defined on atomic formulae. We define inductively an involution $(\_)^\perp$ on $\mathcal{F}_{LL_0}$, which can be considered as a negation:\label{def_perpformula} $(p.X)^\perp =p.X^\perp$, $(p.X^\perp)^\perp =p.X$, $(A \otimes B)^\perp = A^\perp \parr B^\perp$, $(A \parr B)^\perp = A^\perp \otimes B^\perp$, $(\forall X.A)^\perp = \exists X.A^\perp$, $(\exists X.A)^\perp = \forall X.A^\perp$, $(!A)^\perp = ?(A^\perp)$, $(?A)^\perp= \oc (A^\perp)$ and $(\S A)^\perp = \S (A^\perp)$.

The $\S$ connective was introduced by Girard~\cite{girard1995light}. This modality is useful for expressivity. It is difficult to give it a meaning, and it has no real equivalent in usual logics. The integer $p$ in atomic formula $p.X$ means $\underbrace{\S\S \cdots \S}_{p\text{ symbols}}X$. In fact it was used by Baillot and Mazza~\cite{baillot2010linear} to replace the $\S$ connective. In this paper we will prove results which stands for both systems, that is why we allow both notations in $LL_0$. One of the particularity of the $p.X$ notation compared to the $\S$ one is that the indexes can only be on atomic formulae. To change the indexes of all the atomic formulae of a formula, we define for any $q \in \mathbb{N}$ and $A \in \mathcal{F}_{LL_0}$ the formula $q.A$ as follows: $q.(p.X)=(q+p).X$, $q.(p.X^\perp)=(q+p).X^\perp$ and $q.(~)$ commutes with every other connective. For example $q.(A \otimes B)= (q.A)\otimes(q.B)$ and $q.(\exists X.A) = \exists X.(q.A)$.\label{def_decindices}

We can observe that the formulae of linear logic $F_{LL}$ form a subset of $F_{LL_0}$. We can define a forgetful mapping from $F_{LL_0}$ to $F_{LL}$: $(p.X)_{/LL}=X$, $(p.X^\perp)_{/LL}=X^\perp$, $(\S A)_{/LL}=A_{/LL}$ and $(~)_{/LL}$ commutes with all other connectives. Similarly, we define $(~)_{/L^4}$ which erases the indexes on the atomic formulae and $(~)_{/L^4_0}$ which erases the $\S$ connectives.\label{def_proj_fll}

\begin{figure}
  \subfigure[][proof net computing to $n \mapsto n+n$]{
    \begin{tikzpicture}
      \tikzstyle{type}=[opacity=0]
      \tikzstyle{every node}=[inner sep=0.05cm]
      \tikzstyle{neutpos}=[princdoor]
      \tikzstyle{neutneg}=[auxdoor]
      \tikzstyle{exbox}=[]
      \tikzstyle{excont}=[]
      \draw  (0,0) node (parrn) [par] {};
      \draw (parrn) ++(60:1) node (bang) {};
      \draw [ar] (parrn)--++(0,-0.4) node [type] {$! \mathbb{N} \multimap  \mathbb{N}$};
      \draw (bang) ++(60:0.3) node [par] (parrf) {};
      \draw (parrf) ++(60:1.5) node [neutpos, exbox] (neut) {};
      \draw (neut) ++(60:0.7) node [par] (parrx) {};
      \draw [revar] (parrn)--(parrf) node [type] {$\mathbb{N}$};
      \draw [revar, exbox] (parrf)--(neut) node [type, exbox] {$!(\alpha^\perp \parr \alpha) $};
      \draw [revar, exbox] (neut)--(parrx) node [type, exbox] {$\alpha^\perp \parr \alpha$};
      \draw (parrn) ++(175:2)node [cont, excont] (contn) {};
      \draw [opacity=0] (parrn)--(contn) node [type, below, excont] {$? \mathbb{N}^\perp$}; 
      \draw [revar, excont] (parrn) to [out=160, in = -40](contn);
      \draw (parrf) ++(175:1)node [cont] (contf) {};
      \draw [revar] (parrf) -- (contf);    
      
      \draw (parrx) ++ (-1,1) node [ax] (ax3) {};
      \draw (ax3) ++ (-2,0) node [ax] (ax2) {};
      \draw (ax2) ++ (-2,0) node [ax] (ax1) {};
      \draw [revar] (parrx) to [out=90, in=0] (ax3);
      \draw [ar] (ax1) to [out=-170, in = 170] (parrx);
      
      \draw (ax3) ++ (-1,-0.4) node [tensor] (appnfx2) {};
      \draw (appnfx2) ++(-2.5,0) node [tensor] (appnfx1) {};
      \draw [revar] (appnfx2) to [out = 50, in = 180] (ax3);
      \draw [revar] (appnfx2) to [out =130, in = 0] (ax2);
      \draw [revar] (appnfx1) to [out = 40, in = 180] (ax2);
      \draw [revar] (appnfx1) to [out =100, in = 0] (ax1);

      \draw (appnfx2) ++(-0.5,-1.8) node [tensor] (appnf2) {};
      \draw (appnf2) ++(-2.5,0) node [tensor] (appnf1) {};
      \draw [opacity=0, name path=appapp2] (appnfx2) -- (appnf2);
      \draw [opacity=0, name path=appapp1] (appnfx1)--(appnf1);
      \draw [opacity=0, name path=neutline] (neut) --++(-5,0);
      \draw [name intersections={of=appapp1 and neutline}] (intersection-1) node [neutneg, exbox] (neut1) {};
      \draw [name intersections={of=appapp2 and neutline}] (intersection-1) node [neutneg, exbox] (neut2) {};
      \draw [exbox, box, name path= box] (neut) -| ++ (1,1.8)  -| ($(-1,0)+(neut1)$) -- (neut1) -- (neut2)--(neut);
      \draw [ar, exbox] (appnfx1)--(neut1) node [type, exbox] {$\alpha \otimes \alpha^\perp$};
      \draw [ar, exbox] (appnfx2)--(neut2) node [type, exbox] {$\alpha \otimes \alpha^\perp$};
      \draw [ar, exbox] (neut1)--(appnf1) node [type, exbox] {$?(\alpha \otimes \alpha^\perp)$};
      \draw [ar, exbox] (neut2)--(appnf2) node [type, exbox] {$?(\alpha \otimes \alpha^\perp)$};
      \draw [revar, name path=application1, opacity=0] (contn) -- (appnf1);
      \draw [revar, name path=application2, opacity=0] (contn) -- (appnf2);
      \draw [opacity =0, name path=boxbottom] (bang) --++(-5,0);
      \draw [name intersections={of=application1 and boxbottom}] (intersection-1) node [der] (dern1) {};
      \draw [name intersections={of=application2 and boxbottom}] (intersection-1) node [der] (dern2) {};
      \draw [ar] (appnf1)--(dern1) node [type, left] {$\mathbb{N}^\perp$};
      \draw [ar, excont] (dern1) --(contn) node [type, left, excont] {$? \mathbb{N}^\perp$};
      \draw [ar] (appnf2)--(dern2) node [type] {$\mathbb{N}^\perp$};
      \draw [ar, excont] (dern2)--(contn) node [type, excont] {$? \mathbb{N}^\perp$};
      \coordinate (bang2pos) at ($(appnf2)+(-0.4,0.7)$);
      \coordinate (bang1pos) at ($(appnf1)+(-0.4,0.7)$);
      \coordinate (whyn2pos) at ($(bang2pos)+(-1,0)$);
      \coordinate (whyn1pos) at ($(bang1pos)+(-1,0)$);
      
      \nvar{\margedoor}{0.3cm}
      \nvar{\hauteurboite}{0.7cm}
      
      \draw (appnf2) ++(-0.6, 0.3) node [ax] (axapp2) {};
      \draw (appnf1) ++(-0.6, 0.3) node [ax] (axapp1) {};
      \draw [revar] (appnf2) to [out=120,in=0] (axapp2);
      \draw [revar] (appnf1) to [out=120,in=0] (axapp1);
      \draw [ar] (axapp1) to [out=-160, in = 160] (contf);
      \draw [ar] (axapp2) to [out=-160, in = 150](contf);
      
      \draw [opacity=0] (bang)-| ++(1.5,3.3) -| ++(-8,-3.3) -- (dern1) -- (dern2)--(bang);
    \end{tikzpicture}
  }
  \subfigure[][Syntactic tree of $\lambda n. \lambda f. \lambda x . (n f) (nf) x$]{
    \begin{tikzpicture}
      \tikzstyle{every node}=[inner sep=0.05cm]
      \begin{scope}[yscale=0.55, xscale=0.7]
        \draw[] (0,0) node  (lambn) {$\lambda n$};
        \draw[orange, opacity=0] (lambn)++(2.2,-2.2) node (bang) {!};
        \draw[] (bang)++(0.8,-0.8) node (lambf) {$\lambda f$};
        \draw[orange, opacity=0] (bang)--++(-4,0) --++(0,-6.3) --++(8.5,0) |- (bang);
        \draw[] (lambf)++(0.8,-0.8) node (lambx) {$\lambda x$};
        \draw[] (lambx)++(0.8,-0.8) node (appx1) {$@$};
        \draw[] (appx1)++(0.8,-0.8) node (appx2) {$@$};
        \draw[] (appx1)++(-1.5,-0.5) node (nf1) {$@$};
        \draw[] (appx2)++(-1.5,-0.5) node (nf2) {$@$};
        \draw[orange, opacity=0] (lambn)++(-0.5,-0.8) node (contn) {$!D$};
        \draw[orange, opacity=0] (contn)++(-0.5,-1) node (dern1) {$L!$};
        \draw[orange, opacity=0] (contn)++(0.5,-1) node (dern2) {$L!$};
        \draw[orange, opacity=0] (lambf)++(-3,-2.5) node (contf) {$!D$};
        \draw[orange, opacity=0] (contf)++(-0.5,-1) node (derf2) {$L!$};
        \draw[orange, opacity=0] (contf)++(0.5,-1) node (derf1) {$L!$};
        \draw[orange, opacity=0] (nf1)++(-1.5,-1.5) node (redf1) {$R!$};
        \draw[orange, opacity=0] (nf2)++(-2.5,-1.8) node (redf2) {$R!$};
        \draw[orange, opacity=0] (redf1)--++(0.5,0) --++(0,-0.7)--++(-2,0) |- (redf1);
        \draw[orange, opacity=0] (redf2)--++(0.5,0) --++(0,-0.7)--++(-2,0) |- (redf2);
        
        \draw (lambn)--(bang.center)--(lambf)--(lambx)--(appx1)--(appx2);
        \draw [<-] (appx1)--(nf1);
        \draw [<-] (appx2)--(nf2);
        \draw (lambn)--(contn.center);
        \draw (contn.center)--(dern1.center) (contn.center)--(dern2.center);
        \draw (lambf)--(contf.center);
        \draw [->](dern2.center) to [out=-60, in=-150] (nf1);
        \draw [->] (dern1.center) to [out=-100, in =-150] (nf2);
        \draw [->](contf.center) to [out=-120, in= -70] (nf2);
        \draw [->](contf.center) to [out=-60,in=-70] (nf1);
        \draw [->](lambx) to [out=-120, in=-90] (appx2);
        \draw (lambn)--++(0,0.8);
      \end{scope}
    \end{tikzpicture}
  }
  \caption{ \label{curry_howard_iso}We can observe graphically the proofs-as-program correspondence: if we erase the $\oc P$, $\wn P$ and $\wn C$ of the proof net and reverse it, we obtain the syntactic tree of the corresponding $\lambda$-term}
\end{figure}
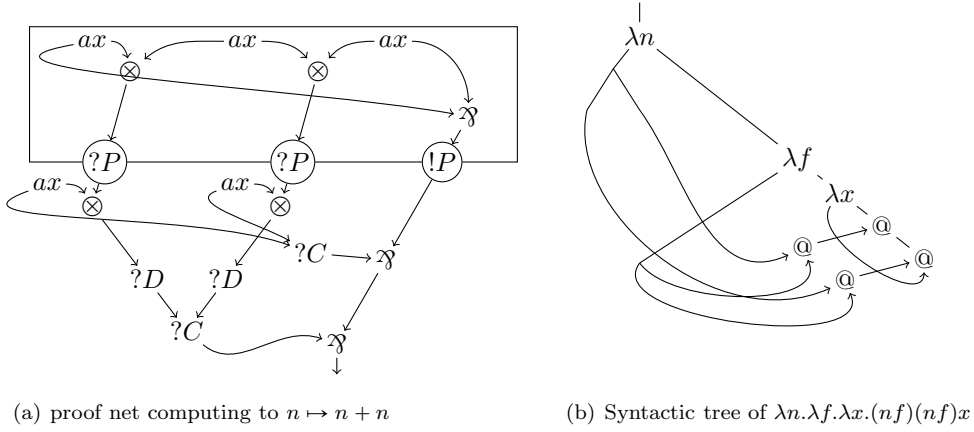  

Proof-nets are the programs of linear logic. They are graph-like structures in which the links correspond to uses of logical rules. We can consider them as $\lambda$-terms, with added information on how the duplication of variables is managed. In fact the proofs-as-programs correspondence gives us a mapping from the intuitionistic fragment of proof-nets to $\lambda$-terms. As an example, you can observe in figure \ref{curry_howard_iso} a proof-net corresponding to the $\lambda$-term $\lambda n. \lambda f. \lambda x. (n f) (n f x) x$ which computes the function $ \left \{ \begin{array}{lcl} \mathbb{N} & \mapsto & \mathbb{N} \\ n & \mapsto & n+n \end{array} \right .$. We can see that if we erase the exponential links ($\wn C$, $\wn P$ and $\oc P$) and reverse the proof net, we get the syntactic tree of the $\lambda$-term. The $\otimes$ link corresponds to the application of a function and the $\parr$ link corresponds to the abstraction.

\begin{definition}[proof-net]\label{def_proofnet}
  A {\em proof-net} is a graph-like structure defined inductively by the graphs of figure \ref{rules_labeling} ($G$ and $H$ being proof-nets). The vertices are called $links$, a link $l$ is labelled by a connector $\alpha(l)$. The set of edges is written $E_G$. The edge $e$ is labelled by a formula $\beta(e)$ in $\mathcal{F}_{LL_0}$.
\end{definition}

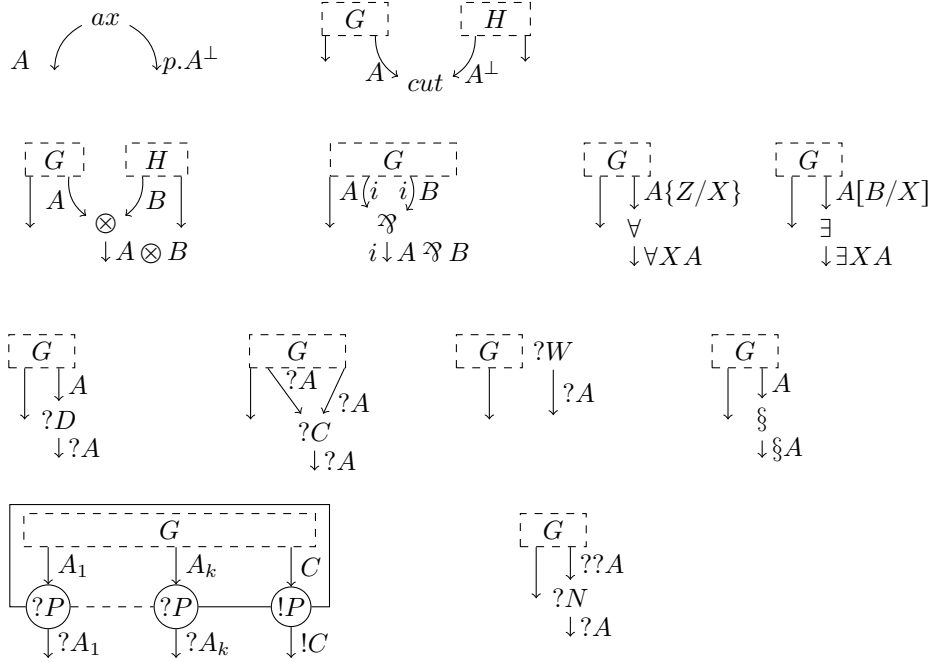
\begin{figure}
  \centering
  \begin{tikzpicture}
    \tikzstyle{level}=[opacity = 0]
    \begin{scope}[scale = 0.85]
      \begin{scope}[shift={(0,1)}]
        \draw (0,0) node (ax) {$ax$};
        \draw[->] (ax) to [bend left] (0.8,-0.8);
        \draw ($(ax)!0.8!(0.8,-0.8)+(0.7,0)$) node {$p.A^\perp$};
        \draw ($(ax)+(0.4,0)$) node [level,right] {$i$};
        \draw[->] (ax)to [bend right](-0.8,-0.8);
        \draw ($(ax)!0.8!(-0.8,-0.8)+(-0.7,0)$) node {$A$};
        \draw ($(ax)+(-0.5,0)$) node [level,left] {$i+p$};
      \end{scope}

    \draw (5,0) node (cut){$cut$};
    \draw (cut) ++ (-1.1,1) node [proofnet ] (G) {$G$};
    \draw (cut) ++ (1.1,1) node [proofnet] (H) {$H$};
    \draw[<-] (cut) to [bend left] (G.-40);
    \draw ($(cut)!0.5!(G.-40)+(-0.4,-0.2)$) node {$A$};
    \draw ($(cut)!0.8!(G.-40)+(0.2,0)$) node [level] {$i$};
    \draw[<-] (cut) to [bend right](H.-140);
    \draw ($(cut)!0.5!(H.-140)+(0.5,-0.2)$) node {$A^\perp$};
    \draw ($(cut)!0.8!(H.-140)+(-0.2,0)$) node [level] {$i$};
    \draw [->] (G.-150) --++ (0,-0.4);
    \draw [->] (H.-30) --++ (0,-0.4);

    \begin{scope}[shift={(-1,-2.2)}]
      \draw (1,0) node (tens){$\otimes$};
      \draw (tens) ++ (-0.8,1) node [proofnet, inner xsep=0.25cm ] (G) {$G$};
      \draw (tens) ++ (0.8,1) node [proofnet, inner xsep=0.25cm] (H) {$H$};
      \draw[<-] (tens) to [bend left] (G.-50);
      \draw ($(tens)!0.5!(G.-50)+(-0.5,0)$) node {$A$};
      \draw ($(tens)!0.8!(G.-50)+( 0.2,0)$) node [level] {$i$};
      \draw[<-] (tens) to [bend right](H.-130);
      \draw ($(tens)!0.5!(H.-130)+( 0.5,0)$) node {$B$};
      \draw ($(tens)!0.8!(H.-130)+(-0.2,0)$) node [level] {$i$};
      \draw[->] (tens) --++ (0,-0.6) node [level, left] {$i$} node [type] {$A \otimes B$};
      \draw [->] (G.-145) --++ (0,-0.8);
      \draw [->] (H.-35) --++ (0,-0.8);
      
      \draw (5.5,1) node [proofnet, inner xsep=0.7cm] (G){$G$};
      \draw (G) ++(-0.1,-1) node (par)  {$\parr$};
      \draw[<-] (par) to [bend left] (G.-145);
      \draw ($(par)!0.7!(G.-145)+(-0.4,0)$) node {$A$};
      \draw ($(par)!0.7!(G.-145)$) node {$i$};
      \draw[<-] (par)to [bend right](G.-45);
      \draw ($(par)!0.7!(G.-45)+(0.4,0)$) node {$B$};
      \draw ($(par)!0.7!(G.-45)$) node {$i$};
      \draw[->] (par) --++ (0,-0.6) node [midway, right] {$A \parr B$} node [midway, left] {$i$};
      \draw[->] (G.-165) --++ (0,-0.8);
      
      \draw (12,1) node [proofnet] (G) {$G$};
      \draw (G.-45)++(0,-0.8) node (ex)   {$\exists$};
      \draw [->](G.-135) --++(0,-0.8);
      \draw[<-] (ex) -- (G.-45) node [midway, right] {$A[B/X]$} node [level, left] {$i$};
      \draw[->] (ex) --++ (0,-0.6) node [midway, right] {$\exists X A$} node [level, left] {$i$};
      
      \draw (9,1) node [proofnet] (G) {$G$};
      \draw (G.-45)++(0,-0.8) node (forall)   {$\forall$};
      \draw[<-] (forall) -- (G.-45) node [midway, right] {$A\{Z/X\}$} node [level, left] {$i$};
      \draw[->] (forall) --++ (0,-0.6) node [midway, right] {$\forall X A$} node [level, left] {$i$};
      \draw[->] (G.-135) --++(0,-0.8);
    \end{scope}

    \begin{scope}[shift={(-1,-4.2)}]
      \draw (0,0) node [proofnet] (G) {$G$};
      \draw (G.-45)++(0,-0.8) node (forall)   {$?D$};
      \draw[<-] (forall) -- (G.-45) node [midway, right] {$A$} node [level, left] {$i$};
      \draw[->] (forall) --++ (0,-0.6) node [midway, right] {$? A$} node [level, left] {$i$-1};
      \draw[->] (G.-135) --++(0,-0.8);
    \end{scope}

    \begin{scope}[shift={(3,-4.2)}]  
      \draw (0,0) node [proofnet, inner xsep=0.5cm] (G) {$G$};
      \draw (G.-45)++(0,-1) node (forall)   {$?C$};
      \draw[<-] (forall) -- (G.-20) node [near start, right] {$?A$} node [level, near end, left] {$i$};
      \draw[<-] (forall) -- (G.-150) node [near end, right] {$?A$} node [level, near start, left] {$i$};
      \draw[->] (forall) --++ (0,-0.6) node [midway, right] {$?A$} node [level, midway, left] {$i$};
      \draw[->] (G.-160) --++(0,-0.8);
    \end{scope}
    
    \begin{scope}[shift={(6,-4.2)}]
      \draw (0,0) node [proofnet] (G) {$G$};
      \draw[->] (G.-90) --++(0,-0.8);
      \draw (G)++(1,0) node (weak) {$?W$};
      \draw[->] (weak) --++(0,-1)  node [midway, right] {$?A$} node [level, midway, left] {$i$};
    \end{scope}

    \begin{scope}[shift={(10,-4.2)}]    
      \draw (0,0) node [proofnet] (G) {$G$};
      \draw (G.-45)++(0,-0.8) node (forall)   {$\S$};
      \draw[<-] (forall) -- (G.-45) node [midway, right] {$A$} node [level, midway, left] {$i$};
      \draw[->] (forall) --++ (0,-0.6) node [midway, right] {$\S A$} node [level, midway, left] {$i$-1};
      \draw[->] (G.-135) --++(0,-0.8);
    \end{scope}

    \begin{scope}[shift={(1,-8)}]
      \draw (0,1) node [proofnet,inner xsep=1.8cm] (G) {$G$};
      \draw (G)++(1.9,-1.2) node [princdoor] (bang)   {};
      \draw (G)++(0.1,-1.2) node [auxdoor] (whyn) {};
      \draw (G)++(-1.9,-1.2) node [auxdoor] (whyn2) {};
      \draw[ar] (whyn2 |- G.south) -- (whyn2) node [type] {$A_1$} node [level] {$i$};
      \draw[ar] (whyn |- G.south) -- (whyn) node [type] {$A_k$} node [level] {$i$};
      \draw[ar] (bang |- G.south) -- (bang) node [type] {$C$} node [level] {$i$};
      \draw[ar] (whyn2) --++(0,-0.8) node [type] {$\wn A_1$} node [level] {$i$-1};
      \draw[ar] (whyn) --++(0,-0.8) node [type] {$\wn A_k$} node [level] {$i$-1};
      \draw[ar] (bang) --++(0,-0.8) node [type] {$\oc C$} node [level] {$i$-1};
      \draw (whyn)--(bang) -| ++(0.6,1.6) -| ($(whyn2)+(-0.6,0)$) -- (whyn2);
      \draw [dashed] (whyn2) -- (whyn);
    \end{scope}
      
    \begin{scope}[shift={(8,-8)}]
      \draw (-1,1) node [proofnet] (G) {$G$};
      \draw (G.-45)++(0,-0.8) node (forall)   {$?N$};
      \draw[<-] (forall) -- (G.-45) node [midway, right] {$??A$} node [level, left] {$i$};
      \draw[->] (forall) --++ (0,-0.6) node [midway, right] {$? A$} node [level, left] {$i$-1};
      \draw[->] (G.-135) --++(0,-0.8);
    \end{scope}

    \end{scope}
  \end{tikzpicture}
  \caption{ \label{rules_labeling} Construction of proof-nets. For the $\forall$ rule, we require $Z$ not to be free in the formulae labelling the other conclusions of $G$}
\end{figure}

\paragraph{Premises and conclusions}\label{def_premise} For any link $l$, the incoming edges of $l$ are called the {\em premises} of $l$.\label{def_conclusion} The outgoing edges of $l$ are called the {\em conclusions} of $l$. $prem(l)$ refers to the premise of $l$ and $concl(l)$ refers to the set of conclusions of link $l$. If $l$ has only one conclusion, we identify $concl(l)$ with its only element.\label{def_tail} The {\em tail} of edge $(l,m)$ refers to $l$,\label{def_head} while the head of $(l,m)$ refers to $m$.
Some edges have no conclusion.\label{def_pending} They are called the pending edges of the net, and by convention we say that their conclusion is $\bullet$.

\paragraph{Boxes}\label{def_box} \label{def_box}The rectangle in the $?P, !P$ rule is called a box. Formally a box is a subset of the links of the proof-net. We say that an edge $(l,m)$ belongs to box $B$ if $l$ is in $B$. We require the boxes to be arborescent: two boxes are either disjoint, or one contains the other.\label{def_depth} The number of boxes containing an element (box, link or edge) $x$ is its depth written $\partial(x)$.\label{def_maxdepth} $\partial_G$ is the maximum depth of an edge of $G$. The set of boxes of $G$ is $B_G$. 

Let us call $B$ the box in figure \ref{rules_labeling}.\label{def_principaldoor} The link labelled $!P$ is the principal door of $B$, its conclusion is written $\sigma(B)$.\label{def_auxiliarydoor} The $?P$ links are the auxiliary doors of box $B$.\label{def_sigmaib} The edge going out of the $i$-th auxiliary door (the count begins by $0$) is written $\sigma_i(B)$.\label{def_sigmab} The edge going out of the principal door is written $\sigma(B)$.\label{def_dgb} $D_G(B)$ is the set of doors of $B$ and \label{def_dg}$D_G=\max_{B \in B_G}|D_G(B)|$. The doors of box $B$ are not considered in box $B$, they are exactly the links which are not in $B$ but whose tails of their premises are in $B$.\label{def_rhoge} $\rho_G(e)$ is the deepest box of $G$ containing $e$. 

\paragraph{Quantifiers}\label{def_eigenvariable}We call {\em eigenvariables} of a proof-net, the variables $Z$ replaced by a $\forall$ link. We will always suppose that they are pairwise distinct. Any proof-net which does not respect this convention can be transformed in a proof-net with pairwise distinct eigenvariables by variables. This is possible because when we add a $\forall$ link to a proof-net, the eigenvariable can not be free in the other pending edges, so even if the eigenvariables are equal, they can not be related. This allows to refer to ``{\em the} link associated to the eigenvariable $Z$''.

\paragraph{Substitutions}\label{def_acrochettheta}Let $A$ be a formula and $\theta=\{(X_1,A_1); \cdots (X_n, A_n)\}$ be a substitution, then $A[\theta]$ refers to $A[A_1/X_1 | \cdots | A_n / X_n]$: the parallel substitution of the variables. Notice that we do not replace the occurences of variables which are bound. Similarly, if $Y \in FV(\theta(X))$ and $X$ appears in the scope of a $\exists Y$ or $\forall Y$, then we use $\alpha$-conversion to distinguish those $Y$s. 

\paragraph{Lists}In the following, we will define many concepts based on lists.\label{def_arobase} $@$ represents concatenation ($[a_1;…;a_n]@[b_1;…;b_k]$ is defined as $[a_1;…;a_n;b_1;…;b_k]$) and \label{def_insertion}$.$ represents ``push'' ($[a_1;…;a_n].b$ is defined as $[a_1;…;a_n;b]$).\label{def_listrestriction} $|[a_1;…;a_j]|$ refers to the length of the list: $j$. If $X$ is a set, $|[a_1; \cdots ; a_j]|_X$ is the number of indices $i$ such that $a_i$ is in $X$. 

\paragraph{Subtree}\label{def_subtree} Let $T,T'$ be two $I,L$-trees, we say that $T$ is a {\em subtree} of $T'$ (written $T \blacktriangleleft T'$) if $T$ is a subgraph of $T'$ such that if $v$ is a vertex of $T$ and $(v,v')$ is an edge of $T'$, then $v'$ is a vertex of $T'$ and $(v,v')$ is an edge of $T$. 

\paragraph{Other notations}\label{def_cardinal} Whenever $X$ is a set, $|X|$ is the cardinal of $X$. We will sometimes manipulate towers of exponentials,\label{def_exponentialtower} and will define $a^c_0$ as $c$ and $a^c_{b+1}$ as $a^{a^c_b}$.\label{def_setimage} Finally, if $f$ is a mapping from $E$ to $F$ and $A \subseteq E$, then $f(A)$ refers to $\Set{f(x)}{x \in E}$

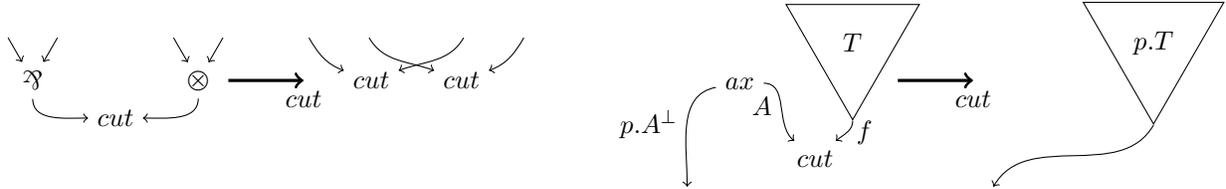
\begin{figure}[h]\centering
  \begin{tikzpicture}
    \begin{scope}[shift={(1,-4)}]
      \draw (0,0) node (cut) {$cut$};
      \draw (cut) ++ (-1.1,0.5) node (par) {$\parr$};
      \draw (cut) ++ ( 1.1,0.5) node (tens){$\otimes$};
      \draw (par) ++ (120:.8) node (restgg) {};
      \draw (par) ++ (60 :.8) node (restgd) {};
      \draw (tens) ++ (120:.8) node (restdg) {};
      \draw (tens) ++ (60 :.8) node (restdd) {};
      
      \draw [->] (tens) to [out=-90, in=0] (cut);
      \draw [->] (par) to [out=-90, in=180] (cut);
      \draw [<-] (tens) -- (restdg);
      \draw [<-] (tens) -- (restdd);
      \draw [<-] (par) -- (restgg);
      \draw [<-] (par) -- (restgd);
      
      \draw (restgg) ++ (4,0) node (restegg) {};
      \draw (restgd) ++ (4,0) node (restegd) {};
      \draw (restdg) ++ (4,0) node (restedg) {};
      \draw (restdd) ++ (4,0) node (restedd) {};

      \draw (par) ++ (4.5,0) node (cutg) {$cut$};
      \draw (tens) ++ (3.5,0) node (cutd) {$cut$};
      
      \draw [->] (restegg) to [out=-60, in=160] (cutg);
      \draw [<-] (cutg) to [out=20, in=-120] (restedg);
      \draw [->] (restegd) to [out=-60, in = 160] (cutd);
      \draw [<-] (cutd) to [out=20, in=-120] (restedd);
      
      \draw [->, very thick] (1.5,0.5) -- ++(1,0) node [below] {$cut$};
    \end{scope}
    
    \begin{scope}[shift={(8.8,-3)}]
      \draw (2,0) node [draw, regular polygon, regular polygon sides=3, shape border rotate=180] (arbre) {$~T~$};
      \draw (arbre.-90) ++ (-1.5,0.5) node (ax) {$ax$};
      \draw (ax) ++ (1,-1) node (cut) {$cut$};
      \draw (ax) ++ (-0.7,-1.5) node (rest) {};
      \draw [->, out=0 , in=140] (ax) to node [type, left] {$A$} (cut);
      \draw [->,out=-90, in= 40] (arbre.-90) to node [edgename, right] {$f$} (cut);
      \draw [->,out=-170,in=90] (ax) to  node [type, left] {$p.A^\perp$} (rest);
      
      \draw (arbre) ++ (4,0) node [draw, regular polygon, regular polygon sides=3, shape border rotate=180] (arbre) {$p.T$};
      \draw (rest) ++ (4,0) node (rest) {};
      \draw [->, out=-120, in=60] (arbre.-90) to (rest);
      \draw[->, very thick] (2.6,-0.5) -- ++(1,0) node [below] {$cut$};
    \end{scope}
  \end{tikzpicture}
  \caption{\label{mult_rules} Multiplicative rules for $cut$-elimination. Explanations for the axiom rule: for any edge $h$ such that there is a directed path from $e$ to $h$, we replace its formula $\beta (h)$ by $p.\beta (h)$  }
\end{figure}

\paragraph{$cut$-elimination}The terms of $\lambda$-calculus correspond, through the proofs-as-programs paradigm, to proof-nets. Intuitively, proof-nets are $\lambda$-terms which have been reversed, where applications and abstraction are respectively replaced by $\otimes$ and $\parr$ and with additionnal information on duplication. It can be observed in figure \ref{curry_howard_iso}. $cut$-elimination, is a a relation proof-nets which corresponds to $\beta$-reduction. The rules of $cut$-elimination can be found in figures \ref{mult_rules}, \ref{exp_rules} and \ref{quant_rules}. Proof-nets are stable under $cut$-elimination.

\begin{figure}[h]\centering
  \begin{tikzpicture}[scale=0.95]
    \tikzstyle{door}=[draw, circle, inner sep=0]
    
    \begin{scope}[shift={(0,0)}]
      \draw (0,0) node (bang) [princdoor] {};
      \draw (bang) ++ (-0.8,0) node  [auxdoor] (auxd) {};
      \draw (auxd) ++ (-1,0) node [auxdoor] (auxg) {};
      \draw (auxd) -- (bang) -|++ (0.3,0.7) -| ($(auxg)+(-0.3,0)$) |- (auxg);
      \draw [dashed] (auxg) -- (auxd);
      \draw (bang) ++ (0.5,-0.8) node (cut) {$cut$};
      \draw (bang) ++ (1,0) node (cont) {$?C$};
      \draw (cont) ++ (60:1) node (contd) {};
      \draw (cont) ++ (120:1)node (contg) {};
      \draw (auxd) ++ (0,-1) node (restd) {};
      \draw (auxg) ++ (0,-1) node (restg) {};
      \draw [->] (bang) to [out=-90, in=150] (cut);
      \draw [->] (cont) to [out=-90, in=30] (cut);
      \draw [<-] (cont) -- (contg);
      \draw [<-] (cont)--(contd);
      \draw [->] (auxg)--(restg);
      \draw [->] (auxd)--(restd);
      
      \draw (bang) ++ (-2,-3)  node (bang1) [door] {$!P$};
      \draw (bang1) ++ (-0.8,0) node  [door] (auxd1) {$?P$};
      \draw (auxd1) ++ (-1,0) node  [door] (auxg1) {$?P$};
      \draw (auxd1) -- (bang1) -|++ (0.3,0.7) -| ($(auxg1)+(-0.3,0)$) |- (auxg1);
      \draw [dashed] (auxg1) -- (auxd1);
      
      \draw (bang1) ++ (3,0)  node (bang2) [door] {$!P$};
      \draw (bang2) ++ (-0.8,0) node  [door] (auxd2) {$?P$};
      \draw (auxd2) ++ (-1,0) node  [door] (auxg2) {$?P$};
      \draw (auxd2) -- (bang2) -|++ (0.3,.7) -| ($(auxg2)+(-0.3,0)$) |- (auxg2);
      \draw [dashed] (auxg2) -- (auxd2);
      
      \draw (auxd1) ++ (-60:1.2) node (contd) {$?C$};
      \draw (auxg1) ++ (-60:1.2) node (contg) {$?C$};
      \draw (contd) ++ (0,-0.8) node (restd) {};
      \draw (contg) ++ (0,-0.8) node (restg) {};
      \draw [->] (auxd1) -- (contd);
      \draw [->] (auxg1) -- (contg);
      \draw [->] (contd) -- (restd);
      \draw [->] (contg) -- (restg);
      \draw [->] (auxd2) -- (contd);
      \draw [->] (auxg2) -- (contg);
  
      \draw (bang2) ++ (0.8,0.8) node (contg) {};
      \draw (bang2) ++ (1.5,0.8) node (contd) {};
      \draw (bang2) ++ (0.5,-0.8) node (cut2) {$cut$};
      \draw (bang1) ++ (1.5,-0.8) node (cut1) {$cut$};
      \draw [->] (bang1) to [out=-60, in=170] (cut1);
      \draw [<-] (cut1) to [out=10, in=-100] (contg);
      \draw [->] (bang2) to [out=-90, in=150] (cut2);
      \draw [<-] (cut2) to [out=30, in = -90] (contd);
      \draw [->, very thick] (-0.5,-1) -- ++(0,-0.6) node [right] {$cut$};
    \end{scope}
    
    \begin{scope}[shift={(5.5,0)}]
      \draw (0,0) node (bang) [princdoor] {};
      \draw (bang) ++ (-0.8,0) node  [auxdoor] (auxd) {};
      \draw (auxd) ++ (-1,0) node [auxdoor] (auxg) {};
      \draw (auxd) -- (bang) -|++ (0.3,0.8) -| ($(auxg)+(-0.3,0)$) |- (auxg);
      \draw [dashed] (auxg) -- (auxd);
      \draw (bang) ++ (0.5,-0.8) node [cut] (cut) {};
      \draw (bang) ++ (1,0) node [der] (der) {};
      \draw (der) ++ (0,1)node (dered) {};
      \draw (auxg) ++ (0,-1) node (restg) {};
      \draw (auxd) ++ (0,-1) node (restd) {};
      \draw [->] (bang) to [out=-90, in=150] (cut);
      \draw [->] (der) to [out=-90, in=30] (cut);
      \draw [<-] (der) -- (dered);
      \draw [->] (auxg)--(restg);
      \draw [->] (auxd)--(restd);
      \draw [<-] (bang) --++ (0,0.5) node (hbang) {};
      \draw [<-] (auxd) --++ (0,0.5) node (hauxd) {};
      \draw [<-] (auxg) --++ (0,0.5) node (hauxg) {};
      
      \draw (auxd) ++ (5,0) node [der] (derd) {};
      \draw (derd) ++ (-1,0) node [der] (derg) {};
      \draw (bang) ++ (5,0) node (bang) {};
      \draw (derd) ++ (0,-1) node (restd) {};
      \draw (derg) ++ (0,-1) node (restg) {};
      \draw [->] (derd)--(restd);
      \draw [->] (derg)--(restg);
      \draw [<-] (derd) --++ (0,0.5) node (hderd) {};
      \draw [<-] (derg) --++ (0,0.5) node (hderg) {};
      \draw (bang) ++ (0,0.6) node (hbang) {};
      \draw (hbang)++(0.5,-1) node (cut) {$cut$};
      \draw (hbang)++(1,0.4) node (rest) {};
      \draw [->] (hbang) to [out=-90, in=140] (cut);
      \draw [->] (rest) to [out=-90, in=50] (cut);
      \draw [->, very thick] ($(der)!0.2!(derg)$) -- ($(der)!0.7!(derg)$) node [below] {$cut$};
    \end{scope}

    \begin{scope}[shift={(5.5,-3)}]
      \draw (0,0) node (bang) [door] {$!P$};
      \draw (bang) ++ (-0.7,0) node  [auxdoor] (auxd) {};
      \draw (auxd) ++ (-1,0) node  [auxdoor] (auxg) {};
      \draw (auxd) -- (bang) -|++ (0.3,0.7) -| ($(auxg)+(-0.3,0)$) |- (auxg);
      \draw [dashed] (auxg)--(auxd);
      \draw (bang) ++ (0.5,-0.8) node (cut) {$cut$};
      \draw (bang) ++ (1,0) node (weak) {$?W$};
      \draw [ar] (auxd) -- ++ (0,-0.8) node (restd) {};
      \draw [ar] (auxg) -- ++ (0,-0.8) node (restg) {};
      \draw [->] (bang) to [out=-90, in=150] (cut);
      \draw [->] (weak) to [out=-90, in=30] (cut);
      
      \draw (auxd) ++ (4,0) node (weakg) {$?W$};
      \draw (weakg) ++ (1,0) node (weakd) {$?W$};
      \draw [ar] (weakg) -- ++ (0,-0.8);
      \draw [ar] (weakd) -- ++ (0,-0.8);
      \draw [->, very thick] ($(weak)!0.2!(weakg)$) -- ($(weak)!0.7!(weakg)$) node [below] {$cut$};
    \end{scope}

    \begin{scope}[shift={(-1,0)}]
      \draw (0,-6.5) node (cut) {$cut$};
      \draw (cut) ++ (-0.7,0.5) node (neutg) {$\S$};
      \draw (cut) ++ ( 0.7,0.5) node (neutd) {$\S$};
      \draw [<-] (neutg) --++ (0,0.5) node (restg) {};
      \draw [<-] (neutd) --++ (0,0.5) node (restd) {};
      \draw [->] (neutg) to [out=-90, in=180] (cut);
      \draw [->] (neutd) to [out=-90, in= 0] (cut);
      
      \nvar{\decNeutr}{3.5cm}
      \draw (restg) ++ (\decNeutr,0) node (restg) {};
      \draw (restd) ++ (\decNeutr,0) node (restd) {};
      \draw (cut) ++ (\decNeutr,0.5) node (cut) {$cut$};
      \draw [->] (restg) to [out=-90, in =180] (cut);
      \draw [->] (restd) to [out=-90, in= 0] (cut);
      \draw [->, very thick] ($(neutd)!0.2!(cut)$) -- ($(neutd)!0.7!(cut)$) node [below] {$cut$};
    \end{scope}
    
    \begin{scope}[shift={(0,-9)}]
      \draw (0,1) node [proofnet, inner xsep=1cm] (G) {$G$};
      \draw (G)++(-1 ,-1.2) node [auxdoor] (whyint1) {};
      \draw (G)++(0.1,-1.2) node [auxdoor] (whyint2) {};
      \draw (G)++(1,-1.2) node [princdoor] (bangint)   {};
      \draw[ar] (whyint1 |- G.south) -- (whyint1);
      \draw[ar] (whyint2 |- G.south) -- (whyint2);
      \draw[ar] (bangint |- G.south) -- (bangint);
      \draw (whyint2)--(bangint) -| ++(0.6,1.8) -| ($(whyint1)+(-0.6,0)$) -- (whyint1);
      \draw [dashed] (whyint1) -- (whyint2); 
      \draw [ar] (whyint1) -- ++ (0,-0.8);
      \draw [ar] (whyint2) -- ++ (0,-0.8);
      
      \draw (bangint) ++ (1.5,0.2) node [dig] (dig) {};
      \draw (bangint) ++(0.7,-0.5) node [cut] (cut) {};
      \draw [ar] (bangint) to [out=-90, in=180] (cut);
      \draw [ar] (dig) to [out=-90, in=0] (cut);
      \draw [revar] (dig) --++ (0,0.7);
      \draw [->, very thick] ($(dig)+(0.7,0)$)-- ++ (1,0) node [below] {$cut$};
    \end{scope}
    \begin{scope}[shift={(7,-9)}]
      \draw (0,1) node [proofnet, inner xsep=1cm] (G) {$G$};
      \draw (G)++(-1 ,-0.8) node [auxdoor] (whyint1) {};
      \draw (G)++(0.1,-0.8) node [auxdoor] (whyint2) {};
      \draw (G)++(1,-0.8) node [princdoor] (bangint)   {};
      \draw (bangint) ++ (0.5,0) node [below] {$B_1$};
      
      \draw[ar] (whyint1 |- G.south) -- (whyint1);
      \draw[ar] (whyint2 |- G.south) -- (whyint2);
      \draw[ar] (bangint |- G.south) -- (bangint);
      \draw (whyint2)--(bangint) -| ++(0.5,1.2) -| ($(whyint1)+(-0.5,0)$) -- (whyint1);
      \draw [dashed] (whyint1) -- (whyint2); 
      
      \draw (whyint1) ++(0,-0.7) node [auxdoor] (whyext1){};
      \draw (whyint2) ++(0,-0.7) node [auxdoor] (whyext2) {};
      \draw (bangint) ++(0,-0.7) node [princdoor] (bangext) {};
      \draw (bangext) ++ (0.5,0) [below] node {$B_2$};
      \draw[ar] (whyint1) -- (whyext1);
      \draw[ar] (whyint2) -- (whyext2);
      \draw[ar] (bangint) -- (bangext);
      \draw (whyext2)--(bangext) -| ++(0.7,2.1) -| ($(whyext1)+(-0.7,0)$) -- (whyext1);
      \draw [dashed] (whyint1)--(whyint2);
      \draw (bangext) ++ (1.5,0) coordinate (dig);
      \draw ($(bangext)!0.5!(dig)$) ++(0,-0.5) node [cut] (cut) {};
      \draw [ar] (bangext) to [out=-90, in=180] (cut);
      \draw [ar] (dig) to [out=-90, in=0] (cut);
      \draw (dig) --++ (0,0.7);
      \node [dig] (dig1) at ($(whyext1)+(0,-0.8)$) {};
      \node [dig] (dig2) at ($(whyext2)+(0,-0.8)$) {};
      \draw [ar] (whyext1) -- (dig1);
      \draw [ar] (whyext2) -- (dig2);
      \draw [ar] (dig1) -- ++ (0,-0.6);
      \draw [ar] (dig2) -- ++ (0,-0.6);
    \end{scope}
    
  \end{tikzpicture}
  \caption{\label{exp_rules} Exponential rules for $cut$-elimination}
\end{figure}
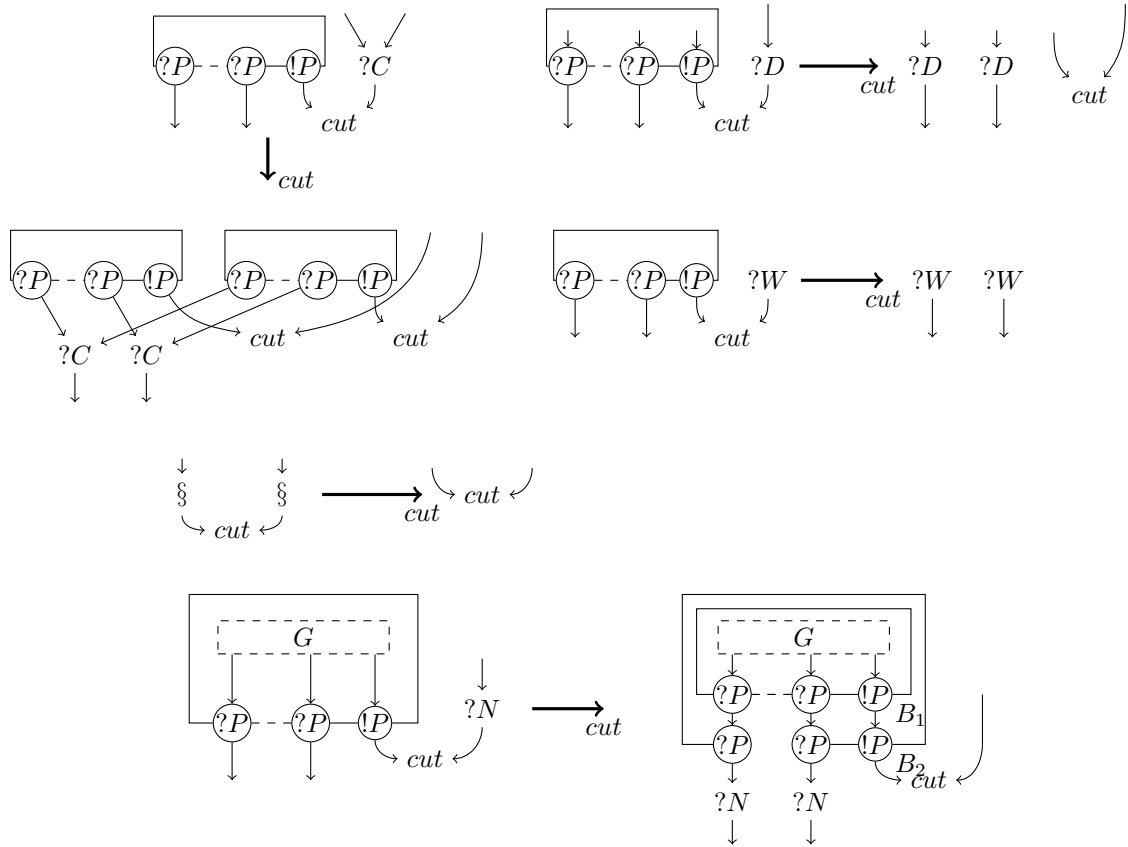

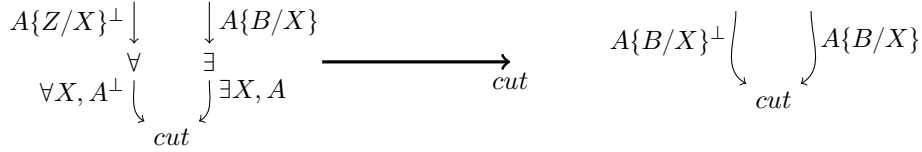
\begin{figure}[h]\centering
  \begin{tikzpicture}
    \begin{scope}
      \draw (0,0) node (cut) {$cut$};
      \draw (cut) ++ (-0.5,1) node (neutg) {$\forall$};
      \draw (cut) ++ ( 0.5,1) node (neutd) {$\exists$};
      \draw [<-] (neutg) --++ (0,0.8) node (restg) {} node [midway, left] {$A \{Z/X \}^\perp$};
      \draw [<-] (neutd) --++ (0,0.8) node (restd) {} node [midway, right] {$A \{B/X \}$};
      \draw [->] (neutg) to [out=-90, in=150] (cut);
      \draw ($(neutg)!0.4!(cut)+(-0.2,0)$) node [left] {$\forall X, A^\perp$}; 
      \draw [->] (neutd) to [out=-90, in= 30] (cut);
      \draw ($(neutd)!0.4!(cut)+(0.2,0)$) node [right] {$\exists X, A$};
      
      \draw (restg) ++ (8,0) node (restg) {};
      \draw (restd) ++ (8,0) node (restd) {};
      \draw (cut) ++ (8,0.5) node (cut) {$cut$};
      \draw [->] (restg) to [out=-90, in =150] (cut);
      \draw ($(restg)!0.4!(cut)+(-0.2,0)$) node [left] {$A\{B/X\}^\perp$};
      \draw ($(restd)!0.4!(cut)+( 0.2,0)$) node [right] {$A\{B/X\}$};
      \draw [->] (restd) to [out=-90, in= 30] (cut);
      
      \draw [very thick, ->] (2,1) --++(2.5,0) node [below] {$cut$};
    \end{scope}
  \end{tikzpicture}
  \caption{\label{quant_rules} Quantifier rule for $cut$-elimination. The substitution of the eigenvariable $Z$ by $B$ takes place on the whole net.}
\end{figure}

\begin{lemma}{\cite{baillot2010linear}}
  Proof-nets are stable under $cut$-elimination.
\end{lemma}
\begin{proof}
  The untyped version of the system corresponds exactly to the untyped version of $meLL$ which is stable under cut reduction~\cite{baillot2010linear}. An analysis of the rules shows that the types of the reduced nets are compatible with the rules. Most case are straightforward. The case of the axiom rule is not. 

  We have to verify that, given a proofnet $G$ an edge $f$ which has no head, if for all edges $h$ such that we have a directed path from $h$ to $f$ we replace the type $\beta(h)$ by $p.\beta(c)$, the resulting graph is still a proofnet. We can do so by induction on the proofnet. If the proofnet is an axiom, then the only edge whose type is replaced is $f$ itself and the axiom is still valid. Else, we can make a disjunction over the last rule used to build the proofnet, we apply the induction hypothesis to the subproofnet(s). The rule used on the original proofnet is still valid now, because each rule is stable by application of the $F \mapsto p.F$ transformation. More details on this rule can be found in~\cite{baillot2010linear}.
\end{proof}

\section{Context semantics}
\label{section_contextsemantics}
\paragraph{}The idea of geometry of interaction and context semantics is to study the reduction of proof-nets (or $\lambda$-terms) by leaving the proof-net unchanged and analysing instead some paths in it. It has been a key tool for the study of optimal reduction in $\lambda$-calculus~\cite{gonthier1992geometry}. Dal Lago~\cite{lago2006context} adapted context semantics so as to use it to prove quantitative properties on proof-net reduction, and applies it to light logics. From this point of view an advantage of context semantics compared to the syntactic study of reduction is its genericity: some common results can be proven for different variants of linear logic, which allows to factor out proofs of complexity results for these various systems.

\paragraph{} In the present work we wil show that context semantics is a powerful method because it allows to prove strong bounds on reduction which we do not know (yet) how to prove directly by syntactical means. The usual method to prove that a given relation strongly normalizes in a bounded number of steps is to define a weight for every object and to show that the weight decreases along any step of the relation. In the case of $cut$-elimination on linear logic proof-nets, we see that the contraction ($?C$ link) will be hard to handle. Indeed, we duplicate a whole box, which can contain an arbitrary proof-net. Whatever the weight associated to the proof-net inside the duplicated box is, it seems like we duplicate the weight when we duplicate the box. So, it seems hard to define any quantity which would decrease during $cut$-elimination.

\paragraph{}During a step of $cut$-elimination, most of the proof-net does not change. Most edges of the reduced net can be related with an edge from the original net. For example we want to say that $e'$ of figure \ref{example_trace_b} comes from $e$ of figure \ref{example_trace_a}. Similarly, $e1$ and $e2$ come from $e'$. We say that $e'$, $e1$ and $e2$ are reducts of $e$. The idea of context semantics is to anticipate $cut$-elimination and speak of the future links and edges. If we consider the reduct of an edge $e$ during a sequence of $cut$-elimination, we can observe that it forms a ``reducts tree'': each reduct has 2 sons in the next proof-net if it is in a box which is duplicated during the $cut$-elimination step, 0 son if it is in a box which is deleted, 1 son otherwise. What we want to capture is the set of leafs of this tree. Indeed, when a box is duplicated, it may increase the number of links and edges, but for every duplicated edge $e$, the leafs of the reducts tree of $e$ is splitted between the two immediate reducts of $e$. So the number of leafs of reducts trees does not increase. For the moment, we will call ``duplicates'' this informal notion of ``leafs of reducts trees''. The following definitions will help us capturing it by formal definitions (duplicates will correspond to ``$\mapsto$-canonical potential edges'').

\subsection{Definition of contexts}

\paragraph{}\label{def_expsignature} The language $Sig$ of {\it exponential signatures} is defined by induction by the following grammar:
\begin{equation*}
Sig = \sige \mid \sigl(Sig) \mid \sigr(Sig) \mid \sigp(Sig) \mid \sign(Sig,Sig)
\end{equation*}
An exponential signature corresponds to a list of choices of premises of $?C$ links. $\sigr(t)$ means: ``I choose the right premise, and in the next $?C$ links I will use $t$ to make my choices''. The construction $n(t,u)$ allows to encapsulate two sequels of choices into one. It corresponds to the digging rule ($??A \vdash B \rightsquigarrow ?A \vdash B$, represented by the $?N$ link in proof-nets) which ``encapsulates'' two $?$ modalities into one. The $p(t)$ construction is a degenerated case of the $n$ construction. Intuitively, $p(t)$ corresponds to $n(\varnothing, t)$. In this paper, we will use the symbols $t,u$ and $v$ to denote exponential signatures.

\paragraph{}\label{def_potential}A { \it potential} is a list of standard exponential signatures: an exponential signature corresponds to the duplication of one box, but a node is duplicated whenever any of the boxes containing it is cut with a $?C$ node. A potential is meant to represent all those possible duplications of a node, so there will be an exponential signature for each box containing it. The set of potentials is $Pot$. In this paper, we will use the symbols $P,Q,R$ for potentials. 

\paragraph{}\label{def_potentialbox} A  {\it potential box} is the couple of a box $B$ and a potential of length $\partial(B)$. For any box $B$, we denote the set of potential boxes whose first component is $B$ by $Pot(B)$.  We define similarly the notions of potential edges and the notation $Pot(e)$, potential links and the notaion $Pot(l)$, and so on. However, this notion does not capture the intuitive notion of ``future duplicate''. Indeed, $(B, [l(l(l(r(l(r(e))))))])$ is a potential box of figure \ref{path_example}, even if it does not correspond to a duplicate of $B$. In section \ref{timecomplexity}, we will define the notion of canonical potential, which fixes this mismatch. 

\paragraph{}\label{def_traceelement} A {trace element} is one of the following characters:  $\parr_l , \parr_r, \otimes_l, \otimes_r, \forall, \exists,\S,!_t,?_t$ with $t$ a potential. A trace element means ``I have crossed an element with this label, from that premise to its conclusion''. The set of trace elements is $TrEl$. ${(TrEl = \S \mid !_{Sig} \mid ?_{Sig} \mid \parr_l  \mid  \parr_r \mid  \otimes_l \mid  \otimes_r \mid  \forall \mid  \exists)}$

\paragraph{}\label{def_trace} A {\it trace } is a non-empty list of trace elements. The set of traces is $Tra$. A trace is a partial memory of the links we have crossed. Intuitively, it remembers the path crossed up to $cut$-eliminations. For example, we want the traces of the three paths of figure \ref{example_trace} to be the same.

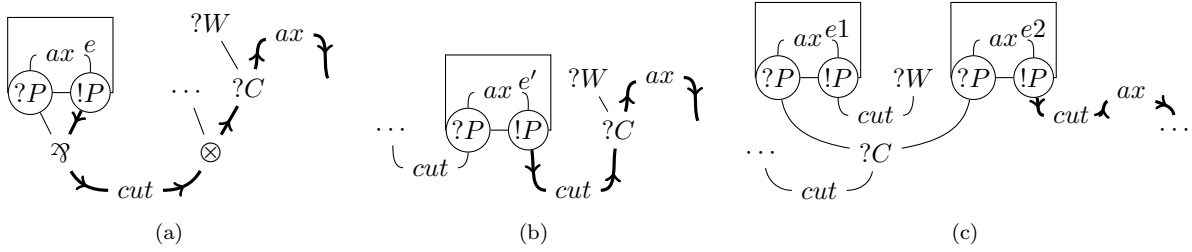
\begin{figure}\centering
  \nvar{\margedoor}{0.3cm}

  \subfigure[]{
    \begin{tikzpicture}
      \label{example_trace_a}
      \begin{scope}[decoration={ markings, mark=at position 0.5 with {\arrow{>}}}] 
        \draw (0,0) node [princdoor] (bang) {};
        \draw (bang) ++(-0.8,0) node [auxdoor] (whyn) {};
        \draw [box] (bang) -| ++(\margedoor, 1) -| ($(whyn)+(- \margedoor, 0)$) -- (whyn) -- (bang);
        \node [ax] (ax) at ($(bang)!0.5!(whyn) + (0,0.5)$) {};
        \draw [out=180, in=90] (ax) to (whyn);
        \draw [out=  0, in=90] (ax) to node [edgename, above] {$e$} (bang);
        \draw ($(bang)!0.5!(whyn)$) ++(0,-0.8) node [par] (par) {};
        \draw (whyn)--(par);
        \draw [very thick, postaction={decorate}] (bang)--(par);
        \draw (par) ++(2,0) node [tensor] (tens) {};
        \draw ($(par)!0.5!(tens)$)++(0,-0.5) node [cut] (cut) {};
        \draw [very thick, postaction={decorate}] (par) to [out=-60, in=180] (cut);
        \draw [very thick, postaction={decorate}] (cut) to [out=0, in=-120] (tens);
        \draw (tens)++(60: 1) node [cont] (cont) {};
        \draw (cont)++(120:1) node [weak] (rest) {};
        \draw (cont) -- (rest);
        \node (sortie) at ($(cont)+(1.1,0)$) {};
        \draw ($(cont)!0.5!(sortie)$)++(0,0.7) node [ax] (ax) {};
        \draw [very thick, postaction={decorate}] (tens) to [out=60, in=-120] (cont);
        \draw [very thick, postaction={decorate}] (cont) to [out=70, in=-170] (ax);
        \draw [very thick, postaction={decorate}] (ax) to [out=0, in=110] (sortie);
        \draw (tens)++(-0.3,0.8) node (inv) {$\cdots$};
        \draw (inv)--(tens);
      \end{scope}
    \end{tikzpicture}
  }
  \subfigure[]{
    \label{example_trace_b}
    \begin{tikzpicture}
      \begin{scope}[decoration={ markings, mark=at position 0.5 with {\arrow{>}}}] 
        \draw (0,0) node [princdoor] (bang) {};
        \draw (bang) ++(-0.8,0) node [auxdoor] (whyn) {};
        \draw [box] (bang) -| ++(\margedoor, 1) -| ($(whyn)+(- \margedoor, 0)$) -- (whyn) -- (bang);
        \node [ax] (ax) at ($(bang)!0.5!(whyn) + (0,0.5)$) {};
        \draw [out=180, in=90] (ax) to (whyn);
        \draw [out=  0, in=90] (ax) to node [edgename, above] {$e'$} (bang);
        \node [cont] (cont) at ($(bang)+(1.2,0)$) {};
        \node [cut] (cut) at ($(bang)!0.5!(cont)+(0,-0.8)$) {};
        \draw [very thick, postaction={decorate}, out=-80, in=180] (bang) to (cut);
        \draw [very thick, postaction={decorate}, out=0, in=-100] (cut) to (cont);
        \draw (cont)++(120:0.8) node [weak] (rest) {};
        \draw (cont) -- (rest);
        \node (sortie) at ($(cont)+(1.1,0)$) {};
        \node [ax] (ax) at ($(cont)!0.5!(sortie)+(0,0.7)$) {};
        \draw [very thick, postaction={decorate}] (cont) to [out=70, in=-170] (ax);
        \draw [very thick, postaction={decorate}] (ax) to [out=0, in=110] (sortie);
        \draw (whyn)++(-1,0) node (inv) {$\cdots$};
        \draw ($(whyn)!0.5!(inv)$) ++(0,-0.5) node [cut] (cutl) {};
        \draw (whyn) to [out=-90, in=0] (cutl);
        \draw (inv) to [out=-90, in=180] (cutl);
      \end{scope}
    \end{tikzpicture}
  }
  \subfigure[]{
    \label{example_trace_c}
    \begin{tikzpicture}
      \begin{scope}[decoration={ markings, mark=at position 0.5 with {\arrow{>}}}] 
        \draw (0,0) node [princdoor] (bang1) {};
        \draw (bang1) ++(-0.8,0) node [auxdoor] (whyn1) {};
        \draw [box] (bang1) -| ++(\margedoor, 1) -| ($(whyn1)+(- \margedoor, 0)$) -- (whyn1) -- (bang1);
        \node [ax] (ax) at ($(bang1)!0.5!(whyn1) + (0,0.5)$) {};
        \draw [out=180, in=90] (ax) to (whyn1);
        \draw [out=  0, in=90] (ax) to node [edgename, above] {$e1$} (bang1);
        \node [weak] (weak) at ($(bang1)+(1,0)$) {};
        \node [cut]  (cut1)  at ($(bang1)!0.5!(weak)+(0,-0.5)$) {};
        \draw (bang1) to [out=-90, in=180] (cut1);
        \draw (weak)  to [out=-90, in=  0] (cut1);

        \draw (2.6,0) node [princdoor] (bang2) {};
        \draw (bang2) ++(-0.8,0) node [auxdoor] (whyn2) {};
        \draw [box] (bang2) -| ++(\margedoor, 1) -| ($(whyn2)+(- \margedoor, 0)$) -- (whyn2) -- (bang2);
        \node [ax] (ax) at ($(bang2)!0.5!(whyn2) + (0,0.5)$) {};
        \draw [out=180, in=90] (ax) to (whyn2);
        \draw [out=  0, in=90] (ax) to node [edgename, above] {$e2$} (bang2);
        \node[ax] (ax2) at ($(bang2)+(1.3,-0.2)$) {};
        \node[cut] (cut2) at ($(bang2)+(0.5,-0.5)$) {};
        \draw [very thick, postaction={decorate}] (bang2) to [out= -80, in= 180] (cut2);
        \draw [very thick, postaction={decorate}] (cut2)  to [out=0   , in=-160] (ax2);
        \node [etc] (etc2) at ($(cut2)+ (1.4,-0.2)$) {};
        \draw [very thick, postaction={decorate}] (ax2) to [out= -20, in=100] (etc2);

        \node[cont] (cont) at ($(whyn1)!0.5!(whyn2)+(0,-1)$) {};
        \draw (whyn1) to [out= -80,in=170] (cont);
        \draw (whyn2) to [out=-100,in= 10] (cont);
        \node [etc] (rest) at ($(cont)+(-1.5,0)$) {};
        \node [cut] (cutb) at ($(cont)!0.5!(rest)+(0,-0.5)$) {};
        \draw (rest) to [out= -80, in=180] (cutb);
        \draw (cont) to [out=-100, in=0]   (cutb);
      \end{scope}        
    \end{tikzpicture}
  }
  \caption{\label{example_trace} The path of subfigures (a) and (b) (from the principal door to the pending edge of the proof net) should lead to the same trace} 
\end{figure}

\paragraph{}\label{def_polarity} A {\it polarity } is either + or -. It will tell us in which way are we crossing the edges. The set of polarities is $Pol$.\label{def_perp} We will use the notation $b^\perp$ for ``the polarity dual to $b$'' ($+^\perp=-$ and $-^\perp=+$). Similarly, $\parr_x^\perp = \otimes_x$, $\otimes_x^\perp = \parr_x$, $!_{i_1 \cdots i_n}^\perp = ?_{i_1 \cdots i_n}$,… We extend the notion of dual on traces by $([e_1; \cdots ; e_k])^\perp = [e_1^\perp; \cdots ; e_k^\perp]$. We also define, for any formula $A$, $A^b$ as $A^+=A$ and $A^-=A^\perp$.

\paragraph{}\label{def_context} A {\it context} is an element $(e,P,T,p)$ of ${C_G=E_G \times Pot \times Tra \times Pol}$. It can be seen as a token that will travel around the net. It is located on edge $e$ (more precisely its duplicate corresponding to $P$) with orientation $p$ and carries information $T$ about its past travel.

\paragraph{}\label{def_rightsquigarrow}\label{def_hookrightarrow} The links define two relations $\rightsquigarrow$ and $\hookrightarrow$ on contexts. The rules are presented in figures \ref{multiplicative_context_semantic} and \ref{exponential_context_semantic}. For any rule $(e,P,T,p) \rightsquigarrow (g,Q,U,q)$ the dual rule $(g,Q,U^\perp,q^\perp) \rightsquigarrow (e,P,T^\perp,p^\perp)$ holds as well. For example, the first $\parr$ rule also gives the following rule: $(c,P,T^\perp.\otimes_l,-)=(c,P,(T.\parr_l)^\perp, +^\perp) \rightsquigarrow (a,P,T^\perp,+^\perp)=(a,P,T^\perp,-)$. So for any $P,T$, $(c,P,T.\otimes_l,-) \rightsquigarrow (a,P,T.\otimes_l,-)$.\label{def_mapsto} $\mapsto$ is the union of $\rightsquigarrow$ and $\hookrightarrow$.

\nvar{\sepColonnes}{3cm}
\nvar{\sepRegles}{4.5cm}
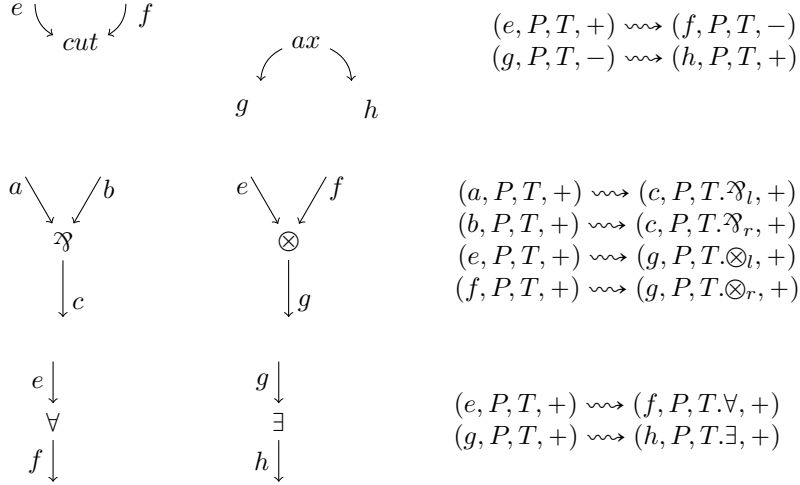
\begin{figure}\centering
  \begin{tabular}{cc}
    \begin{tikzpicture}
      \draw (0,0) node (n) {$cut$};
      \draw (n) ++ (45:0.9) node (r) {};
      \draw (n) ++ (135:0.9) node (l) {};
      \draw[<-] (n) to [bend left]  (l);
      \draw (l)  node [below left] {$e$};
      \draw[<-] (n) to [bend right] (r);
      \draw (r) node [below right] {$f$};
      
      \draw (\sepColonnes,0) node (n) {$ax$};
      \draw (n) ++ (-45:0.9) node (r) {};
      \draw (n) ++ (-135:0.9) node (l) {};
      \draw[->] (n) to [bend right]  (l);
      \draw (l) node [below left] {$g$};
      \draw[->] (n) to [bend left] (r);
      \draw (r) node [below right] {$h$};
      
      \draw (\sepColonnes + \sepRegles,0) node (rules) {
        $\begin{array}{c}
          (e,P,T,+) \rightsquigarrow (f,P,T,-)\\
          (g,P,T,-) \rightsquigarrow (h,P,T,+) 
        \end{array}$
      };
    \end{tikzpicture} 
    
    \\ \\
    
    \begin{tikzpicture}
      \draw (0,0) node (n) {$\parr$};
      \draw[<-] (n) --++(120:1) node [near end, left] {$a$};
      \draw[<-] (n) --++(60:1) node [near end, right] {$b$};
      \draw[->] (n) --++(-90:1) node [near end, right] {$c$};
      
      \draw (\sepColonnes,0) node (n) {$\otimes$};
      \draw[<-] (n) --++(120:1) node [near end, left] {$e$};
      \draw[<-] (n) --++(60:1) node [near end, right] {$f$};
      \draw[->] (n) --++(-90:1) node [near end, right] {$g$};
      
      \draw (\sepColonnes + \sepRegles,0) node (rules) {
        $\begin{array}{c}
          (a,P,T,+) \rightsquigarrow (c,P, T.\parr_l,+) \\
          (b,P,T,+) \rightsquigarrow (c,P, T.\parr_r,+) \\
          (e,P,T,+) \rightsquigarrow (g,P, T.\otimes_l,+) \\
          (f,P,T,+) \rightsquigarrow (g,P, T.\otimes_r,+) 
        \end{array}$
      };
    \end{tikzpicture}\\ \\
    
    \begin{tikzpicture}      
      \draw (0,0) node (forall) {$\forall$};
      \draw[<-] (forall) --++(0,0.8) node [midway, left] {$e$};
      \draw[->] (forall) --++(0,-0.8) node [midway, left] {$f$};
      
      \draw (\sepColonnes,0) node (ex) {$\exists$};
      \draw[<-] (ex) --++(0,0.8) node [midway, left] {$g$};
      \draw[->] (ex) --++(0,-0.8) node [midway, left] {$h$};
      \draw (\sepColonnes + \sepRegles,0) node (rules) {
        $\begin{array}{c}
          (e,P,T,+) \rightsquigarrow (f,P, T.\forall, +) \\
          (g,P,T,+) \rightsquigarrow (h,P, T.\exists, +)
        \end{array}$
      };
      
    \end{tikzpicture} \\
  \end{tabular}
  \caption{ \label{multiplicative_context_semantic}Multiplicative and quantifier rules of context semantics}
\end{figure}

\begin{figure}\centering
  \begin{tabular}{cc}
    \begin{tikzpicture}
      \node[neut] (n) at (0,0) {};
      \draw[<-] (n) --++(0,0.6) node [midway, left] {$h$};
      \draw[->] (n) --++(0,-0.6) node [midway, left] {$i$};
      
      \draw (\sepColonnes,0) node (n) {$?C$};
      \draw[<-] (n) --++(120:0.6) node [edgename] {$e$};
      \draw[<-] (n) --++(60:0.6) node [edgename, right] {$f$};
      \draw[->] (n) --++(-90:0.6) node [edgename] {$g$};
      
      \draw (\sepColonnes + \sepRegles,0) node (rules) {
        $\begin{array}{c}
          (e,P, T.?_t ,+) \rightsquigarrow (g,P, T.?_{\sigl(t)},+)\\
          (f,P, T.?_t ,+) \rightsquigarrow (g,P, T.?_{\sigr(t)},+)\\
          (h,P, T ,+) \rightsquigarrow (i,P,T. \S , +) 
        \end{array}$
      };
    \end{tikzpicture}
    \\ \\
    
    \begin{tikzpicture}
      \draw (0,0) node (n) {$?D$};
      \draw[<-] (n) --++(0,0.6) node [midway, left] {$e$};
      \draw[->] (n) --++(0,-0.6) node [midway, left] {$f$};
      
      \draw (\sepColonnes,0) node (n) {$?N$};
      \draw[<-] (n) --++(90:0.6) node [midway, left] {$g$};
      \draw[->] (n) --++(-90:0.6) node [midway, left] {$h$};

      \draw (\sepColonnes + \sepRegles,0) node (rules) {
        $\begin{array}{c}
          (e,P, T ,+) \rightsquigarrow (f,P, T.?_{\sige}, +) \\
          (g,P, T.?_{t_1}.?_{t_2} ,+) \rightsquigarrow (h,P, T.?_{\sign(t_1,t_2)} , +)\\ 
          (g,P, ?_t ,+) \rightsquigarrow (h,P, ?_{\sigp(t)} , +) 
        \end{array}$
      };
    \end{tikzpicture}
    \\ \\
    
    \begin{tikzpicture}
      \draw (0,0) node (n) {$?P$};
      \draw[<-] (n) --++(0,0.6) node [midway, left] {$e$};
      \draw[->] (n) --++(0,-0.6)  node [midway, left] {$f$};
      \draw (\sepColonnes,0) node  (m) {$!P$};
      \draw[<-] (m) --++(0,0.6) node [midway, left] {$g$};
      \draw[->] (m) --++(0,-0.6) node [midway, left] {$h$};
      \draw (n) -| ++(-.5,.7) (n)--(m) -| ++(.5,.7);
      
      \draw (\sepColonnes + \sepRegles,0) node (rules) {$\begin{array}{c}
          (e, P.t, T,+) \rightsquigarrow (f,P, T.?_t, +) \\
          (g, P.t, T,+) \rightsquigarrow (h,P, T.!_t, +) \\
          (f,P, !_t, -) \hookrightarrow (h, P, !_t, +) \\
        \end{array}$};
    \end{tikzpicture} \\
  \end{tabular}
  \caption{\label{exponential_context_semantic} Exponential rules of the context semantics}
\end{figure}

The behaviour induced by those rule can be understood by observing the path of figure \ref{path_example}. Crossing the $\parr$ downward, we have to keep as an information that we come from the right premise of the $\parr$, so that we know which premise of the $\otimes$  will be cut with $a$ after the reduction of the left cut. In context semantics, the information about the past is contained in the trace, so we put $\parr_r$ on the trace.

The jump from $d$ to $e$ can be surprising: indeed if we reduce the left cut, then the principal door of $B$ is cut with the auxiliary door of $B'$. And if we reduce this cut, the boxes $B$ and $B'$ fuse. It seems that box $B$ has disappeared and will never have any more duplicates. To understand why this jump is necessary, we can look at another sequence of $cut$-eliminations: if we reduce both cuts, the principal door of $B$ will be cut with a contraction link (which comes from $d$). To take all the duplicates of $B$ into account, we have to make this jump.

\begin{figure}\centering
  \begin{tikzpicture}
    \tikzstyle{porte}=[draw, circle, inner sep =0.2]
    \draw (0,0) node [princdoor] (b) {};
    \draw (b) --++ (0.7, 0) --++(0,1) -| ++(-1.4,-1) -- (b);
    \node (bName) at ($(b)+(0.9,0.4)$) {$B$};
    \draw (b) ++ (-0.7,-1) node [par] (par) {};
    \draw (par) ++ (-0.5,1) node (inc) {$\cdots$};
    \draw [->, very thick] (b)--(par) node [midway, below right] {a};
    \draw [->] (inc)--(par);
    \draw (par) ++ (1, -0.5) node (cut1) [cut] {};
    \draw (par) ++ (2, 0) node (tens) [tensor] {};
    \draw [->, very thick, out=-80, in=180] (par) to node [midway, below left] {b} (cut1) ;
    \draw [->, very thick, out=-100, in=0] (tens) to node [midway, below right] {c} (cut1);
    \draw (tens) ++ (120:0.7) node [etc] (inc) {} ;
    \draw (tens) ++ ( 0.7,1) node [auxdoor] (why) {};
    \draw [->] (inc)--(tens);
    \draw [->, very thick] (why)--(tens) node [midway, below right] {d};
    \draw (why) ++ (0.8,0) node [princdoor] (bang) {};
    \draw (why)--(bang) -| ++(0.4,0.8) -| ($(why)+(-0.4,0)$) -- (why);  
    \node (bp) at ($(bang)+(0.6,0.4)$) {$B'$};
    \draw (bang) ++ (1, -0.5) node (cut2) [cut] {};
    \draw (bang) ++ (2, 0) node (der) [cont] {};
    \draw [->, very thick, out=-80, in=180] (bang) to node [midway, below] {e} (cut2) ;
    \draw [->, very thick, out=-100, in=0] (der) to node [midway, below] {f}(cut2) ;
    \draw (der) ++ (120:0.8) node [der] (der1) {};
    \draw (der) ++ (60:0.8) node [der] (der2) {};
    \draw [<-] (der) -- (der1);
    \draw [<-, very thick] (der) -- (der2) node [edgename, right] {$g$};
    \draw (der1)++(0,0.7) node [etc] (etc1) {};
    \draw (der2)++(0,0.7) node [etc] (etc2) {};
    \draw [->] (etc1)--(der1);
    \draw [->] (etc2)--(der2);
    
    \begin{scope}[xshift=9cm]
      \draw (0,0) node [princdoor] (b) {};
      \draw (b) --++ (0.7, 0) --++(0,1) -| ++(-1.4,-1) -- (b);
      \draw (b) ++ (-0.5,-1) node [par] (par) {};
      \draw (par) ++ (-0.7,1) node (inc) {$\cdots$};
      \draw [->] (b)--(par);
      \draw [->] (inc)--(par);
      \draw (par) ++ (1, -0.5) node (cut1) [cut] {};
      \draw (par) ++ (2, 0) node (tens) [tensor] {};
      \draw [->, out=-80, in=180] (par) to (cut1) ;
      \draw [->, out=-100, in=0] (tens) to (cut1);
      \draw (tens) ++ (120:0.7) node [etc] (inc) {} ;
      \draw [->] (inc)--(tens);
      \draw (tens) ++ (60:0.7) node [cont] (cont) {};
      \draw [->] (cont)--(tens);
      \draw (cont)++(110:2) node [auxdoor] (why1) {};
      \draw (cont)++(25:1) node [auxdoor] (why2) {};

      \draw (why1) ++ (0.8,0) node [princdoor] (bang1) {};
      \draw (why1)--(bang1) -| ++(0.4,0.8) -| ($(why1)+(-0.4,0)$) -- (why1);  
      \draw (bang1) ++ (0.7, -0.4) node (cut2) [cut] {};
      \draw (bang1) ++ (1.4, 0) node (der1) [der] {};
      \draw [->, out=-80, in=180] (bang1) to (cut2) ;
      \draw [->, out=-100, in=0] (der1) to (cut2) ;
      \draw (der1)++(0,0.7) node [etc] (etc1) {};
      \draw [->] (etc1)--(der1);

      \draw (why2) ++ (0.8,0) node [princdoor] (bang2) {};
      \draw (why2)--(bang2) -| ++(0.4,0.8) -| ($(why2)+(-0.4,0)$) -- (why2);  
      \draw (bang2) ++ (0.7, -0.4) node (cut3) [cut] {};
      \draw (bang2) ++ (1.5, 1.2) node (der2) [der] {};
      \draw [->, out=-80, in=180] (bang2) to (cut3);
      \draw [->, out=-90, in=10] (der2) to (cut3);
      \draw (der2)++(0,0.7) node [etc] (etc2) {};
      \draw [->] (etc2)--(der2);

      \draw [->] (why1)--(cont);
      \draw [->] (why2)--(cont);
    \end{scope}

    \draw[->, very thick] (6,0)--(7,0);
  \end{tikzpicture}
  
  \caption{ \label{path_example} $(a, [~], !_{[\sigr(e)]}, +)  \rightsquigarrow (b, [~], \parr_r :: !_{[\sigr(e)]}, +)  \rightsquigarrow (c, [~], \parr_r :: !_{[\sigr(e)]} , -)  \rightsquigarrow (d, [~], !_{[\sigr(e)]}, -)  \hookrightarrow (e, [~], !_{[\sigr(e)]}, +)  \rightsquigarrow (f, [~], !_{[\sigr(e)]}, -) \rightsquigarrow (g, [~], !_{\sige}, -)$} 
\end{figure}

\subsection{Dealing with the digging}

\label{def_standard}A {\em standard} exponential signature is one that does not contain the constructor $\sigp$.\label{def_quasistandard} An exponential signature $t$ is {\em quasi-standard} iff for every subtree $\sign(t_1,t_2)$ of $t$, the exponential signature $t_2$ is standard. 

The $n( \_, \_)$ construction corresponds to the $?N$ $cut$-elimination of figure \ref{exp_rules}: we want to be able to speak about the duplicate of some edge in the duplicate $t_1$ of $B_1$, which is itself in the duplicate $t_2$ of $B_2$. But, if we are interested in the box $B_2$, then the $n( \_, \_)$ construction has no meaning, because we are not in any duplicate of $B_1$. When we are only interested in the duplicates of box $B_2$, we use the $p( \_)$ construction. 

These two notions are quite related: if a potential node is in the $\sign (t_1,t_2)$ duplicate of box $B$, then it is somehow contained in the $\sigp(t_2)$ duplicate of box $B$. The $\sqsubseteq$ relation defined below formalizes this relation.

We can notice that if $n(s_1,s_2)$ is a duplicate box $B$ then it means that $s_2$ is a valid duplicate of box $B_2$. So $p(s_2)$ is a valid duplicate of box $B$. More generally if $s$ is a duplicate of any box $B$, and $s \sqsubseteq s'$ then $s'$ is a duplicate of box $B$. So it is enough to describe the duplicates of a box which are minimal for $\sqsubseteq$ (which we call standard) to describe all its duplicates.

We can also notice that the $\sigp( \_)$ construction, has no meaning in potentials. If edge $e$ of $G$ is contained in box $B$, and we want to speak about a future duplicate of $e$ in the normal form of $G$. Not only do we need to know in which duplicates of box $B_2$ this duplicate of $e$ will be, but also in which duplicate of $B_1$. So, only standard signatures will be used in potentials. 

\paragraph{}\label{def_sqsubseteq}The binary relation $\sqsubseteq$ on $Sig$ is defined as follows:
\begin{align*}
  \sige      \sqsubseteq & \sige & &\\
  \sigl(t)   \sqsubseteq & \sigl(t')   &\Leftrightarrow&  t \sqsubseteq t' \\
  \sigr(t)   \sqsubseteq & \sigr(t')   &\Leftrightarrow&  t \sqsubseteq t' \\
  \sigp(t)   \sqsubseteq & \sigp(t')   &\Leftrightarrow&  t \sqsubseteq t' \\
  \sign(t_1,t_2)  \sqsubseteq & \sigp(t') &\Leftrightarrow&  t_2 \sqsubseteq t' \\
  \sign(t_1,t_2) \sqsubseteq & \sign(t'_1,t'_2) &\Leftrightarrow&  t_1 \sqsubseteq t'_1 \text{ and } t_2=t'_2 \\
\end{align*}
If $t \sqsubseteq t'$, then $t'$ is a simplification of $t$.\label{def_sqsubset} We also write $t \sqsubset t'$ for ``$t \sqsubseteq t'$ and $t \neq t'$''. We can observe that $\sqsubseteq$ is an order and $\sqsubset$ a strict order.

Our notation is reversed compared to Dal Lago's notation. Intuitively, $\sqsubseteq$ corresponds to an inclusion of future duplicates, but with the notation of~\cite{lago2006context}, $\sqsubseteq$ corresponds to $\supseteq$. We find this correspondence counter-intuitive, so we reversed the symbol. We found this change really important for the formalization of the ``Dal Lago's weight theorem'' where we manipulate the $\sqsubseteq$ relation really often. The correspondence is made precise through the $\Subset$ relation which encapsulates both $\subseteq$ and $\sqsubseteq$. We will later prove that the $\Subset$ relation is conserved during $cut$-elimination.

\paragraph{}\label{def_Subset}We define $\Subset$ on $Pot(B_G) \times Sig$ by
\begin{equation*}
  (B,P,t) \Subset (B',P',t') \Leftrightarrow  \left \{
  \begin{array}{l}
    B \subset B' \text{ and } P=P'.s'@Q \text{ with } s' \sqsubseteq t'\\
    \text{or} \\
    B=B' \text{, } P=P' \text{ and } t \sqsubset t'
  \end{array}
  \right .
\end{equation*}

\subsection{Context semantics and time complexity}
\label{timecomplexity}
\paragraph{}In this section we want to capture the potential boxes (resp. edges) which really correspond to duplicates of a box (resp. an edge). The definition can be difficult to understand. To give the reader an understanding of it, we will first try to capture intuitively this notion in the case of depth 0. We will do so by successive refinements.

\paragraph{}First, we could say that $(B,[t])$ corresponds to a duplicate iff $t$ corresponds to a sequel of choices of duplicates along a $cut$-elimination sequel, and the duplicate we chose eiter will not be part of a $cut$, or the cut will open it. The relation on contexts which corresponds to $cut$-elimination is $\mapsto$. So, we could make the first following attempt: ``$(B,[s])$ corresponds to a duplicate iff $(\sigma(B),[~],[!_s],+) \mapsto^* (e,P,T,-) \not \mapsto$. In figure \ref{example_trace}, $(B,[\sigl(\sige)])$ and $(B,[\sigr(\sige)])$ would be considered valid duplicates for the box, as they are expected to be.

\paragraph{}However, this definition would allow potential boxes which refuse choices: for example $(B,[\sige])$ in figure \ref{path_example} satisfies this definition, because $(\sigma(B),[~],[\oc_{\sige}],+) \mapsto^5 (f,[~],[\oc_{\sige}],-) \not \mapsto$. We would like the potentials of $B$ in this figure to be $(B,[\sigl(\sige)])$ and $(B,[\sigr(\sige)])$. $(B,[\sige])$ is just an intermediary step. Thus, our second try would be ``$(B,[t])$ corresponds to a duplicate iff $(\sigma(B),[~],[!_t],+) \mapsto (e_1, P_1, T_1, p_1) \cdots \mapsto (e_n,P_n,T_n,p_n) \not \mapsto$ and $\forall u, (\sigma(B),[~],[!_u],p) \mapsto (e_1, Q_1, U_1, q_1) \cdots \mapsto (e_n,Q_n,U_n,q_n) \Rightarrow (e_n,Q_n,U_n,q_n) \not \mapsto$''. 

Thus, if $(e_n,P_n,T_n,p_n)$ has no successor by $\mapsto$ because it faces a contraction or a digging while having an inappropriate exponential signature, then we could find another $u$ which would take the same path from $(e_1,p_1)$ to $(e_n,p_n)$ but with an appropriate exponential signature allowing it to continue the path. Thus, in such a situation, $t$ would not be considered to be a duplicate. On the contrary, if $(e_n,P_n,T_n,p_n)$ has no successor because it is a pending edge, or it arrives at a weakening, or it arrives at a dereliction with a trace equal to $[\oc_{\sige}]$, then changing the exponential signature does not allow us to continue the path, so $t$ would be considered to be a duplicate.

\begin{figure}\centering
  \begin{tikzpicture}
    \draw (0,0) node [princdoor] (bang) {};
    \draw (bang) ++ (0.4,-0.3) node {$B$};
    \draw (bang) -| ++ (0.7,0.9) -| ($(bang)+(-0.7,0)$)--(bang);
    \draw (bang)++ (2.5,0) node [dig] (dig) {};
    \draw ($(bang)!0.5!(dig) + (0,-0.7)$) node [cut] (cut) {};
    \draw [->, out=-90,in=180] (bang) to (cut);
    \draw [->, out=-90,in=0]   (dig)  to (cut);
    \draw (dig) ++(0,0.7) node [auxdoor] (aux) {};
    \draw (dig)--(aux);
    \draw (aux) ++ (0.8,0) node [princdoor] (princ) {};
    \draw (aux) ++ (0,0.6) node [der] (derint) {};
    \draw [<-] (aux)--(derint);
    \draw (derint) ++ (0.4,0.3) node [ax] (ax) {};
    \draw [out=180,in=90] (ax) to (derint);
    \draw [out=0,in=90] (ax) to (princ);
    \draw (aux)--(princ)-| ($(princ)+(0.4,1.2)$) -| ($(aux)+(-0.4,0)$);
    \draw (princ) ++(1,0) node [der] (derext) {};
    \draw ($(princ)!0.5!(derext)+(0,-0.6)$) node [cut] (cut2) {};
    \draw [out=-90,in=180] (princ) to (cut2);
    \draw [out=-90,in= 0] (derext) to (cut2);
    \draw (derext) ++ (0,0.5) node [etc] (etc) {};
    \draw (derext)--(etc);
    
    \draw [->,very thick] (5,0) -- (8,0);
    \begin{scope}[xshift=4cm]
      \draw (6,0) node [princdoor] (bang) {};
      \draw (bang) ++ (0.4,-0.3) node {$B_2$};
      \draw (bang) -| ++ (0.7,1.6) -| ($(bang)+(-0.7,0)$)--(bang);
      \draw (bang)++(0,0.7) node [princdoor] (bangint) {};
      \draw (bangint) -| ++(0.4,0.6) -| ($(bangint)+(-0.4,0)$)--(bangint);
      \draw (bangint) ++ (0.4,-0.3) node {$B_1$};
      \draw [->] (bangint) -- (bang);
      \draw (bang) ++ (1.5,0) node [auxdoor] (aux) {};
      \draw ($(bang)!0.5!(aux) + (0,-0.7)$) node [cut] (cut) {};
      \draw [->, out=-90,in=180] (bang) to (cut);
      \draw [->, out=-90,in=0]   (aux)  to (cut);
      \draw (aux) ++ (0.8,0) node [princdoor] (princ) {};
      \draw (aux) ++ (0,0.6) node [der] (derint) {};
      \draw [<-] (aux)--(derint);
      \draw (derint) ++ (0.4,0.3) node [ax] (ax) {};
      \draw [out=180,in=90] (ax) to (derint);
      \draw [out=0,in=90] (ax) to (princ);
      \draw (aux)--(princ)-| ($(princ)+(0.4,1.2)$) -| ($(aux)+(-0.4,0)$) -- (aux);
      \draw (princ) ++(1,0) node [der] (derext) {};
      \draw ($(princ)!0.5!(derext)+(0,-0.6)$) node [cut] (cut2) {};
      \draw [out=-90,in=180] (princ) to (cut2);
      \draw [out=-90,in= 0] (derext) to (cut2);
      \draw (derext) ++ (0,0.5) node [etc] (etc) {};
      \draw (derext)--(etc);
    \end{scope}
  \end{tikzpicture}
  \caption{\label{copies_simplification_condition}We do not want to consider $(B, \sign( \sige, \sigl( \sigr ( \sigl( \sige )))))$ as a valid duplicate of $B$}
\end{figure}

Similarly, we do not want the potential box to make choices on situations which will never happen. For example, we refuse $(B,[\sigr(\sigl(\sige))])$. So we add the condition ``$\forall u, (\sigma(B),[~],[!_u],p) \mapsto (e_1, Q_1, U_1, q_1) \cdots \mapsto (e_n,Q_n,U_n,q_n) \Rightarrow t \preccurlyeq u$.\label{def_preccurlyeqsig} With $t \preccurlyeq u$ defined inductively by : either $t=\sige$ or $t=x(t')$, $u=x(u')$ and $t' \preccurlyeq u'$ (with $x \in \{ \sigl, \sigr \}$)''.

\begin{definition}
  $t \preccurlyeq^v u$ if we are in one of the following cases:
  \begin{itemize}
    \item $t= \sige$ and $u = v$
    \item $t= \sigl (t')$, $u= \sigl (u')$ and $t' \preccurlyeq^v u'$
    \item $t= \sigr (t')$, $u= \sigr (u')$ and $t' \preccurlyeq^v u'$
    \item $t= \sigp (t')$, $u= \sigp (u')$ and $t' \preccurlyeq^v u'$
    \item $t= \sign (t',w)$, $u= \sign (u',w)$ and $t' \preccurlyeq^v u'$
  \end{itemize}
  We also write $t \preccurlyeq u$ for $\exists v, t \preccurlyeq^v u$
\end{definition}

This means that, every exponential signature leading to the same exact path would be longer than $t$. In our example, $(B,[\sigr(\sigl(\sige))])$ would not be a duplicate of $B$, because $\sigr(\sige)$ would lead to the same path as $\sigr(\sigl(\sige))$ and $\sigr(\sigl(\sige)) \not \preccurlyeq \sigr( \sige)$.

Finally, figure \ref{copies_simplification_condition} shows that those property must be true not only for $t$ but also for all its simplifications. Indeed, we do not want to consider $(B, \sign( \sige,\sigl( \sige )))$ to be a valid potential box because $\sigl( \sige ) $ do not correspond to a copy of $B_2$.

If the depth of box $B$ is $>0$, then $[t_1;\cdots;t_n]$ can not be a valid duplicate for $B$ if $[t_1;\cdots;t_{n-1}]$ is not a valid duplicate for the box containing immediately $B$. What we have described for boxes at depth 0 will give the notion of $\mapsto$-copy. A potential corresponds to a duplicate if all its signatures are copies, in this case it will be called a $\mapsto$-canonical potential. In fact, we will need restricted notions of duplicates which correspond to duplicates when we forbid the elimination of some cuts. In the same way that the $\mapsto$ relation corresponds to full $cut$-elimination, we will define restrictions of $\mapsto$ corresponding to our restrictions of $cut$-elimination. Those relations will satisfy some important properties which we regroup under the notion of ``cut simulation''.

\paragraph{}\label{def_cutsimulation}A relation $\rightarrow$ on contexts is a ``cut simulation'' if:
\begin{itemize}
\item $\rightarrow \subseteq \mapsto$
\item If $(e,P,T.t,-) \rightarrow (f,Q,T@U,q)$ and $|T|\geq 1$ then for every $T'$ with $|T'|>1$ and $P'$ with $|P'|=|P|$, $(e,P',T'.t,p) \rightarrow$
\item If $(e,P,T,+) \rightarrow (f,Q,T@U,q)$ and $|T|>1$ then for every $T'$ with $|T'|>1$ and $P'$ with $|P'|=|P|$, $(e,P',T',p) \rightarrow$
\end{itemize}

\paragraph{}We now consider a cut simulation $\rightarrow$. In this section, the only cut simulation we use is $\mapsto$ itself. However, in section \ref{section_stratification}, we will use other relations.

\paragraph{Copy context}\label{def_copycontext} A context $(e,P,[!_t]@T,p)$ is a $\rightarrow$-copy context if 
\begin{itemize}
  \item ~~ $(e,P,[!_t]@T,p) \rightarrow (e_1, P_1, T_1, p_1) \cdots \rightarrow (e_n,P_n,T_n,p_n) \not \rightarrow$
  \item If $(e,P,[!_u]@T,p) \rightarrow (e_1, Q_1, U_1, p_1) \cdots \rightarrow (e_n,Q_n,U_n,p_n)$, then $t \preccurlyeq u$ and $(e_n, Q_n, U_n, p_n) \not \rightarrow$
\end{itemize}

\begin{definition}[Copy]\label{def_copy}
  A $\rightarrow$-copy of a potential box $(B,P)$ is a standard exponential signature $t$ such that for all $t \sqsubseteq u$, $(\sigma(B),P,[!_u],+)$ is a $\rightarrow$-copy context.
\end{definition}

The set of $\rightarrow$-copies of $(B,P)$ is $C_\rightarrow(B,P)$. The set of simplifications of $\rightarrow$-copies of $(B,P)$ is $Si_{\rightarrow}(B,P)$. Intuitively, $Si_{\mapsto}(B,P)$ corresponds to the duplicates of $B$, knowing the duplicates of the outter boxes we are in.


\begin{definition}[Canonical potential]\label{def_canonicalpotential}Let $e \in B_{\partial(e)} \subset ... \subset B_1 $ , a {\it $\rightarrow$-canonical potential} for $e$ is a potential $[s_{1} ;  ... ; s_{\partial(e)}]$ such that $\forall i\leq \partial(e), s_i \in C_\rightarrow(B_i, [s_{1}; \cdots ; s_{i-1}])$.  The set of $\rightarrow$ canonical sequences for $e$ is $L_{\rightarrow}(e)$. We define $L_\rightarrow(B)=L_\rightarrow (\sigma(B))$. Let $x$ be an edge or a box, we will write $Can_\rightarrow(x)$ for $\{x \} \times L_{\rightarrow}(x)$ and $Can(x)$ for $Can_{\mapsto}(x)$.
\end{definition}

Intuitively, a $\mapsto$-canonical potential for $e$ is the choice, for all box $B_i$ containing $e$, of a copy of $B_i$. $L_{\mapsto}(e)$ corresponds to all the possible duplicates of $e$.\label{def_canonicalbox} If $P \in L_{\rightarrow}(B)$, we say that $(B,P)$ is a {\em canonical box}.\label{def_canonicaledge} We define similarly the notion of {\em canonical edge} and {\em canonical link}.

\begin{definition}[Canonical context]\label{def_canonicalcontext}A context $(e,P,[!_t]@T,p) \in C_G$ is said {$\rightarrow$-\it canonical} if:
  \begin{itemize}
  \item $P \in L_{\rightarrow}(e)$
  \item $t$ is quasi standard and for every $t \sqsubseteq v$, $(e,P,[!_v]@T,p)$ is a $\rightarrow$-copy context
  \item For all cutting of $T$ of the shape $T= U.!_u@V$, $u$ is standard and there exists $v \preccurlyeq u$ such that for every $v \sqsubseteq w$, $(e,P,[!_w]@V~,p~)$ is a $\rightarrow$-copy context
  \item For all cutting of $T$ of the shape $T= U.?_u@V$, $u$ is standard and there exists $v \preccurlyeq u$ such that for every $v \sqsubseteq w$, $(e,P,[!_w]@V^\perp,p^\perp)$ is a $\rightarrow$-copy context
  \item There exists $(B,Q) \in Pot(B_G)$ and $v \in Sig$ such that $(\sigma(B),Q,[\oc_v],+) \rightarrow^* (e,P,[\oc_t]@T,p)$
  \end{itemize}
\end{definition}
  The only paths that interest us in this paper are those which correspond to copies or simplifications of copies of potential boxes. The contexts have special properties, that we need in some lemmas. In particular, the first trace elements of the traces are $\oc_t$ trace elements. If $P \in L_\rightarrow(B)$ and $t \in Si_\rightarrow(B,P)$, then the context $(\sigma(B),P,[\oc_t],+)$ is $\rightarrow$-canonical. Lemma \ref{canonical_context_stable} will show that any context of the path beginning by $(\sigma(B),P,[\oc_t],+)$ is $\rightarrow$-canonical.

\begin{lemma} \label{copy_context_stable}
  If $C \rightarrow D$ then $C$ is a $\rightarrow$-copy context if and only if $D$ is a $\rightarrow$-copy context
\end{lemma}
\begin{proof}
  Straightforward, by definition of $\rightarrow$-context
\end{proof}
\begin{lemma} \label{canonical_context_stable}
  If $C$ is $\rightarrow$-canonical and $C \rightarrow D$ then $D$ is $\rightarrow$-canonical.
\end{lemma}
\begin{proof} 
  The proof is done by a long but straightforward case by case analysis. Here, we consider only one of the cases. We suppose that $C=(e,P.t',[!_t]@T,+) \rightarrow (f,P,[!_t]@T.?_{t'},+)=D$, crossing an auxiliary door of a box $B$ downwards. $P.t'$ is $\rightarrow$-canonical for $e$ so by definition $P$ is $\rightarrow$-canonical for $\sigma (B)$. $\sigma (B)$ and $f$ belong to the same boxes so $P$ is $\rightarrow$-canonical for $f$. $C$ being canonical, $t$ is quasi-standard, and for every $t \sqsubseteq u$, $(e,P.t',[!_u]@T,+)$ is a $\rightarrow$-copy context. Moreover, by definition of cut simulations and by considering that $(e,P.t',[!_t]@T,+) \rightarrow (f,P,[!_t]@T.?_{t'},+)$, we get that $(e,P.t',[!_u]@T,+) \rightarrow (f,P,[!_u]@T.?_{t'},+)$. So $(f,P,[!_u]@T.?_{t'},+)$ is a $\rightarrow$-copy context (lemma \ref{copy_context_stable}). 

Let us suppose that $T.?_{t'}=U.?_u@V$, we have to prove that there exists $v \preccurlyeq u$ such that for all $v \sqsubseteq w$, $(f,P,[!_w]@V^\perp,-)$ is a $\rightarrow$-copy context. The proof depends whether $V=[~]$ or not. If $V$ is not empty, $T=U.?_u@W$ for some $W$, so there exists $v \preccurlyeq u$ such that for every $v \sqsubseteq w$, $(e,P.t',[!_w]@W^\perp,-)$ is a $\rightarrow$-copy context. By lemma \ref{copy_context_stable}, we prove that $(f,P,[!_w]@V^\perp,-)$ is a $\rightarrow$-copy context. Else (if $V=[~]$), we know that $P.t'$ is $\rightarrow$-canonical for $e$, so $t'$ is a $\rightarrow$-copy of $(B,P)$. If $(f,P,!_{t'},-) \rightarrow (\sigma(B),P,!_{t'},+)$, we can take $v=t'$ and prove the needed properties with lemma \ref{copy_context_stable}. Else, we take $v=\sige$. The case $T.?_{t'}=U.!_u@V$ is treated in the same way but is simpler because $V \neq [~]$.
\end{proof}

We define a weight $T_G$. In the following, we will prove that $T_G$ is an upper bound on the length of the longest cut-reductions sequence. This result was already proved in~\cite{lago2006context} for a really similar framework. The proof of Dal Lago was quite convincing for people having a good intuition on proof-net. however it was not really formal. Here, we take a complementary approach: we propose a proof which is quite formal, with many details provided. The downside us that our proof is much more difficult to read than the one in~\cite{lago2006context}.

\begin{definition}[weights]\label{def_tg}For every proof-net G, we define 
\begin{equation*}
T_G = \sum_{e \in E_G} | L_{\mapsto}(e)| + 2. \sum_{B \in B_G} \left( D_G(B) \sum_{P \in L_{\mapsto}(B)} \sum_{t \in Si_{\mapsto}(B,P)} |t| \right ) 
\end{equation*}
\end{definition}

\label{def_acyclic}Let $G$ be a proof net. We say that $G$ is cyclic if there exists $e \in E_G$, $P \in L_G(e)$, $p \in Pol$, $s,s' \in Sig$ such that $(e,P,[\oc_s],p)$ is a $\mapsto$-copy context and $(e,P,[!_s],p) \mapsto^* (e,P,[!_{s'}],p)$. Otherwise, $G$ is acyclic.

Intuitively, a cyclic proof structure may not normalize. In term of $T_G$: suppose that the cycle consumes a part of the exponential signatures and that a path corresponding to a copy $t$ of a potential box $(B,P)$ passes through this cycle. Then we can insert the ``consumed part'' as many times as we want in $t$, creating an infinity of copies for $B$, so $T_G$ would be infinite. 

\label{def_positiveweights}We say that $G$ has positive weights if for every $(B,P) \in Pot(B_G)$, $1 \leq |C_{\mapsto}(B,P) < \infty$.

In the rest of this section, we will prove that our $\mapsto$-copies correspond to the copies of Dal Lago \cite{lago2006context}, all proof-nets are acyclic, have positive weights and that if $G \rightarrow_{cut} H$ then $T_G > T_H$. The readers who are not interested in those proofs may safely skip those parts and jump towards Section \ref{section_stratification}.

\subsection{Copymorphisms: motivations and definitions}
\paragraph{}The idea of context semantics is to anticipate the $cut$-elimination. We said $(e,P)$ and $(e,Q)$ represent two future duplicates of the edge $e$. However, we have not made explicit the link between the potential edges and the duplicates of the edge. We would like to explicit, for all $G \rightarrow_{cut} H$ {\em a mapping $\phi$ from potential edges of $G$ to the potential edges of $H$.}

\paragraph{}Furthermore, we would like to use this mapping to prove that $T_H < T_G$ so this mapping must relate the copies and canonical potential of $G$ and $H$. The copies are defined by paths, {\ so we would like $\phi$ to keep $\mapsto$ intact}.

\paragraph{}Looking at the rules of $\rightarrow_{cut}$ we see that it will be impossible to have a complete mapping which ``keeps $\mapsto$ intact''. For example, in the $\parr / \otimes$ case, what could be the images of the edges cut? So $\phi$ will be a partial mapping. The important paths are the paths corresponding to copies, such paths begin at principal doors, so we would like all $(\sigma(B),P)$ (with $B\in B_G$) to be in the domain of $\phi$. Conversely, we would like all the $(\sigma(B'),P')$ (with $B' \in B_H$) to be images of some $(\sigma(B),P)$ (with $B \in B_G$). Notice that it is not possible to ask the images of $(\sigma(B),P)$ ($B\in B_G$) to be principal doors of $H$ because of the dereliction rule, cf figure \ref{der_copymorphism}

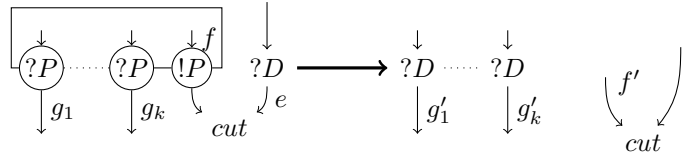
\begin{figure}
  \centering
  \begin{tikzpicture}
    \begin{scope}[shift={(5.5,0)}]
      \draw (0,0) node (bang) [princdoor] {};
      \draw (bang) ++ (-0.8,0) node [auxdoor] (auxd) {};
      \draw (auxd) ++ (-1.2,0) node [auxdoor] (auxg) {};
      \draw (auxd) -- (bang) -|++ (0.4,0.8) -| ($(auxg)+(-0.4,0)$) |- (auxg);
      \draw [dotted] (auxg)--(auxd);
      \draw (bang) ++ (0.5,-0.8) node [cut] (cut) {};
      \draw (bang) ++ (1,0) node [der] (der) {};
      \draw (der) ++ (0,1) node (dered) {};
      \draw (auxd) ++ (0,-1) node (restd) {};
      \draw (auxg) ++ (0,-1) node (restg) {};
      \draw [ar] (bang) to [out=-90, in=150] (cut);
      \draw [ar, out=-90, in=30] (der) to node [midway, right] {$e$} (cut) ;
      \draw [revar] (der) -- (dered);
      \draw [ar] (auxd)--(restd) node [midway, right] {$g_k$};
      \draw [ar] (auxg)--(restg) node [midway, right] {$g_1$};
      \draw [revar] (bang) --++ (0,0.5) node (hbang) {} node [midway, right] {$f$};
      \draw [revar] (auxd) --++ (0,0.5) node (hauxd) {};
      \draw [revar] (auxg) --++ (0,0.5) node (hauxg) {};
      
      \draw (auxd) ++ (5,0) node [der] (derd) {};
      \draw (auxg) ++ (5,0) node [der] (derg) {};
      \draw [dotted] (derg)--(derd);
      \draw (bang) ++ (5,0) node (nbang) {};
      \draw (nbang)++(0.5,0) node (hbang) {};
      \draw (derd) ++ (0,-1) node (restd) {};
      \draw (derg) ++ (0,-1) node (restg) {};
      \draw [ar] (derd)--(restd) node [midway, right] {$g'_k$};
      \draw [ar] (derg)--(restg) node [midway, right] {$g'_1$};
      \draw [revar] (derd) --++ (0,0.5) node (hderd) {};
      \draw [revar] (derg) --++ (0,0.5) node (hderg) {};
      \draw (hbang)++(0.5,-1) node (cut') {$cut$};
      \draw (hbang)++(1,0.4) node (rest) {};
      \draw [ar, out=-90, in=140] (hbang) to node [midway, above right] {$f'$} (cut');
      \draw [ar] (rest) to [out=-90, in=50] (cut');
      \draw [->, very thick] ($(der)!0.2!(derg)$) -- ($(der)!0.8!(derg)$);
    \end{scope}
  \end{tikzpicture}
  \caption{\label{der_copymorphism} The dereliction elimination rules motivates the introduction of a second morphism to capture the relation between a proof net and its reduct}
\end{figure}

\paragraph{} We will prove later that paths corresponding to $\mapsto$-copies end by ``final contexts'' as defined below. Here, we do not need the result but just the intuition to guide our definition. Notice that, in the definition of Dal Lago, the copies are defined by paths finishing by final-contexts. The Dal Lago's definition is a theorem in our framework. This change of definition was necessary because, in the proofs of our main results, we will need to consider copies for relations other than $\mapsto$. 

\begin{definition}\label{def_finalcontext}
  A context $((l,m),P,[!_{\sige}]@T,p)$ is said final if either:
  \begin{itemize}
  \item $p=-$ and $\alpha(l)=?D$ and $T=[~]$
  \item $p=-$ and $\alpha(l)=?W$
  \item $p=+$ and $\alpha(m)=\bullet$
  \end{itemize}
  The set of final contexts of $G$ is written $F_G$.  
\end{definition}

\paragraph{}We would like paths finishing by final contexts to be preserved by $\phi$. So we would like conclusions of $?D$ and $?W$ links, and pending edges of $E_G$ to be in the domain of $\phi$ (for all potentials) and such edges of $E_H$ to be in the codomain of $\phi$ (for all potentials). However, when we analyse the $?D$ and $?W$ cut elimination rules, we see that it will not be possible. In figure \ref{der_copymorphism}, edge $e$ is intuitively transformed in many edges: $\{g'_1, \cdots, g'_k\}$ and if we defined any of those to be the image of $e$, the $\mapsto^*$ relation would not be kept intact (for example $(f,[\sige],T,+) \mapsto^* (e,[~],T .!_{\sige},-)$ but we can not have a similar relation between $f'$ and $g'_i$). Does this mean that we can not explicit a correspondence between the copies of $G$ and $H$? In fact we just have to be more subtle. The paths we are interested in are those who begin by a principal door and finish by contexts $(e,P,[!_{\sige}],-)$. Such a context can not come from the interior of box $B$ (it would have a longer trace) so either it comes from $\sigma(B)$ (in this case, the box being deleted we will deal with this copy in another way) or from a $\hookrightarrow$ rule. In the latter case, we see that the corresponding path would be stopped by a $g'_i$. There is some correspondence between $e$ and the $g'_i$, even if it can not be captured by $\phi$. We will need another mapping $\psi$ to capture it.

\paragraph{}We gather all the properties written above into the concept of ``copymorphism'':

\begin{definition}[Copymorphism]\label{def_copymorphism}
  Let $G$ and $H$ be two proof-nets, a copymorphism from $G$ to $H$ is a tuple $(D_{\phi},D_{\phi}', \phi, \psi)$ with:

  \begin{enumerate}
  \item $D_{\phi} \subseteq E_G$, $D_{\phi}' \subseteq E_H$
  \item $\phi$ a partial surjective mapping from $Pot(D_{\phi}) \times Sig$ to $Pot(D_{\phi}') \times Sig$ such that:
    \begin{enumerate}
    \item \label{copymorphism_path_cons} The paths are the same in $G$ and $H$: 
   \begin{equation*}
     \text{If } \left .
     \begin{array}{c}
       \phi(e,P,t)=(e',P',t')\\ 
       \phi(f,Q,u)=(f',Q',u')
     \end{array}
     \right \} \text{, then } \forall T,U \in Tra, \forall p,q \in Pol
     \begin{array}{c}
       (e,P,[!_t]@T,p) \mapsto^* (f,Q,[!_u]@U,q)\\
       \Leftrightarrow \\
       (e',P',[!_{t'}]@T,p) \mapsto^* (f',Q',[!_{u'}]@U,q)
     \end{array}
   \end{equation*}

 \item \label{copymorphism_subset_cons} $\phi$ preserves the inclusion of boxes: 
   \begin{equation*}
     \text{If } \left .
     \begin{array}{c}
       \phi(\sigma(B),P,t)=(\sigma(B'),P',t')\\ 
       \phi(\sigma(C),Q,u)=(\sigma(C'),Q',u')
     \end{array}
     \right \} \text{, then }
     \begin{array}{c}
       (B, P, t) \Subset (C,Q, u)  \\
       \Leftrightarrow  \\
       (B',P',t') \Subset  (C',Q',u')
     \end{array}
   \end{equation*}
 \item \label{copymorphism_sqsubset} If $\phi(e,P,t)$ is defined and $t \sqsubseteq u$, then $\phi(e,P,u)$ is defined.
 \item \label{copymorphism_boxwise} $\phi$ is boxwise:
   \begin{equation*}
     \left .
     \begin{array}{c}
       B \text{ deepest box containing } e \text{ with } \sigma(B)\in D_{\phi} \\
       \phi(\sigma(B),P,t)=(\sigma(B'),P',t') \\
       \phi(e,P.t@Q,u)=(e',Q',u') \\
       \phi(e,P.t@R,u)=(f',R',v')
     \end{array}
     \right \} \Rightarrow 
     \left \{
     \begin{array}{c}
       e'=f'\\
       u'=v' \\
       \rho(e'=f')=B'  \\
       Q'=R'=P'.t'       
     \end{array}
     \right .
   \end{equation*}
   In the limit case where there is no box $B$ such that $\sigma(B) \in D_{\phi}$ and $e \in B$, we mean that $e'$ is at depth 0.
 \item \label{copymorphism_princ_defined}Images of principal doors are defined: let $(\sigma(B),P) \in Pot(D_{\phi})$ and $t \in C_{\mapsto}(B,P)$.
 \item \label{copymorphism_princ_doors} All principal doors of $H$ are in the image, by $\phi$, of the principal doors of $G$:
   \begin{equation*}
     \Set*{\phi(\sigma(B),P,s) }{ (\sigma(B),P) \in Pot(D_{\phi}),  s\in Sig } = \Set{ (\sigma(B'),P',s') }{ (B',P') \in Pot(B_H), s'\in Sig } 
   \end{equation*}

 \item \label{copymorphism_final_defined}Images of final contexts are defined: 
   \begin{itemize}
   \item Let $(e,P,[\oc_t]@T,p) \in F_G$ with $(e,P) \in Can(D_\phi)$, then $\phi(e,P,t)$ is defined.
   \item Let $(f,Q,[\oc_u]@U,q) = \psi(e,P,[\oc_t]@T,p)$, then $\phi(f,Q,\oc_u)$ is defined.
   \end{itemize}
   
 \item \label{copymorphism_final_cons}Images of final contexts are final
   \begin{equation*}
     \begin{array}{c}
       \Set*{ (e',P',[!_{t'}]@T,b) \in C_H }{ \exists(e,P,[!_t]@T,b) \in F_G  \text{ and }  \phi(e,P,t)=(e',P',t')  } \\
       \cup \\
       \Set*{ (e',P',[!_{t'}]@T,b) \in C_H}{ \exists C \in C_G, \exists (e,P,[!_t]@T,b) \in \psi(C)  \text{ and }  \phi(e,P,t)=(e',P',t') }\\
     \end{array}
     = F_H
   \end{equation*}
 \item \label{copymorphism_phi_division} If $\phi(e,P,t)=(e',P',t')$ and $\phi(f,Q,u)=(e',P',u')$ then $(e,P)=(f,Q)$
 \item \label{copymorphism_phi_copy_not_empty} For all $B \in B_G$, if $\sigma(B) \not \in D_\phi$ then $\forall P \in L_G(B), C_G(B,P) \neq \varnothing$ 
 \end{enumerate}
\item \label{copymorphism_psi} $\psi$ is a partial mapping from  $D_{\psi}= F_G \cap \left( {\mathbf (E_G - D_{\phi})}  \times Pot \times Tra \times Pol \right) $  to $\mathcal{P}(D_{\phi}' \times Pot \times Tra \times Pol)$ such that $\forall (e,P,T,b) \in D_{\psi}$,there exists $k \in \mathbb{N}$ such that: 
  \begin{enumerate}
    \item \label{copymorphism_psi_j} $\forall j \leq k, \left ( (f,Q,U,c) \mapsto^j (e,P,T,b) \right ) \Rightarrow f \not \in (\sigma(B_G) \cap D_\phi)$ 
    \item \label{copymorphism_psi_k} $\left( (f,Q,U,c) \mapsto^k (e,P,T,b) \right ) \Leftrightarrow (f,Q,U,c) \in \psi(e,P,T,b)$
  \end{enumerate}
\end{enumerate}
\end{definition} 

In order to use this definition of copymorphism, we will have to prove that the paths corresponding to a copy (or simplification of a copy) indeed have a final context as last context. This is the goal of the next subsection.

\subsection{Underlying formula of a context}
\label{subsection_underlying}
\subsubsection{Intuition of the proof}
\paragraph{} The idea of this subsection is to prove that, if a context $C$ comes from an edge $e$ with an empty trace, then $C$ still has information about the formula $\beta(e)$. We can define an underlying formula for the contexts, and this underlying formula will be stable by $\rightsquigarrow$ in typed proof-nets. A consequence from this, is that if we begin a path by a context with a ``right trace'', then this path will never be blocked by a mismatch between the trace and the type of link above our edge (the tail of our edge). For example we will never have $(e,P,T.\forall,-)$ with the tail of $e$ being a $\otimes$ link. However, it does not prevent blocking situations such as $(e,P,T.!_{\sige},-)$ with the tail of $e$ being a $?C$ link.

\paragraph{} Our first idea to define the underlying formula of $(e,P,T,+)$ would be to use $T$ to prune the syntactic tree of $\beta(e)$. For example the underlying formula of $(e,P, [!_{\sigl(\sige)};\parr_r;\forall],+)$ with $\beta(e)= \forall X.  ?(X \otimes X^\perp) \parr !(X^\perp \parr X)$ would be $X^\perp \parr X$. Noticing that crossing $cut$ and $axiom$ does not change the trace but transforms the labelling formula into its dual, we would define the underlying formula of $(e,P,T,-)$ as $\beta(e)^\perp$ pruned using $T$.

\paragraph{} However, there is a problem with this definition when we cross a $\exists$ link downwards. For example if $(c,[~],[\parr_r;!_{\sige};\otimes_r],+) \mapsto (d,[~],[\parr_r!_{\sige};\otimes_r;\exists],+)$ with $\beta(c)= ?(X \otimes X^\perp) \otimes !(X^\perp \parr X)$ and $\beta(d)=\exists Y.Y \otimes Y^\perp$ (figure \ref{intuition_underlying_formula}). The underlying formula of $(c,[~],[\parr_r;!_{\sige};\otimes_r],+)$ is $X$, but $(d,[~],[\parr_r;!_{\sige};\otimes_r;\exists],+)$ has no underlying formula: the trace is not compatible with the syntactic tree of $\beta(d)$. A solution would be to index the $\exists$ with the formula they capture. During the paths, this index could be transferred to eigenvariables and $\forall$ links as in figure \ref{intuition_underlying_formula}.
\begin{figure}
  \centering
  \begin{tikzpicture}
    \node [exists] (exists) at (0,0) {};
    \node [forall] (forall) at ($(exists)+(6,0)$) {};
    \node [cut] (cut) at ($(exists)!0.5!(forall)+(0,-0.8)$) {};
    \draw [ar, out=-110, in=  0]   (forall) to node [edgename, above] {$e$} (cut);
    \draw [ar, out= -70, in=180] (exists) to node [edgename, above] {$d$} node [type, below] {$\exists Y. Y^\perp \otimes Y$} (cut);
    \node [tensor] (tensor) at  ($(exists)+(0,0.8)$) {};
    \node [par]    (par)    at  ($(forall)+(0,0.8)$) {};
    \draw [ar] (tensor)--(exists) node [edgename] {$c$} node [type] {$\wn(X \otimes X^\perp) \otimes \oc(X^\perp \parr X)$};
    \draw [ar] (par)   --(forall) node [edgename] {$f$};
    \node [ax] (ax) at ($(par)+(0,1)$) {};
    \draw [ar, out=-20, in=60] (ax) to node [edgename] {$g$} (par);
    \draw [ar, out=-160,in=120](ax) to node [edgename] {$h$} (par);
    \node [princdoor] (bang) at ($(tensor)+(60:1)$) {};
    \draw [ar] (bang)--(tensor) node [type] {$b$};
    \node [etc] (etc) at ($(tensor)+(120:0.8)$) {};
    \draw [ar] (etc)--(tensor);
    \draw (bang)-| ++(1.5,2)-| ($(bang)+(-0.5,0)$)--(bang);
    \node [par] (pari) at ($(bang)+(0,1.1)$) {};
    \draw [ar] (pari) -- (bang) node [edgename] {$a$} node [type] {$X^\perp \parr X$};
    \node [ax] (axi) at ($(pari)+(0,0.7)$) {};
    \draw [out=-20, in=60, ar] (axi) to (pari);
    \draw [out=-160,in=120,ar] (axi) to (pari);
  \end{tikzpicture}
\caption{\label{intuition_underlying_formula} To keep the same underlying formula along the path, we would like to extend the contexts in the following way: $(a,[\sige],[\parr_r],\varnothing, +) \mapsto^2 (c,[~],[\parr_r;\oc_{\sige};\otimes_r],\varnothing, +) \mapsto (d,[~],[\parr_r;!_{\sige};\otimes_r;\exists_{!(Y^\perp \parr Y)}],\varnothing,+) \mapsto (e,[~],[\parr_r;!_{\sige};\otimes_r;\exists_{!(Y^\perp \parr Y)}],\varnothing,-) \mapsto (f,[~], [\parr_r;!_{\sige};\otimes_r], \{Z \mapsto !(Y^\perp \parr Y)\},-) \mapsto^3 (f,[~],[\parr_r;!_{\sige};\parr_l], \{Z \mapsto !(Y^\perp \parr Y) \},+) \mapsto (e,[~],[\parr_r;!_{\sige};\parr_l;\forall_{Y^\perp \parr Y}], \varnothing,+) \mapsto (d,[~],[\parr_r;!_{\sige};\parr_l;\forall_{Y^\perp \parr Y}], \varnothing,-) \mapsto (c,[~],[\parr_r;!_{\sige};\parr_l], \varnothing,-)$}
\end{figure}

\paragraph{}In figure \ref{intuition_underlying_formula}, we extend traces, and we insert substitutions in the contexts to deal with eigenvariables. To define the underlying formula of $(e,P,T,\theta,+)$ we could use $T$ to prune $\beta(e)[\theta]$ and instantiate its ground formulae bound by quantifier connectives. Then, we could hope that the underlying formula would not change along a $\mapsto$ path (as in figure \ref{intuition_underlying_formula}). However, with multiple crossing of $\exists$ and $\forall$ links, the situation becomes a bit more complicated, substitutions will not be precise enough to describe the situations.

\begin{wrapfigure}{l}{1.5cm}
  \centering
  \begin{tikzpicture}
    \node[exists] (ex) at (0,0) {};
    \draw[->] (ex)--++ (0,-1) node [edgename] {$f$} node [type] {$\exists Y. Y$};
    \draw[<-] (ex)--++ (0, 1) node [edgename] {$e$} node [type] {$Z \otimes A$};
  \end{tikzpicture}
\end{wrapfigure}

\paragraph{}Indeed, let us suppose that $(e,P,T,\{Z \mapsto (Z_1^\perp \parr Z_1)\},+) \mapsto (f,P,T.\exists_B,\{Z \mapsto (Z_1^\perp \parr Z_1)\},+)$, crossing a $\exists$ link whose associated formula ($Z \otimes A$, for example)contains all the free $Z$ of $\beta(e)$. How will we define $B$? If we define the list of substitutions of eigenvariables as $\{Z \mapsto Z_1^\perp \parr Z_1, Y \mapsto Z \otimes A \}$, we will have to maintain a tree of dependencies. If we define it as $\{Y \mapsto (Z_1^\perp \parr Z_1) \otimes A \}$ it will be harder to prove, when the path will come back to cross upwards the $\exists$ link, that the formula associated to the $\forall$ in the trace corresponds to the formula associated to the $\exists$ link in the proof net. Our solution is to show that the trees of dependencies  can be considered as {\em truncations} of a {\em ztree}\footnote{The ztrees are trees dealing with eigenvariables, most often named $Z$ in the litterature, hence the ``z'' of ``ztree''} which only depend on the edge and potential of the context (not on the trace). What we mean by ``truncation'' is defined in the following paragraph.

\paragraph{Trees truncations}\label{def_truncation} Let $I$ and $L$ be two sets, we can define the set of $I,L$-trees as the set of directed connex acyclic graphs such that the leafs (nodes of outward arity 0) are labelled by an element of $L$ and the internal nodes (nodes which are not leafs) are labelled by an element of $I$. If $L$ is reduced to one element or $L=I$, we can define a truncation relation on trees: $T \vartriangleleft U$. If $L=I$, it means there is an injection $\phi$ from the nodes of $T$ to the nodes of $U$ such that 
\begin{itemize}
\item The ``son/father'' relation is preserved by $\phi$
\item For all nodes $n$ of $T$, the outward arity of $\phi(n)$ is either 0 or the outward arity of $t$
\item For all nodes $n$, the labels of $t$ and $\phi(t)$ are equal.
\end{itemize}
In the case where $L$ is reduced to one element, we do not require the labels to be equal when $n$ is an internal node and $\phi(n)$ a leaf.

\subsubsection{Substitution trees}
\label{subsection_ztree}

\paragraph{Substitution trees}\label{def_substitutiontree} are finite sets of finitely branching, but potentially infinite, trees whose internal nodes are labelled by a substitution on a unique variable, and whose leafs are labelled by the void function. We define an alternative notation for substitution trees by coinduction. Let $T$ be a substitution tree, by definition it is a set of finitely many trees : $\{T_1, \cdots , T_k \}$. Each $T_i$ admits an unique node $t_i$ with no incoming edge, if $t_i$ is not a leaf it is labelled by $\{X_i \mapsto B_i \}$. Let $\underline{T_1},\cdots ,\underline{T_k}$ be the notations for $T_1,\cdots,T_k$, we define the notation $\underline{T}$ for $T$ by $\underline{T}= \{(X_1,B_1,\underline{T_1}) \cdots (X_k,B_k,\underline{T_k})\}$, the trees reduced to a unique node being omitted.

\begin{definition}[ztree]\label{def_ztree}
  Let $(e,P)$ be a potential edge of $G$, $P=[p_1;\cdots ; p_{\partial(e)}]$ and $Z_1 \cdots Z_k$ the free eigenvariables of $\beta(e)$. The complete substitution of $(e,P)$ (written $ztree(e,P)$) is defined coinductively by the following. Let $f_i$ be the conclusion of the $\forall$ link associated to $Z_i$. We defined the set $E_{(e,P)}$ by :
  \begin{equation*}
    E_{(e,P)} = \Set*{ Z_i }{\exists (g_i,R_i) \in Pot(E_G),\begin{array}{c} (f_i,[p_1; \cdots ; p_{\partial(f_i)}],\forall,+) \mapsto^* (g_i,R_i,\forall,-) \\ \text{ and }g_i\text{ is the conclusion of a }\exists\text{ link }l_i \end{array}}
    \end{equation*}
  Then $ztree(e,P)$ is defined by 
  \begin{equation*}
    ztree(e,P)= \Set*{(Z_i,B_i,ztree(h_i,R_i)) }{ Z_i \in E_{(e,P)} }
  \end{equation*}
  with $B_i$ the formula captured by $l_i$ and $h_i$ its premise.
\end{definition}

\paragraph{}This definition is quite formal. The idea is the following: if the eigenvariable $Z$ is part of the formula $A$ labelling edge $e$ and $l$ is the $\forall$ link associated with $Z$. Then, when $l$ is cut with a $\exists$ with associated formula $B$, then $Z$ will be replaced by $B$ in the formula labelling $e$. It may happen that $l$ is duplicated during the $cut$ elimination process, so $Z$ would be replaced by different formulae in different duplicates of $e$. However, if we fix a potential $P$, then there is at most one $\exists$ link $(m,Q)$ so that $(l,P)$ will be cut with $(m,Q)$ during cut elimination. We can compute this $(m,Q)$ (if there is one) using the paths of context semantics. There may be other eigenvariables inside $B$ (the formula associated to the $\exists$ link), but in the same way, we can compute by which formulae those eigenvariables will be replaced and so on. In the end we have a big tree which says $Y$ will be replaced by $B$, $Z$ will be replaced by $C$, but inside $B$ the eigenvariable $W$ will be replaced by the formula $D$, and inside $D$... A node of this tree may have no child (it corresponds to an edge without eigenvariables in its formula, or such that none of its eigenvariables will be replaced during $cut$-elimination). But, for the moment, we have no assurance that the tree is finite. If $ztree(e,P)$ is finite, it defines a substitution $\sigma$ such that $\beta(e)[\sigma]$ is exactly the label of $(e,P)$ in the normal form of the proof-net. Figure \ref{example_ztree} illustrates the definition $ztree$.

 \begin{figure}
  \centering
  \begin{tikzpicture}
    \node[par] (par) at (0,0) {};
    \nvar{\sep}{1.5cm}
    \node[tensor](tensor) at ($(par)+(\sep,0)$) {};
    \node[ax](axl1) at ($(par)+(0,1)$) {};
    \node[ax](axl2) at ($(tensor)+(0,1)$) {};
    \draw[ar, out=-160,in=120] (axl1) to node [type,left] {$A$} (par);
    \draw[ar, out=-20,in=120] (axl1) to (tensor);
    \draw[ar, out=-160,in= 60] (axl2) to (par);
    \draw[ar, out= -20,in= 60] (axl2) to node [type] {$A$} (tensor);
    \node[exists](expar) at ($(par)+(0,-1)$) {};
    \node[exists](extens)at ($(tensor)+(0,-1)$) {};
    \draw[ar] (par)--(expar) node [type,left] {$A \parr A^\perp$} node [midway, right] {$h$};
    \draw[ar] (tensor)--(extens) node [type] {$A^\perp \otimes A$}node [midway, left] {$g$};
    \node[forall](fatens) at ($(extens)+(3,0)$) {};
    \node[forall](fapar)  at ($(fatens)+(\sep,0)$) {};
    \node[cut] (cuttens) at ($(extens)!0.5!(fatens)+(0,-0.6)$) {};
    \node[cut] (cutpar) at ($(expar)!0.5!(fapar)+(0,-1.5)$) {};
    \draw[ar, out=-70,in=180,looseness=0.7] (extens) to node [type,below=0.22] {$\exists X. X \otimes X^\perp$} (cuttens);
    \draw[ar, out=-90,in=180] (expar) to  node [type,below left] {$\exists X. X^\perp \parr X$} (cutpar);
    \draw[ar, out=-110,in=  0,looseness=0.7] (fatens) to (cuttens);
    \draw[ar, out=-90,in=  0] (fapar)  to (cutpar);
    \node[exists] (exh) at ($(fapar)+(2,0)$) {};
    \node[par] (par) at ($(exh)+(0,0.9)$) {};
    \node[ax] (axm1) at ($(par)!0.5!(fatens)+(0,1.6)$) {};
    \node[ax] (axm2) at ($(par)!0.5!(fapar)+(0,1.3)$) {};
    \draw[ar, out=-170,in=90] (axm1) to node [type,left] {$Z_1^\perp \parr Z_1$}(fatens);
    \draw[ar, out= -20,in=120] (axm1) to (par);
    \draw[ar, out=-160,in=90] (axm2) to node [type,left]{$Z_2 \otimes Z_2^\perp$}(fapar);
    \draw[ar, out= 0,in=60] (axm2) to (par);
    \draw[ar] (par)--(exh) node [type] {$(Z_1 \otimes Z_1^\perp)\parr(Z_2^\perp \parr Z_2)$} node [midway, left] {$f$};
    \node[exists] (exb) at ($(exh)+(0,-1)$) {};
    \draw [ar] (exh) -- (exb) node [type] {$\exists Y.Y$};
    \node[forall] (fah) at ($(exh)+ (5,0)$) {};
    \node[forall] (fab) at ($(fah)+(0,-1)$) {};
    \node[cut] (cut) at ($(exb)!0.5!(fab)+(0,-0.6)$) {};
    \draw[ar, out=-60,in=180, looseness=0.5] (exb) to node [type,below] {$\exists X.\exists Y.Y$}(cut);
    \draw[ar, out=-120,in=  0,looseness=0.5] (fab) to node [type,below] {$\forall X.\forall Y.Y^\perp$} (cut);
    \draw[ar] (fah)--(fab) node [type] {$\forall Y.Y^\perp$};
    \node[etc] (etc) at ($(fah)+(0,1)$) {};
    \draw[ar] (etc)--(fah) node [type] {$Z^\perp$} node [left, midway] {\Large $e$};
  \end{tikzpicture}
  \caption[hey]{\label{example_ztree}$\begin{array}{rl}
     ztree(e,[~]) & = \left \{ \left(Z, (Z_1 \otimes Z_1^\perp) \parr (Z_2^\perp \parr Z_2), ztree(f,[~]) \right) \right \} \\
     & = \left \{ \left(Z, (Z_1 \otimes Z_1^\perp) \parr (Z_2^\perp \parr Z_2), \{(Z_1, A^\perp, ztree(g,[~])), (Z_2, A^\perp, ztree(h,[~])) \} \right) \right \} \\
     & = \left \{ \left(Z, (Z_1 \otimes Z_1^\perp) \parr (Z_2^\perp \parr Z_2), \{(Z_1, A^\perp, \varnothing), (Z_2, A^\perp, \varnothing) \} \right) \right \} 
    \end{array}$}
\end{figure}

\paragraph{} Let $\Theta$ be a finite substitution tree, the substitution induced by $\Theta$ is defined inductively by: If $\Theta=\{(Z_1,B_1,\Theta_1) , \cdots , (Z_k,B_k,\Theta_k)\}$ then the substitution induced by $\Theta$ is $\theta_\Theta = \left \{Z_1 \mapsto B_1[\theta_{\Theta_1}] , \cdots , Z_k \mapsto B_k [\theta_{\Theta_k}] \right \}$.

\paragraph{}We will first define a notion of underlying formulae on objects a bit different from contexts, because it is easier to make an induction definition on these. The definition of the underlying formulae of a context will rely on this first definition. You can notice that we define underlying formula{\textbf e}, it is plural because we consider all the formulae that may replace $\beta(e)$ (with $(e,P)$ the potential edge the context comes from) during $cut$-elimination, for all possible reduction strategies. This exhaustivity is achieved by the union on all truncations of $ztree$s in the $\exists$ and $\forall$ part of the definition. Later, in Subsection \ref{subsection_ztreefinite}, we willl prove that those underlying formulae ``converge'' towards a unique formula, which we will call {\em the} underlying formula of the context.

\label{def_betasetshit}Let $(e,P) \in Pot(E_G)$, $A \in \mathcal{F}_{LL_0}$ and $T,T'$ be two lists of trace elements (traces or empty list) and $p$ a polarity, $\beta_{\{\}}(A,e,P,T,T',p)$ is defined by induction on $|T|$ by:
\begin{itemize}
\item If $T= [~]$, then $\beta_{\{\}}(A,e,P,T,T',p)= \{ A \}$
\item If $A = B \otimes C$ and $T=U.\otimes_l$, then $\beta_{\{\}}(A,e,P,T,T',p)=\beta_{\{\}}(B,e,P,U,[\otimes_l]@T',p)$
\item If $A = B \otimes C$ and $T=U.\otimes_r$, then $\beta_{\{\}}(A,e,P,T,T',p)=\beta_{\{\}}(C,e,P,U,[\otimes_r]@T',p)$
\item If $A = B \parr C$ and $T=U.\parr_l$, then $\beta_{\{\}}(A,e,P,T,T',p)= \beta_{\{\}}(B,e,P,U,[\parr_l]@T',p)$
\item If $A = B \parr C$ and $T=U.\parr_r$, then $\beta_{\{\}}(A,e,P,T,T',p)= \beta_{\{\}}(C,e,P,U,[\parr_r]@T',p)$
\item If $A = \forall X. B$ and $T=U.\forall$, then
  \begin{itemize}
  \item If there exists a potential exists link $(l,Q)$ (whose associated formula will be named $C$) such that $(e,P,[\forall]@T',p) \rightsquigarrow^* (concl(l),Q,[\forall],-)$, then 
    \begin{equation*}
      \beta_{\{\}}(A,e,P,T,T',p)= \bigcup_{\Theta \vartriangleleft ztree(prem(l),Q)} \beta_{\{\}}(B[C[\theta_{\Theta}]_{/LLL}/X],e,P,U,[\forall]@T',p)
    \end{equation*}
    \item If such a potential link does not exist, $\beta_{\{\}}(A,e,P,T,T',p)= \beta_{\{\}}(B,e,P,U,[\forall]@T',p)$.
  \end{itemize}
\item If $A = \exists X. B$ and $T=U.\exists$, then
  \begin{itemize}
    \item If there exists a potential exists link $(l,Q)$ (whose associated formula will be named $C$) such that $(concl(l),Q,[\exists],+) \rightsquigarrow (e,P,[\exists]@T',p)$, then 
      \begin{equation*}\beta_{\{\}}(A,e,P,T,T',p)=  \bigcup_{\Theta \vartriangleleft ztree(prem(l),Q)}\beta_{\{\}}(B[C[\theta_\Theta]_{/LLL}/X],e,P,U,[\exists]@T',p) 
      \end{equation*}
    \item If such a potential link does not exist, $\beta_{\{\}}(A,e,P,T,T',p)= \beta_{\{\}}(B[Y/X],e,P,U,[\exists]@T',p)$ with $Y$ a variable which is not used yet.
  \end{itemize}
\item If $A = ! B$ and $T=U.!_t$, then $\beta_{\{\}}(A,e,P,T,T',p)=\beta_{\{\}}(B,e,P,U,[!_t]@T',p)$
\item If $A = ? B$ and $T=U.?_t$, then $\beta_{\{\}}(A,e,P,T,T',p)=\beta_{\{\}}(B,e,P,U,[?_t]@T',p)$
\item If $A = \S B$ and $T=U.\S$, then $\beta_{\{\}}(A,e,P,T,T',p)=\beta_{\{\}}(B,e,P,U,[\S]@T',p)$
\item Otherwise, it is undefined
\end{itemize}

\begin{definition}[underlying formulae]\label{def_underlyingformulae} Let $(e,P,T,p)$ be a context of $G$. Then, the {\em underlying formulae} of $(e,P,T,p)$, written $\beta_{\{\}}(e,P,T,p)$, is the set $\bigcup_{\Theta_e \vartriangleleft ztree(e,P)} \beta_{\{\}}(\beta(e)^p_{/LLL}[\theta_{\Theta_e}],e,P,T,[~],p)$.
\end{definition}

We would like to state that the set of underlying formulae of a context is stable by $\rightsquigarrow$ (if $C \rightsquigarrow D$, $\beta_{\{\}}(C)=\beta_{\{\}}(D)$. However, it is not true: Let us suppose that $C=(e,P,T,+) \rightsquigarrow (f,P,T.\forall,+)=D$, crossing a $\forall$ link downwards. Let $Z$ be the eigenvariable associated with this $\forall$ link and let us suppose that this $\forall$ link will be cut with an $\exists$ link whose associated formula is $B$. Then $Z$ is replaced by $B$ in every formula of $\beta_{\{\}}(D)$ while we can choose truncations of $ztree(e,P)$ which do not contain the $\{Z \mapsto B \}$ root. So, in this case we have $C \rightsquigarrow D$ and $\beta_{\{\}}(C) \supset \beta_{\{\}}(D)$. In the case of crossing a $\forall$ link upwards, we have $C \rightsquigarrow D$ and $\beta_{\{\}}(C) \subset \beta_{\{\}}(D)$. However, the difference is always on the formulae which use the smallest truncations. The larger truncations are possible for both contexts and lead to the same formulae. It is the meaning of Lemma \ref{lemma_underlying_step}.

\begin{lemma}
\label{lemma_underlying_step}
If $(e,P,T,p) \rightsquigarrow (f,Q,U,q)$, then for all $A \in \beta_{\{\}}(e,P,T,p) \cup \beta_{\{\}}(f,Q,U,q)$ there exists $B \in \beta_{\{\}}(e,P,T,p) \cap \beta_{\{\}}(f,Q,U,q)$ and a substitution $\sigma$ such that $B= A[\sigma]$.
\end{lemma}
\begin{prf}
We make a disjunction on the $\rightsquigarrow$ rule used.

If the rule is neither a $\exists$ nor a $\forall$ rule, then the proof is technical but straightforward. In these cases, we can prove that $\beta_{\{\}}(e,P,T,p) = \beta_{\{\}}(f,Q,U,q)$. Let us take a formula $A \in \beta_{\{\}}(e,P,T,p)$ and show that $A \in \beta_{\{\}}(f,Q,U,q)$. It is important to notice that $ztree(e,P)$ and $ztree(f,Q)$ are equal on the intersection of their domains. And those are the only one interesting, since the other variables (or their replacement) will be deleted by the pruning. So we can take $\Theta_f$ to be $\Theta_e$ restricted to the eigenvariables which are free in $\beta(f)$. 

We have to prove that $A \in \beta_{\{\}}(\beta(f)^q_{/LLL}[\theta_{\Theta_e}],f,Q,U,[~],q)$. The formulae $\beta(f)$ and $\beta(e)$ are almost the same. The only possible differences are a translation of the index of ground variables (in the axiom case) and the addition or deletion of the head connective. The first kind of difference is erased because the formula considered is not $\beta(e)$ but $\beta(e)_{/LLL}$. The second kind of difference is erased during the computations of $\beta_{\{\}}(\beta(f)^q_{/LLL}[\theta_{\Theta_e}],f,Q,U,[~],q)$ and $\beta_{\{\}}(\beta(e)^p_{/LLL}[\theta_{\Theta_e}],e,P,T,p)$. Indeed, we prune the syntactic tree of the formula using the trace, and when we delete the head connective, we also delete the right-most element of the trace. Notice that if $(e,P,T=(T_1.\forall @T_2),p)\rightsquigarrow (f,Q,U=(T_1.\forall @U_2),q)$, then when we compute $\beta_{\{\}}(e,P,T,p)$ and $\beta_{\{\}}(f,Q,U,q)$ we will have as inductive cases respictively $\beta_{\{\}}(B,e,P,T_1.\forall,T_2,p)$ and $\beta_{\{\}}(B,f,Q,T_1.\forall,U_2,q)$. The $\rightsquigarrow$ rules only depend on the right-most trace element so $(e,P,[\forall]@T_2,p) \rightsquigarrow (f,Q,[\forall]@U_2,q)$, so the $\exists$ potential link mentioned in the rule will be the same. Same for the $\exists$ induction cases. As a consequence $A \in \beta_{\{\},\Theta}(f,Q,U,q)$. The other inclusion is done similarly.

\begin{wrapfigure}{l}{2cm}
  \centering
  \begin{tikzpicture}
    \node[exists] (ex) at (0,0) {};
    \coordinate (exup) at ($(ex)+(0, 0.5)$); 
    \coordinate (exdo) at ($(ex)+(0,-0.5)$); 
    \draw [->] ($(exup)+(0,0.5)$)--(exup) node [edgename] {$e$} node [type] {$B[C/X]$};
    \draw (exup)-- (ex); 
    \draw [->] (ex)-- (exdo)  node [edgename] {$f$} node [type] {$\exists X.B $};
    \draw (exdo) -- ($(exdo)+(0,-0.5)$);
  \end{tikzpicture}
\end{wrapfigure}

In the case of crossing an $\exists$ link downward. $(e,P,T,+) \rightarrow (f,P,T.\exists,+)$. 

\begin{align*} 
  \beta_{\{\}}(f,Q,U,q) & = \bigcup_{\Theta_f \vartriangleleft ztree(f,Q)}\beta_{\{\}}(\beta(f)^q_{/LLL}[\theta_{\Theta_f}],f,Q,U,[~],q) \\
  & = \bigcup_{\Theta_f \vartriangleleft ztree(f,Q)} \beta_{\{\}}(\exists X.B_{/LLL} [\theta_{\Theta_f}],f,P,T.\exists,[~],+) \\
  \beta_{\{\}}(f,Q,U,q) & = \bigcup_{\Theta_f \vartriangleleft ztree(f,Q)} \bigcup_{\Theta_e \vartriangleleft ztree(e,P)}\beta_{\{\}}(B_{/LLL}[\theta_{\Theta_f}][C_{/LLL}[\theta_{\Theta_e}]/X],f,P,T,[\exists],+)\\
\end{align*}

\begin{align*} 
  \beta_{\{\}}(e,P,T,p) & = \bigcup_{\Theta_e \vartriangleleft ztree(e,P)}\beta_{\{\}}(\beta(e)^p_{/LLL}[\theta_{\Theta_e}],e,P,T,[~],p) \\
  & = \bigcup_{\Theta_e \vartriangleleft ztree(e,P)} \beta_{\{\}}(B[C/X]_{/LLL} [\theta_{\Theta_e}],e,P,T,[~],+) \\
\end{align*}
The only difference between $ztree(e,P)$ and $ztree(f,Q)$ is that the domain of the former may be bigger because $C$ may contain some eigenvariables. So the truncations of $ztree(f,Q)$ all are truncations of $ztree(e,P)$. The only difference between $\beta_{\{\}}(e,P,T,p)$ and $\beta_{\{\}}(f,Q,U,q)$ is that in the latter we can make a different truncation for the branch corresponding to an eigenvariable $Z$ in $\Theta_e$ (for occurences of $Z$ in $C$) and in $\Theta_f$ (for occurences of $Z$ in $B$). However, even if we use this possibility to pick a formula $A$ in $\beta_{\{\}}(f,Q,U,q)$ which is not in $\beta_{\{\}}(e,P,T,p)$, we can define $\Theta$ as the graph union of $\Theta_e$ and $\Theta_f$. Then  $B=\bigcup_{\Theta \vartriangleleft ztree(e,P)} \beta_{\{\}}(B[C/X]_{/LLL} [\theta_{\Theta}],e,P,T,[~],+)$ is in the intersection of the two sets and $B$ can be obtained from $A$ by a substitution.

\begin{wrapfigure}{l}{2cm}
  \centering
  \begin{tikzpicture}
    \node[forall] (ex) at (0,0) {};
    \coordinate (exup) at ($(ex)+(0, 0.5)$); 
    \coordinate (exdo) at ($(ex)+(0,-0.5)$); 
    \draw [->] ($(exup)+(0,0.5)$)--(exup) node [edgename] {$e$} node [type] {$B[Z/X]$};
    \draw (exup)-- (ex); 
    \draw [->] (ex)-- (exdo)  node [edgename] {$f$} node [type] {$\forall X.B $};
    \draw (exdo) -- ($(exdo)+(0,-0.5)$);
  \end{tikzpicture}
\end{wrapfigure}

  In the case of crossing a $\forall$ link downward. $(e,P,T,+) \rightsquigarrow (f,P,T.\forall,+)$. There are two cases, whether the potential $\forall$ link will be cut with an $\exists$ link or not. However, only the first case is interesting. In the case where the link will not be cut, we can use the same argumentation as in the multiplicative and exponential links. Thus, here we will supose that there exists a potential $\exists$ link $(l,R)$ such that $(f,P,T.\forall, +) \rightsquigarrow (concl(l),R,\forall, -)$. We will denote by $C$ the formula associated to $l$.

  \begin{align*} 
    \beta_{\{\}}(f,Q,U,q) & = \bigcup_{\Theta_f \vartriangleleft ztree(f,Q)}\beta_{\{\}}(\beta(f)^q_{/LLL}[\theta_{\Theta_f}],f,Q,U,[~],q) \\
    & = \bigcup_{\Theta_f \vartriangleleft ztree(f,Q)} \beta_{\{\}}(\forall X.B_{/LLL} [\theta_{\Theta_f}],f,P,T.\forall,[~],+) \\
    \beta_{\{\}}(f,Q,U,q) & = \bigcup_{\Theta_f \vartriangleleft ztree(f,Q)} \bigcup_{\Theta_l \vartriangleleft ztree(prem(l),R)}\beta_{\{\}}(B_{/LLL}[\theta_{\Theta_f}][C_{/LLL}[\theta_{\Theta_l}]/X],f,P,T,[\forall],+)\\
  \end{align*}

  \begin{align*} 
    \beta_{\{\}}(e,P,T,p) & = \bigcup_{\Theta_e \vartriangleleft ztree(e,P)}\beta_{\{\}}(\beta(e)^p_{/LLL}[\theta_{\Theta_e}],e,P,T,[~],p) \\
    & = \bigcup_{\Theta_e \vartriangleleft ztree(f,Q) \cup \{(Z,C,ztree(prem(l),R))\} } \beta_{\{\}}(B[Z/X]_{/LLL} [\theta_{\Theta_e}],e,P,T,[~],+) \\
  \end{align*}
  Here, the difference between $\beta_{\{\}}(e,P,T,p)$ and $\beta_{\{\}}(f,Q,U,q)$ is that in the latter the replacement of $X$ by $C$ is always done, whereas in the former, the truncation of $ztree(e,P)$ may delete this substitution. However, if we take a formula $A$ of $\beta_{\{\}}(e,P,T,p)$, with a substitution tree $\Theta_e$ which does not make a substitution on $Z$, we can define the formula $B$ obtained by extending $\Theta_e$ in $\Theta_e'$ which does not cut the edge $(Z,C)$ but its successors. $B$ belongs to both sets and $B$ can be obtained from $A$ by a substitution.

  The cases of crossing an $\exists$ or a $\forall$ link upwards are quite similar to the downwards cases, the roles of $e$ and $f$ are just swapped.
\qed
\end{prf}

\begin{lemma}
  \label{lemma_underlying_mapsto}
  If $C \mapsto D$ and $\beta_{\{\}}(C) \neq \varnothing$ then $\beta_{\{\}}(D) \neq \varnothing$.
\end{lemma}
\begin{proof}
  If it is a $\rightsquigarrow$ step, then we can use Lemma \ref{lemma_underlying_step}. If it is a $\hookrightarrow$ step, $B=(\sigma_i(B),P,[\oc_t],-) \hookrightarrow (\sigma(B),P,[\oc_t],+)=D$. Then 
  \begin{align*}
    \beta_{\{\}}(C) &= \bigcup_{\Theta_e  \vartriangleleft ztree(prem(l),Q)} \beta_{\{\}}(\beta(\sigma(B))_{/LLL}^+[\theta_{\Theta_e}],\sigma(B),P,[\oc_t],[~],+)\\
    \beta_{\{\}}(C) &\supset \beta_{\{\}}(\beta(\sigma(B))_{/LLL}[\theta_{\varnothing}],\sigma(B),P,[\oc_t],[~],+)\\
    \beta_{\{\}}(C) &\supset \beta_{\{\}}(\oc A,\sigma(B),P,[\oc_t],[~],+)\\    
    \beta_{\{\}}(C) &\supset \beta_{\{\}}(A,\sigma(B),P,[],[\oc_t],+)\\    
    \beta_{\{\}}(C) &\supset \{A\}
  \end{align*}
\end{proof}

\paragraph{}This lemma implies that if we begin a path by a context whose trace corresponds to its formula (for example $(\sigma(B), P, !_t, +)$), the path will never be blocked by a mismatch between the top element of the trace and the link encountered. For instance, we will never reach a context $(e,Q,T.\forall,-)$ with $\beta(e)=A \otimes B$. This lemma will allow us to prove the corollary \ref{charac_final_contexts}  which says that our copies are exactly the copies in the Dal Lago's definition. It will be important to prove the Dal Lago's weight theorem in our framework. But this is not the only purpose of this lemma. It will be used to prove the strong bound for $L^4_0$ in Section \ref{section_l40}.

\subsubsection{Our copies are the same as those of Dal Lago's}
Here, we will show that our definition of $\mapsto$-copies matches the definition of the copies of Dal Lago.

\begin{lemma}[Subtree property~\cite{lago2006context}]\label{subtree_property}
  For any exponential signature $t$ and any $u \blacktriangleleft t$, there exists $v$ such that $t \sqsubseteq v$ and: If $\forall t \sqsubseteq w, \exists f,Q,T,q, (e,P,!_w,+) \mapsto^* (f,Q,[!_{\sige}]@T,q)$, Then there is $(g,R)$ such that $(e,P,!_v,+) \mapsto^* (g,R,!_u,-)$.
\end{lemma}

\begin{definition}\label{def_skeleton}
  The {\em skeleton} of a trace element is defined by: the skeleton of $!_t$ is $!$, the skeleton of $?_t$ is $?$, the skeleton of other trace elements are themselves. We define the skeleton of a trace as the list of the skeletons of its trace elements.
\end{definition}
\begin{lemma}
\label{lemma_skeletons_cons}
  If $(e_0,P_0,T_0,p_0) \mapsto (e_1,P_1,T_1,p_1) \cdots \mapsto (e_n,P_n,T_n,p_n)$, $(e_0,Q_0,U_0,p_0) \mapsto (e_1, Q_1, U_1,p_1) \cdots \mapsto (e_n,Q_n, U_n,p_n)$ and the skeletons of $T_1$ and $U_1$ are equal, then the skeletons of $T_n$ and $U_n$ are equal.
\end{lemma}
\begin{proof}
  It is enough to prove for the case $n=1$, which can be done by analysis of all the $\mapsto$ rules. 
\end{proof}

\begin{lemma}\label{lemma_prec_left_right}
  If $t_n \preccurlyeq^v t_n'$, $T_n$ has the same skeleton as $U_n$ and $(e_1,P_1,[!_{t_1}]@T_1,p_1) \mapsto (e_2,P_2,[!_{t_2}]@T_2,p_2) \mapsto \cdots \mapsto (e_n,P_n,[!_{t_n}]@T_n,p_n)$ then there exists $t_1 \preccurlyeq^v t_1'$ such that $(e_1,P_1,[!_{t'_1}]@T_1,p_1) \mapsto \cdots \mapsto (e_n,P_n,[!_{t'_n}]@T_n,p_n)$.

  If $t_n' \preccurlyeq^v t_n$, $T_n$ has the same skeleton as $U_n$ and $(e_1,P_1,[!_{t_1}]@T_1,p_1) \mapsto (e_2,P_2,[!_{t_2}]@T_2,p_2) \mapsto \cdots \mapsto (e_n,P_n,[!_{t_n}]@T_n,p_n)$ then there exists $t_1' \preccurlyeq^v t_1$ such that $(e_1,P_1,[!_{t'_1}]@T_1,p_1) \mapsto \cdots \mapsto (e_n,P_n,[!_{t'_n}]@T_n,p_n)$.
\end{lemma}
\begin{proof}
  We can prove the lemma by induction on $n$ and case-by-case analysis of the $\mapsto$ rules. This works because the $\mapsto$ only take into account the surface of the exponential signatures. And during a $\mapsto$ step, exponential signatures are decreasing for the $\blacktriangleleft$ relation.
\end{proof}

\begin{lemma}\label{lemma_prec_right_left}
  If $t_n \preccurlyeq^v t_n'$, $(e_1,P_1,[!_{t_1}]@T_1,p_1) \mapsto (e_2,P_2,[!_{t_2}]@T_2,p_2) \mapsto \cdots \mapsto (e_n,P_n,[!_{t_n}]@T_n,p_n)$ and $(e_1,Q_1,[!_{t'_1}]@U_1,p_1) \mapsto \cdots \mapsto (e_n,Q_n,[!_{t'_n}]@U_n,p_n)$ then $t_1 \preccurlyeq^v t'_1$
\end{lemma}

\begin{lemma}\label{lemma_copy_context_final_context}
  Let $C=(e,P,[\oc_t]@T,p)$ be a context, with $\beta_{\{\}}(C) \neq \varnothing$. $C$ is a $\mapsto$-copy context if and only if there exists a final context $C_f \in F_G$ such that $C \mapsto^* C_f$
\end{lemma}
\begin{proof}
  Let us suppose that $(e,P,[\oc_t]@T,p) \mapsto (e_1,P_1,T_1,p_1) \cdots \mapsto (e_n,P_n,T_n,p_n) =C_f$ and $C_f$ is final. Looking at the rules of $\mapsto$ and the definition of final contexts, $C_f \not \mapsto$. Let us suppose that there is an exponential signature $u$ such that $(e,P,[\oc_u]@T,p) \mapsto (e_1, Q_1,U_1,p_1) \cdots (e_n,Q_n,U_n,p_n)$. The skeletons of $[!_t]@T$ and $[\oc_u]@T$ are equal, so the skeletons of $T_n$ and $U_n$ are equal (Lemma \ref{lemma_skeletons_cons}). Thus, no matter which case of final context is $C_f$, $(e_n,Q_n,U_n,p_n) \not \mapsto$. We know that $T_n$ has the shape $[!_{\sige}]@T'$ so $U_n$ has the shape $[!_v]@U'$. By definition of $\preccurlyeq$, $\sige \preccurlyeq v$. So, according to lemma \ref{lemma_prec_right_left}, $t \preccurlyeq u$.

  Now, let us suppose that $(e,P,[\oc_t]@T,p)$ is a copy context, then $(e,P,[\oc_t]@T,p) \mapsto (e_1,P_1,T_1,p_1) \cdots \mapsto (e_n,P_n,T_n,p_n) \not \mapsto$ and forall $u$ such that $(e,P,[\oc_u]@T,p) \mapsto (e_1, Q_1,U_1,p_1) \cdots (e_n,Q_n,U_n,p_n)$, $t \preccurlyeq u$ and $(e_n,Q_n,U_n,p_n) \not \mapsto$. We will show that $(e_n,P_n,T_n,p_n)$ is a final context. We supposed that $\beta_{\{\}}(e,P,[!_t]@T,p)$ is  not empty, so $\beta_{\{\}}(e_n,P_n,T_n,p_n)$ is not empty by Lemma~\ref{lemma_underlying_mapsto}. So, knowing that $(e_n,P_n,T_n,p_n) \not \mapsto$ and that the leftmost trace element is a $!$ element, the only possibilities for $(e_n,P_n,T_n,p_n)$ are $((\_,\bullet),P_n,T_n,+)$ or $(e_n,P_n,[\oc_y],-)$ with the tail of $e_n$ being a $?C$, $?D$, $?N$ or $?W$ link.

  We will prove that the tail of $e_n$ can not be a $\wn C$ or $\wn N$ by absurd. If the tail of $e_n$ was a $?C$ or $?N$ and $\sige \neq y$ then notice that $\sige \preccurlyeq y$ so there exists $t' \preccurlyeq^y t$ such that $(e,P,[!_{t'}]@T,p) \mapsto \cdots (e_n,R',!_{\sige},-)$ (Lemma \ref{lemma_prec_left_right}). However, because $(e,P,[!_{t}]@T,p)$ is a copy context, $t \preccurlyeq t'$. Then, $\preccurlyeq$ being an order, $t'=t$ so $y = \sige$ which is a contradiction. If we suppose that the tail of $e$ was a $?C$ or $?N$ and $y = \sige$ then notice that $\sige \preccurlyeq \sigl(\sige)$ and $\sige \preccurlyeq \sigp(\sige)$ so, by Lemma \ref{lemma_prec_left_right}, there exists $u$ such that $(e,P,[\oc_u]@T,p) \mapsto (e_1, Q_1,U_1,p_1) \cdots (e_n,Q_n,U_n,p_n) \mapsto$ which would contradict our assumption of $(e,P,[\oc_t]@T,p)$ being a $\mapsto$-copy context. So the tail of $e_n$ is neither $\wn C$ nor $\wn N$.

  So the only fact left to prove is that the leftmost trace element is $!_{\sige}$. If it is not, we can find an exponential signature $u$ such that $u \preccurlyeq t$, $u \neq t$ and $(e,P,!_u,+) \mapsto (e_1, Q_1,U_1,p_1) \cdots (e_n,Q_n,U_n,p_n)$ (Lemma \ref{lemma_prec_left_right}). 
\end{proof}

\begin{theorem}
  \label{charac_final_contexts}
  Let $t$ be a standard exponential signature. $t \in C_{\mapsto}(B,P)$ if and only if for any $t \sqsubseteq u$, there exists a final context $C_u$ such that $(\sigma(B),P,[\oc_u],+) \mapsto^* C_u$.
\end{theorem}
\begin{proof}
  If $t \in C_{\mapsto}(B,P)$, let us consider $u \in Sig$ such that $t \sqsubseteq u$. By definition of copies, $(\sigma(B),P,[\oc_t],+)$ is a $\mapsto$-copy context. So, by Lemma \ref{lemma_copy_context_final_context}, there exists a final context $C_u$ such that $(\sigma(B),P,[\oc_u],+) \mapsto^* C_u$.

  Now, let us suppose that for any $t \sqsubseteq u$, there exists a final context $C_u$ such that $(\sigma(B),P,[\oc_u],+) \mapsto^* C_u$. Then, according to Lemma \ref{lemma_copy_context_final_context}, for any $t \sqsubseteq u$, $(\sigma(B),P,[\oc_u],+)$ is a $\mapsto$-copy context. So $t \in C_\rightarrow(B,P)$.
\end{proof}

\subsection{Proof of Dal Lago's weight theorem}
\label{subsection_proof_dal_lago}
In this subsection, we will assume that $G$ and $H$ are two proof-nets and that there is a copymorphism $(D_\phi, D_\phi', \phi, \psi)$ from $G$ to $H$. We will first exhibit a correspendence between the $\mapsto$-canonical edges of $G$ and the $\mapsto$-canonical edges of $H$ by a serie of lemma. Then we will prove that the weight $T_G$ decreases along $cut$-elimination.

\begin{lemma}\label{lemma_copymorphism_cons_copy_context}
  If $(\sigma(B),P,[!_t],p)$ is a $\mapsto$-copy context and $\phi(\sigma(B),P,t)=(\sigma(B'),P',t')$ then $(\sigma(B'),P',[!_{t'}],p)$ is a $\mapsto$-copy context.
\end{lemma}
\begin{proof}
  By Theorem \ref{charac_final_contexts}, there exists $C_e=(e,Q,[\oc_{u}]@U,q) \in F_G$ such that $(\sigma(B),P,[!_t],+)\mapsto^* C_e$.
  \begin{itemize}
  \item If $e \in D_{\phi}$, then by rule \ref{copymorphism_final_defined} $\phi(e,Q,u)$ is defined. Let $(e',Q',u')$ be $\phi(e,Q,u)$. By rule \ref{copymorphism_path_cons}, $(\sigma(B'),P',[!_{t'}]@T,+) \mapsto (e',Q',[\oc_{u'}]@U,q)$. And by rule \ref{copymorphism_final_cons}, $(e',Q',[\oc_{u'}]@U,q)$ is a final context for $H$.
  \item Else, by rule \ref{copymorphism_psi}, $(e,Q, [\oc_{u}]@U,q) \in D_{\psi}$. There exists $k$ such that  $\forall (f,R,[!_v]@V,r) \in C_{G} ,(f,R,[!_v]@V,r) \mapsto^k (e,Q,[\oc_{u}]@U,q) \Leftrightarrow (f,R,[!_v]@V,r) \in \psi(e,Q,[!_{u}]@U,q)$ (rule \ref{copymorphism_psi_k}) and if $j < k$, then $(f,R,[!_v]@V,r) \mapsto^{j} (e,Q,[!_{u}]@U,q) \Rightarrow f \not \in \sigma(B_G)$. So there is $(f,R,[\oc_v]@V,r) \in \psi(e,Q,[\oc_{u}]@U,q)$ such that $(\sigma(B),P,[\oc_t],+) \mapsto^* (f,R,[!_v]@V,r) \mapsto^k (e,Q,[\oc_{u}]@U,q)$. By rule \ref{copymorphism_final_defined}, $\phi(f,R,v)$ is defined. Let $(f',R',v')$ be $\phi(f,R,v)$, according to rule \ref{copymorphism_final_cons}, $(f',R',[\oc_{v'}]@V,r)$ is final. According to rule \ref{copymorphism_path_cons}, $(\sigma(B'),P', [\oc_{t'}],+) \mapsto^* (f',R',[\oc_{v'}]@V,r)$.
  \end{itemize}
  So $(\sigma(B'),P',[\oc_{t'}],+)$ is a $\mapsto$-copy context.
\end{proof}

\begin{lemma}\label{lemma_copymorphism_revcons_copy_context}
  If $(\sigma(B'),P',[!_{t'}],+)$ is a $\mapsto$-copy context and $\phi(\sigma(B),P,t)=(\sigma(B'),P',t')$ then $(\sigma(B),P,[!_{t}],+)$ is a $\mapsto$-copy context.
\end{lemma}
\begin{proof}
  $(\sigma(B'),P',[!_{t'}],+)$ is a $\mapsto$-copy context so $(\sigma(B'),P',[!_{t'}],+) \mapsto^* (e',Q',[!_{u'}]@U,q)$ with $(e',Q',[!_{u'}]@U,q) \in F_H$. So, according to rule \ref{copymorphism_final_cons}, either  there exists $(e,Q) \in Pot(E_G)$ and $u \in Sig$ such that $(e,Q,[!_u]@U,q) \in F_G$ and $\phi(e,Q,u)=(e',Q',u')$ or there exists $e,Q,u$ and $(f,R,[!_v]@V,r)\in F_G$ such that $(e,Q,[!_u]@U,q) \in \psi(f,R,[!_v]@V,r)$ and $\phi(e,Q,u)=(e',Q',u')$. We examine both cases.
  \begin{itemize}
  \item If we suppose that there exists $e,Q,u$ such that $(e,Q,[!_u]@U,q) \in F_G$ and $\phi(e,Q,u)=(e',Q',u')$. So, according to rule \ref{copymorphism_path_cons}, $(\sigma(B),P,[!_{t}],+) \mapsto^* (e,Q,[!_u]@U,q)$ and $(e,Q,[!_u]@U,q) \in F_G$ so $(\sigma(B),P,[\oc_t],+)$ is a $\mapsto$-copy context.
  \item If we suppose that there exists $e,Q,u$ and $(f,R,[!_v]@V,r) \in F_G$ such that $(e,Q,[!_u]@U,q) \in \psi(f,R,[!_v]@V,r)$ and $\phi(e,Q,u)=(e',Q',u')$. According to rule \ref{copymorphism_path_cons}, $(\sigma(B),P,[!_t],+) \mapsto^* (e,Q,[!_u]@U,q)$. According to rule \ref{copymorphism_psi_k}, $(e,Q,[!_u]@U,q)\mapsto^*(f,R,[!_v]@V,r)$ and $(f,R,[!_v]@V,r)$ is final. By transitivity $(\sigma(B),P,[!_t],+) \mapsto^* (f,R,[!_v]@V,r)$ so $(\sigma(B),P,[\oc_t],+)$ is a $\mapsto$-copy context of $(B,P)$.
  \end{itemize}
\end{proof}

\begin{lemma}\label{copymorphism_phi_standard} If $\phi ( \sigma (B), P, t)=(\sigma(B'), P', t')$ and $t$ is standard, then $t'$ is standard.
\end{lemma}
\begin{proof}
  Let us assume that $t'$ is not standard. Then, there exists $u' \in Sig$ such that $u' \sqsubset t'$. We have $(B',P',u') \Subset (B',P',t')$. Let $C,Q,u$ such that $\phi(\sigma(C),Q,u)=(B',P',u')$. Then, by rule \ref{copymorphism_subset_cons}, $(C,Q,u) \Subset (B,P,t)$. However, because of rule \ref{copymorphism_phi_division}, we know that $(C,Q)=(B,P)$. So $u \sqsubset t$. This is impossible, because $t$ is standard.
\end{proof}

\begin{lemma}\label{lemma_copymorphism_cons_copy}
  If $(B,P) \in Can_{\mapsto}(B_G)$, $t \in C_{\mapsto}(B,P)$ and $\phi(\sigma(B),P,t)=(\sigma(B'),P',t')$ then $t' \in C_{\mapsto}(B',P')$.
\end{lemma}
\begin{proof}
  We know that $t$ is a copy so $t$ is standard. By Lemma \ref{copymorphism_phi_standard}, $t'$ is standard. Let $u' \in Sig$ such that $t' \sqsubseteq u'$. By rule \ref{copymorphism_princ_doors}, there exists $(C,Q,u)$ such that $\phi(\sigma(C),Q,u)=(\sigma(B'),P',t')$. 

Either $t' \sqsubset u'$ or $t'=u'$. If $t' \sqsubset u'$, $(B',P',t') \Subset (B',P',u')$. Then, by rule \ref{copymorphism_subset_cons}, $(B,P,t) \Subset (C,Q,u)$. Either $(B,P)=(C,Q)$ and $t \sqsubset u$ or $B \subset C$ and $P=Q.v@Q'$ with $v \sqsubseteq u$. In both cases, we will prove that $(\sigma(C),Q,[!_u],+)$ is a copy context. Lemma \ref{lemma_copymorphism_cons_copy_context} will then give us that $(\sigma(C'),Q',[!_{u'}],+)$ is a $\mapsto$-copy context, concluding the proof.
    \begin{itemize}
    \item In the first case, $B=C$, $P=Q$ and $t \sqsubset u$. We already know that $t \in C_{\mapsto}(B,P)$ so, by definition of $C_{\mapsto}(B,P)$, $(\sigma(C),Q,[!_u],+)$ is a $\mapsto$-copy context.
    \item If $B \subset C$ and $P=Q.v@Q'$ with $v \sqsubseteq u$. We know that $P \in L_G(B)$ so the signatures composing it are copies of their respective corresponding boxes. So $v \in C_{\mapsto}(C,Q)$, so $(\sigma(C),Q,[!_u],+)$ is a copy context.
    \end{itemize}

If $t' = u'$. Then, it is similar to the first case of $t' \sqsubset u'$: $(\sigma(B),P,[\oc_t],+)$ is a $\mapsto$-copy context. So $(\sigma(B'),P',[\oc_{t'}],+)$ is a $\mapsto$-copy context.
\end{proof}

\begin{lemma}\label{lemma_copymorphism_revcons_copy}
  If $(B',P') \in Can_{\mapsto}(B')$, $t' \in C_{\mapsto}(B',P')$, $\phi(\sigma(B),P,t)=(\sigma(B'),P',t')$ and $t \sqsubseteq u$ then $(\sigma(B),P,[!_u],+)$ is a copy context.  
\end{lemma}
\begin{proof}
  If $t=u$, then $(\sigma(B'),P',[\oc_{t'}],+)$ is a $\mapsto$-copy context so, by lemma \ref{lemma_copymorphism_revcons_copy_context}, $(\sigma(B),P,[\oc_t],+)$ is a $\mapsto$-copy context.

If $t \sqsubset u$, $(B,P,t) \Subset (B,P,u)$. According to rule \ref{lemma_copymorphism_subseteq}, $\phi(\sigma(B),P,u)$ is defined. Let $(\sigma(C'),Q',u')$ be $\phi(\sigma(B),P,u)$. Then, according to rule \ref{copymorphism_subset_cons}, $(B',P',t') \Subset (C',Q',u')$. Either $(B',P')=(C',Q')$ and $t' \sqsubset u'$ or $B' \subset C'$ and $P'=Q'.v'@R'$ with $v' \sqsubseteq u'$. In both cases we will prove that $(\sigma(C'),Q',[\oc_{u'}],+)$ is a $\mapsto$-copy context. Lemma \ref{lemma_copymorphism_revcons_copy_context} will then give us that $(\sigma(B),P,[!_{u}],+)$ is a $\mapsto$-copy context, concluding the proof.
  \begin{itemize}
  \item If $(B',P')=(C',Q')$ and $t' \sqsubset u'$ then, because $t' \in C_{\mapsto}(B',P')$, $(\sigma(B'),P',[\oc_{u'}],+)$ is a $\mapsto$-copy context.
  \item If $B' \subset C'$ and $P'=Q'.v'@R'$ with $v' \sqsubseteq u'$. Then, knowing that $P' \in L_{\mapsto}(B')$, we know that $v' \in C_{\mapsto}(C',Q')$ so $(\sigma(C'),Q',[\oc_{u'}],+)$ is a $\mapsto$-copy context.
  \end{itemize}
\end{proof}

\begin{coro}\label{lemma_copymorphism_cons_canonical}
  If $\phi(e,P,t)=(e',P',t')$ and  $P \in L_{\mapsto}(e)$, then $P' \in L_{\mapsto}(e')$
\end{coro}
\begin{proof}
  We will make the proof by induction on $\partial(e)$. Let us suppose that $\phi(e,P,t)=(e',P',t')$ and $P \in L_{\mapsto}(e)$. Then, 
  \begin{itemize}
  \item If $P = [~]$ then by rule \ref{copymorphism_boxwise}, $e'$ has depth $0$ so $P'=[~]$, which is canonical for $e'$.
  \item Else, let $B$ the deepest box including $e$ such that $\sigma(B) \in D_{\phi}$. Then we can decompose $P$ in $P=Q.t@R$ with $Q \in L_{\mapsto}(B)$ and $t \in C_{\mapsto}(B,Q)$. Let $(\sigma(B'),Q',t')$ be $\phi(\sigma(B),Q,t)$ then according to rule \ref{copymorphism_boxwise}, $\rho(e')=B'$ and $P'=Q'.t'$. According to lemma \ref{lemma_copymorphism_cons_copy}, $t' \in C_{\mapsto}(B',Q')$. Moreover, by induction hypothesis, $Q' \in L_{\mapsto}(B')$. So $P'$ is canonical for $e'$.
  \end{itemize}
\end{proof}

\begin{coro}\label{lemma_copymorphism_cons_canonical2}
  If $(B',P') \in Can(B_H)$ and $t' \in Sig$, then there exists $t \in Sig$ and $(B,P) \in Can(B_G)$ such that $\phi(\sigma(B),P,t)=(\sigma(B'),P',t')$.
\end{coro}
\begin{proof}
  We prove this by induction on the depth of $B'$. Let us assume that the property is true for every $(C',Q') \in Pot(B_H)$ with $\partial(C') < \partial(B')$. Let $C'$ be the deepest box containing $B'$, and $u'$ the exponential signature such that $Q'.u'=P'$. Then, by induction hypothesis, there exists $u \in Sig$ and $(\sigma(C),Q) \in Can(B_G)$ such that $\phi(\sigma(C),Q,u)=(\sigma(C'),Q',u')$. 

Let $(\sigma(B),P_1,t) \in \phi^{-1}(\sigma(B'),P',t')$. $(B',P',t') \Subset (C',Q',u')$ so $(B,P_1,t) \Subset (C,Q,u)$. We prove by contradiction that $C$ is the deepest box containing $B$ such that $\sigma(C) \in D_{\phi}$. Let $D$ be the deepest box containing $B$ such that $\sigma(D) \in D_{\phi}$ and let us suppose that $D \subset C$. Then we extend $Q.u$ in a potential for $D$, $R$. Let $(\sigma(D'),R',v')=\phi(\sigma(D),R,v)$. $(D,R,v) \Subset (C,Q,u)$ so, according to rule \ref{copymorphism_subset_cons}, $(D',R',v') \Subset (C',Q',u)$. $u$ is a $\mapsto$-copy so is standard so $D' \subset C'$. According to rule \ref{copymorphism_boxwise}, $\rho(B')=D'$ which contradicts our hypothesis that $\rho(B)=C'$. So $C$ is the deepest box containing $B$ such that $\sigma(C) \in D_{\phi}$.

Thanks to rule \ref{copymorphism_phi_copy_not_empty}, we can complete $Q.u$ in a canonical potential $P$ for $B$. And, thanks to rule \ref{copymorphism_boxwise}, $\phi(\sigma(B),P,t)=\phi(\sigma(B),P_1,t)=(B',P',t')$.
\end{proof}

Let us suppose that there exists a copymorphism between two proof-nets $G$ and $H$. If we compute the differences of weights between $G$ and $H$ ($W_G-W_H$ and $T_G-T_H$), there are many simplifications. So, the differences depends mostly on the edges of $G$ (resp $H$) which are not in $D_{\phi}$ (resp $D_{\phi}'$). In the cut reduction rules, most of the edges are in $D_{\phi}$ or $D_{\phi}'$, so we have only a few edges to consider to compute those differences.

We will separate the weight $T_G$ into two subweights $T_G=T^1_G+2\cdot T^2_G$ with $T^1_G=\sum_{e \in E_G}|L_{\mapsto}(e)|$ and $T^2_G=\sum_{ \substack{B \in B_G \\ P \in L_{\mapsto}(B)\\ t \in Si_{\mapsto}(e,P)}}D_G(B).|t|$.

\begin{align*}
  T_G^1-T_H^1 &= \sum_{e \in E_G}|L_{\mapsto}(e)| - \sum_{e \in E_H}|L_{\mapsto}(e)|\\
  &= \sum_{\substack{e \in E_G \cap D_{\phi}\\ P \in L_{\mapsto}(e)}}1 + \sum_{e\in E_G \cap \overline{D_{\phi}}}|L_{\mapsto}(e)| - \sum_{\substack{f \in E_H \cap D_{\phi}' \\ Q \in L_{\mapsto}(f)}}1 -\sum_{f \in E_H \cap \overline{D_{\phi}'}}|L_{\mapsto}(f)|\\
  &= \sum_{\substack{e \in D_\phi\\ P \in L_{\mapsto}(e)}}1 - \sum_{\substack{f \in D'_\phi \\ Q \in L_{\mapsto}(f)}}|\Set{(e,P,t) }{ \phi(e,P,t)=(f,Q,\sige)}| + \sum_{e\in  \overline{D_{\phi}}}|L_{\mapsto}(e)| -\sum_{f \in  \overline{D_{\phi}'}}|L_{\mapsto}(f)| \\
  &= \sum_{\substack{e \in D_\phi \\ P \in L_{\mapsto}(e)}}1 - \sum_{\substack{e \in D_\phi \\ P \in L_{\mapsto}(e)}}|\Set{ t }{ \phi(e,P,t)=(\_,\_, \sige) }| + \sum_{ e \in  \overline{D_{\phi}}}|L_{\mapsto}(e)| -\sum_{f \in  \overline{D_{\phi}'}}|L_{\mapsto}(f)| \\
  T_G^1-T_H^1&= \sum_{\substack{e \in D_\phi \\ P \in L_{\mapsto}(e)}}1-|\Set{ t }{ \phi(e,P,t)=(\_,\_, \sige)}|+ \sum_{e\in  \overline{D_{\phi}}}|L_{\mapsto}(e)| -\sum_{f \in  \overline{D_{\phi}'}}|L_{\mapsto}(e)| \\
\end{align*}

For the transformation between the second and third line, notice that we have $|\Set{(e,P,t)}{\phi(e,P,t)=(f,Q,\sige)}|=1$ for every $(f,Q) \in Can(D'_{\phi})$ because of the rule \ref{copymorphism_phi_division} of the definition of copymorphisms. The transformation between the third and fourth line, use Corollaries \ref{lemma_copymorphism_cons_canonical} and \ref{lemma_copymorphism_cons_canonical2}.

The formula of the last line may seem more complex than the first line. However, when we will use this formula, we will notice that most potential edges of $G$ will have exactly one image in $Pot(E_H)$ with exponential signature $\sige$. So, we will immediately notice that most of the terms of the sum are equal to 0.

\begin{align*}
  T_G^2 - T_H^2 & = \sum_{\substack{B \in B_G \\ P \in L_{\mapsto}(B) \\ s \in Si_{\mapsto}(B,P)}}D_G(B).|s| - \sum_{\substack{B \in B_H \\ Q \in L_{\mapsto}(B) \\ t \in Si_{\mapsto}(B,Q)}}D_H(B).|t| \\
  T_G^2 - T_H^2 & = \sum_{\substack{\sigma(B) \in D_{\phi} \\ P \in L_{\mapsto}(B) \\ s \in Si_{\mapsto}(B,P)}}D_G(B).|s| + \sum_{\substack{\sigma(B) \in \overline{D_{\phi}} \\ P \in L_{\mapsto}(B) \\ s \in Si_{\mapsto}(B,P)}}D_G(B).|s| - \sum_{\substack{B \in B_H \\ Q \in L_{\mapsto}(B) \\ t \in Si_{\mapsto}(B,Q)}}D_H(B).|t| \\
  T_G^2 - T_H^2 & \geq \sum_{\substack{\sigma(B) \in D_{\phi} \\ P \in L_{\mapsto}(B) \\ s \in Si_{\mapsto}(B,P)\\(C,Q,t)=\phi(B,P,s) }}D_G(B).|s| - \sum_{\substack{B \in B_H \\ Q \in L_{\mapsto}(B) \\ t \in Si_{\mapsto}(B,Q)}}D_H(B).|t| + \sum_{\substack{\sigma(B) \in \overline{D_{\phi}} \\ P \in L_{\mapsto}(B) \\ s \in Si_{\mapsto}(B,P)}}D_G(B).|s|\\
  T_G^2 - T_H^2  & \geq  \sum_{\substack{\sigma(B) \in D_{\phi} \\ P \in L_{\mapsto}(B) \\ s \in Si_{\mapsto}(B,P)\\ (C,Q,t)=\phi(B,P,s)}}D_G(B).|s|-D_H(C).|t| + \sum_{\substack{\sigma(B) \in \overline{D_{\phi}} \\ P \in L_{\mapsto}(B) \\ s \in Si_{\mapsto}(B,P)}}D_G(B).|s|
\end{align*}

Here, we use the fact that every $(\sigma(B'),P') \in Pot(B_H)$ are images of $(\sigma(B),P) \in Pot(B_G)$ (rule \ref{copymorphism_princ_doors} of the definition of copymorphisms). Similar to the case $T_G^1$, the formula of the second line may seem complex, but most of the terms of the left sum will be equal to $0$.

\begin{theorem}[Dal Lago's weight theorem]\label{dallago_weight}
  Suppose that $G \rightarrow_{cut} H$, $T_H$ is finite and $H$ has positive weights then $T_G > T_H$, $T_G$ is finite and $G$ has positive weights.
\end{theorem}
\begin{proof}
  We will examine every rule of reduction. For each rule, we will exhibit a copymorphism to prove the inequalities. We will only present the non trivial images. On the drawings, the edges of $G$ (resp. $H$) which are not in $D_{\phi}$ (resp. $D_{\phi}'$) will be drawn dashed. The edges of $D_{\phi}$ which have different images by $\phi$ in $E_H$ depending on the associated potential and signature, will be drawn thicker. 

  Let us consider a box of $B$. In most of the cases, $\sigma(B)$ will be in $D_\phi$. The copies of $(B,P)$ will correspond, via $\phi$ to distinct copies of potential boxes of $H$. $H$ having positive weights, there is a non null and finite number of those copies so there is a non null and finite number of copies of $(B,P)$. For the cases where $\sigma(B) \not \in D_\phi$: in the weakening and dereliction cases, it is straightforward that the box disappearing has exactly one copy: $\sige$. In the case of the box fusion, the copies of the deleted box correspond exactly to the copies of the fused box (in $H$) so there is a non null, finite number of copies.
  
  The difference between $T_G$ and $T_H$ can always be expressed as sums over sets of canonical potentials or sets of copies. Those are finite sums, because $G$ and $H$ are supposed to have positive weights.

  \begin{itemize}
  \item 
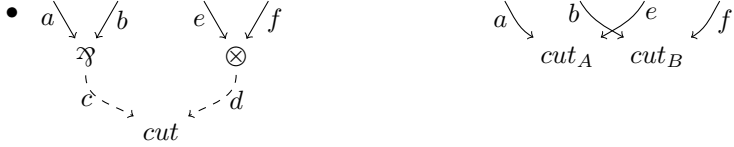
    \begin{tikzpicture}[baseline=1.5cm]
        \draw (0,0) node (cut) {$cut$};
        \draw (cut) ++ (-1,1) node (par) {$\parr$};
        \draw (cut) ++ ( 1,1) node (tens){$\otimes$};
        \draw (par) ++ (120:1) node (restgg) {};
        \draw (par) ++ (60 :1) node (restgd) {};
        \draw (tens) ++ (120:1) node (restdg) {};
        \draw (tens) ++ (60 :1) node (restdd) {};
        
        \draw [->, out=-90, in=30, dashed] (tens) to  node [right, midway] {$d$} (cut);
        \draw [->,out=-90, in=150, dashed] (par) to node [left, midway] {$c$} (cut) ;
        \draw [<-] (tens) -- (restdg) node [left,midway] {$e$};
        \draw [<-] (tens) -- (restdd) node [right,midway] {$f$};
        \draw [<-] (par) -- (restgg) node [left,midway] {$a$};
        \draw [<-] (par) -- (restgd) node [right, midway] {$b$};
        
        \draw (restgg) ++ (6,0) node (restegg) {};
        \draw (restgd) ++ (6,0) node (restegd) {};
        \draw (restdg) ++ (6,0) node (restedg) {};
        \draw (restdd) ++ (6,0) node (restedd) {};
        
        \draw (par) ++ (6.4,0) node (cutg) {$cut_A$};
        \draw (tens) ++ (5.6,0) node (cutd) {$cut_B$};
        
        \draw [->,out=-60, in=150] (restegg) to node [left, midway] {$a$} (cutg);
        \draw [<-,out=30, in=-120] (cutg) to node [right, near end] {$e$} (restedg);
        \draw [->,out=-60, in = 150] (restegd) to node [left, near start] {$b$} (cutd);
        \draw [<-,out=30, in = -120] (cutd) to node [right, midway] {$f$}(restedd);
    \end{tikzpicture}

    All edges except $c$ and $d$ are in $D_\phi$. For every $(e,P) \in Pot(D_\phi)$ and $t \in Sig$, $\phi(e,P,t)=(e,P,t)$.
    
    \begin{align*}
      T_G^1-T_H^1 &= \sum_{\substack{e \in D_\phi \\ P \in L_{\mapsto}(e)}}1-|\Set{ t }{ \phi(e,P,t)=(\_,\_, \sige)}|+ \sum_{e\in  \overline{D_{\phi}}}|L_{\mapsto}(e)| -\sum_{f \in  \overline{D_{\phi}'}}|L_{\mapsto}(e)| \\
      T_G^1-T_H^1 &= \sum_{\substack{e \in D_\phi \\ P \in L_{\mapsto}(e)}}(1-1)+ \sum_{e\in \{c,d\}}|L_{\mapsto}(e)| -\sum_{f \in \varnothing}|L_{\mapsto}(e)| \\
      T_G^1-T_H^1 &= |L_{\mapsto}(c)|+|L_{\mapsto}(d)|
    \end{align*}
    
    \begin{align*}
      T_G^2 - T_H^2 & =  \sum_{\substack{\sigma(B) \in D_{\phi} \\ P \in L_{\mapsto}(B) \\ s \in Si_{\mapsto}(B,P)\\ (f,Q,t)=\phi(B,P,s)}}D_G(B).|s|-D_H(f).|t| + \sum_{\substack{\sigma(B) \in \overline{D_{\phi}} \\ P \in L_{\mapsto}(B) \\ s \in Si_{\mapsto}(B,P)}}D_G(B).|s| \\
      T_G^2 - T_H^2 & =  \sum_{\substack{\sigma(B) \in D_{\phi} \\ P \in L_{\mapsto}(B) \\ s \in Si_{\mapsto}(B,P)\\ (f,Q,t)=\phi(B,P,s)}}0 + \sum_{\substack{B \in \varnothing \\ P \in L_{\mapsto}(B) \\ s \in Si_{\mapsto}(B,P)}}D_G(B).|s| \\
      T_G^2 - T_H^2 & = 0
    \end{align*}
    
    So $T_G > T_H$.
    
  \item 
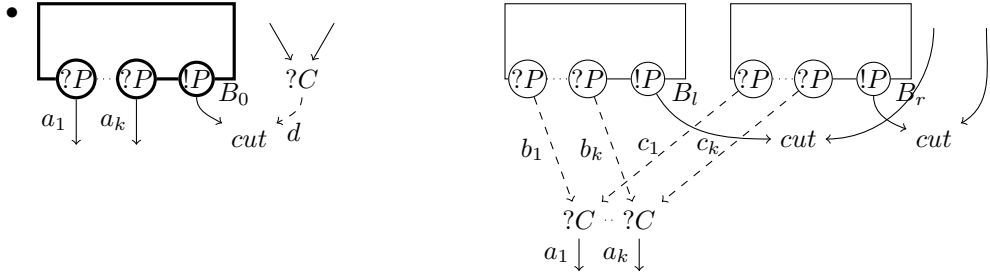
    \begin{tikzpicture}[baseline=0.8cm]
      \tikzstyle{door}=[draw, circle, inner sep=0]
      
      \begin{scope}[shift={(0,0)}]
        \draw (0,0) node (bang) [princdoor, very thick] {};
        \draw (bang) ++ (-0.8,0) node  [auxdoor, very thick] (auxd) {};
        \draw (auxd) ++ (-0.8,0) node [auxdoor, very thick] (auxg) {};
        \draw [very thick] (auxd) -- (bang) -|++ (0.5,1) --++(-2.6,0) |- (auxg);
        \draw (bang) ++ (0.5,-0.2) node {$B_0$};
        \draw [dotted] (auxg) -- (auxd);
        \draw (bang) ++ (0.7,-0.8) node (cut) {$cut$};
        \draw (bang) ++ (1.4,0) node (cont) {$?C$};
        \draw (cont) ++ (60:1) node (contd) {};
        \draw (cont) ++ (120:1)node (contg) {};
        \draw (auxg) ++ (0,-1) node (restg) {};
        \draw (auxd) ++ (0,-1) node (restd) {};
        \draw [->] (bang) to [out=-90, in=150] (cut);
        \draw [->, dashed, out=-90, in=30] (cont) to node [midway, below] {$d$} (cut);
        \draw [<-] (cont) -- (contg) node [midway, left]  {};
        \draw [<-] (cont)--(contd) node [midway, right] {};
        \draw [->] (auxg)--(restg) node [midway, left] {$a_1$};
        \draw [->] (auxd)--(restd) node [midway, left] {$a_k$};

        \draw (bang) ++ (6,0)  node (bang1) [door] {$!P$};
        \draw (bang1) ++ (-0.8,0) node  [door] (auxd1) {$?P$};
        \draw (auxd1) ++ (-0.8,0) node [door] (auxg1) {$?P$};
        \draw (auxd1) -- (bang1) -|++ (0.5,1) --++(-2.4,0) |- (auxg1);
        \draw (bang1) ++ (0.5,-0.2) node {$B_l$};
        \draw [dotted] (auxg1) -- (auxd1);
        
        \draw (bang1) ++ (3,0)  node (bang2) [door] {$!P$};
        \draw (bang2) ++ (-0.8,0) node  [door] (auxd2) {$?P$};
        \draw (auxd2) ++ (-0.8,0) node [door] (auxg2) {$?P$};
        \draw (auxd2) -- (bang2) -|++ (0.5,1) --++(-2.4,0) |- (auxg2);
        \draw (bang2) ++ (0.5,-0.2) node {$B_r$};
        \draw [dotted] (auxg2) -- (auxd2);
        
        \draw (auxg1) ++ (-70:2) node (contg) {$?C$};
        \draw (auxd1) ++ (-70:2) node (contd) {$?C$};
        \draw (contg) ++ (0,-0.8) node (restg) {};
        \draw (contd) ++ (0,-0.8) node (restd) {};
        \draw [->, dashed] (auxg1) -- (contg) node [midway, left] {$b_1$};
        \draw [->, dashed] (auxd1) -- (contd) node [midway, left] {$b_k$};
        \draw [->] (contg) -- (restg) node [midway, left] {$a_1$};
        \draw [->] (contd) -- (restd) node [midway, left] {$a_k$};
        \draw [->, dashed] (auxg2) -- (contg) node [midway, left] {$c_1$};
        \draw [->, dashed] (auxd2) -- (contd) node [midway, left] {$c_k$};
        \draw [dotted] (contg)--(contd);
        
        \draw (bang2) ++ (0.8,0.8) node (contg) {};
        \draw (bang2) ++ (1.5,0.8) node (contd) {};
        \draw (bang2) ++ (0.8,-0.8) node (cut2) {$cut$};
        \draw (bang1) ++ (2,-0.8) node (cut1) {$cut$};
        \draw [->] (bang1) to [out=-60, in=180] (cut1);
        \draw [<-, out=0, in=-90] (cut1) to node [near end, right] {} (contg);
        \draw [->, out=-90, in=160] (bang2) to (cut2);
        \draw [<-, out=30, in=-90] (cut2) to node [near end, right] {} (contd);
      \end{scope}
    \end{tikzpicture}
    
    The only particuliar case is for $\sigma(B_0)$. We set: $\phi(\sigma(B_0),P,\sigl(s)) = (\sigma(B_l), P, s)$ and $\phi(\sigma(B_0),P,\sigr(s)) = (\sigma(B_r), P, s)$. In the other cases, $(\sigma(B_0),P,s)$ will be outside the domain of $\phi$. Edges inside $B_0$ will be separated between $B_l$ and $B_r$ according to the exponential signature corresponding to $B_0$ in their potential: $\phi(e,P.\sigl(t)@Q,u)=(e_l,P.t@Q,u)$ and $\phi(e,P.\sigr(t)@Q,u)=(e_r,P.t@Q,u)$ (with $e_l$ and $e_r$ the edges corresponding to $e$ in $B_l$ and $B_r$).
    
    \begin{align*}
      T_G^1-T_H^1 &= \sum_{\substack{e \in D_\phi \\ P \in L_{\mapsto}(e)}}1-|\Set{ t }{ \phi(e,P,t)=(\_,\_, \sige) }|+ \sum_{e\in  \overline{D_{\phi}}}|L_{\mapsto}(e)| -\sum_{f \in  \overline{D_{\phi}'}}|L_{\mapsto}(e)| \\
      T_G^1-T_H^1 &= \left( \sum_{P \in L_{\mapsto}(\sigma(B_0))}1-2 \right )+ |L_{\mapsto}(d)| -\sum_{i=1}^{i \leq k}|L_{\mapsto}(b_i)|-\sum_{i=1}^{i \leq k}|L_{\mapsto}(c_i)| \\
      T_G^1-T_H^1 &= -|L_{\mapsto}(\sigma(B_0))| + |L_{\mapsto}(\sigma(B_0))| -2 \cdot k |L_{\mapsto}(\sigma(B_0))| \\
      T_G^1-T_H^1 &= -2\cdot k \cdot |L_{\mapsto}(\sigma(B_0))| \\
    \end{align*}
    
    \begin{align*}
      T_G^2 - T_H^2 & =  \sum_{\substack{\sigma(B) \in D_{\phi} \\ P \in L_{\mapsto}(B) \\ s \in Si_{\mapsto}(B,P)\\ (f,Q,t)=\phi(B,P,s)}}D_G(B).|s|-D_H(f).|t| + \sum_{\substack{\sigma(B) \in \overline{D_{\phi}} \\ P \in L_{\mapsto}(B) \\ s \in Si_{\mapsto}(B,P)}}D_G(B).|s| \\
      T_G^2 - T_H^2 & =  \sum_{\substack{P \in L_{\mapsto}(\sigma(B_0)) \\ s \in Si_{\mapsto}(B_0,P)\\ (f,Q,t)=\phi(\sigma(B_0),P,s)}} (k+1).|s|-(k+1).(|s|-1) + \sum_{\substack{B \in \varnothing \\ P \in L_{\mapsto}(B) \\ s \in Si_{\mapsto}(B,P)}}D_G(B).|s| \\
      T_G^2 - T_H^2 & =  \sum_{\substack{P \in L_{\mapsto}(\sigma(B)) \\ s \in Si_{\mapsto}(B,P)}} (k+1) + 0 \\
      T_G^2 - T_H^2 & = (k+1) \cdot \sum_{P \in L_{\mapsto}(\sigma(B))}|Si_{\mapsto}(B,P)|
    \end{align*}
    \begin{align*}
      T_G - T_H & = -2\cdot k \cdot |L_{\mapsto}(\sigma(B))| + 2\cdot (k+1) \cdot  \sum_{P \in L_{\mapsto}(\sigma(B))}|Si_{\mapsto}(B,P)| \\
      T_G - T_H & \geq -2 \cdot k \cdot |L_{\mapsto}(\sigma(B))| + 2 \cdot (k+1) \cdot \sum_{P \in L_{\mapsto}(\sigma(B))}1 \\
      T_G - T_H & \geq 2 |L_{\mapsto}(\sigma(B)| > 0
    \end{align*}

  \item 
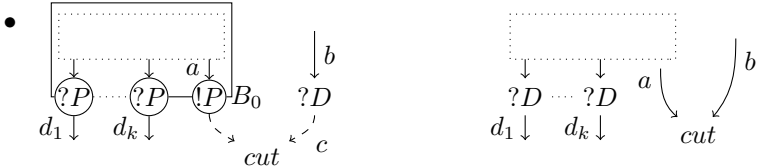
    \begin{tikzpicture}[baseline=0.9cm]
      \tikzstyle{door}=[draw, circle, inner sep=0]
        \draw (0,0) node (bang) [door] {$!P$};
        \draw (bang) ++ (-0.8,0) node  [door] (auxd) {$?P$};
        \draw (auxd) ++ (-1,0) node [door] (auxg) {$?P$};
        \draw (auxd) -- (bang) -|++ (0.3,1.25) --++(-2.4,0) |- (auxg);
        \draw [dotted] (auxg) -- (auxd);
        \draw (bang) ++ (0.5,0) node {$B_0$};
        \draw (bang) ++ (0.7,-0.8) node (cut) {$cut$};
        \draw (bang) ++ (1.4,0) node (der) {$?D$};
        \draw (der) ++ (0,1)node (dered) {};
        \draw (auxg) ++ (0,-0.7) node (restg) {};
        \draw (auxd) ++ (0,-0.7) node (restd) {};
        \draw [->,out=-90, in=150, dashed] (bang) to (cut);
        \draw [->,out=-90, in= 30, dashed] (der) to node [edgename,below right] {$c$} (cut);
        \draw [<-] (der) -- (dered) node [midway, right] {$b$};
        \draw [->] (auxg)--(restg) node [edgename] {$d_1$};
        \draw [->] (auxd)--(restd) node [edgename] {$d_k$};
        \draw [<-] (bang) --++ (0,0.5) node (hbang) {} node [midway, left] {$a$};
        \draw [<-] (auxd) --++ (0,0.5) node (hauxd) {};
        \draw [<-] (auxg) --++ (0,0.5) node (hauxg) {};
        \draw [dotted] (hbang) -|++ (0.2,0.6) -| ($(hauxg)+(-0.2,0)$) -- cycle;
        
        \draw (auxg) ++ (6,0) node (derg) {$?D$};
        \draw (auxd) ++ (6,0) node (derd) {$?D$};
        \draw (bang) ++ (6,0) node (bang) {};
        \draw (derg) ++ (0,-0.7) node (restg) {};
        \draw (derd) ++ (0,-0.7) node (restd) {};
        \draw [->] (derg)--(restg) node [edgename]{$d_1$};
        \draw [->] (derd)--(restd) node [edgename]{$d_k$};
        \draw [dotted] (derg)--(derd);
        \draw [<-] (derg) --++ (0,0.5) node (hderg) {};
        \draw [<-] (derd) --++ (0,0.5) node (hderd) {};        
        \draw (bang) ++ (0,0.5) node (hbang) {};
        \draw [dotted] (hbang) -|++ (0.2,0.6) -| ($(hderg)+(-0.2,0)$) -- cycle;
        \draw (hbang)++(0.5,-1) node (cut) {$cut$};
        \draw (hbang)++(1,0.4) node (rest) {};
        \draw [->, out=-90, in=140] (hbang) to node [near start, left] {$a$} (cut);
        \draw [->, out=-90, in= 50] (rest) to node [near start, right] {$b$} (cut);
    \end{tikzpicture}

    The only edges of $E_G$ which are not in $D_\phi$ are the premises of the $cut$. For the edges in $B_0$, we simply delete the exponential signature corresponding to $B_0$ in their potential: $\phi(e,P.t@Q,u)=(e,P@Q)$. Moreover, for every $(c,P) \in Pot(c)$, $(c,P,[\oc_{\sige}],-) \in F_H$ so $(c,P,[\oc_{\sige}],-) \in D_{\psi}$. We set $\psi(c,P,[\oc_{\sige}],-)= \{(d_1,P,[\oc_{\sige}],-),\cdots,(d_k,P,[\oc_{\sige}],-)\}$.
    \begin{align*}
      T_G^1-T_H^1 &= \sum_{\substack{e \in D_\phi \\ P \in L_{\mapsto}(e)}}1-|\Set{ t }{ \phi(e,P,t)=(\_,\_, \sige) }|+ \sum_{e\in  \overline{D_{\phi}}}|L_{\mapsto}(e)| -\sum_{f \in  \overline{D_{\phi}'}}|L_{\mapsto}(e)| \\
      T_G^1-T_H^1 &= \sum_{\substack{e \in D_\phi \\ P \in L_{\mapsto}(e)}}(1-1)+ 2 \cdot |L_{\mapsto}(\sigma(B))| -0 \\
      T_G^1-T_H^1 &= 2 \cdot |L_{\mapsto}(\sigma(B))|
    \end{align*}
    
    \begin{align*}
      T_G^2 - T_H^2 & =  \sum_{\substack{\sigma(B) \in D_{\phi} \\ P \in L_{\mapsto}(B) \\ s \in Si_{\mapsto}(B,P)\\ (f,Q,t)=\phi(B,P,s)}}D_G(B).|s|-D_H(f).|t| + \sum_{\substack{\sigma(B) \in \overline{D_{\phi}} \\ P \in L_{\mapsto}(B) \\ s \in Si_{\mapsto}(B,P)}}D_G(B).|s| \\
      T_G^2 - T_H^2 & =  \sum_{\substack{\sigma(B) \in D_{\phi} \\ P \in L_{\mapsto}(B) \\ s \in Si_{\mapsto}(B,P)\\ (f,Q,t)=\phi(B,P,s)}}0 + \sum_{\substack{P \in L_{\mapsto}(\sigma(B)) \\ s \in Si_{\mapsto}(B,P)}}D_G(\sigma(B)).|s| \\
      T_G^2 - T_H^2 & =  \sum_{\substack{P \in L_{\mapsto}(\sigma(B)) \\ s \in Si_{\mapsto}(B,P)}}D_G(\sigma(B)).|s| \\
    \end{align*}
    \begin{align*}
      T_G- T_H &= 2 |L_{\mapsto}(\sigma(B))| + 2  \sum_{\substack{P \in L_{\mapsto}(\sigma(B)) \\ s \in Si_{\mapsto}(B,P)}}D_G(\sigma(B)).|s| \\
      T_G- T_H &> 0
    \end{align*}
    
  \item 
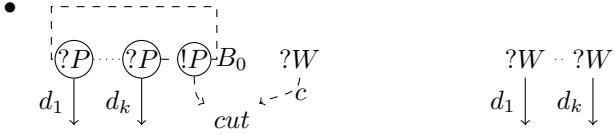
    \begin{tikzpicture}[baseline=0.6cm]
      \tikzstyle{door}=[draw, circle, inner sep=0]
        \draw (0,0) node (bang) [door] {$!P$};
        \draw (bang) ++ (0.5,0) node {$B_0$};
        \draw (bang) ++ (-0.7,0) node  [door] (auxd) {$?P$};
        \draw (auxd) ++ (-0.9,0) node [door] (auxg) {$?P$};
        \draw[dashed] (auxd) -- (bang) -|++ (0.3,0.7) --++(-2.2,0) |- (auxg);
        \draw [dotted] (auxg) -- (auxd);
        \draw (bang) ++ (0.5,-0.8) node (cut) {$cut$};
        \draw (bang) ++ (1.4,0) node (weak) {$?W$};
        \draw (auxg) ++ (0,-1) node (restg) {};
        \draw (auxd) ++ (0,-1) node (restd) {};
        \draw [->, dashed] (bang) to [out=-90, in=150] (cut);
        \draw [->, dashed] (weak) to [out=-90, in=30] node [edgename,right] {$c$} (cut);
        \draw [->] (auxg)--(restg) node [edgename] {$d_1$};
        \draw [->] (auxd)--(restd) node [edgename] {$d_k$};
        
        \draw (auxg) ++ (6,0) node (weakg) {$?W$};
        \draw (auxd) ++ (6,0) node (weakd) {$?W$};
        \draw (weakg) ++ (0,-1) node (restg) {};
        \draw (weakd) ++ (0,-1) node (restd) {};
        \draw [->] (weakg)--(restg) node [edgename] {$d_1$};
        \draw [->] (weakd)--(restd) node [edgename] {$d_k$};
        \draw [dotted] (weakg)--(weakd);
    \end{tikzpicture}
    
    Here, many edges of $E_G$ are not in $D_\phi$: the premises of the $cut$ and the edges inside the deleted box. For every $(c,P) \in Pot(e)$ and $T \in Tra$, $(c,P,[\oc_{\sige}]@T,-) \in F_G$ so $(c,P,[\oc_{\sige}]@T,-) \in D_\psi$. We set $\psi(c,P,[\oc_{\sige}]@T,-)=\{(d_1,P,[\oc_{\sige}]@T,-),\cdots,(d_k,P,[\oc_{\sige}]@T,-)\}$.

    \begin{align*}
      T_G^1-T_H^1 &= \sum_{\substack{e \in D_\phi \\ P \in L_{\mapsto}(e)}}1-|\Set{ t }{ \phi(e,P,t)=(\_,\_, \sige) }|+ \sum_{e\in  \overline{D_{\phi}}}|L_{\mapsto}(e)| -\sum_{f \in  \overline{D_{\phi}'}}|L_{\mapsto}(e)| \\
      T_G^1-T_H^1 &= \sum_{\substack{e \in D_\phi \\ P \in L_{\mapsto}(e)}}(1-1)+ \sum_{e\in  \overline{D_{\phi}}}|L_{\mapsto}(e)| -\sum_{f \in \varnothing}|L_{\mapsto}(e)| \\
      T_G^1-T_H^1 & =  \sum_{e\in  \overline{D_{\phi}}}|L_{\mapsto}(e)|\\
      T_G^1-T_H^1 & > 0
    \end{align*}
    
    \begin{align*}
      T_G^2 - T_H^2 & =  \sum_{\substack{\sigma(B) \in D_{\phi} \\ P \in L_{\mapsto}(B) \\ s \in Si_{\mapsto}(B,P)\\ (f,Q,t)=\phi(B,P,s)}}D_G(B).|s|-D_H(f).|t| + \sum_{\substack{\sigma(B) \in \overline{D_{\phi}} \\ P \in L_{\mapsto}(B) \\ s \in Si_{\mapsto}(B,P)}}D_G(B).|s| \\
      T_G^2 - T_H^2 & =  \sum_{\substack{\sigma(B) \in D_{\phi} \\ P \in L_{\mapsto}(B) \\ s \in Si_{\mapsto}(B,P)\\ (f,Q,t)=\phi(B,P,s)}}D_G(B).|s|-D_G(B).|t| + \sum_{\substack{P \in L_{\mapsto}(\sigma(B_0)) \\ s \in Si_{\mapsto}(B,P_0)}}D_G(\sigma(B_0)).|s| \\
      T_G^2-T_H^2 & =  0 + \sum_{\substack{P \in L_{\mapsto}(\sigma(B_0)) \\ s \in Si_{\mapsto}(B_0,P)}}D_G(\sigma(B_0)).|s| )\\
      T_G^2-T_H^2 & > 0
    \end{align*}
    
    Donc $T_G > T_H$
  \item 
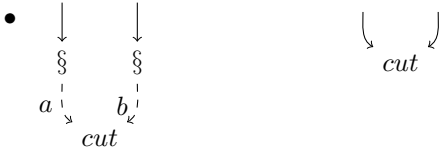
    \begin{tikzpicture}[baseline=1.5cm]
      \tikzstyle{door}=[draw, circle, inner sep=0]
      \begin{scope}
        \draw (0,0) node (cut) {$cut$};
        \draw (cut) ++ (-0.5,1) node (neutg) {$\S$};
        \draw (cut) ++ ( 0.5,1) node (neutd) {$\S$};
        \draw [<-] (neutg) --++ (0,0.8) node (restg) {};
        \draw [<-] (neutd) --++ (0,0.8) node (restd) {};
        \draw [->,out=-90, in=150, dashed] (neutg) to node [midway, left] {$a$} (cut);
        \draw [->,out=-90, in=30, dashed] (neutd) to node [midway, left] {$b$} (cut);
        
        \draw (restg) ++ (4,0) node (restg) {};
        \draw (restd) ++ (4,0) node (restd) {};
        \draw (cut) ++ (4,1) node (cut) {$cut$};
        \draw [->] (restg) to [out=-90, in =150] (cut);
        \draw [->] (restd) to [out=-90, in= 30] (cut);
      \end{scope}
    \end{tikzpicture}
    
    It is quite similar to the $\parr$/$\otimes$ case: we just delete two edges.
    \begin{align*}
      T_G^1-T_H^1 &= \sum_{\substack{e \in D_\phi \\ P \in L_{\mapsto}(e)}}1-|\Set{ t }{ \phi(e,P,t)=(\_,\_, \sige) }|+ \sum_{e\in  \overline{D_{\phi}}}|L_{\mapsto}(e)| -\sum_{f \in  \overline{D_{\phi}'}}|L_{\mapsto}(e)| \\
      T_G^1-T_H^1 &= \sum_{\substack{e \in D_\phi \\ P \in L_{\mapsto}(e)}}(1-1)+ \sum_{e\in  \overline{D_{\phi}}}|L_{\mapsto}(e)| -\sum_{f \in \varnothing \cap \overline{D_{\phi}'}}|L_{\mapsto}(e)| \\
      T_G^1-T_H^1 &= 0+  |L_{\mapsto}(a)|+|L_{\mapsto}(b)|- 0 \\
      T_G^1-T_H^1 &> 0 \\
    \end{align*}

    \begin{align*}
      T_G^2 - T_H^2 & =  \sum_{\substack{\sigma(B) \in D_{\phi} \\ P \in L_{\mapsto}(B) \\ s \in Si_{\mapsto}(B,P)\\ (f,Q,t)=\phi(B,P,s)}}D_G(B).|s|-D_H(f).|t| + \sum_{\substack{\sigma(B) \in \overline{D_{\phi}} \\ P \in L_{\mapsto}(B) \\ s \in Si_{\mapsto}(B,P)}}D_G(B).|s| \\
      T_G^2 - T_H^2 & =  \sum_{\substack{\sigma(B) \in D_{\phi} \\ P \in L_{\mapsto}(B) \\ s \in Si_{\mapsto}(B,P)\\ (f,Q,t)=\phi(B,P,s)}}0 + \sum_{\substack{B \in \varnothing \\ P \in L_{\mapsto}(B) \\ s \in Si_{\mapsto}(B,P)}}D_G(B).|s| \\
      T_G^2 - T_H^2 & = (0 + 0) = 0
    \end{align*}
    
    So $T_G > T_H$. The quantifier case and the axiom case are the same, so we will not present them.

  \item 
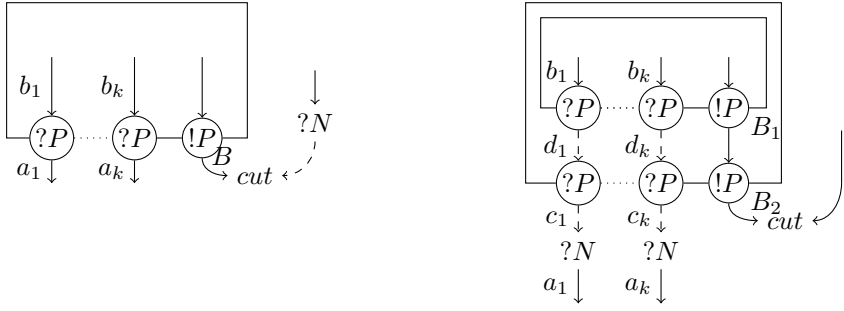
    \begin{tikzpicture}[baseline=1.5cm]
        \draw (0,1) node [inner xsep=1cm] (G) {};
        \draw (G)++(-1 ,-1.2) node [auxdoor] (whyint1) {};
        \draw (G)++(0.1,-1.2) node [auxdoor] (whyint2) {};
        \draw (G)++(1,-1.2) node [princdoor] (bangint)   {};
        \draw[ar] (whyint1 |- G.south) -- (whyint1) node [midway, left] {$b_1$};
        \draw[ar] (whyint2 |- G.south) -- (whyint2) node [midway, left] {$b_k$};
        \draw[ar] (bangint |- G.south) -- (bangint);
        \draw (whyint2)--(bangint) -| ++(0.6,1.8) -| ($(whyint1)+(-0.6,0)$) -- (whyint1);
        \draw [dotted] (whyint1) -- (whyint2); 
        \draw (bangint) node [below right] {$B$};
        
        \draw (bangint) ++ (1.5,0.2) node [dig] (dig) {};
        \draw (bangint) ++(0.7,-0.5) node [cut] (cut) {};
        \draw [ar] (bangint) to [out=-90, in=180] (cut);
        \draw [ar,out=-90,in=0,dashed] (dig) to (cut);
        \draw [revar] (dig) --++ (0,0.7);
        \draw [ar] (whyint1) --++(0,-0.6) node [midway, left] {$a_1$};
        \draw [ar] (whyint2) --++(0,-0.6) node [midway, left] {$a_k$};
      \begin{scope}[shift={(7,0)}]
        \draw (0,1) node [inner xsep=1cm] (G) {};
        \draw (G)++(-1 ,-0.8) node [auxdoor] (whyint1) {};
        \draw (G)++(0.1,-0.8) node [auxdoor] (whyint2) {};
        \draw (G)++(1,-0.8) node [princdoor] (bangint)   {};
        \draw (bangint) ++ (0.5,0) node [below] {$B_1$};
        
        \draw[ar] (whyint1 |- G.south) -- (whyint1) node [midway, left] {$b_1$};
        \draw[ar] (whyint2 |- G.south) -- (whyint2) node [midway, left] {$b_k$};
        \draw[ar] (bangint |- G.south) -- (bangint);
        \draw (whyint2)--(bangint) -| ++(0.5,1.2) -| ($(whyint1)+(-0.5,0)$) -- (whyint1);
        \draw [dotted] (whyint1) -- (whyint2); 
        
        \draw (whyint1) ++(0,-1) node [auxdoor] (whyext1){};
        \draw (whyint2) ++(0,-1) node [auxdoor] (whyext2) {};
        \draw (bangint) ++(0,-1) node [princdoor] (bangext) {};
        \draw (bangext) ++ (0.5,0) [below] node {$B_2$};
        \draw[ar, dashed] (whyint1) -- (whyext1) node [midway, left] {$d_1$};
        \draw[ar, dashed] (whyint2) -- (whyext2) node [midway, left] {$d_k$};
        \draw[ar] (bangint) -- (bangext);
        \draw [dotted] (whyext1)--(whyext2);
        \draw (whyext2)--(bangext) -| ++(0.7,2.4) -| ($(whyext1)+(-0.7,0)$) -- (whyext1);
        \draw [dotted] (whyint1)--(whyint2);
        \draw (bangext) ++ (1.5,0) coordinate (dig);
        \draw ($(bangext)!0.5!(dig)$) ++(0,-0.5) node [cut] (cut) {};
        \draw [ar] (bangext) to [out=-90, in=180] (cut);
        \draw [ar,out=-90,in=0] (dig) to (cut);
        \draw (dig) --++ (0,0.7);
        \draw (whyext1) ++(0,-0.9) node [dig] (dig1) {};
        \draw (whyext2) ++(0,-0.9) node [dig] (digk) {};
        \draw [ar, dashed] (whyext1)--(dig1) node [midway, left] {$c_1$};
        \draw [ar, dashed] (whyext2)--(digk) node [midway, left] {$c_k$};
        \draw [ar] (dig1) --++(0,-0.7) node [midway, left] {$a_1$};
        \draw [ar] (digk) --++(0,-0.7) node [midway, left] {$a_k$};
      \end{scope}
    \end{tikzpicture}

    In this case the images of $(\sigma(B),P,t)$ will depend on $t$: $\phi(\sigma(B),P,\sigp(t))=(\sigma(B_2),P,t)$ and $\phi(\sigma(B),P,\sign(t_1,t_2))= (\sigma(B_1),P.t_2,t_1)$. In the other cases, $(\sigma(B_0),P,s)$ will be outside the domain of $\phi$. Similarly, for the edges inside $B$, their exponential signature are divided in two: $\phi(e,P.\sign(t_1,t_2)@Q,u)=(e,P.t_1.t_2@Q,u)$.
    
    \begin{align*}
      T_G^1-T_H^1 &= \sum_{\substack{e \in D_\phi \\ P \in L_{\mapsto}(e)}}1-|\Set{ t }{ \phi(e,P,t)=(\_,\_, \sige) }|+ \sum_{e\in  \overline{D_{\phi}}}|L_{\mapsto}(e)| -\sum_{f \in \overline{D_{\phi}'}}|L_{\mapsto}(e)| \\
      T_G^1-T_H^1 &= 0 - |L_{\mapsto}(B)|+ 0 -\sum_{i=1}^k|L_{\mapsto}(c_i)| -\sum_{i=1}^k|L_{\mapsto}(d_i)|\\
      T_G^1-T_H^1 &= -|L_{\mapsto}(B)| -|L_{\mapsto}(B_2)|.k - |L_{\mapsto}(B_1)|.k\\
    \end{align*}
    
    \begin{align*}
      T_G^2 - T_H^2 & =  \sum_{\substack{\sigma(B) \in D_{\phi} \\ P \in L_{\mapsto}(B) \\ s \in Si_{\mapsto}(B,P)\\ (f,Q,t)=\phi(B,P,s)}}D_G(B).|s|-D_H(f).|t| + \sum_{\substack{\sigma(B) \in \overline{D_{\phi}} \\ P \in L_{\mapsto}(B) \\ s \in Si_{\mapsto}(B,P)}}D_G(B).|s| \\
      T_G^2 - T_H^2 & =  \sum_{\substack{P \in L_{\mapsto}(B) \\ n(t_1,t_2) \in Si_{\mapsto}(B,P)}}(k+1).(|t_1|+|t_2|+1)-(k+1).|t_1|+\sum_{\substack{P \in L_{\mapsto}(B) \\ p(s) \in Si_{\mapsto}(B,P)}}(k+1).(|t|+1)-(k+1).|s|\\
      T_G^2 - T_H^2 & =  \sum_{\substack{P \in L_{\mapsto}(B) \\ n(t_1,t_2) \in Si_{\mapsto}(B,P)}}(k+1).(|t_2|+1)+\sum_{\substack{P \in L_{\mapsto}(B_2) \\ s \in Si_{\mapsto}(B_2,P)}}(k+1))\\
      T_G^2 - T_H^2 & \geq  |L_{\mapsto}(B)|. (k+1).2+ \sum_{P.s \in L_{\mapsto}(B_1)}(k+1)\\
      T_G^2 - T_H^2 & \geq  |L_{\mapsto}(B)|.(k+1)+|L_{\mapsto}(B_2)|.(k+1) + |L_{\mapsto}(B_1)|.(k+1) \\
    \end{align*}
    
    \begin{align*}
      T_G-T_H & \geq -|L_{\mapsto}(B)| -|L_{\mapsto}(B_2)|.k - |L_{\mapsto}(B_1)|.k + 2 |L_{\mapsto}(B_2)|.(k+1)+2.|L_{\mapsto}(B)|.(k+1) + |L_{\mapsto}(B_1)|.(k+1) \\
      T_G-T_H & \geq (2k+1)\cdot |L_{\mapsto}(B)| + (k+2) \cdot |L_{\mapsto}(B_2)| + (k+2)\cdot |L_{\mapsto}(B_1)|\\
      T_G-T_H & \geq |L_{\mapsto}(B)| > 0      
    \end{align*}
    So $T_G > T_H$
    
  \item
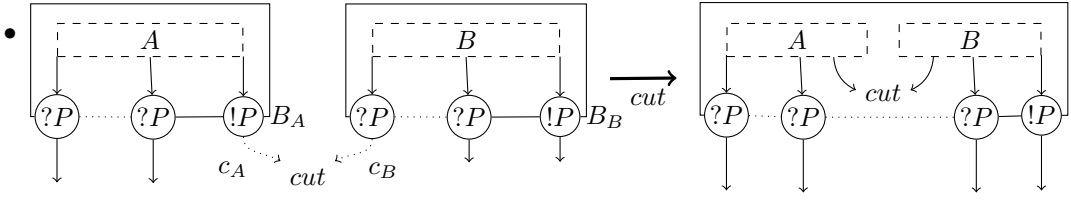
    \begin{tikzpicture}[baseline=1cm]
      \draw (-0.4,1) node [proofnet, inner xsep=1.1cm] (G) {$A$};
      \node [princdoor] (bang) at ($(G.-10)+(0,-0.8)$) {};
      \node [auxdoor] (auxd) at ($(G.-80)+(0,-0.8)$) {};
      \node [auxdoor] (auxg) at ($(G.-170)+(0,-0.8)$) {};
      \draw [ar] (G.-10)--(bang);
      \draw [ar] (G.-90)--(auxd);
      \draw [ar] (G.-170)--(auxg);
      \draw (auxd) -- (bang) -|++ (0.35,1.5) -| ($(auxg)+(-0.35,0)$) |- (auxg);
      \node at ($(bang)+(0.6,0)$) {$B_A$};
      \draw [dotted] (auxg) -- (auxd);
      \draw (auxg) ++ (0,-1) node (restg) {};
      \draw (auxd) ++ (0,-1) node (restd) {};
      \draw [ar] (auxg)--(restg);
      \draw [ar] (auxd)--(restd);

      \draw ($(G)+(4.2,0)$) node [proofnet, inner xsep=1.1cm] (H) {$B$};
      \node [princdoor] (bangh) at ($(H.-10)+(0,-0.8)$) {};
      \node [auxdoor] (auxdh) at ($(H.-80)+(0,-0.8)$) {};
      \node [auxdoor] (auxgh) at ($(H.-170)+(0,-0.8)$) {};
      \draw [ar] (H.-10)--(bangh);
      \draw [ar] (H.-90)--(auxdh);
      \draw [ar] (H.-170)--(auxgh);
      \draw (auxdh) -- (bangh) -|++ (0.35,1.5) -| ($(auxgh)+(-0.35,0)$) |- (auxgh);
      \node at ($(bangh)+(0.6,0)$) {$B_B$};
      \draw [dotted] (auxgh) -- (auxdh);
      \draw [ar] (auxdh) --++ (0,-0.6) node (restdh) {};
      \draw [ar] (bangh) --++ (0,-0.6) node (restgh) {};
      
      \node [cut] (cut) at ($(bang)!0.5!(auxgh)+(0,-0.8)$) {};
      \draw [ar,dotted] (bang) to [out=-90, in=150] node [edgename,below left] {$c_A$} (cut);
      \draw [ar,dotted] (auxgh) to [out=-90, in=30] node [edgename,below right]{$c_B$} (cut);
      
      \nvar{\decder}{5.2cm}
      \node [proofnet, inner xsep=0.8cm] (Gr) at ($(H)+(4.4,0)$)    {$A$};
      \node [proofnet, inner xsep=0.8cm] (Hr) at ($(Gr)+(2.3,0)$) {$B$};
      \node [auxdoor]   (auxgd) at ($(Gr. -70)+(0,-0.8)$) {};
      \node [auxdoor]   (auxgg) at ($(Gr.-168)+(0,-0.8)$) {};
      \node [princdoor] (bangr) at ($(Hr. -12)+(0,-0.8)$) {};
      \node [auxdoor]   (auxdd) at ($(Hr. -70)+(0,-0.8)$) {};

      \node [cut] (cutr)  at ($(Gr.-12)!0.5!(Hr.-168)+(0,-0.5)$) {};
      \draw [ar] ($(Hr.-155)$) to [bend left]  (cutr);
      \draw [ar] ($(Gr. -25)$)  to [bend right] (cutr);
      
      \draw (auxdd) -- (bangr) -| ++(0.35,1.5) -| ($(auxgg)+(-0.35,0)$) -- (auxgg);
      \draw [dotted] (auxgg)--(auxgd)--(auxdd);
      \draw [ar] (Gr.-168) -- (auxgg);
      \draw [ar] (Gr. -80) -- (auxgd);
      \draw [ar] (Hr. -80) -- (auxdd);
      \draw [ar] (Hr. -12) -- (bangr);
      \draw [ar] (auxgg) --++ (0,-1) node (restgg) {};
      \draw [ar] (auxgd) --++ (0,-1) node (restgd) {};
      \draw [ar] (auxdd) --++ (0,-1) node (restdd) {};
      \draw [ar] (bangr) --++ (0,-1) node (restbang) {};
      \draw [->, very thick] ($(bangh)!0.3!(auxgg)+(0,0.5)$) -- ($(bangh)!0.7!(auxgg)+(0,0.5)$) node [below left] {$cut$};
    \end{tikzpicture}

    Here, two edges are deleted, but $T_G^2-T_H^2$ is not as simple to compute as in the $\parr$/$\otimes$ or the $\S$ case because two boxes fuse. 
    \begin{align*}
      T_G^1-T_H^1 &= \sum_{\substack{e \in D_\phi \\ P \in L_{\mapsto}(e)}}1-|\Set{ t }{ \phi(e,P,t)=(\_,\_, \sige) }|+ \sum_{e\in  \overline{D_{\phi}}}|L_{\mapsto}(e)| -\sum_{f \in \overline{D_{\phi}'}}|L_{\mapsto}(e)| \\
      T_G^1-T_H^1 &= 0 + |L_{\mapsto}(c_A)|+|L_{\mapsto}(c_B)| - 0\\
      T_G^1-T_H^1 &> 0 \\
    \end{align*}
    
    \begin{align*}
      T_G^2 - T_H^2 & =  \sum_{\substack{\sigma(B) \in D_{\phi} \\ P \in L_{\mapsto}(B) \\ s \in Si_{\mapsto}(B,P)\\ (f,Q,t)=\phi(B,P,s)}}D_G(B).|s|-D_H(f).|t| + \sum_{\substack{\sigma(B) \in \overline{D_{\phi}} \\ P \in L_{\mapsto}(B) \\ s \in Si_{\mapsto}(B,P)}}D_G(B).|s| \\
      T_G^2 - T_H^2 & =  \sum_{\substack{P \in L_{\mapsto}(B_B) \\ s \in Si_{\mapsto}(B_B,P)}}(k_B+1).|s|-(k_A+k_B+1).|s|+\sum_{\substack{P \in L_{\mapsto}(B_A) \\ s \in Si_{\mapsto}(B_A,P)}}(k_A+1).|s|\\
      T_G^2 - T_H^2 & =  \sum_{\substack{P \in L_{\mapsto}(B) \\ s \in Si_{\mapsto}(B,P)}}|t|\\
      T_G^2 - T_H^2 & >  0 \\
    \end{align*}

    So $T_G-T_H>0$.
  \end{itemize}
\end{proof}

\begin{coro}\label{lemma_tg_finite}
  If $G$ is a proof-net, then $G$ has positive weights and $T_G$ is finite. 
\end{coro}
\begin{proof}
  We first prove that whenever $G$ is normal with respect to $cut$-elimination, $G$ has positive weights and $T_G$ is finite. The proof-net has no $cut$, so for every $(B,P) \in Pot(B_G)$ and $t \in Sig$, the paths beginning by $(\sigma(B),P,[\oc_t],+)$ are always going downwards, in particular we never cross a $\wn C$ or $\wn N$ link upwards. So for every $(B,P) \in Pot(B_G)$, $C_{\mapsto}(B,P) = \{\sige\}$. So $G$ has positive weights and for every $e \in E_G$, $L_{\mapsto}(e)=1$. Thus, $T_G=|E_G|+\sum_{B \in B_G}D_G(B)$. $T_G$ is finite.

  Then, let us consider any proof-net. According to \cite{girard1987linear}, proof-nets strongly normalize so we can consider a sequence $G \rightarrow_{cut} G_1 \rightarrow_{cut} G_n \not \rightarrow$. We know that $G_n$ is normal so $G_n$ has positive weights and $T_{G_n}$ is finite. We can use Theorem \ref{dallago_weight}, to prove that $G_{n-1}$ has positive weights and $T_{G_{n-1}}$ is finite, so $G_{n-2}$ has positive weights and $T_{G_{n-2}}$ is finite, and so on. So $G$ has positive weights and $T_G$ is finite.
\end{proof}

\begin{coro}
\label{coro_dal_lago_theo}
If $G$ is a proof-net, then the length of any path of reduction is bounded by $T_G$
\end{coro}
\begin{proof}
Suppose $G$ is a proof-net and $G \mapsto G_1 \mapsto \cdots \mapsto G_n$. Then, $G_n$ is a proof-net (stability of proof nets with respect to $cut$-elimination). According to Lemma \ref{dallago_weight}, $T_{G} > T_{G_1} > \cdots > T_{G_n}  \geq 0$. So $n < T_G$.
\end{proof}

In~\cite{lago2006context}, the acyclicity of proof nets is proved along with the decrease of $T_G$. Here, we chose to separate the two results because the acyclicity needs to define another relation. We define $\Rightarrow_{cut}$ on proof-nets, which is the $\rightarrow$ relation, where the $?W$ rewriting steps are allowed only if all the cuts of the nets are $?W/!P$ cuts.

\begin{lemma} \label{lemma_weak_acyclic}
  If all the cuts of $G$ are $?W/!P$ cuts, then $G$ is acyclic
\end{lemma}
\begin{proof}
  Let us suppose that $(e,P,[!_t],p) \mapsto^+ (e,P,[!_u],p)$. By construction of proof-nets, there must be at least a change of direction to go back to the same edge. We return with the same direction. So there must have been at least two changes of direction. The $+$ to $-$ change has been done by crossing a cut, $(f,Q,T,+) \mapsto (g,Q,T,-)$. So either we go from a conclusion of a $!P$ link to a conclusion of a $?W$ link (in this case, the path can not continue so $(g,Q,T,-)=(e,P,[!_u],p)$, but it is impossible because $(g,Q,[!_t],-) \not \mapsto$) or we go from a conclusion of a $?W$ to a conclusion of a principal door (in this case this is the first step, so $(f,Q,T,+)=(e,P,[!_t],p)$, but it is impossible because $\not \mapsto (f,Q,[!_u],+)$). We have a contradiction, so there are no such cycle, $G$ is acyclic.
\end{proof}

\begin{lemma}
  \label{lemma_acyclicity} All proof-nets are acyclic 
\end{lemma}
\begin{proof}
  Let $G$ be a proof-net then, linear logic being strongly normalizing, there exists a sequence $G \Rightarrow G_1 \Rightarrow G_2 \Rightarrow \cdots \Rightarrow G_n \not \Rightarrow$. By lemma \ref{lemma_weak_acyclic}, $G_n$ is acyclic. We now have to show that if $G \Rightarrow H$ and $H$ is acyclic, then $G$ is  acyclic. 

If it is a $?W$ reduction, then the only cuts of $G$ are $?W/!P$ cuts, so $G$ is acyclic. 

Else, if $G$ has a cycle, then it must be on an edge $e$ which is not in the $D_\phi$ of the copymorphism associated to this cut elimination step or such that $\phi(e,P,t)$ depends on $t$, otherwise by rule \ref{copymorphism_path_cons} of the copymorphism definition, there would be a cycle in $H$. Almost all edges $e$ which are not in $D_\phi$ are ruled out because $(e,P,[!_t],p)\not \mapsto^2$. The only remaining possible edges are: the premises of the cut in a $?C$ sep (but it would mean one of the premises of the $?C$ link is also part of a cycle which is impossible), the premises of the cut in an $ax$ rule (but it would mean the other conclusion of the axiom is also part of a cycle which is impossible) or one of the premises of the cut in a $?N$ rule (but it would mean the premise of the $?N$ link is also part of a cycle which is impossible).
\end{proof}


\subsection{The $ztrees$ are finite}
\label{subsection_ztreefinite}
In subsection \ref{subsection_ztree}, we defined $ztrees$ of a potential edge as (potentially infinite) trees of substitutions. It allowed us to prove that the last contexts of paths corresponding to simplifications of copies are final contexts. Here, we will prove that those $ztrees$ are, in fact, always finite. It will allow us to define {\textbf the} formula associated to a context, as opposed to the set $\bigcup_{\Theta \vartriangleleft ztree(C)} \beta_{\{\},\Theta}(C)$ of formulae associated to the context. To define this formula, each time we are asked (in the definition of the previous set) to choose a truncation of a $ztree$, we will directly use the $ztree$ (it is a finite truncation because the $ztrees$ will be proved to be finite).

\begin{lemma}
\label{lemma_underlying_finite}
If $G \rightarrow_{cut} H$, let $(D_{\phi},D_{\phi}',\phi, \psi)$ be the copymorphism from $G$ to $H$ chosen in the proof of theorem \ref{dallago_weight} and $(e,P) \in Pot(D_{\phi})$, $t \in Sig$ and $(e',P',t')=\phi (e,P,t)$. 

If $ztree(e',P')$ is finite, then $ztree(e,P)$ is finite.
\end{lemma}

\begin{figure}
  \subfigure{\label{contraction_a}
    \begin{tikzpicture}
      [
        level 1/.style={sibling distance=2.5cm},
        level 2/.style={sibling distance=1cm}
      ]
      \node {} [level distance = 2.5cm]
      child{ node {}
        child{ node{} edge from parent node [above, sloped]  {$Z_1 \rightarrow B_1$} }
        child{ node{} edge from parent node [above, sloped] {$Z_2 \rightarrow B_2$} }
        edge from parent node[above, sloped]{$Z_4 \rightarrow B_4$}
      }
      child{ node {}
        child{node {} edge from parent node[above, sloped]{$Z_2 \rightarrow B_2$}}
        child{node {} edge from parent node[above, sloped]{$Z_1 \rightarrow B_1$}}
        edge from parent node[above, sloped]{$Z \rightarrow B$}      
      }
      child{ node {}
        child{ node {}
          child{ node{} edge from parent node [above, sloped]  {$Z_4 \rightarrow B_5$} }
          child{ node{} edge from parent node [above, sloped] {$Z_2 \rightarrow B_6$} }
          edge from parent node[above, sloped]{$Z \rightarrow B$}
        }
        edge from parent node[above, sloped]{$Z_3 \rightarrow B_3$}
      };
    \end{tikzpicture}
  }
  \subfigure{ \label{contraction_b}
    \begin{tikzpicture}
      [
        level 1/.style={sibling distance=2.5cm},
        level 2/.style={sibling distance=1cm}
      ]
      \node {} [level distance = 3.7cm]
      child{ 
        node {}
        child{ node{} edge from parent node [above, sloped]  {$Z_1 \rightarrow B_1$} }
        child{ node{} edge from parent node [above, sloped] {$Z_2 \rightarrow B_2$} }
        edge from parent node[above, sloped]{$Z_4 \rightarrow B_4$}
      }
      child{node {} edge from parent node[above, sloped]{$Z_2 \rightarrow B_2$}}
      child{node {} edge from parent node[above, sloped]{$Z_1 \rightarrow B_1$}}
      child{ node {}
        child{ node{} edge from parent node [above, sloped] {$Z_4 \rightarrow B_5$} }
        child{ node{} edge from parent node [above, sloped] {$Z_2 \rightarrow B_6$} }
        edge from parent node[above, sloped]{$Z_3 \rightarrow B_3[B/Z]$}
      };
    \end{tikzpicture}
  }
  \caption{\label{fig_example_contraction_ztree}In case of a cut between the $\forall$ link associated to the eigenvariable $Z$ and a $\exists$ link whose associated formula is $B$, we transform the left tree into the right tree}
\end{figure}
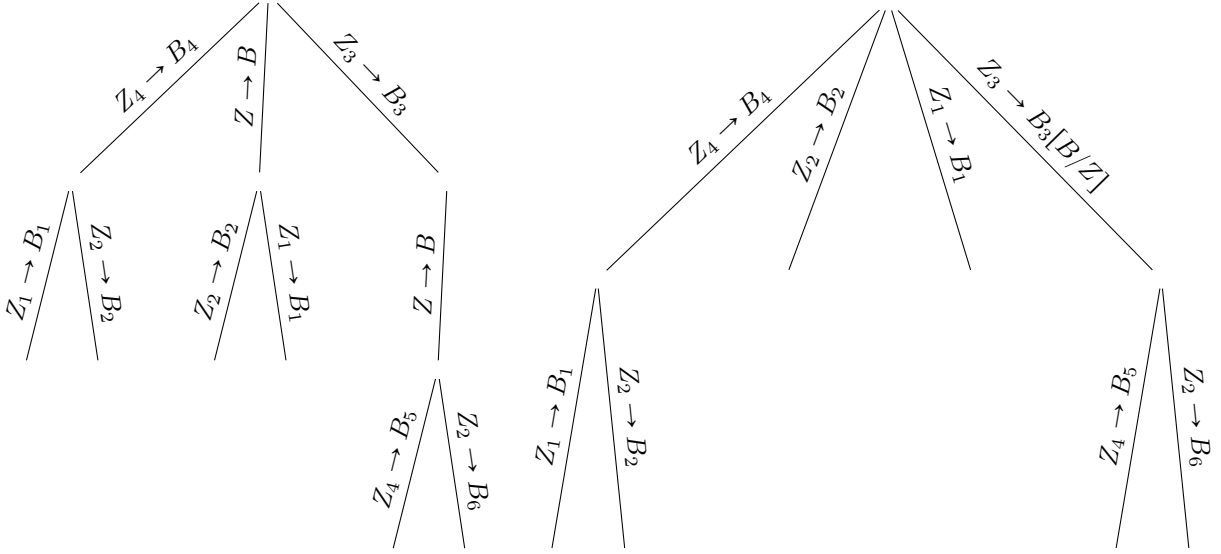

\begin{proof}
  We will only prove the statement in the case of a $\exists / \forall$ cut elimination. All the other cases are simpler. Let $Z$ be the eigenvariable corresponding to the reduced $\forall$ link and $B$ be the formula corresponding to the reduced $\exists$ link. Let $g$ be the conclusion of the $\exists$ link and $h$ its premise. We prove by coinduction that, when $\phi(e,P,t)=(e',P',t')$, $ztree(e',P')$ is equal to the tree obtained from $ztree(e,P)$ by contracting the branches whose label is a substitution on $Z$ (as shown in Figure \ref{fig_example_contraction_ztree}).

Let $[p_1; \cdots ; p_{\partial(e)}]=P$ and $[p'_1;\cdots;p_{\partial(e')}]=P'$. For $1\leq i \leq \partial(e)$ we define $P_i$ as $[p_1;\cdots;p_i]$, and for $1 \leq i \leq \partial(e')$ we define $P'_i$ as $[p'_1;\cdots;p'_i]$. Let $E$ (resp. $E'$) be the  subset of the free eigenvariables of $\beta(e)$ (resp. $\beta(e')$) whose associated $\forall$ link has a corresponding $\exists$ link,
\begin{align*}
  E=\Set{ Z_i }{\exists (g_i,R_i)\in Pot(E_G), (f_i,[p_1;\cdots ; p_{\partial (f_i)}],\forall,+) \mapsto^* (g_i,R_i,\forall,-) \text{ with }g_i\text{ the conclusion of a }\exists\text{ link $l_i$} }\\
  E'=\Set{ Z_i }{\exists (g'_i,R'_i)\in Pot(E_H), (f_i',[p_1';\cdots ; p'_{\partial (f'_i)}],\forall,+) \mapsto^* (g_i',R_i',\forall,-) \text{ with }g_i'\text{ the conclusion of a }\exists\text{ link $l_i$} }\\
\end{align*}
 with $f_i$ (resp. $f'_i$) the $\forall$ link associated to eigenvariable $Z_i$ in $G$ (resp. $H$). Then, by definition, 
 \begin{align*}
   ztree(e,P) = \Set{(Z_i,B_i,ztree(h_i,R_i))}{ Z_i \in E }\\
 \end{align*}
 with $h_i$ the premise of the $\exists$ link corresponding to $Z_i \in E$ and $B_i$ the formula associated to this $\exists$ link. We want to prove that if there exists $(h,R) \in Pot(E_G)$ such that $(Z,B,ztree(h,R)) \in ztree(e,P)$, then 
 \begin{equation*}
  ztree(e',P')= 
  \begin{array}{c}
    \Set{(Z_i,B_i,ztree(h_i',R_i'))}{ Z_i \in E \text{ and } Z_i \neq Z }\\
    \cup \\
    \Set{(Y,C,ztree(d',S'))}{ (Y,C,ztree(d,S)) \in ztree(h,R) }
  \end{array}
 \end{equation*}
 And else,
 \begin{equation*}
   ztree(e',P') = \Set{(Z_i,B_i,ztree(h_i',R_i'))}{ Z_i \in E }
 \end{equation*}
 
 We first consider the case where there exists $(h,R) \in Pot(E_G)$ such that $(Z,B,ztree(h,R)) \in ztree(e,P)$. We will consider the eigenvariables which are in $E$ and show that they are in $E'$ with their corresponding $\exists$ link for $(e',P')$ being the reduct of their corresponding $\exists$ link for $(e,P)$. Then, we will consider the elements of $ztree(h,R)$ and prove that they belong to $ztree(e',P')$. Finally we will prove that the other eigenvariables do not belong to $ztree(e',P')$. These three results put together, show that $ztree(e',P')$ is equal to the expected result.
 \begin{itemize}
  \item If $Z_i \in E$ and $Z_i \neq Z$, then $Z_i$ is a free variable of $\beta_H(e')$. The conclusion $(f_i,P_i)$ of the $\forall$ link associated to $Z_i$ in $G$ is in $D_{\phi}$. We know that $(f_i,P_i,[\forall],+) \mapsto^* (g_i,R_i,[\forall],-)$. So, $(f'_i,P'_i,[\forall],+) \mapsto^* (g'_i,R'_i,[\forall],-)$.




 The tail of $g'_i$ is an $\exists$ link whose associated formula is $B$. Its premise is $h'_i$, with $\phi(h_i,R_i,\sige)=(h'_i,R_i,\sige)$. So $(Z_i,B,ztree(h'_i,R'_i)) \in ztree(e',P')$.
  \item $Z \not \in FV(\beta(e'))$, however $Z$ being replaced by $B$ in the whole net, all free eigenvariables of $B$ are free variables of $\beta(e')$. By construction of the proof-nets, $Z$ can not be a free eigenvariable of $B$. So we can use the arguments of the first case to prove that for all $(Y,C,ztree(d,S)) \in ztree(h,R)$, ${(Y,C,ztree(d',S'))} \in ztree(e',P')$.
  \item If $Y \not \in E$ and $Y$ is not in the first branches of $ztree(h,R)$. So either $Y \not \in FV(\beta(e)) \cup FV(B)$ (so $Y \not \in FV(\beta(e'))$) or $(f_Y,R_Y,[\forall],+)$ does not end with a $\exists$ link. In this case, let $(f_Y,P_{\partial(f_Y)})$ be the conclusion of the $\exists$ link. If the path beginning by $(f_Y,R_Y,[\forall],+)$ does not end with a $\exists$ link, by Theorem \ref{lemma_underlying_mapsto}, this means that either it runs an infinite path or it arrives at a pending edge. In both cases, it means that if we set $(f_Y',R_Y',t'_Y)=\phi(f_Y,R_Y,t_Y)$, the path beginning by $(f_Y',R_Y',[\forall],+)$ does not end or arrives at a pending edge. So $Y$ is not in the first branches of $ztree(e',P')$.
\end{itemize}

If $ztree(e,P)$ is infinite, then there is an infinite branch path. The corresponding branch path in $ztree(e',P')$ is also infinite because at most one over two branch is contracted.
\end{proof}

\begin{definition}\label{def_underlyingformula}
  Let $C$ be a context, the {\em underlying formula} of $C$ (written $\beta(C)$) is the element of $\beta_{\{\}}(C)$ which we obtain by following the definition of $\beta_{\{\}}(C)$ and, whenever we have to choose a truncation of some $ztree(f,Q)$, we choose $ztree(f,Q)$ itself.
\end{definition}

\begin{theorem}\label{theorem_underlying_formula}
If $C \rightsquigarrow D$ and $\beta(C)$ is defined, then $\beta(D)$ is defined and $\beta(C)=\beta(D)$.
\end{theorem}
\begin{proof}
  In the induction proof, we observe that for the base case we can choose any truncation of $ztree(e,P)$ we want. So, in particular, we can choose $ztree(e,P)$ itself. The induction steps always extend the truncation, they never restrain it.
\end{proof}

\section{Stratification}
\label{section_stratification}
\subsection{History and motivations}
\paragraph{} A stratification designs a restriction of a framework, which forbids the contraction (or identification) of two objects belonging to two morally different ``strata''. Russell's paradox in naive set theory relies on the identification of two formulae which belong morally to different strata. The non-terminating $\lambda$-term $(\lambda x. x x)(\lambda x. x x)$ depends on the identification of an argument with the function duplicating it. In recursion theory, to create from the elementary sequences $\theta_m(n)=2^n_m$ (tower of exponential of height $m$ in $n$), the non elementary sequence $n \mapsto 2^n_n$, we also need to identify $n$ and $m$ which seem to belong to different strata. Stratification restrictions might be applied to those frameworks (naive set theory, linear logic, lambda calculus and recursion theory) to entail coherence or complexity properties~\cite{baillot2010linear}.

\paragraph{}The first example of a stratification condition in linear logic appears in~\cite{girard1995light}, though Girard did not use the word ``stratification'' at that time. Girard's inspiration came from a sharp analysis of Russell's paradox in naive set theory. This paradox needs the contraction of two formulae, the second being obtained from the first by the application of a ``specification rule''. Therefore, we can avoid the paradox if:
\begin{enumerate}
\item We index each formula in the sequents with a natural number (called the stratum of the formula)
\item The use of the specification rule on a formula increases its stratum 
\item We only allow contraction between formulae with the same stratum
\end{enumerate}

\paragraph{}Concretely, in~\cite{girard1995light} and~\cite{danos2003linear}, the stratification condition is ``use neither digging nor dereliction ($?N$ and $?D$ links)'' and is presented as a subsystem of linear logic, named $ELL$. Any proof-net of $ELL$ reduces to its normal form in a number of step bounded by an elementary function of its size. In~\cite{girard1995light} and~\cite{danos2003linear}, the stratum of an occurence of a formula in a proof is the depth of its corresponding edge (in the proof net corresponding to the proof) in terms of box inclusion. The name ``stratification'' is given in~\cite{danos2003linear} for this technique, but in this work the only kind of stratification considered is still the one where strata correspond to depths. In~\cite{baillot2010linear}, Baillot and Mazza present an analysis of the concept of stratification, and a generalization of the stratification of~\cite{girard1995light} and~\cite{danos2003linear}. Their stratification condition is enforced by a labelling of edges. It also enforces elementary time.

\paragraph{} In this paper, we present an even more general stratification. This generalization is not given by a new linear logic subsystem but by a criterion on proof-nets. Then, to prove that a system is elementary time sound, we only have to prove that all the proof-nets of the system satisfy the criterion. Here, we apply the criterion to $ELL$ and $L^3$, the only two linear logic subsystems discovered characterizing elementary time. However, if a better system was discovered, it might satisfy our criterion. To prove the soundness of this new system, we would only have to prove that it satisfies our criterion. Our work may simplify proofs of soundness of several systems by factoring out a common part.

\subsection{Stratification on $\lambda$-calculus}
Our definition of stratification is based on context semantics paths and may be difficult to grasp at first read. To motivate the criterion, we first state a criterion on $\lambda$-calculus, the formal system whose terms are generated by $\Lambda = x \mid \lambda x. \Lambda \mid \Lambda \Lambda$. Where $x$ ranges over a countable set of variables. Parentheses are added when a term is ambiguous. We think this criterion corresponds to the criteria on proof-nets. Unfortunately, we did not prove any statement precising this equivalence yet. Thus, the criterion on $\lambda$-term can only be taken as a guide for intuition. 

Let $t,t',u,u' \in \Lambda$ such that $t' \rightarrow_{\beta} u'$, $u$ is a subterm of $t$ and $u'$ is a subterm of $t'$. We say that $u'$ is a {\em residue} of $u$ if it is a ``copy by $\beta$-reduction'' of $u$ where, possibly, the free variables have been substituted. Complete definition can be found in appendix~\ref{def_residues}. Here we give two examples: 
\begin{itemize}
\item If $t=(\lambda x.x x) (\lambda y.y) (\lambda v. \lambda w. w) \rightarrow_\beta (\lambda y.y) (\lambda y.y) (\lambda v. \lambda w. w)= t'$. Then, the  residues of $\lambda y.y$ through $t \rightarrow_\beta t'$ are the two occurrences of $\lambda y.y$ in $t'$.
\item If $t=(\lambda x. \lambda y. x y) (\lambda z.z)$, $t'= \lambda y. (\lambda z.z) y$ and $t \rightarrow_\beta t'$. Then, the only residue of $\lambda y. x y$ through $t \rightarrow_\beta t'$ is the occurrence of $\lambda y. (\lambda z.z) y$ in $t'$.
\end{itemize}

We define ``hole-terms'' as $\lambda$-terms $h$ with a special variable $\circ$ which appears free exactly once in $h$. Then, if $t$ is a $\lambda$-term, $h[t]$ designs $h[t/\circ]$.

A $\lambda$-term is said stratified if the following $\twoheadrightarrow$ relation on subterms is acyclic. Intuitively, $v \twoheadrightarrow w$ if, during $\beta$-reduction a residue of $w$ will be applied to a term containing a residue of $v$. With the additional constraint that $v$ is not on the left of an application.
\begin{definition}
  Let $v$, $w$ be subterms of $t$, then $v \twoheadrightarrow w$ if there exists hole-terms $h_1,h_2$ and $\lambda$-terms $v',w'$ such that: $t \rightarrow_\beta^* h_1[w' h_2[v']]$, $v',w'$ are residues of $v,w$ along the $\beta$-reduction. With the additional constraint that this residue of $v$ is not applied to something, i.e. we do not have  $(v) (u)$ but either $(u) (v)$ or $\lambda x.v$.
\end{definition}

A $\lambda$-term is said stratified if $\twoheadrightarrow$ relation on subterms is acyclic. 

As an example, we can observe that $(\lambda x. x x) (\lambda y. y y)$ is not stratified, because $(\lambda y. y y) \twoheadrightarrow (\lambda y. y y)$. Indeed $t=(\lambda x. x x) \underbrace{(\lambda y. y y)}_{v=w} \rightarrow_\beta \underbrace{(\lambda y. y y)}_{w'} \underbrace{(\lambda y. y y)}_{v'}$

Similarly, let $\underline{n}= \lambda f. \lambda x. \underbrace{f(f(\cdots (f}_{n\text{ applications}} x)))$ be the Church-numeral corresponding to $n$ and $S=\lambda n. \lambda f. \lambda x. n f(f x)$ the successor on Church-numeral. Then, the $\lambda$-term $(\lambda n. n (\lambda a. \lambda k. k a (a \underline{1})) S n)  \underline{2}$, which represents the Ackermann function applied to $2$ is not stratified. Indeed, the following $\beta$-reduction sequence shows that $(a \underline{1}) \twoheadrightarrow (a \underline{1})$, tracking the residues of $(a \underline{1})$ with braces.

\begin{align*}
  &(\lambda n. n (\lambda a. \lambda k. k a \underbrace{(a \underline{1})}_{}) S n)  \underline{2} ~\rightarrow_{\beta}^* ~ 
  (\lambda a. \lambda k. k a \underbrace{(a \underline{1})}_{})   ( (\lambda a. \lambda k. k a \underbrace{(a \underline{1})}_{}) S  ) \underline{2} \rightarrow_{\beta}^* \\
  &(\lambda a. \lambda k. k a \underbrace{(a \underline{1})}_{})   (\lambda k. k S \underbrace{(S \underline{1})}_{} ) \underline{2} ~\rightarrow_{\beta}^*~
  (\lambda k. k (\lambda k. k S \underbrace{\underline{2}}) \underbrace{((\lambda k. k S \underbrace{\underline{2}}) \underline{1}))} \underline{2} \\
  &\underline{2} (\lambda k. k S \underbrace{\underline{2}}) \underbrace{(\underline{1} S \underbrace{\underline{2}})} ~\rightarrow_{\beta}^*~ 
  (\lambda k. k S \underbrace{\underline{2}}) ( (\lambda k. k S \underbrace{\underline{2}}) \underbrace{\underline{3}}) ~\rightarrow_{\beta}^*~\\&(\lambda k. k S \underbrace{\underline{2}}) (\underbrace{\underline{3}}_{w'} S \underbrace{\underline{2}}_{v'})
\end{align*}

We think that this stratification on $\lambda$-terms corresponds to the notion of stratification on proof-net which we will define in the next subsection.
\begin{conjecture} \label{conj_eq}
  Let $t$ be a $\lambda$-term typable in System F, and $G$ be the proof-net obtained by encoding the type derivation of $t$ in linear logic (by Girard's encoding, transforming $A \rightarrow B$ into $!A \multimap B$~\cite{girard1987linear}). Then, $t$ is stratified if and only if $G$ is stratified. 
\end{conjecture}

Notice that, for any $\lambda$-term $t$ typable in System F, there may be many proof-nets $G$ whose underlying $\lambda$-term is $t$. The proof-net obtained by Girard's encoding is very special in the sense that every function is supposed non-linear (if $A \multimap B$ appears in the proof-net, $A=\oc A'$ for some $A'$). There are stratified proof-nets whose corresponding $\lambda$-term is not stratified according to our definition. For example $(\lambda <f,g>.<(f)g,(g)f>) <\lambda x.x, \lambda y. <y,y>>$ can be decorated in a stratified proof-net even if $\lambda x.x \twoheadrightarrow \lambda y.<y,y> \twoheadrightarrow \lambda x.x$.
    
\subsection{Definition of ``principal door'' stratification}

\paragraph{}We will define a relation $\twoheadrightarrow$ between boxes of proof nets. Intuitively, $B \twoheadrightarrow B'$ means that $B$ can be duplicated before being passed to $B'$ as an argument. In terms of context semantics paths, it means that there is a path beginning by the principal door of $B$ which enters $B'$ by its principal door.

    \begin{equation*}\label{def_twoheadrightsquigarrow}
      B \twoheadrightarrow B' \Leftrightarrow \exists P,P' \in Pot, s\in Sig, T \in Tra, (\sigma(B), P, !_s, +) \rightsquigarrow^* (\sigma(B'),P',T,-)
    \end{equation*}

\paragraph{} This definition of stratification may not be the most general possible for linear logic. Maybe we will find better, more general, simpler conditons for elementary time. Because we anticipate future definitions, we want to distinguish ``stratification'', which is the general idea of forbidding the identification of objects belonging to different levels, and this particuliar version of stratification in linear logic, which we name ``principal door stratification''. However, as we will write about prinicipal door stratification dozens of times in this article, we will use ``stratification'' (respectively ``stratified'') as a shortcut for ``principal door stratification'' (respectively ``principal door stratified'').

\begin{definition}\label{def_stratified}
  A proof net $G$ is principal door stratified if $\twoheadrightarrow$ is acyclic.
\end{definition}

\begin{definition}[strata of a box/context]\label{def_strata}
  The strata of a box $B$, written $S(B)$, is the depth of $B$ in terms of $\twoheadrightarrow$, i.e. $S(B)= \max \Set{k \in \mathbb{N}}{\exists B_1,\cdots,B_k, B \twoheadrightarrow B_1 \twoheadrightarrow \cdots \twoheadrightarrow B_k}$. Let $C$ be a context such that $(\sigma(B),P,!_t,+) \rightsquigarrow^* C$, the stratum of $C$ (written $S(C)$ ) is the stratum of $B$. We also write $S_G$ for $\max_{B \in B_G}S(B)$.
\end{definition}
Notice that the definition of the strata of a context is not ambiguous because $\rightsquigarrow$ is bideterministic and $(\sigma(B),P,[\oc_t],+)$ can not have an antecedent by $\rightsquigarrow$.

We will prove that stratified proof-nets terminate in elementary time, the height of the exponential tower depending only on the depth of the $\twoheadrightarrow$ relation. To prove this, let us consider a path beginning by $(\sigma(B),P,[\oc_t],+)$. Such a path can not go through two contexts of the shape $(e,Q,[\oc_u],q)$ and $(e,Q,[\oc_v],q)$ (because proof-nets are acyclic by Theorem \ref{lemma_acyclicity}). In fact, we can refine the result. Let us assume $e \in B_{\partial(e)} \subset \cdots \subset B_1$, then such a path can not go through two contexts of the shape $(e,[q_1;\cdots;q_{\partial(e)}],[\oc_u],q)$ and $(e,[r_1;\cdots;r_{\partial(e)}],[\oc_v],q)$ where $q_i=r_i$ for every $B \twoheadrightarrow B_i$. We will refer to this result as the strong acyclicity lemma.

This bounds the number of times we can go through the same $\wn C$ or $\wn N$ link with a trace of one element, by $\max_{\substack{B \twoheadrightarrow B'\\ P' \in L_{\mapsto}(B)}}C_{\mapsto}(B',P')^{\partial_G}$. So the height of any copy of $(B,P)$ will be inferior to $|E_G| \cdot \max_{\substack{B \twoheadrightarrow B'\\P' \in L_{\mapsto}(B')}}|C_{\mapsto}(B',P')|^{\partial_G}$.

Finally, we will use this inequality to prove that the number of copies of a potential box $(B,P)$ is bounded by an elementary function on the maximal number of copies of potential boxes $(B',P')$ such that $B \twoheadrightarrow B'$. The depth of $\twoheadrightarrow$ being finite (at most equal to the number of boxes), this entails an elementary bound on the maximum number of copies of potential boxes. 

\begin{figure}
  \centering
  \begin{tikzpicture}
    \node [princdoor] (bprinc) at (0,0) {};
    \node [below]     (bname)  at ($(bprinc)+(-0.4,0)$) {$\mathbf{B}$};
    \node [auxdoor]   (baux)   at ($(bprinc)+(1,0)$) {};
    \node [ax]        (bax)    at ($(bprinc)!0.5!(baux)+(0,0.5)$) {};
    \draw [ar] (bax) to [out=180,in= 70] (bprinc);
    \draw [ar] (bax) to [out=  0,in=110] (baux);
    \draw (bprinc) -| ++(-0.35,0.8) -| ($(baux)+(+0.35,0)$) -- (baux)--(bprinc);
    \node [par]       (bpar)   at ($(bprinc)!0.5!(baux)+(0,-0.8)$) {};
    \draw [ar] (bprinc) -- (bpar); \draw [ar] (baux) -- (bpar);
    \node [princdoor] (cprinc) at ($(bpar)+(0,-0.8)$) {};
    \draw [ar] (bpar)   -- (cprinc);
    \draw (cprinc) -| ++(1,3) -| ($(cprinc)+(-1,0)$) -- (cprinc);

    \node [auxdoor]   (bpaux)   at ($(cprinc) +(2.5,0)$) {};
    \node [princdoor] (bpprinc) at ($(bpaux)  +(1.5,0)$) {};    
    \node [below]     (bpname)  at ($(bpprinc)+(0.5,0)$) {$\mathbf{B'}$};
    \draw (bpprinc) -| ++(0.6,3.85) -| ($(bpaux)+(-1,0)$) -- (bpaux) -- (bpprinc);
    \node [tensor]    (bptens)  at ($(bpaux)  +(0,0.8)$) {};
    \node [par]       (bppar)   at ($(bpprinc)+(0,0.8)$) {};
    \draw [ar] (bptens) -- (bpaux); \draw [ar] (bppar)  -- (bpprinc);
    \node [ax]        (bpax)    at ($(bptens)!0.5!(bppar)+(0,0.4)$) {};
    \draw [ar] (bpax) to [out=180,in=60] (bptens); \draw [ar] (bpax) to [out=0,in=120] (bppar);
    \node [cont]      (bpcont)  at ($(bptens)+(-0.1,0.8)$) {};
    \draw [ar] (bpcont) -- (bptens) node [edgename] {$\mathbf{e}$};

    \node [auxdoor]   (daux1)  at ($(bpcont)+(120:0.8)$) {};
    \node [auxdoor]   (daux2)  at ($(bpcont)+( 60:0.8)$) {};
    \node [princdoor] (dprinc) at ($(daux2)+(1.3,0)$) {};
    \draw [ar] (daux1) -- (bpcont); \draw [ar] (daux2)--(bpcont);
    \draw (dprinc) -| ++(0.35,1.45) -| ($(daux1)+(-0.35,0)$) -- (daux1) -- (daux2) -- (dprinc); 
    \node [tensor]    (dtens)  at ($(dprinc)+( 0,0.6)$) {};
    \node [ax]        (dax1)   at ($(daux1) +( 0.4,1.3)$) {};
    \node [ax]        (dax2)   at ($(daux2) +( 0.2,1.05)$) {};
    \draw [ar] (dax1) to [out=-140,in=90] (daux1); \draw [ar] (dax1) to [out=0, in=60] (dtens);
    \draw [ar] (dax2) to [out=-140,in=90] (daux2); \draw [ar] (dax2) to [out=-10, in=120] (dtens);        
    \draw [ar] (dtens) -- (dprinc);

    \draw [ar] (dprinc) to [out=-70,in=80] (bppar);
    \node [cut] (cutl) at ($(bpaux)!0.5!(cprinc)+(0,-0.7)$) {};
    \draw [ar] (cprinc) to [out=-80,in=180] (cutl); \draw [ar] (bpaux) to [out=-100,in=0] (cutl);

    \node [auxdoor]  (eaux)   at ($(bpprinc)+(0,-1.1)$) {};
    \node [princdoor](eprinc) at ($(eaux) +(-4.5,0)$) {};
    \node [below]     (ename)  at ($(eprinc)+(0.5,0)$) {$\mathbf{D}$};
    \draw (eaux) -| ++(1.5,5) -| ($(eprinc)+(-1,0)$) -- (eprinc) -- (eaux);
    \draw [ar] (bpprinc)--(eaux);
    \node [weak]     (weak) at ($(eprinc)+(0,0.8)$) {};
    \draw [ar] (weak) -- (eprinc);
    \node [dig]      (extdig)  at ($(eprinc)+(-2,0)$) {};
    \node [cut]      (extcut2) at ($(eprinc)!0.5!(extdig)+(0,-0.8)$) {};
    \node [ax]       (extax)   at ($(extdig)+(-0.5,0.8)$) {};
    \draw [ar] (extax) to [out=  0, in=90] (extdig);
    \draw [ar] (extax) to [out=180, in=90] ($(extax)+(-0.5,-0.8)$);
    \draw [ar]       (extdig) to [out=-90,in=170] (extcut2);
    \draw [ar]       (eprinc)   to [out=-90,in= 10] (extcut2);
    
    \node [cont]     (extcont) at ($(eaux)   +(-120:1.1)$) {};
    \node [weak]     (extweak) at ($(extcont)+( 120:0.8)$) {};
    \draw [ar]       (eaux) -- (extcont);     \draw [ar] (extweak) -- (extcont);
    \node [princdoor](fprinc)  at ($(eaux)+(4,0)$)   {};
    \node [below]     (fname)  at ($(fprinc)+(0.5,0)$) {$\mathbf{C}$};
    \draw (fprinc) -| ++(1.5,5) -| ($(fprinc)+(-1.9,0)$) -- (fprinc);
    \node [cont]     (fcont)   at ($(fprinc)+(0,0.7)$) {};
    \node [der]      (fder1)   at ($(fcont)+(120:0.8)$) {};
    \node [der]      (fder2)   at ($(fcont)+( 60:0.8)$) {};
    \node [tensor]   (ftens1)  at ($(fder1)+(110:0.8)$){};
    \node [tensor]   (ftens2)  at ($(fder2)+( 70:0.8)$){};
    \node [ax]       (fax)     at ($(ftens1)!0.5!(ftens2)+(0,0.5)$) {};
    \draw [ar] (fcont) -- (fprinc);
    \draw [ar] (ftens1) to [out=-90,in=120] (fder1);
    \draw [ar] (ftens2) to [out=-90,in= 60] (fder2);
    \draw [ar] (fder1) -- (fcont);
    \draw [ar] (fder2) -- (fcont);
    \draw [ar] (fax) to [out=180,in=60] (ftens1); \draw [ar] (fax) to [out=0,in=120] (ftens2);
    \node [princdoor] (gprinc) at ($(ftens1)+(120:0.9)$) {};
    \node [par]       (gpar)   at ($(gprinc)+(0,0.8)$)   {};
    \node [ax]        (gax)    at ($(gpar)  +(0,0.7)$)   {};
    \draw [ar] (gax)  to [out= -60,in= 60] (gpar);
    \draw [ar] (gax)  to [out=-120,in=120] (gpar);
    \draw [ar] (gpar)   -- (gprinc);
    \draw [ar] (gprinc) -- (ftens1);
    \draw (gprinc) -| ++(0.5,1.7) -| ($(gprinc)+(-0.5,0)$) -- (gprinc);
    \node [weak]      (fweak)  at ($(ftens2)+(60:0.8)$) {};
    \draw [ar]        (fweak) -- (ftens2);

    \node [cut]      (extcut)  at ($(fprinc)!0.5!(eaux)+(0,-1.5)$) {};
    \draw [ar] (extcont) to [out=-50, in=180] (extcut);
    \draw [ar] (fprinc)  to [out=-90, in=  0] (extcut);
  \end{tikzpicture}
  \caption{\label{fig_ex_stratified}$(\sigma(B),[x_D;\sigl(\sige)],[\oc_{\sigl(\sigr(\sige))}],+)\mapsto^5 (e,[x_D;\sigl(\sige)],[\oc_{\sigl(\sigl(\sige))}],-) \mapsto^{7} (\sigma(C),[],[\oc_{\sigr(\sige)};\parr_r;\oc_{\sigl(\sige)};\wn_{\sigr(x_D)}],-) \mapsto^9 (\sigma(C),[],[\oc_{\sigr(\sige)};\otimes_l;\wn_{\sigr(\sige)};\oc_{\sigr(x_D)}],+) \mapsto^3 (\sigma(B'),[x_D],[\oc_{\sigr(\sige)};\otimes_l; \wn_{\sigr(\sige)}],-) \mapsto^8 (\sigma(B),[x_D;\sigr(\sige)],[\oc_{\sigr(\sigr(\sige))}],+) \mapsto^5 (e,[x_D;\sigr(\sige)],[\sigr(\sige)],-)  \mapsto^{11} (w,[\sigr(x_D)],[\oc_{\sige}],-)$}
 \end{figure}

As an example, we can observe the path presented in Figure \ref{fig_ex_stratified}. We have $B \twoheadrightarrow B'$ but $B \not \twoheadrightarrow C$. And indeed, in this path, there are not two contexts of the shape $(e,[p_B;p_{B'}],[\oc_u],-)$ and $(e,[q_B;q_{B'}],[\oc_v],-)$ with $p_{B'}=q_{B'}$. On this proof-net, we can get the intuition underlying the strong acyclicity lemma. Let us suppose $(e,[p_1,…,p_{\partial(e)}],[\oc_u],p) \mapsto^* (e,[q_1,…q_{\partial(e)}],[\oc_v],p)$ and $p_i=q_i$ for all $i$ such that $B \twoheadrightarrow B_i$. Then, we can take the path between those two contexts backward (we will name this reverse path an {\em antipath}), forgetting the exponential signatures $q_i$ corresponding to boxes $B_i$ with $B \not \twoheadrightarrow B_i$ (as in the path of Figure \ref{fig_ex_stratified} where we replaced the exponential signature corresponding to box $D$ by a generic variable named $x_D$). Then, we can observe that we have enough information to do the antipath, because the $q_i$s we forgot are never really used. In Figure \ref{fig_ex_stratified}, if we supposed $(e,[p_{D};p_{B'}],[\oc_u],-) \mapsto^* (e,[q_D;p_{B'}],[\oc_v],-)$, we could follow the antipath beginning by $(e,[x_D;p_{B'}],[\oc_v],-)$. This antipath leaves $D$ by its auxiliary door with the contexts $(\sigma(B'),[x_D],[\oc_v;\otimes_l;\wn_{p_{B'}}],-) \mapsfrom (\sigma_1(D),[~],[\oc_v;\otimes_l;\wn_{p_{B'}};\oc_{x_D}],-)$. Then the antipath crosses a contraction node downwards, in this direction, there is no choice to make and $\oc_{x_D}$ is transformed into $\oc_{\sigr(x_D)}$. A bit later, the antipath crosses the contraction node upwards, so we have to know where we came from, so we have to look at our trace. But our trace is $\sigr(x_D)$ so we do not need to know $x_D$ to make the choice. The only possibility we could imagine where we would need to know $x_D$ is if the antipath crossed the $\wn N$ link upwards, but it would mean that the antipath left $D$ by its principal door, so the path would enter $D$ by its principal door. But in this case we would have $B \twoheadrightarrow D$, so $p_{D}=q_D$.

So, as we said we do not need to know the values of the $q_i$s corresponding to the boxes $B_i$ with $B \not \twoheadrightarrow B_i$. So, we could do the same antipath by replacing those $q_i$s by the corresponding $p_i$s, which would give us a context $(e,[o_1;\cdots;o_{\partial(e)}],[\oc_t],-)$ such that $(e,[o_1;\cdots;o_{\partial(e)}],[\oc_t],-) \mapsto^* (e,[p_1;\cdots;p_{\partial(e)}],[\oc_u],-)$ and for all $i$ such that $B \twoheadrightarrow B_i$, $o_i=p_i$. Then we could repeat the same antipath again and again until we get a cycle $(e,[r_1;\cdots;r_{\partial(e)}],[\oc_w],-) \mapsto^* (e,[r_1;\cdots;r_{\partial(e)}],[\oc_x],-)$. This is a contradiction, because proof-nets are acyclic. So our assumption was false, there are not two contexts of the shape $(e,[p_1;\cdots;p_{\partial(e)}],[\oc_u],p)$ and $(e,[q_1;\cdots;q_{\partial(e)}],[\oc_v],p)$ where $p_i=q_i$ for every $B \twoheadrightarrow B_i$.

To see the relationship between $\twoheadrightarrow$ and the number of copies of box, we can notice that we could replace the contraction in box $C$ of Figure \ref{fig_ex_stratified} by a tree of $n-1$ contraction (with $n$ derelictions above them and $n$ tensors above the derelictions). Thus $B'$ would have $n$ copies, so a path beginning by the principal door of $B$ could go $n$ times through $e$ making independent choices each times. So $B$ would have $2^n$ copies. Intuitively, for each additionnal copy of $B'$, we multiply the number of copies of $B$ by $2$. We can build proof-nets where there is a box $B$ such that $B \twoheadrightarrow B$ and there is a similar relationship between some copies of $B$ and other copies of $B$ (the more copies of $B$ there are, the more copies there are). This is the case for the proof-net of Figure \ref{ackermann}, representing the Ackermann function applied to $3$, where $B_1\twoheadrightarrow B_1$. The proof-net representing the Ackermann function does not normalize in elementary time, as this function is not even primitive recursive.

 \begin{figure}\centering
  \begin{tikzpicture}
    \node [princdoor] (Bbang) at (0,0) {};
    \node at ($(Bbang)+(0.6,-0.4)$) {$\mathbf B_1$}; 
    \node [auxdoor]   (Bwhyn) at ($(Bbang)+(-3,0)$) {};
    \draw (Bbang) -| ++ (3,6.5) -| ($(Bwhyn)+(-1,0)$) -- (Bwhyn) -- (Bbang);
    \node [par]       (parm)  at ($(Bbang)+(0,0.8)$) {};
    \draw [ar] (parm) -- (Bbang);
    \node [cut]       (cutSm) at ($(parm)+(0.2,0.7)$) {};
    \node [tensor]    (tensSm)at ($(cutSm)+(-0.4,0.6)$) {};
    \draw [ar, out=-90,in=180] (tensSm) to (cutSm); 
    \node             (S)     at ($(cutSm)+( 0.5,0.6)$) {S};
    \draw [ar, out=-90,in=0] (S) to (cutSm);
    \node [ax]        (axM)   at ($(tensSm)+(-0.9,0.7)$) {};
    \draw [ar, out=0,in=120] (axM) to (tensSm);
    \draw [ar, out=-140,in=120] (axM) to (parm);
    \node [exists]    (exSm)  at ($(tensSm)+( 0.9,0.7)$) {};
    \draw [ar, out=-90, in=60] (exSm) to (tensSm);
    \node [tensor]    (Smf)   at ($(exSm)+(0,0.8)$) {};
    \draw [ar] (Smf) -- (exSm);
    \node [ax]        (axf)   at ($(Smf)+(-3,1.2)$){};
    \draw [ar, out=-170,in=100] (axf) to (Bwhyn);
    \draw [ar, out= -10,in=130] (axf) to (Smf);
    \node [der]       (derSmf)at ($(Smf)+(0.4,0.7)$) {};
    \draw [ar] (derSmf)--(Smf);
    \node [tensor]    (tensSmf)  at ($(derSmf)+(0,0.7)$) {};
    \draw [ar] (tensSmf)--(derSmf);
    \node             (n1)    at ($(tensSmf)+(120:1)$) {1};
    \draw [ar] (n1) -- (tensSmf);
    \node [ax]        (axSmf1) at ($(tensSmf)+( 70:1)$) {};
    \draw [ar, out=180, in=60] (axSmf1) to (tensSmf);
    \draw [ar, out=-20, in=40] (axSmf1) to (parm);

    \node [par]       (parf)  at ($(Bbang)!0.5!(Bwhyn)+(0,-1.2)$) {};
    \draw [ar] (Bbang) to (parf);
    \node [dig]       (digf)  at ($(Bwhyn)!0.5!(parf)$) {};
    \draw [ar] (Bwhyn) to (digf);
    \draw [ar] (digf)  to (parf);
    \node [princdoor] (BEbang)at ($(parf)+(0,-0.9)$) {};

    \draw [ar] (parf) to (BEbang);
    \draw (BEbang) -| ++ (5.5,9) -| ($(BEbang)+(-3,0)$) -- (BEbang);
    \node [tensor]    (tensNaux) at ($(BEbang)+(1,-4)$) {};
    \node [exists]    (exN)      at ($(tensNaux)+(-0.7,-0.7)$) {};
    \node [der]       (derNr)     at ($(exN)+(-0.7,-0.7)$) {};
    \node [cont]      (contN)    at ($(derNr)+(-0.9,-0.9)$) {};
    \node [par]       (parN)     at ($(contN)+(2,-1)$) {};
    \node [der]       (derIter)  at ($(tensNaux)+(0.7,0.7)$) {};
    \node [tensor]    (tensIterS)at ($(derIter)+(0.7,0.7)$) {};
    \node [princdoor] (BSbang)   at ($(tensIterS)+(-0.7,0.7)$) {};
    \node             (Sbang)    at ($(BSbang)+(0,0.7)$) {S};
    \draw (BSbang) -| ++ (0.5,1) -| ($(BSbang)+(-0.5,0)$) -- (BSbang) {};
    \node [der]       (derIterS) at ($(tensIterS)+(0.7,0.7)$) {};
    \node [tensor]    (tensIterSn)at ($(derIterS)+(0.7,0.7)$) {};
    \node [ax]        (axN)      at ($(tensIterSn)+(-1,0.8)$) {};
    \node [der]       (derNl)    at ($(axN)+(-2,-0.7)$) {};
    \node [ax]        (axResFun) at ($(tensIterSn)+(1,0.8)$) {};
    \draw [ar] (tensNaux) to (exN);
    \draw [ar] (exN) to (derNr);
    \draw [ar,out=-120,in= 60] (derNr)to(contN);
    \draw [ar, out=0, in=120] (axN) to node [edgename, pos=0.3, right] {$\mathbf e$} (tensIterSn);
    \draw [ar] (tensIterSn) to (derIterS);
    \draw [ar] (derIterS) to (tensIterS);
    \draw [ar] (tensIterS) to (derIter);
    \draw [ar]  (BSbang) to (tensIterS);
    \draw [ar]  (Sbang) to (BSbang);
    \draw [ar] (derIter) to (tensNaux);
    \draw [ar,out=180,in=50] (axN) to (derNl);
    \draw [ar,out=-110,in=110] (derNl) to (contN);
    \draw [ar] (contN) to (parN);
    \draw [ar,out=180,in=60] (axResFun) to (tensIterSn);
    \draw [ar] (axResFun) to [out=-20, in=60] (parN); 
    \draw [ar] (BEbang) to (tensNaux);

    \node [tensor]    (appn)    at ($(parN)+(11,0)$) {};
    \node [cut]       (cutappn) at ($(parN)!0.5!(appn)+(0,-0.8)$) {};
    \draw [ar, out=-50,in=180] (parN) to (cutappn);
    \draw [ar, out=-130,in=  0] (appn) to (cutappn);

    \node [princdoor]    (CEbang)   at ($(appn)+(-0.4,3)$) {};
    \draw (CEbang) -| ++ (2.8,6.5) -| ($(CEbang)+(-4.5,0)$) -- (CEbang);
    \node [forall]       (forall)   at ($(CEbang)+(0,0.8)$) {};
    \node [par]          (parg)     at ($(forall)+(0,0.8)$) {};
    \node [princdoor]    (princ3g)  at ($(parg)+(1, 2)$) {};
    \node at ($(princ3g)+(0.6,-0.4)$) {$\mathbf B_2$}; 
    \node [auxdoor]      (aux1g)    at ($(princ3g)+(-4,0)$) {};
    \node [auxdoor]      (aux2g)    at ($(aux1g)!0.333!(princ3g)$) {};
    \node [auxdoor]      (aux3g)    at ($(aux1g)!0.666!(princ3g)$) {};
    \node [cont]         (cont1g)   at ($(aux1g)!0.5!(aux2g)+(0,-0.9)$) {};
    \node [cont]         (cont2g)   at ($(cont1g)!0.5!(parg)$) {};
    \nvar{\hautTens}{1.3cm}
    \node [tensor]       (tens1g)   at ($(aux1g)+(0,\hautTens)$) {};
    \node [tensor]       (tens2g)   at ($(aux2g)+(0,\hautTens)$) {};
    \node [tensor]       (tens3g)   at ($(aux3g)+(0,\hautTens)$) {};
    \nvar{\decAx}{0.6cm}
    \node [ax]        (ax1n)     at ($(tens1g)+(-0.5,\decAx)$) {};
    \node [ax]        (ax2n)     at ($(tens1g)!0.5!(tens2g)+(0,\decAx)$) {};
    \node [ax]        (ax3n)     at ($(tens2g)!0.5!(tens3g)+(0,\decAx)$) {};
    \node [ax]        (ax4n)     at ($(princ3g)+(0,\decAx + \hautTens)$) {};
    \node [par]          (parx)     at ($(princ3g)+(0,0.7)$) {};
    \draw [ar] (CEbang) -- (appn.120) ;
    \draw [ar] (forall) -- (CEbang);
    \draw [ar] (parg) -- (forall);
    \draw [ar] (cont2g) -- (parg);
    \draw [ar] (cont1g) -- (cont2g);
    \draw [ar] (aux1g) -- (cont1g);
    \draw [ar] (aux2g) -- (cont1g);
    \draw [ar] (aux3g) -- (cont2g);
    \draw [ar] (parx) -- (princ3g);
    \draw [ar] (tens1g) -- (aux1g);
    \draw [ar] (tens2g) -- (aux2g);
    \draw [ar] (tens3g) -- (aux3g);
    \draw [ar,out=-140,in= 160] (ax1n) to (parx);
    \draw [ar,out=-30 ,in= 120] (ax1n) to (tens1g);
    \draw [ar,out=-150 ,in=  60] (ax2n) to (tens1g);
    \draw [ar,out=0   ,in=120 ] (ax2n) to (tens2g);
    \draw [ar,out=180 ,in=  60] (ax3n) to (tens2g);
    \draw [ar,out=0   ,in=120 ] (ax3n) to (tens3g);
    \draw [ar,out=180 ,in=  60] (ax4n) to (tens3g);
    \draw [ar,out=-20 ,in=  60] (ax4n) to (parx);
    \draw (princ3g) -| ++(1,\hautTens + \decAx +0.5cm) -| ($(aux1g)+(-1.2,0)$) -- (aux1g) -- (aux2g) -- (aux3g) -- (princ3g);
    \draw (princ3g) -- (parg);
    \coordinate (concl) at ($(appn)+(2,0)$);
    \node  [ax]    (axConcl) at ($(appn)!0.5!(concl)+(0,1)$) {};
    \draw [ar,out=  0,in=100] (axConcl) to (concl);
    \draw [ar,out=180,in= 80] (axConcl) to (appn.60);
  \end{tikzpicture}
  \caption{\label{ackermann}This proof-net, representing the ackermann function applied to 3, is not stratified. Indeed $
(\sigma(B_1),[\sigl(\sigl(\sige))], [\oc_{\sign(\sigl(\sige),\sign(\sigl(\sigr(\sige)),\sige))}],+) \rightsquigarrow^{14}
(\sigma_1(B_2),[\sigr(e)],[\oc_{\sign(\sigl(\sige),\sign(\sigl(\sigr(\sige)),\sige))};\parr_r; \oc_{\sige}],-) \rightsquigarrow^5 
(\sigma_2(B_2),[\sigr(\sige)],[\oc_{\sign(\sigl(\sige),\sign(\sigl(\sigr(\sige)),\sige))};\otimes_l;\wn_{\sige}],+) \rightsquigarrow^{15} 
(\sigma_1(B_1),[\sigl(\sigr(\sige))],[\oc_{\sigl(\sige)};\oc_{\sign(\sigl(\sigr(\sige)),\sige)}],-) \rightsquigarrow^{19}
(\sigma(B_1),[\sigl(\sigr(\sige))],[\oc_{\sige};\otimes_l;\exists;\parr_l;\oc_{\sign(\sigl(\sigr(\sige)),\sige)}],+) \rightsquigarrow^{53}
(\sigma(B_1),[\sigr(\sige)],[\oc_{\sige};\otimes_l;\exists;\parr_l;\oc_{\sigr(\sige)};\otimes_l;\exists;\parr_l;\oc_{\sige}],+) \rightsquigarrow^{18} 
(\sigma(B_2),[\sigr(\sige)],[\oc_{\sige};\otimes_l;\exists;\parr_l;\oc_{\sigr(\sige)};\otimes_l;\exists;\parr_l;\oc_{\sige};\parr_r;\oc_{\sige}],+) \rightsquigarrow^{14} 
(e,[], [\oc_{\sige};\otimes_l;\exists;\parr_l; \oc_{\sigr(\sige)}; \otimes_l;\exists],-) \rightsquigarrow^{10} 
(\sigma_3(B_2),[\sigl(\sige)],[\oc_{\sige};\otimes_l;\exists;\parr_l;\oc_{\sige}],-) \rightsquigarrow^{16} 
(e,[],[\oc_{\sige};\otimes_l;\exists;\otimes_r;\wn_{\sigl(\sigr(\sige))};\parr_l ;\forall],+) \rightsquigarrow^{14} 
(\sigma(B_2),[\sigr(\sige)],[\oc_{\sige};\otimes_l;\exists;\otimes_r;\wn_{\sigl(\sigr(\sige))};\parr_l ;\forall;\otimes_l;\wn_{\sige};\otimes_r;\wn_{\sige}],-) \rightsquigarrow^5 
(\sigma_3(B_2),[\sigr(\sige)],[\oc_{\sige};\otimes_l;\exists;\otimes_r;\wn_{\sigl(\sigr(\sige))};\parr_l ;\forall;\otimes_l; \wn_{\sige}; \otimes_l;\wn_{\sige}],+) \rightsquigarrow^{13} 
(\sigma(B_1),[\sigr(\sige)],[\oc_{\sige};\otimes_l;\exists;\otimes_r;\wn_{\sigl(\sigr(\sige))};\parr_l ;\forall;\otimes_l;\wn_{\sige}],-)$}
\end{figure}
    
    \paragraph{}In this section, we will assume that the proof-nets we work on are stratified. In order to prove elementary soundness for stratified proof-nets, we will make a careful analysis of paths of context semantics in stratified proof-nets. The weak bounds for systems such as $ELL$ and $L^3$ were proved using a strata by strata strategy (our notion of strata corresponds to depths in $ELL$ and corresponds to levels in $L^3$). They prove that reducing the cuts at strata $\leq i$ does not increase too much the size of the proof-net at stratum $i+1$. Here we will prove the strong bound for stratified proof nets in a similar way: we will bound the number of copies of a box has when we only reduce cuts in the strata $\leq i+1$ by the maximum number of copies of a box when reducing only cuts in strata $\leq i$.  Moreover, we need a notion of copies telling us whether a copy still corresponds to a duplicate if we only fire exponential cuts in strata $\leq i$. This is exactly what a $\mapsto_i$-copy will be.
    
    \begin{definition}\label{def_mapstos}
      Let $G$ be a stratified proof-net. For all $s \in \{0, \cdots n \}$, we define $\mapsto_s$ as follows:
      \begin{equation*}
        C \mapsto_s D \Leftrightarrow \left \{ \begin{array}{c} C \mapsto D \\ S(D) \leq s \end{array} \right.
      \end{equation*}
    \end{definition}
     
    \paragraph{}Concretely, it will prevent $\hookrightarrow$ jumps over a box whose stratum is too high. We define similarly $\rightsquigarrow_s$. Notice that if $S(D)$ is undefined (there is no box $B$ such that $(\sigma(B),P,\oc_t,+) \rightsquigarrow D$) then $C \not \mapsto_s D$.

    \begin{lemma}\label{lemma_impossible_jump}
      If $G$ is stratified and $(\sigma_i(B),P,[\oc_t],-) \mapsto_s (\sigma(B),P,[\oc_t],+)$ then $S(B) \leq s$
    \end{lemma}
    \begin{proof}
      If $(\sigma_i(B),P,[\oc_t],-) \mapsto_s (\sigma(B),P,[\oc_t],+)$ then, by definition of $\mapsto_s$, $S(\sigma(B),P,[\oc_t],+) \leq s$. By definiton of the strata of a context, $ S(\sigma(B),P,[\oc_t],+) = S(B) \leq s$.
    \end{proof}

    \begin{lemma}\label{lemma_impossible_twohead}
      If $G$ is stratified and $(\sigma(B),P,[\oc_t]@T.\wn_u,-) \mapsto_s (e,P.u,[\oc_t]@T,-)$ then $S(B) < s$.
    \end{lemma}
    \begin{proof}
      If $(\sigma(B),P,[\oc_t]@T.\wn_u,-) \mapsto (e,P.u,[\oc_t]@T,-)$ then, by definition of $\mapsto_s$, $S(e,P.u,[\oc_t]@T,-) \leq s$. So there exists $C,Q,v$ such that $(\sigma(C),Q,[\oc_v],+) \rightsquigarrow^* (e,P.u,[\oc_t]@T,-)$ and $S(C) \leq s$. $C \twoheadrightarrow B$ so $S(B) < S(C) \leq s$.
    \end{proof}

    \begin{lemma}\label{def_copy_rescriction}
      For all $B \in B_G$ and $s,s' \in \mathbb{N}$ with $s' \leq s$:
      \begin{itemize}
        \item For all $P=[p_1; \cdots ; p_{\partial(B)}] \in L_{\mapsto_{s}}(B)$, there exists a unique $P^{/B,s'}=[p'_1; \cdots ; p'_{\partial (B)}] \in L_{\mapsto_{s'}}(B)$ such that for all $0 < i \leq \partial (B)$, $p_i' \vartriangleleft p_i$
        \item For all $P \in L_{\mapsto_s}(B)$ and $t$ such that for all $t \sqsubseteq u$, $(\sigma(B),P,[\oc_u],+)$ is a $\mapsto_s$-copy context of $(B,P)$, there exists a unique $t^{/B,P,s'} \vartriangleleft t$ such that for all $t^{/B,P,s'} \sqsubseteq u'$, $(\sigma(B),P^{/B,s'},[\oc_u'],+)$ is a $\mapsto_{s'}$-copy context. 
      \end{itemize}
    \end{lemma}
    \begin{proof}
      We prove the statement by induction on $\partial (B)$.
      \begin{itemize}
      \item If $\partial (B)=0$, then for all $P \in L_{\mapsto_s}(B)$, $P=[~]$. So we can take $[~]^{/B,s'}=[~]$. Else, $P=Q.p_{\partial(B)}$ with $C$ the deepest box containing $B$, $Q \in L_{\mapsto_s}(C)$ and $p_{\partial(B)} \in C_{\mapsto_s}(C,Q)$. By induction hypothesis, there exists a unique $Q^{/C,s'}=[p_1'; \cdots ;p'_{\partial(B)-1}] \in L_{\mapsto_{s'}}(C)$ such that for every $1 < i < \partial (B)$, $p_i' \vartriangleleft p_i$. Moreover, $p_{\partial(B)} \in C_{\mapsto}(C,Q)$, so by induction hypothesis there exists $p_{\partial(B)}^{/C,Q,s'} \vartriangleleft p_{\partial (B)}$ which is standard (the truncation of a standard signature is always standard) and such that all its simplifications are copy simplifications. So $p_{\partial(B)}^{/C.Q,s'} \in C_{\mapsto_{s'}}(C,Q^{/C,s'})$. So we can take $Q^{/C,s'}.p_{\partial(B)}^{/C,Q,s'}$. If it was not unique, it would break the unicity of either $Q^{/C,s'}$ or $p_{\partial(B)}^{/C,Q,s'}$, which are guaranteed by induction hypothesis.
      \item Now, we show the second property by induction on $\sqsupseteq$. Let us take $t$ such that for all $t \sqsubseteq u$, $u$ is a copy simplification of $(B,P)$. We suppose that the property is true for every $u \sqsupseteq t$ satisfying the hypothesis. We build an exponential signature $t_0'$ in the following way:
        \begin{itemize}
        \item If $t$ is minimal for $\sqsupseteq$ (i.e. there are no $\sign(~,~)$ in $t$), then $t_0' = t$ 
        \item Else we consider $t_0$ the exponential signature obtained by transforming the deepest leftmost $\sign(v_l,v_r)$ of $t$ into $\sigp(v_r)$. Then, $t_0'$ is obtained by replacing in $t_0^{/B,P,s'}$ the $\sigp(v_r')$ corresponding to $\sigp(v_r)$ (if if has not been cut) by $\sign(v_l,v_r')$.
        \end{itemize}
        We now consider the path beginning by $(\sigma(B),P, [\oc_{t_0'}],+)$. The underlying formula of $(\sigma(B),P,[\oc_{t_0'}],+)$ is well-defined so, by lemma \ref{lemma_underlying_mapsto}, the underlying formula of all the contexts in the path are well-defined. So the path will not be stopped by a mismatch between the right-most trace element and the top-most connective of the formula labelling the edge.
        
        Moreover, we can prove that for any context $(e,Q,[!_v]@U,p)$ of the path, $Q$ is a $\mapsto_{s'}$-canonical potential for $e$ and for any $v \sqsubset w$, $(e,Q,[!_w]@U,p)$ is a copy context. So the path will not be stopped by a mismatch between an exponential link and the root of the exponential signature on the right-most $\oc$ trace element (except if it is the left-most trace element, which we will deal with in the next paragraph).
        
        We know that $(\sigma(B),P,[\oc_{t'_0}],+)$ is a $\mapsto_s$-copy context so the $\mapsto_s$ beginning by this context ends with $\oc_{\sige}$ as its first trace element. We know that $\mapsto_{s'} \subseteq \mapsto_s$, so the $\mapsto_{s'}$ path will end and will not be stopped by a mismatch between an exponential link and the root of the exponential signature on the left-most $\oc$ trace element.
        
        So, there are four possibilites:
        \begin{itemize}
        \item $(\sigma(B),P,[\oc_{t'_0}],+) \mapsto^*_{s'} (e,Q,[!_{\sige}]@W,-) \not \mapsto_{s'}$ with the tail of $e$ being a $\wn W$ link.
        \item $(\sigma(B),P,[\oc_{t'_0}],+) \mapsto^*_{s'} (e,Q,[!_{\sige}],-) \not \mapsto_{s'}$ with the tail of $e$ being a $\wn D$ link.
        \item $(\sigma(B),P,[\oc_{t'_0}],+) \mapsto^*_{s'} (e,Q,[!_w],-) \not \mapsto_{s'}$ with the tail of $e$ being a $\wn P$ link of a box of stratum strictly greater than $s'$.
        \item $(\sigma(B),P,[\oc_{t'_0}],+) \mapsto^*_{s'} (e,Q,[!_{\sige}]@W,+) \not \mapsto_{s'}$ with $e$ being a pending edge.
        \end{itemize} 
        In each case, by Lemma \ref{lemma_prec_left_right}, we know that there exists $t'$ such that $t' \preccurlyeq^w t'_0$ and $(\sigma(B),P,!_{t'},+) \mapsto^*_{s'} (e,Q,[!_{\sige}]@W,p) \not \mapsto^*_{s'}$. Then, verifiying that $(\sigma(B),P,[\oc_{t'}],+)$ is a $\mapsto_{s'}$-copy context is straight forward. The induction hypothesis gives us that for every $t' \sqsubset u'$, $(\sigma(B),P,[\oc_{u'}],+)$ is a $\mapsto_{s'}$-copy context. So the property is true for $t$ because we did not touch any branches except the deepest left-most $\sign(~,~)$.

        To prove the unicity of such a $t'$, let us suppose there is another exponential signature $t''$ verifying the properties. Then, we define $u'$ (resp. $u''$) as $t'$ (resp. $t''$) if the deepest left-most $\sign(~,~)$ of $t$ is no longer in $t'$ (resp. $t''$). Else, we define $u'$ (resp. $u''$) as the exponential signature obtained by replacing this $\sign(v_l',v_r')$ by $\sigp(v_r')$ (resp. $\sign(v_l'',v_r'')$ by $\sigp(v_r'')$). We can notice that $u'$ (resp. $u''$) is a truncation of $t_0$ which satisfy the properties of $t_0^{/B,P,s'}$ so, by unicity, $u''=t_0^{/B,P,s'}=u'$.

        In particular the deepest left-most $\sign(~,~)$ of $t$ is in $t'$ if and only if it is in $t''$. So the only possibility for $t'$ and $t''$ to be different is that this $\sign$ is present in both and $v_l'\neq v_l''$. This exponential signature does not contain any $\sign(~,~)$ and they are both truncations of $v_l$ so either $v_l' \preccurlyeq v_l''$ or $v_l'' \preccurlyeq v_l'$. So either $t' \preccurlyeq t''$ or $t'' \preccurlyeq t'$. In both cases, knowing that $(\sigma(B),P,[\oc_{t'}],+)$ and $(\sigma(B),P,[\oc_{t''}],+)$ are $\mapsto_{s'}$-copy contexts, we can use the definition of copy-contexts to get that $t' \preccurlyeq t''$ and $t'' \preccurlyeq t'$. So $t'=t''$.
      \end{itemize}
    \end{proof}

    \paragraph{}\label{def_ls}For matters of readability, we will often write $L_s(x)$ for $L_{\mapsto_s}(x)$ and $C_s(x,P)$ for $C_{\mapsto_s}(x,P)$. We will also write $P^{/s}$ for $P^{/B,s}$ and $t^{/s}$ for $t^{/B,P,s}$ when the box $B$ (and the potential $P$) can be guessed. Let us notice that if $s \geq S(G)$ then $\mapsto_s = \mapsto$. So the upper bounds on $|L_{\mapsto_{S(G)}}(x)|$ and $|C_{\mapsto_{S(G)}}(x,P)|$ will give us upper bounds on $|L_{\mapsto}(x)|$ and $C_{\mapsto}(x,P)|$. Then we will use these upper bounds to prove an upper bound on $T_G$, so on the maximum length of the reduction paths of $G$.

    \subsection{Definition of $\sim$ equivalence between contexts}
    \paragraph{}The idea of the injection lemma is the following. Suppose that two different $\mapsto_s$-copies of $(B,P)$ ($t$ and $t'$) lead respectively to the final contexts $(g, Q, [!_{\sige}], -)$ and $(g,Q', [!_{\sige}], -)$ and $Q^{/s}=Q'^{/s}$. Let us go back from $g$ to $\sigma(B)$ by the two paths. Because we are trying to follow paths, beginning by their end, words like ``beginning (or end) of the path '' can be confusing: in which way are we taking the paths? If we go from $\sigma(B)$ to $g$ we will talk about the ``paths'', but if we go from $g$ to $\sigma(B)$ we talk about the ``antipaths''.

    On the begining of our antipaths, the contexts are on the same edge. The only way for the antipaths to separate is to cross a $?C$, $\parr$ or $\otimes$ link upward with a different right-most trace element. The only way to have different traces between the two antipaths is to go out from a box and that the potentials of the two contexts for this box are different. The potential for boxes of stratum $< s$ are the same in the two antipaths. So the only way to have different traces between the two antipaths is to leave a box of stratum $\geq s$. It is only possible by their auxiliary door (the strata of the contexts are $\leq s$ along the path), so the only difference between the traces of the contexts is on exponential stacks of $!$ trace element. So, the antipaths will never separate (notice that for the antipaths to be separated by a $?C$ link, the difference must be on a $?$ signature), the two copies take exactly the same path, they are equal. 
    
In fact, it is a little bit more complex, as we can see in figure \ref{ex_for_equivs} supposing $s(B_2) \geq s$ and $s(B_1) < s$. Indeed, a difference on the potential $[t]$ associated to $a$ in the beginning of the antipath transforms into a difference on the signatures of a $\oc$ trace element (if we take another potential $[t']$ for $a$ in the beginning of the antipath, it leads to a $\oc_{t'}$ in the trace of the context in $b$). This difference can, in turn, lead to a difference in the potential of the context corresponding to a copy of a box of stratum $<s$ (we would have $[\sigr(t')]$ as a potential when entering $B_1$). This allows, by leaving the boxes by the principal door in the antipaths, to have different exponential signatures on $?$ trace element, so that the antipaths would separate on a $?C$. Although this possibility complexify the proof, the antipaths will not separate because the surface of both exponential signatures is the same and will lead them back into the box where the difference originated (whatever $t'$ we choose, we always take the right premise of the contraction and go back to $B_2$). And, this box being in a strata $\geq s$, it is still impossible to leave the box by its principal door, hence impossible to have a difference on a $?$ trace element so to make the antipaths separate.

\begin{figure}\centering
      \begin{tikzpicture}
        \node[princdoor] (princ1) at (0,0) {};
        \node at ($(princ1)+(0.8,-0.4)$) {$B_1$}; 
        \node[par] (par1) at ($(princ1)+(0,1)$) {};
        \node[ax] (ax1) at ($(par1)+(0,1)$) {};
        \draw[ar,out= -20,in= 60] (ax1) to (par1);
        \draw[ar,out=-160,in=120] (ax1) to (par1);
        \draw[ar] (par1) -- (princ1) node [edgename] {$d$};
        \node[cont] (cont) at ($(princ1)+(4,0)$) {};
        \node[cut] (cut) at ($(princ1)!0.5!(cont) + (0,-1)$) {};
        \draw[ar,out=-85,in=175] (princ1) to (cut);
        \draw[ar,out=-95,in=  5] (cont)   to [edgename, below right] node {$c$} (cut);
        \draw (princ1) -| ($(princ1)+(1,2.7)$) -| ($(princ1)+(-1,0)$) -- (princ1);
        \node[etc]   (etcCont) at  ($(cont)+(140:1.5)$)  {};
        \node[auxdoor]  (aux2l) at ($(cont)+( 40:1.5)$) {};
        \draw [ar] (etcCont) -- (cont);
        \draw [ar] (aux2l)   -- (cont) node [edgename, below right]{$b$} ;
        \node[auxdoor]  (aux2r) at ($(aux2l)+(1.5,0)$)  {};
        \node[princdoor](princ2) at ($(aux2r)+(1,0)$) {};
        \node at ($(princ2)+(0.8,-0.4)$) {$B_2$}; 
        \node[tensor] (tens2) at ($(aux2l)+(0,1)$) {};
        \node[der] (der2) at ($(tens2)+(120:1)$) {};
        \node[ax]  (ax2r) at ($(tens2)+( 60:1)+(0.6,0)$) {};
        \draw[ar] (tens2)--(aux2l);
        \draw[ar, out=-80 , in=120] (der2) to node [edgename] {$a$} (tens2);
        \draw[ar, out=-160, in= 60](ax2r) to (tens2);
        \node[ax] (ax2l) at ($(ax2r)+(0,0.7)$) {};
        \draw[ar,out=-180,in=80] (ax2l) to (der2);
        \draw[ar,out=-20 ,in=90] (ax2l) to (princ2);
        \draw[ar,out=-40 ,in=90] (ax2r) to (aux2r);
        \draw (princ2) -| ($(princ2)+(0.8,3)$) -| ($(aux2l)+(-0.8,0)$) -- (aux2l)--(aux2r)--(princ2);
        \draw [ar] (princ2)--++(0,-1);
        \draw [ar] (aux2r) --++(0,-1) node [edgename] {$e$};
      \end{tikzpicture}
      \caption{\label{ex_for_equivs} $(a,[t],[!_{\sige}],-) \leftsquigarrow $}

    \end{figure}

In order to prove that the antipaths never separate, we will prove that their contexts are pairwise equivalent for a complex equivalence relation: $\sim_s$. The idea of the $\sim_s$ equivalence between contexts is: their edges and their $\mapsto_{s-1}$ canonical potential are equal. And, indeed, the actual $\sim_s$ definition will be (almost) equivalent to this when the traces are reduced to one element. But, the actual definition must be a bit more complex because we want this equivalence to be stable under anti-$\rightsquigarrow_s$ step. We will need two relations ($\approx_{s}$ and $\simeq_s$) to define $\sim_s$
    
Let $C$ and $C'$ be two contexts, $C \approx_s C'$ means: ``It is possible to make the same number of $\mapsto_s$ steps beginning by those two contexts and to reach contexts which have the same trace or a box of stratum $\geq s$. The edges of the contexts in the paths must be pairwise equal''. $(e,P) \simeq_s (e,Q)$ means: ``the exponential signatures of $P$ and $Q$ corresponding to box of stratum $< s$ must be either pairwise equal, or at least their surface is pairwise equal because a $\mapsto_s$ path leaving this box will arrive at a box of stratum $\geq s$''. These definitions are made to take into account cases similar to figure \ref{ex_for_equivs} where a difference on an exponential signature corresponding to the copy of a box of stratum greater than $s$ is transformed into a harmless difference in an exponential signature corresponding to the copy of a box of stratum strictly lower than $s$.
    
    \begin{definition}
      Let $(e,P,T,p)$ and $(e,Q,U,p)$ be two contexts and $s \in \mathbb{N}$. $(e,P,T,p) \approx_s (e,Q,U,p)$ if one of the following conditions holds:
      \begin{itemize}
      \item $T = U$
      \item $(e,P,T,p) \not \rightsquigarrow_s$, $(e,Q,U,p) \not \rightsquigarrow_s$, $(e,P,T,p) \not \hookrightarrow_{s-1}$ and $(e,Q,U,p) \not \hookrightarrow_{s-1}$
      \item $(e,P,T,p) \rightsquigarrow_s (e',P',T',p')$, $(e,Q,U,p) \rightsquigarrow_s (e',Q',U',p')$ and $(e',P',T',p') \approx_s (e',Q',U',p')$
      \item $(e,P,T,p) \hookrightarrow_{s-1} (e',P',T',p')$, $(e,Q,U,p) \hookrightarrow_{s-1} (e',Q',U',p')$ and $(e',P',T',p') \approx_s (e',Q',U',p')$
      \end{itemize}
    \end{definition}

    \begin{definition}
      Let $P,Q$ be canonical potentials of edge $e$ and $s \in \mathbb{N}$. $(e,P) \simeq_{s} (e,Q)$ is defined by:
      \begin{itemize}
      \item If $P=Q= [~]$, $(e,P) \simeq_{s} (e,Q)$
      \item If $P=P'.t$, $Q=Q'.u$ and $B$ is the deepest box containing $e$
        \begin{itemize}
        \item If $S(B) \geq s$, then $(e,P)\simeq_{s}(e,Q) \Leftrightarrow (\sigma(B),P') \simeq_{s}(\sigma(B),Q')$
        \item If $S(B) < s$, then $(e,P) \simeq_{s}(e,Q) \Leftrightarrow \left\{\begin{array}{l}(\sigma(B),P') \simeq_{s}(\sigma(B),Q') \\ (\sigma(B),P',[\oc_t],+) \approx_s (\sigma(B), Q',[\oc_u],+) \end{array} \right.$
        \end{itemize}
      \end{itemize}

      \paragraph{}We can notice that whenever $e$ and $f$ belong to the same boxes, $(e,P) \simeq_{s} (e,Q) \Leftrightarrow (f,P) \simeq_{s} (f,Q)$. We will often write $P \simeq_{s} Q$ for $(e,P) \simeq_{s} (e,Q)$ when the edge $e$ we refer to (or at least the boxes containing it) can be deduced from the sentences around it.
    \end{definition}

    \paragraph{} The definition of $\sim_s$ says: the contexts are the same except that the potentials and traces may differ on exponential signatures corresponding to boxes of strata $\geq s$ or on other exponential signatures in such a way that it will not make the antipaths separate until the respective signatures go in the potential of a $\geq s$ box.

    \begin{definition}
      $(e,P,[\oc_t]@T,p) \sim_s (e,P',[\oc_{t'}]@T',p)$ if all the following conditions stand:
      \begin{enumerate}
      \item \label{sim_simeq} $(e,P) \simeq_{s} (e,P')$
      \item \label{sim_exp} The skeletons of $T$ and $T'$ are equal.
      \item \label{sim_bang} If $T = U@[!_u]@V$, $T'=U'@[!_{u'}]@V'~$ and $|U|=|U'|$ then $(e,P,[!_u]@V~~,p~) \approx_s (e,P',[!_{u'}]@V'~~,p~~)$
      \item \label{sim_why}  If $T = U@[?_u]@V$, $T'=U'@[?_{u'}]@V'$ and $|U|=|U'|$  then $(e,P,[!_u]@V^\perp,p^\perp) \approx_s (e,P',[!_{u'}]@V'^\perp,p^\perp)$
      \end{enumerate}
    \end{definition}

    \begin{lemma}\label{injection_lemma}
      Let $C_e,C_e',C_f,C_f' \in C_G$, if $\left \{\begin{array}{l} C_e \rightsquigarrow_s C_f \\ C_e' \rightsquigarrow_s C_f' \\C_f \sim_s C_f' \end{array} \right \}$ then $C_e \sim_s C_e' $.
    \end{lemma}
    \begin{prf}
      The $\rightsquigarrow$-steps for which the result is hard to prove are the ones entering and leaving boxes.
      
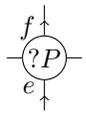
\begin{wrapfigure}{l}{1.5cm}
        \centering
        \begin{tikzpicture}
          \node[auxdoor] (aux) at (0,0) {};
          \coordinate (auxup) at ($(aux)+(0, 0.5)$);
          \coordinate (auxdo) at ($(aux)+(0,-0.5)$);
          \draw[->] (aux)--(auxup) node [pos=1,edgename] {$f$};
          \draw (auxup)--++(0, 0.2);
          \draw (aux)--(auxdo) node [pos=1,edgename] {$e$};
          \draw[<-] (auxdo)--++(0,-0.2);
          \draw (aux)--++(0.5,0) (aux)--++(-0.5,0);
        \end{tikzpicture}
      \end{wrapfigure}
        
      Suppose $C_e= (e,P,[!_{t}]@T.!_{u},-) \rightsquigarrow (f,P.u,[!_t]@T,-)=C_f$ (crossing an auxiliary door of box $B$ upwards). Then $C_f'= (f,P'.u',[!_{t'}]@T',-)$ with $(f,P.u) \simeq_{s} (f,P'.u')$ and the skeletons of $T$ and $T'$  are equal (definition of $\sim_s$). The predecessor of $C_f'$ is unique, so we have $C_e'=(e,P',[!_{t'}]@T.!_{u'},-)$. We can verify that $C_e \sim_s C_e'$. The only interesting point to prove is that $(e,P,[!_u],-)\approx_s (e,P',[!_{u'}],-)$ (for condition \ref{sim_bang} ):

      \begin{itemize}
      \item If $S(B) < s$, $(e,P,[!_u],-) \hookrightarrow_{s-1} (\sigma(B),P, [!_u],+)=C_g$ and $(e,P',[!_{u'}],-) \hookrightarrow_{s-1} (\sigma(B),P',[!_{u'}],+)=C_g'$. We know that $C_g \approx_s C_g'$ because $(f,P.u) \simeq_s (f,P'.u')$. So $(e,P,[!_u],-) \approx_s (e,P',[!_{u'}],-)$ (by rule 3 of the definition of $\approx_s$).
      \item If $S(B) \geq s$, $(e,P,[!_u],-) \approx_s (e,P',[!_{u'}],-)$ by point 2 of the definition of $\approx_s$
      \end{itemize}

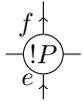
      \begin{wrapfigure}{l}{1.5cm}
        \begin{tikzpicture}
          \node[princdoor] (aux) at (0,0) {};
          \coordinate (auxup) at ($(aux)+(0, 0.5)$);
          \coordinate (auxdo) at ($(aux)+(0,-0.4)$);
          \draw[->] (aux)--(auxup) node [pos=1,edgename] {$f$};
          \draw (auxup)--++(0, 0.2);
          \draw (aux)--(auxdo) node [pos=1,edgename] {$e$};
          \draw[<-] (auxdo)--++(0,-0.2);
          \draw (aux)--++(0.5,0) (aux)--++(-0.5,0);
        \end{tikzpicture}
      \end{wrapfigure}

      \paragraph{}Suppose $C_e=(e,P,[!_t]@T.?_{u},-) \rightsquigarrow (f,P.u,[!_{t}]@T,-)=C_f$ (crossing the principal door of $B$ upwards). Because of the hypothesis and the stratification of the proof net, $S(B)<s$. So $C_f'= (f,P'.u',[!_{t'}].T',-)$ with $(e,P,[!_{u}],+) \approx_s (e,P',[!_{u'}],+)$ (because $(f,P.u) \simeq (f,P'.u')$). So $(e,P, [!_u],(-)^\perp) \approx_s (e,P', [!_{u'}],(-)^\perp)$. This gives us the condition \ref{sim_why} for the only new exponential signature on the trace. Because we know that $(f,P.u) \simeq_s (f,P'.u')$, we have that $(\sigma(f)=e,P) \simeq_s (\sigma(f)=e,P')$. The other conditions are straightforward.

      
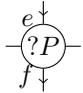
\begin{wrapfigure}{l}{1.5cm}
        \begin{tikzpicture}
          \node[auxdoor] (aux) at (0,0) {};
          \coordinate (auxup) at ($(aux)+(0, 0.4)$);
          \coordinate (auxdo) at ($(aux)+(0,-0.5)$);
          \draw (aux)--(auxup) node [pos=1,edgename] {$e$};
          \draw[<-] (auxup)--++(0, 0.2);
          \draw[->] (aux)--(auxdo) node [pos=1,edgename] {$f$};
          \draw (auxdo)--++(0,-0.1);
          \draw (aux)--++(0.5,0) (aux)--++(-0.5,0);
        \end{tikzpicture}
      \end{wrapfigure}

      \paragraph{}Suppose $C_e=(e,P.u, [!_t]@T,+) \mapsto (f,P, [!_t]@T.?_u,+)=C_f$ (crossing an auxiliary door of box $B$ downwards). Then $C_f'$ has the shape: $C_f'=(f,P',[!_{t'}]@T'.?_{u'},+)$. The only possibility for $C_e'$ is $C_e'=(e,P'.u',[!_{t'}]@T',+)$. We will show that $C_e \sim_s C'_e$. Only the first point offers some difficulties, the others are straightforward.
      \begin{itemize}
      \item If $S(B) \geq s$, then we only have to show that $(\sigma(B),P)\simeq_s (\sigma(B),P')$. $\sigma(B)$ and $f$ are included in the same boxes, so we only have to prove that $(f,P) \simeq_s (f,P')$. This is given by the hypothesis of $C_f \sim_s C'_f$
      \item If $S(B) < s$, then we have to show that $(\sigma(B),P) \simeq_s (\sigma(B),P')$ (which we can show as in the previous case) and that $(\sigma(B),P,[!_u],+) \approx_s (\sigma(B),P',[!_{u'}],+)$. We know that $C_f \sim_s C'_f$, so $(f,P,[!_u],-) \approx_s (f,P',[!_{u'}],-)$. Which of the four possible conditions of the definition of $\approx_s$ holds? It can not be the second or third one. So, this means that either $!_u = !_{u'}$ (and in this case $(\sigma(B),P,[!_u],+)\approx_s (\sigma(B),P',[!_{u'}],+)$ because of the first condition) or the $\hookrightarrow_{s-1}$ successors of $C_f$ and $C_f'$ (which are respectively $(\sigma(B),P,!_u,+)$ and $(\sigma(B),P',!_{u'},+)$) are $\approx_s$ equivalent which is the result needed.
      \end{itemize}

      
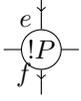
\begin{wrapfigure}{l}{1.5cm}
      \begin{tikzpicture}
        \node[princdoor] (aux) at (0,0) {};
        \coordinate (auxup) at ($(aux)+(0, 0.5)$);
        \coordinate (auxdo) at ($(aux)+(0,-0.4)$);
        \draw (aux)--(auxup) node [pos=1,edgename] {$e$};
        \draw[<-] (auxup)--++(0, 0.2);
        \draw[->] (aux)--(auxdo) node [pos=1,edgename] {$f$};
        \draw (auxdo)--++(0,-0.2);
        \draw (aux)--++(0.5,0) (aux)--++(-0.5,0);
      \end{tikzpicture}
      \end{wrapfigure}

      \paragraph{}Suppose $C_e= (e,P.u, [!_t]@T,+) \mapsto_s (f,P,[!_t]@T.!_u,+)=C_f$ (crossing the principal door of box $B$ downwards). We supposed that $C_f \sim_s C_f'$, so $C'_f$ has shape $C_f'=(f,P',[!_{t'}]@T'.!_{u'},+)$ and $C_e'=(e,P'.u',[\oc_{t'}]@T',+)$. We have $(f,P,[!_u],+) \approx_s (f,P',[!_{u'}],+)$ (condition 3 of the definition of $\sim$). Moreover we know that $(\sigma(B),P) \simeq_s (\sigma(B),P')$ (because of condition \ref{sim_simeq} of $C_f \sim_s C_f'$). So $(e,P.u) \simeq_s (e,P'.u')$. 

The other conditions are straightforward, for example, if you consider the decompositions $T=U.\oc_v@V$ and $T'=U'.\wn_{v'}@V'$ with $|U|=|U'|$. Then, we can find a similar decomposition for the traces of $C_f$ and $C_f'$: $T.!_u=U.\oc_v@(V.!_u)$ and $T'.!_{u'}=U'.\oc_{v'}@(V'.\oc_{u'})$. So according to condition \ref{sim_bang} of $C_f \sim_s C'_f$, we have $(f,P,[\oc_v]@V.\oc_u,+) \approx_s (f,P',[\oc_{v'}]@V'.\oc_{u'},+)$. Moreover, $(e,P.u,[\oc_v]@V,+) \rightsquigarrow_s (f,P,[\oc_v]@V.\oc_u,+)$ and $(e,P'.u',[\oc_{v'}]@V',+) \rightsquigarrow_s (f,P',[\oc_{v'}]@V'.\oc_{u'},+)$. So $(e,P.u,[\oc_v]@V,+) \approx_s (e,P'.u',[\oc_{v'}]@V',+)$. This proves the condition \ref{sim_bang}.

      \paragraph{}In the other cases, it is straightforward after unfolding the definitions.
      $\square$
    \end{prf}

    \subsection{Elementary bound for stratified proof-nets}
    \begin{lemma}[strong acyclicity]\label{lemma_strong_acyclicity}
      If $(e,P,[\oc_{t}],p) \mapsto_{s}^+ (e,Q,[\oc_{u}],p)$ then $P^{/s-1} \neq Q^{/s-1}$. 
    \end{lemma}
    \begin{proof}
      We will make a proof by contradiction. Let us suppose that $C_0=(e,P,[\oc_{t}],p) \mapsto_{s}^+ (e,Q,[\oc_{u}],q)=D_0$ and $P^{/s-1}=Q^{/s-1}$. Then, $C_0 \sim_s D_0$. Then, we define $D_1$ as the last-but-one context in the $C_0 \mapsto_s^+ D_0$ path. According to Lemma \ref{injection_lemma}, there exists a context $C_1$ such that $C_1 \mapsto C_0$ and $D_1 \sim_s C_1$. Moreover, $C_1 \mapsto^+_s D_1$. We can repeat this, creating an infinitely long path. In particular, this path will go through infinitely many contexts of shape $(e,R,[\oc_v],p)$. According to Theorem \ref{dallago_weight}, the number of canonical potentials for an edge is finite. So there is some $(e,R) \in Can(E_G)$, $v,v' \in Sig$ and $r \in Pol$ such that $(e,R,[\oc_v],p) \mapsto^* (e,R,[\oc_{v'}],p)$. This is impossible as we proved proof-nets to be acyclic (Lemma \ref{lemma_acyclicity}).
    \end{proof}

    \begin{theorem}\label{theoStratElementaryBound}
      If a proof-net is stratified, then the length of its longest path of reduction is bounded by  $2^{3.|E_G|}_{3.S_G+1}$
    \end{theorem}
    \begin{proof}
      Let us consider a $\mapsto_s$-copy simplification $u$ of a potential box $(B,P)$, as a tree. During the path beginning by $(\sigma(B),P,[\oc_u],+)$, the height of the left-most branch of $t$ (viewed as a tree) decreases to $0$ (the height of $\sige$). The height of the left-most branch decreases only by crossing a $\wn C$ or $\wn N$ upwards (which correspond to contexts of the shape $(e,Q,[\oc_v],q)$) and during those steps it decreases by exactly $1$. So the height of the left-most branch of $t$ is inferior to the number of contexts of the shape $(e,Q,[\oc_v],q)$ by which the path go through. From the strong acyclicity lemma, we can deduce that the height of the left-most branch is inferior to $\Set*{(e,[q_{i_1};…;q_{i_k}])}{
  \begin{array}{c} 
    e \in B_{\partial(e)} \subset … \subset B_1 \\
    \Set{j}{S(B_j) \leq s} = \{i_1,\cdots,i_k\} \\
    ~[q_1 ; \cdots ; q_{\partial(e)}] \in L_{\mapsto}(e)\\
  \end{array}
}$ which is itself inferior to $|Pot_{s-1}(E_G)|$. Let $t$ be a $\mapsto_s$-copy of $(B,P)$, then there exists a simplification $u$ of $t$ such that the heigth of $t$ is equal to the heigth of the left-most branch of $u$. So the size of such a copy is bounded by $2^{|Pot_{s-1}(E_G)|}$. Copies are standard signatures so there are $4$ possible symbols for each nodes of the signatures. Thus, for any potential box $(B,P)$, $|C_s(B,P)| \leq 4^{2^{|Pot_{s-1}(E_G)|}}$.
      \begin{align*}
        \max_{e \in E_G}|Pot_{s}(e)| \leq& \left(\max_{(C,Q) \in Pot(B_G)}|C_s(C,Q)|\right)^{\partial_G}\\
        \max_{e \in E_G}|Pot_{s}(e)| \leq& \left(4^{2^{|Pot_{s-1}(E_G)|}}\right)^{\partial_G}\\       
        \max_{e \in E_G}|Pot_{s}(e)| \leq& 2^{\partial_G \cdot 2^{1+|Pot_{s-1}(E_G)|}}\\       
        \max_{e \in E_G}|Pot_{s}(e)| \leq& 2^{\partial_G \cdot 2^{1+|E_G|\cdot \max_{e \in E_G}|Pot_{s-1}(e)|}}\\       
      \end{align*}

      We define $u_n$ as $2^{3.E_G}_{3 \cdot n}$. We will show by induction that, for every $n \in \mathbb{N}$, $\max_{e \in E_G}|Pot_{n-1}(e)| \leq u_n$. For $n=0$, for every $e \in E_G$, we have $|Pot_{-1}(e)|=1$ (the only canonical potentials are lists of $\sige$) and $u_0=2^{3.|E_G|}_0=3.|E_G|$.

      If $n \geq 0$, we have the following inequalities:
      \begin{align*}
        \max_{e \in E_G}|Pot_{n}(e)| \leq& 2^{\partial_G \cdot 2^{1+|E_G|\cdot \max_{e \in E_G}|Pot_{n-1}(e)|}} \\
        \max_{e \in E_G}|Pot_{n}(e)| \leq& 2^{2^{\log(\partial_G)+1+|E_G|\cdot u_n}} \\        
        \max_{e \in E_G}|Pot_{n}(e)| \leq& 2^{2^{u_n+|E_G|\cdot u_n}} \\        
        \max_{e \in E_G}|Pot_{n}(e)| \leq& 2^{2^{(1+|E_G|)\cdot u_n}} \\        
        \max_{e \in E_G}|Pot_{n}(e)| \leq& 2^{2^{\frac{u_n^2}{2}}} \\        
        \max_{e \in E_G}|Pot_{n}(e)| \leq& 2^{2^{2^{u_n}}} \\        
        \max_{e \in E_G}|Pot_{n}(e)| \leq& 2^{3.|E_G|}_{3.n} \\        
      \end{align*}
      
      Now that we have bounded canonical potentials, we can bound $T_G$. 
      \begin{align*}
        T_G =& \sum_{e \in E_G}|L_{\mapsto}(e)|+\sum_{B \in B_G}\left( |D_G(B)|\cdot \sum_{P \in L_{\mapsto}(B)}\sum_{t \in C_{\mapsto}(B,P)}|t| \right)\\
        T_G \leq& |E_G|.2^{3.|E_G|}_{3.S_G}+2 S_G \cdot \max_{B \in B_G}|D_G(B)| \cdot 2^{3.|E_G|}_{3.S_G} \cdot 2^{3.|E_G|}_{3.S_G}\\
        T_G \leq& |E_G|.2^{3.|E_G|}_{3.S_G}+ |E_G| \cdot (2^{3.|E_G|}_{3.S_G})^2\\
        T_G \leq& 2^{3.|E_G|}_{3.S_G+1}\\
      \end{align*}
    \end{proof}

    \section{Dependence control}
    \label{section_dependence}
    \paragraph{}Though stratification gives us a bound on the length of the reduction, elementary time is not considered as a reasonable bound. Figure \ref{exp} shows us a way for the complexity to arise, despite stratification. On this proof net, the box A duplicates the box B. Each copy of B duplicates C, each copy of C… In~\cite{roversi2009some}, this situation is called a chain of ``spindle''\label{def_spindle}. We call ``dependence control condition'' any restriction on linear logic which aims to tackle this kind of spindle chains. The solution chosen by Girard~\cite{girard1995light} was to limit the number of auxiliary doors of each $!$-boxes to 1. To keep some expressivity, he introduced a new modality $\S$ with $\S$-boxes which can have an arbitrary number of auxiliary doors.

    \begin{figure}
      \centering
      \begin{tikzpicture}
        \tikzstyle{door}=[draw, circle, inner sep=0.02cm]
        \draw (0,0) node [door] (bang1) {!P};
        \draw (bang1) node [below right] {$\mathbf{C}$}; 
        \draw (bang1)++(-1.2,-0.7) node (why1) {?C};
        \draw (why1)++(-0.4,0.7) node [door] (ancl1) {?P};
        \draw (why1)++(0.4,0.7) node [door] (ancr1) {?P};
        \draw (bang1) ++(0,0.6) node (tens1) {$\otimes$};
        \draw (tens1) ++(-0.6,0.3) node (ax1b) {$ax$};
        \draw (tens1) ++(-0.6,0.6) node (ax1h) {$ax$};
        \draw (bang1) ++ (1,-1) node (cut1) {$cut$};
        \draw [<-] (bang1)--(tens1);
        \draw [<-] (tens1) to [bend right] (ax1h);
        \draw [<-] (tens1) to [bend right] (ax1b);
        \draw [->] (ax1h) to [bend right] (ancl1);
        \draw [->] (ax1b) to [bend right] (ancr1);
        \draw (3,0) node (bang2) [door] {!P};
        \draw (bang2) node [below right] {$\mathbf{B}$}; 
        \draw (bang2)++(-1.2,-0.7) node (why2) {?C};
        \draw (why2)++(-0.4,0.7) node [door]  (ancl2) {?P};
        \draw (why2)++(0.4,0.7) node  [door] (ancr2) {?P};
        \draw (bang2) ++(0,0.6) node (tens2) {$\otimes$};
        \draw (tens2) ++(-0.6,0.3) node (ax2b) {$ax$};
        \draw (tens2) ++(-0.6,0.6) node (ax2h) {$ax$};
        \draw (bang2) ++ (1,-1) node (cut2) {$cut$};
        \draw [<-] (bang2)--(tens2);
        \draw [<-] (tens2) to [bend right] (ax2h);
        \draw [<-] (tens2) to [bend right] (ax2b);
        \draw [->] (ax2h) to [bend right] (ancl2);
        \draw [->] (ax2b) to [bend right] (ancr2);
        \draw (6,0) node (bang3)  [door] {!P};
        \draw (bang3) node [below right] {$\mathbf{A}$}; 
        \draw (bang3)++(-1.2,-0.7) node (why3) {?C};
        \draw (why3)++(-0.4,0.7) node [door] (ancl3) {?P};
        \draw (why3)++(0.4,0.7) node [door] (ancr3) {?P};
        \draw [->](bang1) to [out=-80, in=180] (cut1);
        \draw [<-](cut1) to [out=0, in=-140] (why2);
        \draw [->](bang2) to [out=-80, in=180] (cut2);
        \draw [<-] (cut2) to [out=0, in=-140] (why3);
        \draw (bang3) ++(0,0.6) node (tens3) {$\otimes$};
        \draw (tens3) ++(-0.6,0.3) node (ax3b) {$ax$};
        \draw (tens3) ++(-0.6,0.6) node (ax3h) {$ax$};
        \draw [<-] (bang3)--(tens3);
        \draw [<-] (tens3) to [bend right] (ax3h);
        \draw [<-] (tens3) to [bend right] (ax3b);
        \draw [->] (ax3h) to [bend right] (ancl3);
        \draw [->] (ax3b) to [bend right] (ancr3);
        \draw [->] (bang3) -- ++ (0,-1);
        \draw (bang1) --++(0.4,0) --++(0,1.5) --++(-2.4,0) |- (ancl1)--(ancr1)--(bang1);
        \draw (bang2) --++(0.4,0) --++(0,1.5) --++(-2.4,0) |- (ancl2)--(ancr2)--(bang2);
        \draw (bang3) --++(0.4,0) --++(0,1.5) --++(-2.4,0) |- (ancl3)--(ancr3)--(bang3);
        \draw [<-](why1) to [bend left] (ancl1);
        \draw [<-](why1) to [bend right](ancr1);
        \draw [<-](why2) to [bend left] (ancl2);
        \draw [<-](why2) to [bend right] (ancr2);
        \draw [<-](why3) to [bend left] (ancl3);
        \draw [<-](why3) to [bend right](ancr3);
        \draw (why1) ++(-1,-0.5) node (suite) {};
        \draw [->,dashed] (why1) to [bend left] (suite);
        
      \end{tikzpicture}
      \caption{This proof-net (if extended to $n$ boxes) reduces in $2^n$ steps}
      \label{exp}
    \end{figure}

    Baillot and Mazza generalized $ELL$ with $L^3$, a system capturing elementary time~\cite{baillot2010linear}. Contrary to $ELL$, $L^3$ allows dereliction and digging ($\wn D$ and $\wn N$ links). The presence of digging allows another way to create an exponential blow up, shown in Figure \ref{exp2}. Notice that in this second proof-net, all the boxes have at most one auxiliary door. So, contrary to the case of $ELL$ where the ``one auxiliary door'' condition alone ensures polynomial time, Baillot and Mazza added another restriction. They defined the $L^4$ proof-nets as the $L^3$ proof-nets without digging and with at most one auxiliary door by box. $L^4$ proof-nets normalize in polynomial time. However, we think that having all the links of linear logic (with some restriction on them) in $L^3$ was a nice feature and it is unfortunate that the authors could not keep the digging in $L^4$. In fact, there are no implicit characterization of polynomial time by subsystems of linear logic which keeps the digging. The criterion enforcing polynomial time normalization that we define in this paper does not forbid the digging. This could lead to a subsystem of linear logic characterizing $Ptime$ with a digging link.

    \begin{figure}
      \centering
      \begin{tikzpicture}
        \tikzstyle{door}=[draw, circle, inner sep=0.02cm]
        \node [princdoor] (bang1) at (0,0) {};
        \node (name1) at ($(bang1)+(0.4,-0.3)$) {$\mathbf{C}$}; 
        \node [auxdoor]   (aux1)  at ($(bang1)+(-1.4,  0)$) {};
        \node [cont]      (cont1) at ($(aux1) +(   0,0.7)$) {};
        \node [tensor]    (tens1) at ($(bang1)+(   0,0.8)$) {};
        \node [ax]        (ax1b)  at ($(tens1)+(-0.7,0.4)$) {};
        \node [ax]        (ax1h)  at ($(tens1)+(-0.7,0.7)$) {};
        \node [dig]       (dig1)  at ($(aux1) +(   0,-0.7)$){};
        \draw (bang1) -| ++(0.4,1.8) -| ($(aux1)+(-0.4,0)$) -- (aux1) -- (bang1);
        \draw [ar] (tens1)--(bang1);
        \draw [ar] (ax1b) to [out=  0,in=120] (tens1);
        \draw [ar] (ax1h) to [out=-10,in=70]  (tens1);
        \draw [ar] (ax1b) to [out=-170,in=60] (cont1);
        \draw [ar] (ax1h) to [out=-170,in=120](cont1);
        \draw [ar] (cont1) to (aux1);        
        \draw [ar] (aux1)  to (dig1);
        \draw [ar,dashed] (dig1) to [out=-90,in=0] ($(dig1)+(-0.8,-0.5)$);
        
        \node [princdoor] (bang2) at ($(bang1)+(3.5,0)$) {};
        \node (name2) at ($(bang2)+(0.4,-0.3)$) {$\mathbf{B}$}; 
        \node [auxdoor]   (aux2)  at ($(bang2)+(-1.4,  0)$) {};
        \node [cont]      (cont2) at ($(aux2) +(   0,0.7)$) {};
        \node [tensor]    (tens2) at ($(bang2)+(   0,0.8)$) {};
        \node [ax]        (ax2b)  at ($(tens2)+(-0.7,0.4)$) {};
        \node [ax]        (ax2h)  at ($(tens2)+(-0.7,0.7)$) {};
        \node [dig]       (dig2)  at ($(aux2) +(   0,-0.7)$){};
        \draw (bang2) -| ++(0.4,1.8) -| ($(aux2)+(-0.4,0)$) -- (aux2) -- (bang2);
        \draw [ar] (tens2)--(bang2);
        \draw [ar] (ax2b) to [out=  0,in=120] (tens2);
        \draw [ar] (ax2h) to [out=-10,in=70]  (tens2);
        \draw [ar] (ax2b) to [out=-170,in=60] (cont2);
        \draw [ar] (ax2h) to [out=-170,in=120](cont2);
        \draw [ar] (cont2) to (aux2);        
        \draw [ar] (aux2)  to (dig2);
        \node [cut](cut12) at ($(bang1)!0.5!(aux2)+(0,-1.2)$) {};
        \draw [ar] (bang1) to [out=-90,in=180] (cut12);
        \draw [ar] (dig2)  to [out=-100,in=  0] (cut12);

        \node [princdoor] (bang3) at ($(bang2)+(3.5,0)$) {};
        \draw [ar] (bang3) --++ (0,-1);
        \node (name3) at ($(bang3)+(0.4,-0.3)$) {$\mathbf{A}$}; 
        \node [auxdoor]   (aux3)  at ($(bang3)+(-1.4,  0)$) {};
        \node [cont]      (cont3) at ($(aux3) +(   0,0.7)$) {};
        \node [tensor]    (tens3) at ($(bang3)+(   0,0.8)$) {};
        \node [ax]        (ax3b)  at ($(tens3)+(-0.7,0.4)$) {};
        \node [ax]        (ax3h)  at ($(tens3)+(-0.7,0.7)$) {};
        \node [dig]       (dig3)  at ($(aux3) +(   0,-0.7)$){};
        \draw (bang3) -| ++(0.4,1.8) -| ($(aux3)+(-0.4,0)$) -- (aux3) -- (bang3);
        \draw [ar] (tens3)--(bang3);
        \draw [ar] (ax3b) to [out=  0,in=120] (tens3);
        \draw [ar] (ax3h) to [out=-10,in=70]  (tens3);
        \draw [ar] (ax3b) to [out=-170,in=60] (cont3);
        \draw [ar] (ax3h) to [out=-170,in=120](cont3);
        \draw [ar] (cont3) to (aux3);        
        \draw [ar] (aux3)  to (dig3);
        \node [cut](cut23) at ($(bang2)!0.5!(aux3)+(0,-1.2)$) {};
        \draw [ar] (bang2) to [out=-90,in=180] (cut23);
        \draw [ar] (dig3)  to [out=-100,in=  0] (cut23);
      \end{tikzpicture}
      \caption{This proof-net (if extended to $n$ boxes) reduces in $2^n$ steps}
      \label{exp2}
    \end{figure}

    The ``one auxiliary door''condition forbids a great number of proof-nets where there are boxes with more than one auxiliary door but whose complexity is still polynomial. The complexity explosion in Figure \ref{exp} comes from the fact that two copies of a box $B$ fuse with the same box $A$. A box with several auxiliary doors is only harmful if two of its auxiliary edges are contracted as in Figure \ref{exp}. Moreover, let us recall that we are interested in the complexity of functions, not stand-alone proof-nets. We say that the complexity of proof-net $G$ is polynomial if there is a polynomial $P_G$ such that whenever $G$ is cut with a proof-net $H$ in normal form, the resulting proof-net normalizes in at most $P_G(|E_H|)$ cut-elimination steps. $G$ is fixed and $P_G$ depends on $G$, so we can create a proof-net which has Figure \ref{exp} as a subproof-net and still is in $Ptime$. In fact, as Figure \ref{eightIsConst} shows, such a proof-net can even normalize in constant time.

    \begin{figure}
      \centering
      \begin{tikzpicture}
        \tikzstyle{door}=[draw, circle, inner sep=0.02cm]
        \draw (0,0) node [door] (bang1) {!P};
        \draw (bang1) node [below right] {$\mathbf{C}$}; 
        \draw (bang1)++(-1.2,-0.7) node (why1) {?C};
        \draw (why1)++(-0.4,0.7) node [door] (ancl1) {?P};
        \draw (why1)++(0.4,0.7) node [door] (ancr1) {?P};
        \draw (bang1) ++(0,0.6) node (tens1) {$\otimes$};
        \draw (tens1) ++(-0.6,0.3) node (ax1b) {$ax$};
        \draw (tens1) ++(-0.6,0.6) node (ax1h) {$ax$};
        \draw (bang1) ++ (1,-1) node (cut1) {$cut$};
        \draw [<-] (bang1)--(tens1);
        \draw [<-] (tens1) to [bend right] (ax1h);
        \draw [<-] (tens1) to [bend right] (ax1b);
        \draw [->] (ax1h) to [bend right] (ancl1);
        \draw [->] (ax1b) to [bend right] (ancr1);
        \draw (3,0) node (bang2) [door] {!P};
        \draw (bang2) node [below right] {$\mathbf{B}$}; 
        \draw (bang2)++(-1.2,-0.7) node (why2) {?C};
        \draw (why2)++(-0.4,0.7) node [door]  (ancl2) {?P};
        \draw (why2)++(0.4,0.7) node  [door] (ancr2) {?P};
        \draw (bang2) ++(0,0.6) node (tens2) {$\otimes$};
        \draw (tens2) ++(-0.6,0.3) node (ax2b) {$ax$};
        \draw (tens2) ++(-0.6,0.6) node (ax2h) {$ax$};
        \draw (bang2) ++ (1,-1) node (cut2) {$cut$};
        \draw [<-] (bang2)--(tens2);
        \draw [<-] (tens2) to [bend right] (ax2h);
        \draw [<-] (tens2) to [bend right] (ax2b);
        \draw [->] (ax2h) to [bend right] (ancl2);
        \draw [->] (ax2b) to [bend right] (ancr2);
        \draw (6,0) node (bang3)  [door] {!P};
        \draw (bang3) node [below right] {$\mathbf{A}$}; 
        \draw (bang3)++(-1.2,-0.7) node (why3) {?C};
        \draw (why3)++(-0.4,0.7) node [door] (ancl3) {?P};
        \draw (why3)++(0.4,0.7) node [door] (ancr3) {?P};
        \draw [->](bang1) to [out=-80, in=180] (cut1);
        \draw [<-](cut1) to [out=0, in=-140] (why2);
        \draw [->](bang2) to [out=-80, in=180] (cut2);
        \draw [<-] (cut2) to [out=0, in=-140] (why3);
        \draw (bang3) ++(0,0.6) node (tens3) {$\otimes$};
        \draw (tens3) ++(-0.6,0.3) node (ax3b) {$ax$};
        \draw (tens3) ++(-0.6,0.6) node (ax3h) {$ax$};
        \draw [<-] (bang3)--(tens3);
        \draw [<-] (tens3) to [bend right] (ax3h);
        \draw [<-] (tens3) to [bend right] (ax3b);
        \draw [->] (ax3h) to [bend right] (ancl3);
        \draw [->] (ax3b) to [bend right] (ancr3);
        \draw (bang1) --++(0.4,0) --++(0,1.5) --++(-2.4,0) |- (ancl1)--(ancr1)--(bang1);
        \draw (bang2) --++(0.4,0) --++(0,1.5) --++(-2.4,0) |- (ancl2)--(ancr2)--(bang2);
        \draw (bang3) --++(0.4,0) --++(0,1.5) --++(-2.4,0) |- (ancl3)--(ancr3)--(bang3);
        \draw [<-](why1) to [bend left] (ancl1);
        \draw [<-](why1) to [bend right](ancr1);
        \draw [<-](why2) to [bend left] (ancl2);
        \draw [<-](why2) to [bend right] (ancr2);
        \draw [<-](why3) to [bend left] (ancl3);
        \draw [<-](why3) to [bend right](ancr3);
        \draw (why1) ++(-1,-0.5) node (suite) {};

        \node [par]    (par)  at ($(why1)!0.5!(bang3)+(0,-1.5)$) {};
        \draw [ar]     (why1)  to [out=- 65,in=175] (par);
        \draw [ar]     (bang3) to [out=-115,in= 15] (par);
        \node [tensor] (tens) at ($(par)+(5,0)$) {};
        \node [cut]    (cut)  at ($(par)!0.5!(tens)+(0,-0.5)$) {};
        \draw [ar]     (tens)  to [out=-145,in=  0] (cut);
        \draw [ar]     (par)   to [out= -35,in=180] (cut);
        \node [ax]     (finax) at ($(tens)+(0.8,0.5)$) {};
        \draw [ar]     (finax) to [out=180,in= 60] (tens);
        \draw [ar]     (finax) to [out=  0,in=120] ($(finax)+(0.8,-0.5)$);
        \node [proofnet] (H) at ($(tens)+(0,1.5)$) {$H$};
        \draw [ar] (H) -- (tens);
      \end{tikzpicture}
      \caption{If $H$ is in normal form, this proof-net reduces in exactly 32 $cut$-elimination steps}
      \label{eightIsConst}
    \end{figure}
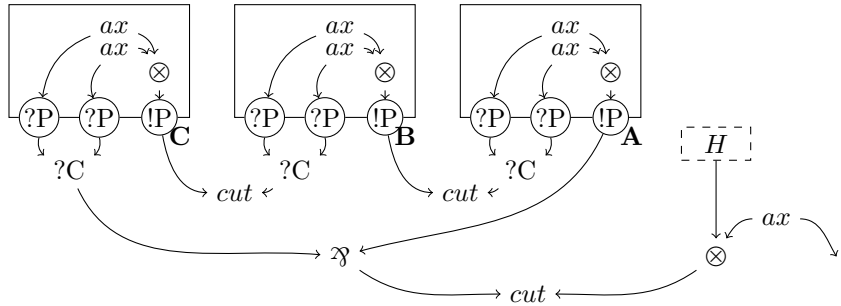

    What really leads to an exponential blow up is when the length of such a chain of spindles depends on the input, as in Figure \ref{reallyExp}. If we replace the sub proof-net $H$ (which represents $3$) by a proof-net $H'$ representing $n$, the resulting proof-net normalizes in time $2^n$.

    \begin{figure}
      \centering
      \begin{tikzpicture}
        \node [princdoor]  (bang) at (0,0) {};
        \node at ($(bang)+(0.4,-0.3)$) {$\mathbf{B}$};
        \node [auxdoor]    (aux2)  at ($(bang)+(-1.6,0)$) {};
        \node [auxdoor]    (aux1)  at ($(aux2) +(-0.8,0)$) {};
        \node [cont]       (cont)  at ($(aux1)!0.5!(aux2)+(0,-0.7)$) {};
        \node [tensor]     (tens)  at ($(bang)+(0,0.7)$)  {};
        \node [ax]         (axb)   at ($(tens)+(-0.7,0.35)$) {};
        \node [ax]         (axh)   at ($(tens)+(-0.7,0.7)$) {};
        \draw (bang) -| ++(0.4,1.65) -| ($(aux1)+(-0.4,0)$) -- (aux1) -- (aux2) -- (bang);
        \draw [ar] (tens)--(bang);
        \draw [ar] (axh) to [out=  0,in= 75]  (tens);
        \draw [ar] (axb) to [out=  0,in=120]  (tens);
        \draw [ar] (axh) to [out=180,in= 90] (aux1);
        \draw [ar] (axb) to [out=180,in= 90] (aux2);
        \draw [ar] (aux1) -- (cont);
        \draw [ar] (aux2) -- (cont);
        \node [par]        (par)   at ($(cont)!0.5!(bang)+(0,-1.7)$) {};
        \node [forall]     (fa)    at ($(cont)!0.5!(par)$) {};
        \draw [name path=bangPar,opacity=0] (bang)--(par);
        \draw [name path=horizontFa,opacity=0] ($(fa)+(-2,0)$) -- ($(fa)+(2,0)$);
        \node [exists, name intersections={of=bangPar and horizontFa}] (ex) at (intersection-1) {};
        \draw [ar] (cont) -- (fa); \draw [ar] (fa)--(par);
        \draw [ar] (bang) -- (ex); \draw [ar] (ex)--(par);

        \node [princdoor] (bangf) at ($(par)+(0,-0.7)$) {};
        \draw (par) -- (bangf);
        \draw (bangf) -| ++(1.7,4.8) -| ($(bangf)+(-1.9,0)$) -- (bangf);
        \node [tensor]    (tensf) at ($(bangf)+(-18:3.2)$) {};
        \node [der]       (derf)  at ($(tensf)+( 50:0.8)$) {};
        \node [tensor]    (tensfx)at ($(derf) +( 60:0.8)$) {};
        \node [exists]    (exx)   at ($(tensfx)+(120:0.8)$) {};
        \node [par]       (parx)  at ($(exx)  +(0,1)$) {};
        \node [ax]        (axx)   at ($(parx) +(0,1)$) {};
        \node [ax]        (axfx)  at ($(tensfx)+(0.8,0.5)$) {};
        \node [exists]    (exf)   at ($(tensf)+(-60:0.8)$)  {};
        \node [par]       (parf)  at ($(exf)  +(-60:0.8)$)  {};
        \draw [ar] (axx) to [out= -60,in= 60] (parx);
        \draw [ar] (axx) to [out=-120,in=120] (parx);
        \draw [ar] (parx)-- (exx); 
        \draw [ar] (exx) -- (tensfx);
        \draw [ar] (tensfx)--(derf);
        \draw [ar] (derf)  --(tensf);
        \draw [ar] (bangf) to [out=-80, in=162] (tensf);
        \draw [ar] (tensf) --(exf);
        \draw [ar] (exf)   --(parf);
        \draw [ar] (axfx) to [out=-170,in=60] (tensfx);
        \draw [ar] (axfx) to [out=-10, in=60] (parf);
        
        \node [tensor]  (app) at ($(parf)+(7,0)$) {};
        \node [cut]     (cut) at ($(parf)!0.5!(app)+(0,-0.9)$) {};
        \draw [ar] (parf) to [out=-50,in=180] (cut);
        \draw [ar] (app) to [out=-130,in=  0] (cut);
        \node [ax]     (finax) at ($(app)+(0.8,0.5)$) {};
        \draw [ar]     (finax) to [out=180,in= 60] (app);
        \draw [ar]     (finax) to [out=  0,in=120] ($(finax)+(0.8,-0.5)$);

        \node [forall]       (forall)   at ($(app)+(-0.3,1.4)$) {};
        \draw (forall) -- (app);
        \node [par]          (parg)     at ($(forall)+(0,0.8)$) {};
        \node [princdoor]    (princ3g)  at ($(parg)+(1, 2)$) {};
        \node at ($(princ3g)+(0.6,-0.4)$) {$\mathbf B_2$}; 
        \node [auxdoor]      (aux1g)    at ($(princ3g)+(-4,0)$) {};
        \node [auxdoor]      (aux2g)    at ($(aux1g)!0.333!(princ3g)$) {};
        \node [auxdoor]      (aux3g)    at ($(aux1g)!0.666!(princ3g)$) {};
        \node [cont]         (cont1g)   at ($(aux1g)!0.5!(aux2g)+(0,-0.9)$) {};
        \node [cont]         (cont2g)   at ($(cont1g)!0.5!(parg)$) {};
        \nvar{\hautTens}{1.3cm}
        \node [tensor]       (tens1g)   at ($(aux1g)+(0,\hautTens)$) {};
        \node [tensor]       (tens2g)   at ($(aux2g)+(0,\hautTens)$) {};
        \node [tensor]       (tens3g)   at ($(aux3g)+(0,\hautTens)$) {};
        \nvar{\decAx}{0.6cm}
        \node [ax]        (ax1n)     at ($(tens1g)+(-0.5,\decAx)$) {};
        \node [ax]        (ax2n)     at ($(tens1g)!0.5!(tens2g)+(0,\decAx)$) {};
        \node [ax]        (ax3n)     at ($(tens2g)!0.5!(tens3g)+(0,\decAx)$) {};
        \node [ax]        (ax4n)     at ($(princ3g)+(0,\decAx + \hautTens)$) {};
        \node [par]          (parx)     at ($(princ3g)+(0,0.7)$) {};
        \draw [ar] (parg) -- (forall);
        \draw [ar] (cont2g) -- (parg);
        \draw [ar] (cont1g) -- (cont2g);
        \draw [ar] (aux1g) -- (cont1g);
        \draw [ar] (aux2g) -- (cont1g);
        \draw [ar] (aux3g) -- (cont2g);
        \draw [ar] (parx) -- (princ3g);
        \draw [ar] (tens1g) -- (aux1g);
        \draw [ar] (tens2g) -- (aux2g);
        \draw [ar] (tens3g) -- (aux3g);
        \draw [ar,out=-140,in= 160] (ax1n) to (parx);
        \draw [ar,out=-30 ,in= 120] (ax1n) to (tens1g);
        \draw [ar,out=-150 ,in=  60] (ax2n) to (tens1g);
        \draw [ar,out=0   ,in=120 ] (ax2n) to (tens2g);
        \draw [ar,out=180 ,in=  60] (ax3n) to (tens2g);
        \draw [ar,out=0   ,in=120 ] (ax3n) to (tens3g);
        \draw [ar,out=180 ,in=  60] (ax4n) to (tens3g);
        \draw [ar,out=-20 ,in=  60] (ax4n) to (parx);
        \draw (princ3g) -| ++(1,\hautTens + \decAx +0.5cm) -| ($(aux1g)+(-1.2,0)$) -- (aux1g) -- (aux2g) -- (aux3g) -- (princ3g);
        \draw (princ3g) -- (parg);

        \draw [dashed] ($(aux1)+(-0.7,2.2)$) rectangle ($(parf)+(1.8,-0.5)$);
        \node at ($(parf)+(2,0)$) {$\mathbf{G}$};
        \draw [dashed] ($(forall)+(2.5,-0.5)$) rectangle ($(ax1n)+(-1,1.2)$);
        \node at ($(forall)+(2.7,0)$) {$\mathbf{H}$};
      \end{tikzpicture}
      \caption{\label{reallyExp}The sub proof-net $G$ is not polynomial time}
      \end{figure}
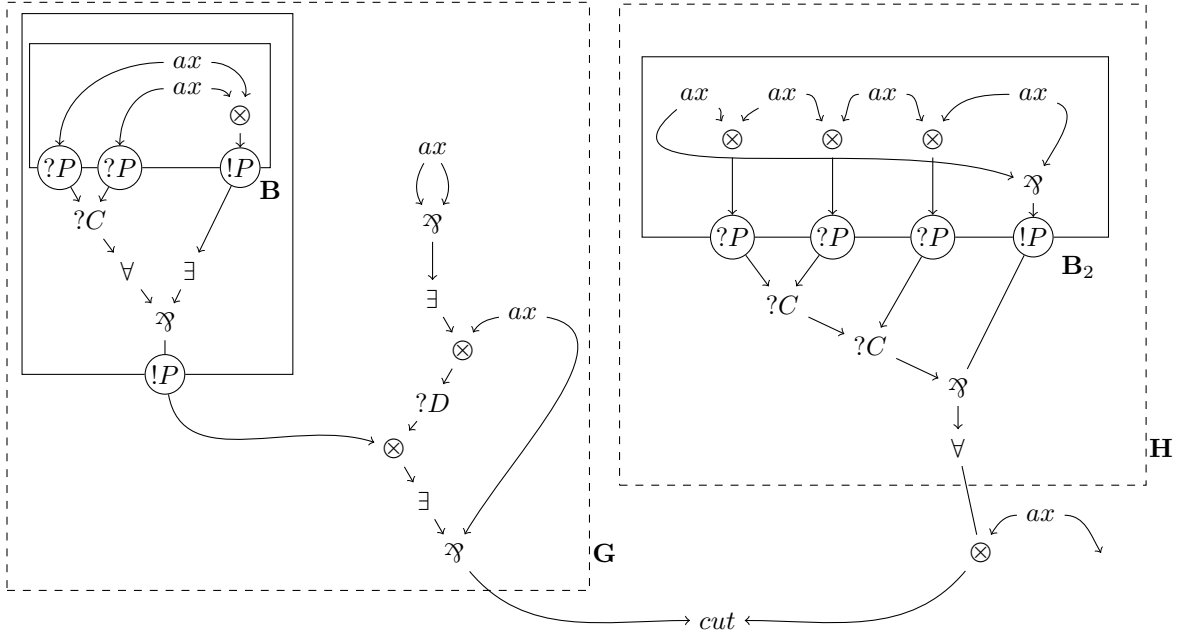
 That is the reason why, in the system $L^{3a}$~\cite{dorman2009linear}, Dorman and Mazza replaced the ``one auxiliary door'' condition by a looser dependence control condition: each edge is labelled with an integer, the label of an auxiliary edge must be greater or equal to the label of the principal edge of the box, for a given box at most one auxiliary edge can have the same label as the principal edge. Thus, in figure \ref{exp}, either $\sigma_0(A)$ or $\sigma_1(A)$ has a label greater than the label of $\sigma(A)$. They are contracted so they must have the same label, which will also be the label of $\sigma(B)$. Thus the label of $\sigma(A)$ is inferior to the label of $\sigma(B)$ which is inferior to the label of $\sigma(C)$. In general, the length of chains of spindles is bounded by the maximum label of the proof-net, which does not depend on the input. The dependence control of $L^{3a}$ seems to give a greater expressive power than the dependence control of $LLL$. In our view, the main limitation of $L^{3a}$ is that it uses the same labels to control dependence and to enforce stratification. This entails useless constraints on the strata corresponding to the auxiliary edges of boxes.

 Our dependence control condition is closer to $MS$.

In ~\cite{roversi2009some}, Roversi and Vercelli proposed to relax this discipline by considering a framework of logics, $MS$. $MS$ is defined as a set of subsystems of $ELL$ with indexes on $\oc$ and $\wn$ connectives. Roversi and Vercelli provide a sufficient criterion on those systems to ensure that a system is $Ptime$. This criterion intuitively says that a $MS$ system is $Ptime$ if and only if one of the two following condition holds:
\begin{itemize}
\item If $\wn_{i}A$ and $\wn_jA$ can be contracted in $\wn_k A$, then $i \geq k$, $j \geq k$ {\bf and at least one of those comparison is strict}. And for every boxes, the indexes on the $\wn$-s of the auxiliary doors are greater or equal to the index of the $\oc$ of the principal doors.
\item If $\wn_{i}A$ and $\wn_jA$ can be contracted in $\wn_k A$, then $i \geq k$, $j \geq k$. And for every boxes, the indexes on the $\wn$-s of the auxiliary doors are greater or equal to the index of the $\oc$ of the principal doors, {\bf with all but (at most) one of those comparisons being strict}.
\end{itemize}
In the following, we propose instead a criterion on proof-nets implying a polynomial time bound. Our criterion is more general, every proof-net satisfying the criterion of~\cite{roversi2009some} satisfies our criterion. On $ELL$ proof-nets, our criterion seems close to the $MS$ criterion. However, our dependence criterion entails polynomial time normalization on any stratified proof-net, while the $MS$ criterion entails polynomial time entails polynomial time normalization only on $ELL$ proof-nets. Intuitively, their criterion only deals with the kind of blow-up of Figure \ref{exp} and does not deal with the kind of blow-up of Figure \ref{exp2}. At the end of Section \ref{section_applications}, we will give more comparison between the two approaches.

We try to have as few false negatives as possible for our criterion (proof-nets which are in $Ptime$ but do not satisfy the criterion) so we will only forbid proofnets where, along the cut elimination, two (or more) duplicates of a box $B$ join the same duplicate of a box $B$. Indeed, suppose a chain of spindles appears during cut elimination and that the boxes of the sequence are duplicates of pairwise distinct boxes of the original proofnet. Then, the length of the sequence is bounded by the number of boxes of the original proofnet, so the sequence would be harmless. Our condition is given in the following way: we first define a relation $B \succcurlyeq_k B'$ ($B$ $k$-joins $B'$) on boxes meaning that at least $k$ duplicates of $B$ join $B'$ and we say that a proof net controls dependence if $\succcurlyeq_2$ is acyclic.

    \begin{definition}[$B$ $k$-joins $B'$]\label{def_kjoins}~\\
      \begin{itemize}
      \item
        We define $(B,P)$ $k$-joins directly $(B',P')$ in stratum $s$ as
        \begin{equation*}
          (B,P) \succcurlyeq_k^s (B',P') \Leftrightarrow k = \left | \Set*{t \in C_{s}(B,P) }{ \exists u, t\sqsubseteq u \text{ and } (\sigma(B),P, [!_{u}], +) \rightsquigarrow^*_s ( \sigma_i(B'), P', [!_{\sige}], -) } \right|.
        \end{equation*}
      \item We define $B$ $k$-joins $B'$ in stratum $s$ as: 
        \begin{equation*}
          B \succcurlyeq_k^s B' \Leftrightarrow k \leq \max \limits_{\substack{P \in L_{\mapsto}(B)\\ P' \in L_{\mapsto}(B')}} \sum_{(B,P) \succcurlyeq_{k_1}^s (B_1,P_1) \succcurlyeq_{k_2}^s \cdots \succcurlyeq_{k_n}^s (B',P')}{k_1 \cdot  k_2 \cdots \cdot k_n}
        \end{equation*}
      \end{itemize}
    \end{definition}

    \begin{definition}\label{def_controlsdependence}
      A principal door stratified proof net $G$ controls dependence if $\succcurlyeq_2^{S(G)}$ is irreflexive.
    \end{definition}
    For example, in Figure \ref{exp2}, we have $(B,[]) \succcurlyeq^0_2 (A,[])$ because $\sign(\sigr(\sige),\sige)$ and $\sign(\sigl(\sige),\sige)$ are $\mapsto_0$-copies of $(B,[])$, $\sigp(\sige)$ is a simplification of both and $(\sigma(B),[],[\oc_{\sigp(\sige)}],+) \mapsto_0^2 (\sigma_0(A),[],[\oc_{\sige}],-)$.

    The proof-net of Figure \ref{reallyExp} does not control dependence because $(B,[\sigl(\sigr(\sige))]) \succcurlyeq^s_2 (B,[\sigr(\sige)])$ so $B \succcurlyeq^0_2 B'$. Indeed, $\sigl(\sige)$ and $\sigr(\sige)$ are $\mapsto_0$-copies of $(B,[\sigl(\sigr(\sige))])$, any signature is its own simplification, $(\sigma(B),[\sigl(\sigr(\sige))],[\oc_{\sigl(e)}],+) \mapsto^{29} (\sigma_1(B),[\sigr(\sige)],[\oc_{\sige}],-)$ and $(\sigma(B),[\sigl(\sigr(\sige))],[\oc_{\sigr(e)}],+) \mapsto^{29} (\sigma_2(B),[\sigr(\sige)],[\oc_{\sige}],-)$.
    
    As in Section \ref{section_stratification}, we defined our criterion as the acyclicity of a relation on boxes. The methodology will be similar, we will prove that the number of $\mapsto_s$-copies of a box $B$ can be bounded by the number of $\mapsto_s$-copies of the boxes $B'$ with $B \succcurlyeq^s_2 B'$.\label{def_nest} If a proof net controls dependence, then for every box $B$ we define the nest of $B$ at stratum $s$ (written $N_s(B)$) as the depth of $B$ in terms of the $\succcurlyeq^s_2$ relation. So the notion of nest is the equivalent for dependence control of the notion of stratum for stratification. We can notice that for every $s \leq s'$, $\succcurlyeq^s_2 \subseteq \succcurlyeq^{s'}_2$, so $N_s(B) \leq N_{s'}(B)$. We will write $N(B)$ for $N_{S_G}(B)$, thus for every $s \in \mathbb{N}$, $N_s(B) \leq N(B)$. Finally, $N_G$ will stand for $\max_{B \in B_G}N(B)$.
   
    The proof will be done in two main steps. First, in Subsection \ref{subsection_itineraries}, given a potential box $(B,P)$ and a potential edge $(e,Q)$, we will bound the number of ``different paths'' that the $\mapsto_s$-copies of $(B,P)$ whose associated paths go through the context $(e,Q,[\oc_{\sige}],-)$ can take. Those ``different paths'' will be captured by the notion of itinerary. This subsection deals with the proof-nets similar to Figure \ref{exp}. 
    Then, in Subsection \ref{subsection_depcontrol_digging}, given a potential edge $(e,Q)$ and an itinerary $I$ from $(B,P)$ to $(e,Q)$, we will bound the number of $\mapsto_s$-copies $t$ of $(B,P)$ whose associated paths take the itinerary $I$ from $(B,P,[\oc_{t}],+)$ to $(e,Q,[\oc_{\sige}],-)$. This subsection is heavily involved with digging and deals with the proof-nets similar to Figure \ref{exp2}.
    
    Finally, in Subsection \ref{subsection_depcontrol_final}, we will compose those two results to give a polynomial bound on stratified proof-nets controlling dependence.

\subsection{Bound on the number of itineraries}  
\label{subsection_itineraries}
As we said, this subsection deals with the proof-nets similar to Figure \ref{exp2}, we will refine the $\succcurlyeq^s_k$ relation on potential boxes to deal only with this kind of dependence.

\label{def_geqsk}Let $(B,P)$ and $(B',P')$ be potential boxes, 
\begin{equation*}
  (B,P) \geq^s_k (B',P') \Leftrightarrow k = \left| \Set{(\sigma_i(B'),P')}{\exists t \in Sig, (\sigma(B),P,[\oc_t],+) \rightsquigarrow^*_s (\sigma_i(B'),P',[\oc_{\sige}],-)}\right|
\end{equation*}

We can notice that if $(B,P) \succcurlyeq^s_j (B',P')$ and $(B,P) \geq^s_k (B',P')$ then $j \geq k$.

In Section \ref{section_stratification}, we proved Lemma \ref{injection_lemma} which tells us that, if $(\sigma(B),P,[\oc_t],+) \rightsquigarrow^*_s (e,P,[\oc_{\sige}],-)$ and $(\sigma(B),P,[\oc_u],+) \rightsquigarrow^*_s (e,P',[\oc_{\sige}],-)$ with $P^{/s-1}=P'^{/s-1}$, then we can follow the paths from their end and we can observe that the edges of the paths are pairwise equal. So, we crossed exactly the same $\wn N$ and $\wn C$ nodes in the two paths. Thus, we can deduce that $t=u$. To be precise, we can do so only in the case where there is no digging, but we will deal with it in the next subsection.

Now, let us suppose that $(\sigma(B),P,[\oc_t],+) \mapsto^*_s (e,P,[\oc_{\sige}],-)$ and $(\sigma(B),P,[\oc_u],+) \mapsto^*_s (e,P',[\oc_{\sige}],-)$ with $P^{/s-1}=P'^{/s-1}$. We can follow the paths from their end and we will deduce that we are in the situation $(\sigma(C),R,[\oc_{v}],+) \mapsto^*_s (e,P,[\oc_{\sige}],-)$ and $(\sigma(C),R',[\oc_{v}],+) \mapsto^*_s (e,P',[\oc_{\sige}],-)$ with $R^{/s-1}=R'^{s-1}$. However, if $B \neq C$ and the box $C$ has more than one auxiliary door, then we do not know if the contexts $(\sigma(C),R,[\oc_{v}],+)$ and $(\sigma(C),R',[\oc_{v}],+)$ come from the same auxiliary door. So to know exactly the edges by which the paths has gone through between $(\sigma(B),P,[\oc_{t}],+)$ and $(e,Q,[\oc_{\sige}],-)$ we not only have to know $(e,Q^{/s-1})$, but also from which auxiliary door we came from for each $\hookrightarrow$ step. We will capture this notion of choices of auxiliary door with the notion of {\em itineraries}.

    \begin{definition}\label{def_itinerary}
      Let $C,C'$ be contexts of $G$ such that $C \mapsto C'$, the itinerary between $C$ and $C'$ (written $I(C,C')$) is the list of natural numbers $[i_1; i_2; \cdots ;  i_n]$ such that 

        \begin{tabular}{cccc}
          & $C$ &$\rightsquigarrow^*$ & $(\sigma_{\boldsymbol{i_1}}(B_1),P_1,[!_{t_1}],-)$ \\
          $\hookrightarrow$  & $(\sigma(B_1),P_1, [!_{t_1}],+)$ &$\rightsquigarrow^*$ & $(\sigma_{\boldsymbol{i_2}}(B_2),P_2,[!_{t_2}],-)$ \\
          $\hookrightarrow$  & $\cdots$  & $\rightsquigarrow^*$ &$(\sigma_{\boldsymbol{i_n}}(B_n),P_n,[!_{t_n}],-)$ \\
          $\hookrightarrow$  & $(\sigma(B_n),P_n, [!_{t_n}],+)$ &$\rightsquigarrow^*$ & $C'$
        \end{tabular}
        If $t \in Sig$ and $(\sigma(B),P,[\oc_t],+) \mapsto^* C' \not \mapsto$,then $I(B,P,t)$ refers to $I((\sigma(B),P,[\oc_t],+),C')$. 

        We will also write $I_s(B,P)$ for $\Set{I(B,P,t) }{ t \in Si_{s}(B,P) }$ and for every $(e,Q_{s-1}) \in Can_{s-1}(E_G)$ and $(B,P) \in Pot(E_G)$ we define
        \begin{equation*}
          I_s((B,P),(e,Q_{s-1}))= \Set{I((\sigma(B),P,[\oc_t],+),(e,Q,[\oc_{\sige}],-))}{ t \in Si_{s}(B,P) \text{ and } Q^{/s-1}=Q_{s-1}}
        \end{equation*}
    \end{definition}

    \label{def_colony}We will need more details on $\geq^s_2$, so we define for any potential box $(B,P)$, the colonies of $(B,P)$ at stratum $s$ (written $Col_s(B,P)$). The colonies of $(B,P)$ are the first auxiliary doors that a path from $(\sigma(B),P,[\oc_t],+)$ can reach (with $t \in Sig$) which belong to a box $B'$ with $N(B)>N(B')$.
    \begin{equation}
      Col_s(B,P)=\Set*{ (\sigma_i(B'),P') \in Can_{s-1}(E_G)}{
        \begin{array}{c}
          (B,P)=(B_0,P_0) \geq^s_1 (B_1,P_1) \geq^s_1 \cdots \geq^s_1 (B_n,P_n) \\
          \exists t \in Sig, (\sigma(B_n),P_n,!_t,+) \rightsquigarrow^* (\sigma_i(B'),P',!_{\sige},-) \\
          N(B)=N(B_1)= \cdots = N(B_n) > N(B') 
        \end{array}
      }
    \end{equation}

    \begin{lemma}
      If a stratified proof net controls dependence, then for all $s \in \mathbb{N}$, 
      \begin{equation*}
        |I_s((B,P),(e,Q))| \leq (D_G.|Can_{s-1}(E_G)|)^{2.N_s(B)} \cdot |Col_s(B,P)|
      \end{equation*}
    \end{lemma}
    \begin{proof}
      In fact we will prove $|I_s((B,P),(e,Q))| \leq (D_G.l.|Can_{s-1}(E_G)|)^{N(B)} \cdot |Col_s(B,P)|$ with $l$ the maximum length of a $\geq^s_1$ sequence beginning by potential box $(B,P)$. The announced result is then immediate because, if we suppose there is a $\geq^s_1$ path longer than $|Can_{s-1}(E_G)|$, then there are two contexts $C_1=(\sigma(B'),Q,[\oc_t],+)$ and $C_2(\sigma(B'),Q',[\oc_{t'}],+)$ such that $C_1 \mapsto^+_s C_2$ and $C_1 \sim_s C_2$ which is impossible because of the strong acyclicity lemma.

      We will prove this by induction on $l$, the depth of $(B,P)$ in terms of the $\geq^s_1$ relation. The relation $\geq^s_1$ may be cyclic on boxes. But it is acyclic on potential boxes (by the strong acyclicity lemma). So the induction is well founded. Explanations of the calcultations done between each line are given at the end of the calculation.
      \begin{align}
          |I_s((B,P),(e,Q))| &\leq  1+ \sum_{(B,P) \geq^s_j (B',P')}j.|I_s((B',P'),(e,Q))| \label{depcont_i_1}\\
          |I_s((B,P),(e,Q))| & \leq 1+ \sum_{(B,P) \geq^s_j (B',P')} j.\left ( D_G.(l-1).|Can_{s-1}(E_G)| \right)^{N(B')}|Col_s(B',P')| \label{depcont_i_2}\\
          |I_s((B,P),(e,Q))| & \leq 1 + A + B
      \end{align}

      With $A$ and $B$ quantities defined and bound below (the size of the expressions made it impossible to keep both on the same line). The previous calculations may need some explanations. 
      
      The first line is given by a counting argument. Let us choose an itinerary from $(B,P)$ to $(e,Q)$. Either it has some $\hookrightarrow$ rule in it, jumping over a $(B',P')$ box, or it goes directly to $(e,Q)$. In the first case we still have to choose an itinerary from $(B',P')$ to $(e,Q)$, in the second case there is only $[]$.

      The second line of the previous calculus was obtained by induction hypothesis. The last line is obtained by the separation of the set of potential boxes $(B',P')$ directly joined by $(B,P)$ into two disjoint sets.
      
      \begin{align}
        A &= \sum_{\substack{(B,P) \geq^s_j (B',P')\\ N(B')=N(B)}}j \left( D_G.(l-1).|Can_{s-1}(E_G)| \right)^{N(B')}|Col_s(B',P')| & \notag \\
        & \leq\sum_{\substack{(B,P) \geq^s_j (B',P')\\(\sigma_i(C),Q) \in Col_s(B',P')\\ N(B')=N(B)}} j \left(D_G.(l-1).|Can_{s-1}(E_G)|\right)^{N(B')} & \notag \\
        & \leq\sum_{\substack{(B,P) \geq^s_1 (B',P')\\(\sigma_i(C),Q) \in Col_s(B',P')\\ N(B')=N(B)}} \left(D_G.(l-1).|Can_{s-1}(E_G)|\right)^{N(B')} & \text{ if $j$ was $>1$, $N(B')<N(B)$} \notag \\
        & \leq \sum_{\substack{(B,P) \geq^s_1 (B',P')\\(\sigma_i(C),Q) \in Col_s(B',P')\\ N(B')=N(B) \\ (\sigma_i(C),Q) \in Col_s(B,P)}}\left(D_G.(l-1).|Can_{s-1}(E_G)|\right)^{N(B)} & \label{ext_col_a}
      \end{align}
      
      To obtain the inequality \ref{ext_col_a}, we notice that if $(\sigma_i(C),Q) \in Col_s(B',P')$, there is a sequence $(B',P') \geq^1 (B_1,P_1) \geq^1 \cdots \geq^1 (B_n,P_n)$ and $(\sigma(B_n),P_n,[!_t],+) \rightsquigarrow^* (\sigma_i(C),Q,[\oc_{\sige}],-)$. We can extend the sequence in the following way: $(B,P) \geq^1 (B',P') \geq^1 (B_1,P_1) \geq^1 \cdots \geq^1 (B_n,P_n)$. We know that $N(B)=N(B')$, so the condition on the nests is also respected and this sequence proves that $(\sigma_i(C),Q) \in Col_s(B,P)$.

      \begin{align}
        B & = \sum_{\substack{ (B,P) \geq^s_j (B',P')\\ N(B') < N(B)}}j\left(D_G.(l-1).|Can_{s-1}(E_G)|\right)^{N(B')}|Col_s(B',P')| & \notag \\
        & \leq \sum_{\substack{ (B,P) \geq^s_j (B',P')\\N(B')<N(B)}}j\left(D_G.(l-1).|Can_{s-1}(E_G)|\right)^{N(B)-1}|Col_s(B',P')| &  \notag \\
        & \leq \sum_{\substack{(B,P) \geq^s_j (B',P')\\N(B')<N(B)}} j \left(D_G.(l-1).|Can_{s-1}(E_G)|\right)^{N(B)-1}|Can_{s-1}(E_G)| \notag\\
        B & \leq \sum_{\substack{(B,P) \geq^s_j (B',P')\\(\sigma_i(B'),P')\in Col_s(B,P)}} \left(D_G.(l-1).|Can_{s-1}(E_G)|\right)^{N(B)} \label{final_b}
      \end{align}
      

      Now, it is possible to assemble the two inequalities. Indeed, there can be considered as sums of the same term $(D_G (l-1).|Can_{s-1}(E_G)|)^{N(B)}$ over disjoint sets. More precisely, we will show that the following set is a partition of $Col_s(B,P)$:
      \begin{equation*}
        \left\{
        \Set{(\sigma_i(C),Q) \in Col_s(B,P)}{(B,P) \geq^s_1 (C,Q) }
        \right\} 
        \cup
        \bigcup_{\substack{(B,P) \geq^s_1 (B',P')\\N(B')=N(B)}}\left\{\Set*{(\sigma_i(C),Q) \in Col_s(B,P)}{
          (\sigma_i(C),Q) \in Col_s(B',P')
        } \right\}
      \end{equation*}

      It is enough to prove that for all $(\sigma_i(C),Q) \in Col_s(B,P)$, the sequence $(B,P) \geq^s_1 (B_1,P_1) \geq^s_1 \cdots \geq^s_1 (B_n,P_n)$ such that $(\sigma(B_n),P_n,[\oc_t],+) \rightsquigarrow_s (\sigma_i(C),Q,[\oc_{\sige}],-)$ is unique.

      Suppose $(\sigma_i(C),Q) \in Col_s(B,P)$, let us consider two sequences $(B,P) \geq^s_1 (B_1,P_1) \geq^s_1 \cdots \geq^s_1 (B_n,P_n)$ and $(B,P) \geq^s_1 (B_1',P'_1) \geq^s_1 \cdots \geq^s_1 (B'_{n'},P'_{n'})$ such that there exists $t,t' \in Sig$ such that $(\sigma(B_n),P_n,[\oc_t],+) \rightsquigarrow_s (\sigma_i(C),Q,[\oc_{\sige}],-)$ and $(\sigma(B'_{n'}),P'_{n'},[\oc_{t'}],+) \rightsquigarrow_s (\sigma_i(C),Q,[\oc_{\sige}],-)$. Then, by Lemma \ref{injection_lemma}, we have $(B'_{n'},P'_{n'})=(B_n,P_n)$. We know that $N_s(B_n,P_n)=N_s(B'_{n'},P'_{n'})=N_s(B,P)$. So if we write $\sigma_{i_n}(B_n)$ and $\sigma_{i'_{n'}}(B_n)$ the auxiliary doors such that, respectively, $(\sigma(B_{n-1}),P_{n-1},[\oc_{t_{n-1}}],+) \rightsquigarrow^*_s (\sigma_{i_n}(B_n),P_n,[\oc_{\sige}],-)$ and $(\sigma(B'_{n'-1}),P'_{n'-1},[\oc_{t'_{n'-1}}],+) \rightsquigarrow^*_s (\sigma_{i'_{n'}}(B_n),P_n,[\oc_{\sige}],-)$. Then $i_n=i'_{n'}$. So, using Lemma \ref{injection_lemma}, we can show that $(B_{n-1},P_{n-1})=(B'_{n'-1},P'_{n'-1})$. By induction, we show that for every $i \leq min(n,n')$, we have $(B_{n-i},P_{n-i})=(B'_{n'-i},P'_{n'-i})$. So one of the sequence is a suffix of the other. Moreover, we know that $\geq^s_1$ is acyclic, so the sequences $(B,P) \geq^s_1 (B_1,P_1) \geq^s_1 \cdots \geq^s_1 (B_n,P_n)$ and $(B,P) \geq^s_1 (B_1',P'_1) \geq^s_1 \cdots \geq^s_1 (B'_{n'},P'_{n'})$ are equal.
      
      \begin{align}
        |I_s((B,P),(e,Q))| & \leq 1+\sum_{(\sigma_i(B'),P') \in Col_s(B,P)}\left(D_G.(l-1).|Can_{s-1}(E_G)|\right)^{N(B)} \label{depcont_i_8}\\ 
        |I_s((B,P),(e,Q))| & \leq 1+\left(D_G.(l-1).|Can_{s-1}(E_G)|\right)^{N(B)}|Col_s(B,P)| \label{depcont_i_9}\\
        |I_s((B,P),(e,Q))| & \leq \left(D_G.l.|Can_{s-1}(E_G)|\right)^{N(B)}|Col_s(B,P)| \label{depcont_i_10}
      \end{align}
    \end{proof}

    \subsection{Digging and dependence control}
    \label{subsection_depcontrol_digging}
    As we said in the previous subsection, if there are no $\wn N$ links, $(\sigma(B),P,[\oc_t],+) \rightsquigarrow^*_s (e,P,[\oc_{\sige}],-)$ and $(\sigma(B),P,[\oc_u],+) \rightsquigarrow^*_s (e,P',[\oc_{\sige}],-)$ with $P^{/s-1}=P'^{/s-1}$, then $t=u$. To understand why the $\wn N$ links break this properties, we can take an example in Figure \ref{exp2}. We have $(\sigma(C),[],[\oc_{\sign(\sigl(\sige),\sign(\sigr(\sige),\sige))}],+) \rightsquigarrow_0 (e,[\sign(\sigr(\sige),\sige)],[\oc_{\sige}],-)$ and $(\sigma(C),[],[\oc_{\sign(\sigl(\sige),\sign(\sigl(\sige),\sige))}],+) \rightsquigarrow_0 (e,[\sign(\sigl(\sige),\sige)],[\oc_{\sige}],-)$ and $[\sign(\sigl(\sige),\sige)]^{/-1}=[\sign(\sigr(\sige),\sige)]^{/-1}=[\sige]$. However $\sign(\sigl(\sige),\sign(\sigl(\sige),\sige)) \neq \sign(\sigl(\sige),\sign(\sigr(\sige),\sige))$. If we follow the paths backwards we see that the crucial step is $(\sigma_1(B),[],[\oc_{\sign(\sigl(\sige),\sige)};\oc_{\sign(\sigr(\sige),\sige)}],-) \leftsquigarrow (f,[],[\oc_{\sign(\sign(\sigl(\sige),\sige), \sign(\sigr(\sige))}],-)$ where a difference on the second trace element (which comes from a box $B$ of same strata than $C$) becomes a difference on the first trace element, which will correspond to the copy. The paths of  $\sign(\sign(\sigl(\sige),\sige), \sign(\sigl(\sige))$ and $\sign(\sign(\sigl(\sige),\sige), \sign(\sigr(\sige))$ may be the same, but their simplifications are different and have different paths. 

    So if we choose the $\mapsto_{-1}$-potential edge $(e,[\sige])$ and the itinerary $[]$, there are as many $\mapsto_0$-copies of $(\sigma(C),[])$ $t$ going through a context of the shape $(e,P,[\oc_{\sige}],-)$ with $P^{-1}=[\sige]$ and $I((\sigma(C),[],[\oc_t],+),(e,P,[\oc_{\sige}],-))=[]$ as there are $\mapsto_0$-copies of $(\sigma(B),[])$. Similarly, if we choose the $\mapsto_{-1}$-potential edge $(f,[])$ and the itinerary $[]$, there are as many $\mapsto_0$-copies of $(\sigma(B),[])$ $t$ going through a context of the shape $(f,P,[\oc_{\sige}],-)$ with $P^{-1}=[\sige]$ and $I((\sigma(C),[],[\oc_t],+),(f,P,[\oc_{\sige}],-))=[]$ as there are $\mapsto_0$-copies of $(\sigma(A),[])$. Let us notice that if the number of $\mapsto_s$-copies of $(B,P)$ depend on the number of $\mapsto_s$-copies of $(C,Q)$ in this manner, then $(B,P) \succcurlyeq_k^s (C,Q)$ with $k \geq |C_{s}(C,Q)|$. For example $(C,[]) \succcurlyeq^0_{2} (B,[])$ and $(B,[])  \succcurlyeq^0_1 (A,[])$.

    So if we fix a $\mapsto_{s-1}$ potential edge $(e,Q_{s-1})$, an itinerary $I$, and a $\mapsto_s$-copy $t_i$for every potential box $(C_i,R_i)$ containing $(e,Q)$ such that $B \succcurlyeq^s_2 C_i$, then there are at most one $\mapsto_s$-copy $t$ of $(B,P)$ such that $(\sigma(B),P,[\oc_t],+) \mapsto^s (e,Q,[\sige],-)$, $I((\sigma(B),P,[\oc_t],+),(e,Q,[\sige],-))$, $Q^{/s-1}=Q_{s-1}$ and the exponential signature of $Q$ corresponding to $C_i$ is equal to $t_i$. So we bound the number of $\mapsto_s$-copies of a box $B$ by $\mapsto_{s-1}$-copies of some boxes and $\mapsto_{s}$-copies of boxes $C$ such that $B \succcurlyeq_2^s C$. If we make an induction on $N(B)$ inside an induction on $s$, we can bound the number of $\mapsto_s$-copies of any potential box.

    To prove this, we will keep finely track of the exponential signature which ``come from'' a box of nest greater than $n$ during a $\mapsfrom$ path. To do this we will need some kind of pointer to refer to a precise location in an exponential signature in a context $C$. This is exactly what $Pos(C)$ will be. First we define the notion of positions on a single exponential signature. An element of $Pos(t)$ represents the path from the root of the exponential signature (viewed as a tree) to the location we want to point to. A $0$ means ``take the left branch (or the only branch if there is only one)'', a $1$ means ``take the right branch''.

    \begin{definition}\label{def_positions}
      Let $t$ be an exponential signature, the set of positions of $t$ (written $Pos(t)$) is defined by induction by: 
      \begin{itemize}
      \item $Pos(\sige)=\{[]\}$
      \item $Pos(\sigl(u))=Pos(\sigr(u))=Pos(\sigp(u))=\{[]\} \cup \Set{[0]@p}{p \in Pos(u)}$
      \item $Pos(\sign(u,v))= \{[]\} \cup \Set{[0]@p}{p \in Pos(u)} \cup \Set{[1]@p}{p \in Pos(v)}$
      \end{itemize}
      \label{def_tbarp}Let $p \in Pos(t)$, $t_{|p}$ is the exponential signature defined by induction on $p$ by: $t_{|[]}=t$, $\sigl(u)_{|[0]@q}=\sigr(u)_{|[0]@q}=\sigp(u)_{|[0]@q}=\sign(u,v)_{|[0]@q}=u_{|q}$ and $\sign(u,v)_{|[1]@q}=v_{|q}$.
    \end{definition}
      
    We also write $Pos$ the set of lists of $0$ and $1$. Now that we can point to a precise location in an exponential signature, we can replace the exponential signature at this place by another exponential signature. 

    \begin{definition}
      \label{def_parallelpositions}If $p_1,…p_n$ are parallel positions of $t$ (i.e. for every $1 \leq i < j \leq n$, if there are no $q$ such that $p_i=p_j@q$ or $p_j=p_i@q$), and $f$ a mapping from $\{p_1,\cdots,p_n\}$ to $Sig$, then we define $f_0$ as the function $p \mapsto f([0]@p)$ and $f_1$ as the function $p \mapsto f([1]@p)$. Then, we define $t[f]$ by:\label{def_tcrochetf}
      \begin{itemize}
      \item $t[\varnothing]=t$
      \item $t[\{[]\mapsto u\}]=u$
      \item $\sigl(t)[f]=\sigl(t[f_0])$, $\sigr(t)[f]=\sigr(t[f_0])$, $\sigp(t)[f]=\sigp(t[f_0])$
      \item $\sign(t,u)[f]=\sign(t[f_0],u[f_1])$
      \end{itemize}
    \end{definition}

    Then, we define those notions on contexts. A position in a context must first explain if the location we want to point to is in the potential (we will then set the first component to $POT$) or in the trace (we will then set the first component to $TRA$). Then we have to point to some exponential signature in the potential or some trace element. We do so with an integer representing the indice of the object in the list it belongs to. Finally we have to precise the location inside the exponential signature $t$ we defined by the two first components. We do so with some element of $Pos(t)$.

    \begin{definition}[exponential position]\label{def_contextpositions}
      Let $C=(e,[p_1;\cdots;p_k],[t_1;\cdots;t_n],p)$ be a context of a proof-net $G$. An exponential position of $C$ is:
      \begin{itemize}
      \item Either $(POT,i,q)$ with $1 \leq i \leq k$ and $q \in Pos(p_i)$.
      \item Either $(TRA,i,q)$ with $1 \leq i \leq n$ such that $t_i$ is either of the shape $t_i=\oc_t$ or $t_i=\wn_t$, and $q \in Pos(t)$.
      \end{itemize}
      The set of the exponential positions of $C$ is writen $Pos(C)$. We define $C_{|p}$ as \label{def_cbarp}
      \begin{itemize}
      \item $(e,[p_1;\cdots;p_k],[t_1;\cdots;t_n],p)_{|(POT,i,q)}=(p_i)_{|q}$
      \item $(e,[p_1;\cdots;p_k],[t_1;\cdots;t_n],p)_{|(TRA,i,q)}=(t)_{|q}$ (with $t_i=\oc_t$ or $t_i=\wn_t$)
      \end{itemize}
      \end{definition}
    
    \begin{definition}\label{def_ccrochetf}
      If $q_1,\cdots,q_m$ are parallel positions of $(e,[p_1;\cdots;p_k],[t_1;\cdots;t_n],p)$ (i.e. for every $i \neq j$, either the two first components of $q_i$ and $q_j$ design different exponential signatures or their third components are parallel) and $f$ is a mapping from $\{q_1,\cdots,q_m\}$, then we define $(e,[p_1;\cdots;p_k],[t_1;\cdots;t_n],p)[f]$ as the context $(e,[p_1[\{x \mapsto f(POT,1,x)\}];\cdots;p_k[\{x \mapsto f(POT,k,x)\}]],[t_1[\{x \mapsto f(TRA,1,x)\}];\cdots;t_n[\{x \mapsto f(TRA,n,x)\}]],p)$.
    \end{definition}

    For example, if we set $C=(e,[\sige;\sign(\sigl(\sige),\sige)],[\parr_r;\wn_{\sigr(\sige)}],+)$, then 
    \begin{equation*}
      Pos(C) = \{(POT,1,[]), (POT,2,[]), (POT,2,[0]), (POT,2,[0;0]), (POT,2,[1]), (TRA,2,[]), (TRA,2,[0])\}
    \end{equation*}
    We also have $C_{|(POT,2,[0])}=\sigl(\sige)$ and $C[\{(TRA,2,[0])\mapsto \sigr(\sige)\}]=(e,[\sige;\sign(\sigl(\sige),\sige)],[\parr_r;\wn_{\sigr(\sigr(\sige))}],+)$.

    We need a last notation. Let us suppose that $C=(e,P,[\oc_{\sign(t,u)}],-) \mapsto (f,P,[\oc_t;\oc_u],-)=C'$ (crossing a $\wn N$ link upwards) and that for every $v \in Sig$, $(f,P,[\oc_v],-) \mapsto^* E$. Then, in this subsection, we will want to state this property without mentioning the contents of $C'$ ($f$, $P$, $t$, $u$ and $-$). A way of saying it is: ``If we restrict the trace of $C'$ to its last trace element and replace its exponential signature by $v$, giving us a context $D'$, then $D' \mapsto^* E$''. If we look at the context $C$, the equivalent property is that for every $v \in Sig$, $(e,P,[\oc_{\sigp(v)}],-) \mapsto^* E$. Making this statement without mentioning the contents of $C$ will be troublesome. So, to make such a statement in a general, yet concise, manner we will define an operation $(~)\downarrow^t_p$ first on signature and then on contexts such that $(e,P,[\oc_{\sign(t,u)}],-)\downarrow^v_{(TRA,1,[2])}=(e,P,[\oc_{\sigp(v)}],-)$. This will allow us to simply state ``$\forall v \in Sig, C \downarrow^v_{(TRA,1,[2])} \mapsto_s E$''. In general $C^t_p$ represents the context obtained by replacing the exponential signature at position $p$ by $v$, replacing the $\sign(t_1,t_2)$ above it by $\sigp(t_2)$, and (if $p$ refers to a trace element) delete the trace elements on the left of $p$.

    \begin{definition}\label{def_downarrowtp}
      Let $t \in Sig$ and $p \in Pos(t)$, we define $t\downarrow_{p}$ as:
      \begin{itemize}
      \item If there is a prefix $q'$ of $q$ such that $t_{|q'}$ is of the shape $\sign(t,u)$, let us consider the longest such $q'$. Then $t \downarrow_p=\sigp(u)$.
      \item Else, $t \downarrow_p=t$
      \end{itemize}
      
      Let $e \in B_{\partial(e)} \subset \cdots \subset B_1$, $C=(e,[p_1;\cdots;p_k],[t_1;\cdots;t_n],p)$, $p=(X,i,q) \in Pos(C)$ and $t \in Sig$, we define $C\downarrow_{p}^t$ as:
      \begin{itemize}
      \item If $p=(TRA,i,q)$ and $t_i=\oc_u$, then $C\downarrow_p^t= (e,[p_1; \cdots ;p_k],[\oc_{u[q \mapsto t] \downarrow_q};t_{i+1};\cdots;t_n],p)$
      \item If $p=(TRA,i,q)$ and $t_i=\wn_u$, then $C\downarrow_p^t= (e,[p_1; \cdots ;p_k],[\oc_{u[q \mapsto t] \downarrow_q};t^\perp_{i+1};\cdots;t^\perp_n],p^\perp)$
      \item If $p=(POT,i,q)$, then $C\downarrow_p^t= (\sigma(B_i),[p_1; \cdots ;p_{i-1}],[\oc_{p_i[q \mapsto t] \downarrow_q}],+)$
      \end{itemize}
    \end{definition}

    Now, we will prove the core lemma of this subsection. We state that if $C \mapsto^n_s (e,P,[\oc_{\sige}],-)$ then we can find the exponential signatures $p_i$ of $P=[p_1;\cdots;p_{\partial(e)}]$ inside $C$ in a position $\phi(i)$. And replacing $p_i$ by $p'_i$ in $C$ lead to a context $(e,[p_1;\cdots;p_{i-1};p'_i;p_{i+1};\cdots;p_{\partial(e)}],[\oc_{\sige}],-)$.

    \begin{lemma}\label{lemma_digging_depcont}
      Let $C$ be a context of $G$, $e \in B_{\partial(e)} \subset \cdots \subset B_1$, $P=[p_1;\cdots;p_{\partial(e)}] \in L_s(e)$ and $C_e=(e,P,[\oc_{\sige}],-)$. If $C \mapsto_s^n C_e$ then there exists an injective mapping $\phi$ from $\{1,\cdots,\partial(e)\}$ to $Pos(C)$ and a mapping $\psi$ which, to $1 \leq i \leq \partial(e)$ associates an element of $Pos(p_i)$ such that:
      \begin{itemize}
      \item For every contexts $C'$, $P'=[p'_1;\cdots;p'_{\partial(e)}] \in L_s(e)$ with $P^{/s-1}=P'^{/s-1}$ and $C_e'=(e,P',[\oc_{\sige}],-)$, then
        \begin{equation*}
          \left(C' \mapsto^n_s C_e' \text{ and } I(C',C_e')=I(C,C_e) \right) \Leftrightarrow C[\{\phi(i) \mapsto (p'_i)_{|\psi(i)}\}]=C'
        \end{equation*}
      \item For every $1\leq i \leq \partial(e)$ there exists a potential box $(D_i,Q_i)$ such that:
        \begin{itemize}
        \item Either $\exists j, \forall t \in Sig, C\downarrow_{\phi(i)}^t \mapsto^*_s (\sigma_j(D_i),Q_i,[\oc_{t}],-)$.
        \item Or $\phi(i)=(\_,\_,[])$ does not correspond to a $\oc$ trace element and $\forall t \in Sig,(\sigma(D_i),Q_i,[\oc_t],+) \rightsquigarrow^* C\downarrow_{\phi(i)}^t$.
        \end{itemize}
      \end{itemize}
    \end{lemma}
    \begin{proof}
      We will prove the result by induction on $n$ (the length of the path from $C$ to $C_e$). 
      \begin{itemize}
      \item If $n=0$, then we can consider $\phi: i \mapsto (POT,i,[])$ and $\psi: i \mapsto []$. 
        \begin{itemize}
        \item Let us suppose that $P^{/s-1}=P'^{s-1}$, $C' \mapsto^0_s C'_e$ and $I(C',C'_e)=I(C,C_e)$.  Then, by definition of $(~)^{/s-1}$, $\forall i, p_i=p'_i$. So $C_e[\{(POT,i,[]) \mapsto p'_i\}]=C'_e$. Moreover, we know that $C=C_e$ and $C'=C'_e$. So $C[\{(POT,i,[]) \mapsto p'_i\}]=C'$. Finally $\phi(i)=\{(POT,i,[])\}$ and $(p'_i)_{|\psi(i)}=(p'_i)_{[]}=p'_i$ so $C[\{\phi(i) \mapsto (p'_i)_{|\psi(i)}\}]=C'$.
        \item Let us suppose that $P^{/s-1}=P'^{/s-1}$ and $C[\{\phi(i) \mapsto (p'_i)_{|\psi(i)} \}]=C'$. We defined $\phi(i)$ as $(POT,i,[])$ and $\psi(i)=[]$ so $C[\{(POT,i,[]) \mapsto p'_i\}]=C'$. We know that $C \mapsto^0_s C_e$, so $C_e[\{(POT,i,[]) \mapsto p'_i\}]=C'$. Moreover, $P^{/s-1}=P'^{/s-1}$ so for every $B_i$, $p_i=p'_i$, thus $C_e[\{(POT,i,[]) \mapsto p'_i\}]=C'_e$. By transitivity of equality, $C'=C'_e$ so $C' \mapsto^s_0 C'_e$ and $I(C,C_e)=I(C',C'_e)=[]$.
        \item Let $1 \leq i \leq \partial(e)$, we set $(D_i,Q_i)=(\sigma(B_i),[p_1;\cdots;p_{i-1}])$. Let $t \in Sig$, then $C\downarrow^t_{\phi(i)}=C\downarrow^t_{(POT,i,[])}=(\sigma(B_i),[p_1;\cdots;p_{i-1}],[\oc_{t}],-)$ so $(\sigma(D_i),Q_i,[\oc_t],+) \rightsquigarrow^0 C\downarrow^t_{\phi(i)}$.
        \end{itemize}
      \item If $n>0$, then $C \mapsto_s C_1 \mapsto^n_s C_e$. So, the hypothesis of induction gives us mappings $\phi_1$ from $\{1,\cdots,\partial(e)\}$ to $Pos(C_1)$ and $\psi_1$ from $\{1,\cdots,\partial(e)\}$ to $Pos$. We will transform it into mappings $\phi$ from $\{1,\cdots,\partial(e)\}$ to $Pos(C)$ and $\psi$ from $\{1,\cdots,\partial(e)\}$ to $Pos$. This transformation will depend on the $\mapsto_s$ step between $C$ and $C_1$.
        \begin{itemize}
        \item For the steps which are not involved with the $\oc$ and $\wn$ exponentials (crossing a $cut$, $ax$, $\parr$, $\otimes$, $\forall$, $\exists$ or $\S$ link), we set $\phi=\phi_1$ and $\psi=\psi_1$. All these cases being similar, we will only consider one of them: $C = (f,Q,T,+) \mapsto_s (g,Q,T.\parr_l,+)=C_1$ (crossing a $\parr$ link downwards). 
          \begin{itemize}
          \item Let us suppose that $P^{/s-1}=P'^{s-1}$, $C' \mapsto_s C'_1 \mapsto_s^{n-1} C'_e$ and $I(C',C'_e)=I(C,C_e)$.  We have $I(C_1,C_e)=I(C,C_e)$ (because the step between $C$ and $C_1$ is not a $\hookrightarrow$ step). We also have $C_1 \sim_s C_1'$ (Lemma \ref{injection_lemma} so the step from $C'$ to $C_1'$ is not a $\hookrightarrow$ so $I(C_1',C_e')=I(C',C'_e)$). So $I(C'_1,C'_e)=I(C_1,C_e)$. So, by induction hypothesis, $C_1[\{\phi_1(i) \mapsto (p'_i)_{|\psi_1(i)}\}]=C'_1$. We have $C'_1=(g,Q',T'.\parr_l,+)=C_1[\{\phi(i) \mapsto (p'_i)_{|\psi(i)}\}$, so $C'=(f,Q',T',+)=C[\{\phi(i) \mapsto (p'_i)_{|\psi(i)}\}]=C[\{\phi(i) \mapsto (p'_i)_{|\psi(i)}\}]$.
          \item Let us suppose that $P^{/s-1}=P'^{/s-1}$ and $C[\{\phi(i) \mapsto (p'_i)_{|\psi(i)} \}]=C'$. So $C'$ is of the shape $C'=(f,Q',T',+)$, let us set $C'_1=(f,Q',T'.\parr_l,+)=C_1[\{\phi(i) \mapsto (p'_i)_{|\psi(i)} \}]$. We can use the induction hypothesis, and we get that $C_1' \mapsto^{n-1}_s C_e'$ and $I(C'_1,C'_e)=I(C_1,C_e)$. Moreover, $C' \mapsto_s C'_1$, $I(C'_1,C'_e)=I(C',C'_e)$ and $I(C_1,C_e)=I(C,C_e)$. So $C' \mapsto_s^n C'_e$ and $I(C',C'_e)=I(C,C_e)$.
          \item Let $1 \leq i \leq \partial(e)$, then we take the same $(D_i,Q_i)$ as in the $C_1$ case
            \begin{itemize}
            \item If $\phi_1(i)$ corresponds to a $\oc$ trace element and there exists $j$ such that, $\forall t \in Sig, C_1\downarrow^t_{\phi_1(i)} \mapsto^*_s (\sigma_j(D_i),Q_i,[\oc_{t}],-)$, then for any $t \in Sig$, $C\downarrow^t_{\phi(i)} \mapsto_s C_1 \downarrow^t_{\phi_1(i)}$ so $C \downarrow^t_{\phi(i)} \mapsto^*_s (\sigma_j(D_i),Q_i,[\oc_t],-)$.
            \item If $\phi_1(i)$ corresponds to a $\wn$ trace element and there exists $j$ such that $\forall t \in Sig, C_1\downarrow^t_{\phi_1(i)} \mapsto^*_s (\sigma_j(D_i),Q_i,[\oc_{t}],-)$, then $C_1\downarrow^t_{\phi_1(i)} \mapsto_s C \downarrow^t_{\phi(i)}$. Let us suppose that $C_1 \downarrow^t_{\phi_1(i)} \mapsto_s^{\mathbf{0}} (\sigma_j(D_i),Q_i,[\oc_{t}],-)$, then either $C_1$ is equal to $(\sigma_j(D_i),Q_i,T_1@[\oc_{u}]@T_2,-)$ with $\phi(i)=(TRA,|T_1|+1,q)$ or $C_1$ is equal to $(\sigma_j(D_i),Q_i,T_1@[\wn_{u}]@T_2,+)$ with $\phi(i)=(TRA,|T_1|+1,q)$. The first case is ruled out because we supposed that $\phi(i)$ corresponds to a $\wn$ trace element, the second case is ruled out because we supposed that the step between $C$ and $C_1$ does not involve exponentials. So we have a contradiction, our supposition was false, so $C_1 \downarrow^t_{\phi_1(i)} \mapsto_s^{\mathbf{+}} (\sigma_j(D_i),Q_i,[\oc_{t}],-)$. So $C \downarrow^t_{\phi(i)} \mapsto_s^{*} (\sigma_j(D_i),Q_i,[\oc_{t}],-)$.
            \item If $\phi_1(i)$ corresponds to a $\wn$ trace element and $\forall t \in Sig, (\sigma(D_i),Q_i,[\oc_t],+) \rightsquigarrow^* C_1\downarrow^t_{\phi(i)}$. Then, for any $t \in Sig$, $C_1 \downarrow^t_{\phi_1(i)} \rightsquigarrow C \downarrow^t_{\phi(i)}$. So $(\sigma(D_i),Q_i,[\oc_t],+) \rightsquigarrow^* C \downarrow^t_{\phi(i)}$. 
            \item If $\phi_1(i)=(POT,\_,\_)$, then $C \downarrow^t_{\phi(t)} = C_1 \downarrow^t_{\phi_1(t)}$. So whatever the case we were in for $C_1$, we are in the same case for $C$.
            \end{itemize}
          \end{itemize}
          
        \item We will now consider the steps which cross $\wn C$, $\wn D$, $\wn N$, $\wn P$ or $\oc P$ links. In each case, we will only detail some part of the proofs. The parts that we do not detail are quite similar to the non-exponential cases described above.
          
          Let us suppose that $C=(\sigma_k(B),Q,T.\oc_{u},-) \mapsto_s (g,Q.u,T,-)= C_1$, crossing an auxiliary door upwards. We define $\psi=\psi_1$ and $\phi$ as a function almost equal to $\phi_1$, the only difference being that when $\phi_1(POT,i,[])$ corresponds to some position in $u$, we have to change the image to find the corresponding position in $u$ in the context $C$.
          \begin{equation*}
            \phi: i \mapsto \left\{ \begin{array}{l}(TRA,|T|+1,q)\text{, if }\phi_1(i)=(POT,|Q|+1,q)\\ \phi_1(i)\text{ otherwise}\end{array}\right.
          \end{equation*}
          Here the particular point is to prove that for every $1 \leq i \leq \partial(e)$ either there exists $j$ such that $\forall t \in Sig, C \downarrow^t_{\phi(i)} \mapsto^*_s (\sigma_j(D_i),Q_i,[\oc_{p_i}],-)$ or $\forall t \in Sig (\sigma(D_i),Q_i,[\oc_{t}],-) \rightsquigarrow^* C \downarrow^t_{\phi(i)}$. Let us consider some $1 \leq i \leq \partial(e)$. Then, by induction hypothesis, there exists $(D_i^1,Q_i^1)$ such that either there exists $j$ such that $\forall t \in Sig, C_1 \downarrow^t_{\phi_1(i)} \mapsto^*_s (\sigma_j(D_i^1),Q_i^1,[\oc_{t}],-)$ or $\phi_1(i)$ does not correspond to a $\oc$ trace element and $\forall t \in Sig (\sigma(D_i^1),Q_i^1,[\oc_{t}],-) \rightsquigarrow^* C_1 \downarrow^t_{\phi_1(i)}$. We set $(D_i,Q_i)=(D^1_i,Q^1_i)$ and make a disjunction over the case we are in,
          \begin{itemize}
          \item In the first case,
            \begin{itemize}
            \item If $\phi_1(i)= (POT,|Q|+1,q)$, we know that $C\downarrow^t_{\phi(i)}=C\downarrow^t_{(TRA,|T|+1,q)}= (\sigma_k(B),Q,[\oc_{u\downarrow^t_q}],-)$. So $C\downarrow^t_{\phi(i)} \mapsto_s (\sigma(B),Q,[\oc_{u \downarrow^t_q}],+)= C_1 \downarrow^t_{\phi_1(i)}$. Moreover, by hypothesis, $C_1 \downarrow^t_{\phi_1(i)} \mapsto^*_s (\sigma_j(D^1_i),Q^1_i,[\oc_t],-)$. So, by transitivity, $C \downarrow^t_{\phi(i)} \mapsto^*_s (\sigma_j(D_i),Q_i,[\oc_t],-)$.
            \item If $\phi_1(i) \neq (POT,|Q|+1,q)$, then as in the non-exponential cases, we have either $C \downarrow^t_{\phi(i)} \rightsquigarrow C_1 \downarrow^t_{\phi_1(i)}$ or $C_1 \downarrow^t_{\phi_1(i)} \rightsquigarrow C \downarrow^t_{\phi(i)}$ which gives us the expected result.
            \end{itemize}
          \item In the second case,
            \begin{itemize}
            \item If $\phi_1(POT,i,[])= (POT,|Q|+1,q)$, then $\forall t \in Sig, (\sigma(D_i^1),Q_i^1,[\oc_{u\downarrow^t_q}],+) \rightsquigarrow^* (\sigma(B),Q,[\oc_t],+)$. So $(D_i,Q_i)=(B,Q)$ and $q=[]$. So for any $t \in Sig$, $C \downarrow^t_{\phi(i)}= C \downarrow^t_{(TRA,|T|+1,q)} = (\sigma_k(B),Q,[\oc_t],-)= (\sigma_k(D_i),Q_i,[\oc_t],-)$. Thus, there exists some $j$ (precisely, $j=k$) such that $\forall t \in Sig, C \downarrow^t_{\phi(i)} \mapsto^0_s (\sigma_j(D_i),Q_i,[\oc_t],-)$.
            \item If $\phi_1(POT,i,[]) \neq (POT,|Q|+1,q)$, then $\forall t \in Sig, C_1 \downarrow^t_{\phi_1(i)} \rightsquigarrow C \downarrow^t_{\phi(i)}$, so $\forall t \in Sig, (\sigma(D_i),Q_i,[\oc_t],+) \rightsquigarrow^* C\downarrow^t_{\phi(i)}$.
            \end{itemize}
          \end{itemize}
          
        \item Let us suppose that $C=(g,Q.u,T,+) \mapsto_s (\sigma_k(B),Q,T.\wn_{u},+)$, crossing an auxiliary door downwards. We set $\psi=\psi_1$ and 
          \begin{equation*}
            \phi: i \mapsto \left\{ \begin{array}{l}(POT,|Q|+1,q)\text{, if }\phi_1(i)=(TRA,|T|+1,q)\\ \phi_1(i)\text{ otherwise}\end{array}\right.
          \end{equation*}
          Here the important point is in the case where we have $1 \leq i \leq \partial(e)$ such that $\phi_1(i)=(TRA,|T|+1,q)$.
          \begin{itemize}
          \item If there exists $(D_i^1,Q_i^1,j) \neq (B,Q,k)$ such that $\forall t \in Sig, C_1\downarrow^t_{\phi_1(i)} \mapsto^*_s (\sigma_j(D_i^1),Q_i^1,[\oc_t],-)$. Then we set $(D_i,Q_i)=(D^1_i,Q^1_i)$. For any $t \in Sig$, $C_1\downarrow^t_{\phi(i)}=(\sigma_k(B),Q,[\oc_{u\downarrow^t_q}],-) \mapsto_s^{\mathbf{+}} (\sigma_j(D_i),Q_i,[\oc_t],-)$ and $(\sigma_k(B),Q,[\oc_{u \downarrow^t_q}],-) \mapsto_s (\sigma(B),Q,[\oc_{u \downarrow^t_q}],+)= C \downarrow^t_{\phi(i)}$. So $C \downarrow^t_{\phi(i)} \mapsto^*_s (\sigma_j(D_i),Q_i,[\oc_t],-)$.
          \item If $\forall t \in Sig, C_1 \downarrow^t_{\phi_1(i)} \mapsto^* (\sigma_k(B),Q,[\oc_t],-)$. Then $q=[]$ (otherwise the proof-net would be cyclic). We set $(D_i,Q_i)=(D^1_i,Q^1_i)$. For any $t \in Sig$, $(\sigma(D_i),Q_i,[\oc_t],+) \mapsto^0 C \downarrow^t_{(POT,|Q|+1,[])}= C \downarrow^t_{\phi(i)}$.
          \item If there exists $(D^1_i,Q^1_i)$ such that $\forall t \in Sig, (\sigma(D_i^1),Q_i^1,[\oc_t],+) \rightsquigarrow^* C_1 \downarrow^t_{\phi_1(i)}$. Then, $q=[]$. We set $(D_i,Q_i)=(B,Q)$. For any $t \in Sig$, $(\sigma(D_i),Q_i,[\oc_t],+) \rightsquigarrow^0 (\sigma(B),Q,[\oc_t],+) = C \downarrow^t_{(POT,|Q|+1,[])}=C \downarrow^t_{\phi(i)}$.
          \end{itemize}

        \item Let us suppose that $C=(\sigma(B),Q,T.\wn_{u},-) \mapsto_s (g,Q.u,T,-)= C_1$, crossing a principal door upwards. We define $\psi=\psi_1$ and $\phi$ as a function almost equal to $\phi_1$, the only difference being that when $\phi_1(i)$ corresponds to some position in $u$, we have to change the image to find the corresponding position in $u$ in the context $C$.
          \begin{equation*}
            \phi: i \mapsto \left\{ \begin{array}{l}(TRA,|T|+1,q)\text{, if }\phi_1(i)=(POT,|Q|+1,q)\\ \phi_1(i)\text{ otherwise}\end{array}\right.
          \end{equation*}
          Let $1\leq i \leq \partial(e)$, then we set $(D_i,Q_i)=(D_i^1,Q^1_i)$. Let us notice that if $\phi_1(i)=(POT,|Q|+1,q)$, then $C\downarrow^t_{\phi(i)}=C_1 \downarrow^t_{\phi_1(i)}$.

        \item Let us suppose that $C=(g,Q.u,T,+) \mapsto_s (\sigma(B),Q,T.\oc_u,+)= C_1$, crossing a principal door downwards. We define $\psi=\psi_1$ and $\phi$ as a function almost equal to $\phi_1$, the only difference being that when $\phi_1(i)$ corresponds to some position in $u$, we have to change the image to find the corresponding position in $u$ in the context $C$.
          \begin{equation*}
            \phi: (POT,i,[]) \mapsto \left\{ \begin{array}{l}(POT,|Q|+1,q)\text{, if }\phi_1(i)=(TRA,|T|+1,q)\\ \phi_1(i)\text{ otherwise}\end{array}\right.
          \end{equation*}
          Let $1\leq i \leq \partial(e)$, then we set $(D_i,Q_i)=(D_i^1,Q^1_i)$. Let us notice that if $\phi_1(i)=(TRA,|T|+1,q)$, then $C\downarrow^t_{\phi(i)}=C_1 \downarrow^t_{\phi_1(i)}$.

        \item Let us suppose that $C=(\sigma_k(B),Q,[\oc_u],-) \mapsto_s (\sigma(B),Q,[\oc_u],-)$, jumping from an auxiliary door of a box to its principal door, then we define $\psi=\psi_1$ and $\phi=\phi_1$.

        \item Let us suppose that $C= (f,Q,T.\wn_{u},+) \mapsto_s (g,Q,T.\wn_{\sigl(u)},+)=C_1$, crossing a $\wn C$ link downwards.  The only difference between $\psi$ and $\psi_1$ is in the case of $\phi_1(i)=(TRA,|T|+1,[])$. In this case, there are no position of $C$ corresponding to  $\sigl(u)$. In this case, we will define $\psi(i)=psi_1(i)@[0]$. Then the only difference between $\phi$ and $\phi_1$ is that when $\phi_1(POT,i,[])$ refers to a position in $t$ in $C_1$, we have to delete the first $0$ so that $\phi(POT,i,[])$ corresponds to the same position in $C$.
          \begin{align*}
            \psi:& i \mapsto \left\{ \begin{array}{l}\psi_1(i)@[0]\text{, if }\phi_1(i)=(TRA,|T|+1,[])\\ \psi_1(i)\text{ otherwise}\end{array}\right.\\
            \phi:& i \mapsto \left\{ \begin{array}{l}(TRA,|T|+1,q)\text{, if }\phi_1(i)=(TRA,|T|+1,[0]@q)\\ \phi_1(i)\text{ otherwise}\end{array}\right.
          \end{align*}
          \begin{itemize}
          \item Let $P'=[p'_1;\cdots;p'_{\partial(e)}] \in L_s(e)$, let $C'$ and $C'_e=(e,P',[\oc_{\sige}],-)$ be contexts such that $C' \mapsto^n_s C'_e$ and $I(C',C'_e)=I(C,C_e)$. Then, we define $C'_1$ as the context such that $C' \mapsto_s C'_1$. By induction hypothesis, $C_1[\{\phi_1(i) \mapsto (p'_i)_{|\psi_1(i)}\}]=C'_1$. We want to prove that $C[\{\phi(i) \mapsto (p'_i)_{|\psi(i)}\}]=C'$. We have $C[\{\phi(i) \mapsto (p'_i)_{|\psi(i)}\}]=(f,Q,T.\wn_u,+)[\{\phi(i) \mapsto (p'_i)_{|\psi(i)}\}]$. The only interesting thing to prove is that if there exists $i$ with $\phi_1(i)=(TRA,|T|+1,q_1)$ (so $\phi(i)=(TRA,|T|+1,q)$ for some $q$) and  $\sigl(u)[q_1 \mapsto (p'_i)_{|\psi_1(i)}]=\sigl(u')$ then $u[q \mapsto (p'_i)_{|\psi(i)}] = u'$. If $q_1=[]$, then $q=[]$ and $\psi(i)=\psi_1(i).0$ so $u[q \mapsto (p'_i)_{|\psi(i)}] = (p'_i)_{|\psi_1(i).0}=((p'_i)_{|\psi_1(i)})_{|[0]}=(\sigl(u)[q_1 \mapsto (p'_i)_{|\psi_1(i)}])_{|[0]}$. So if $\sigl(u)[q_1 \mapsto (p'_i)_{|\psi_1(i)}]=\sigl(u')$, then $u[q \mapsto (p'_i)_{|\psi(i)}]=(\sigl(u'))_{|[0]}=u'$. If $q_1 \neq []$, $q_1=0.q$ and $\psi_1(i)=\psi(i)$. Let us suppose that $\sigl(u)[q_1 \mapsto (p'_i)_{|\psi_1(i)}]=\sigl(u')$, then $\sigl(u[q \mapsto (p'_i)_{|\psi_1(i)}])=\sigl(u')$. So $u[q \mapsto (p'_i)_{|\psi(i)}]=u'$.
          \item Let $P'=[p'_1;\cdots;p'_{\partial(e)}] \in L_s(e)$, let $C'$ and $C'_e=(e,P',[\oc_{\sige}],-)$ be contexts such that $C[\{\phi(i) \mapsto (p'_i)_{|\psi(i)} \}]=C'$. Then, there is a context $C'_1$ such that $C' \mapsto_s C'_1$. Repeating the calculus done for the other implication and using the hypothesis, we can deduce that $C_1[\{\phi_1(i) \mapsto (p'_i)_{|\psi_1(i)}\}]=C'_1$. So we can use the induction hypothesis, we get $C_1' \mapsto^{n-1}_s C_e'$ and $I(C'_1,C'_e)=I(C_1,C_e)$. So $C' \mapsto^n_s C_e$ and $I(C',C'_e)=I(C,C_e)$.
          \end{itemize}
          
          \item Let us suppose that $C= (g,Q,T.\oc_{\sigl(t)},-) \mapsto_s (f,Q,T.\oc_{t},-)=C_1$, crossing a $\wn C$ link upwards. Then $\psi=\psi_1$ and the only difference between $\phi$ and $\phi_1$ is that when $\phi_1(i)$ refers to a position in $t$ in $C_1$, we have to add a $0$ on the left so that $\phi(i)$ corresponds to the same position in $C$.
          \begin{equation*}
            \phi: i \mapsto \left\{ \begin{array}{l}(TRA,|T|+1,[0]@q)\text{, if }\phi_1(i)=(TRA,|T|+1,q)\\ \phi_1(i)\text{ otherwise}\end{array}\right.
          \end{equation*}
          
        \item Let us suppose that $C= (f,Q,T.\wn_{u_1}.\wn_{u_2},+) \mapsto_s (g,Q,T.\wn_{\sign(u_1,u_2)},+)$, crossing a $\wn N$ link downwards. Then,
          \begin{align*}
            \psi:& i \mapsto \left\{ \begin{array}{l}\psi_1(i)@[0]\text{, if }\phi_1(i)=(TRA,|T|+1,[])\\ \psi_1(i)@[1]\text{, if }\phi_1(i)=(TRA,|T|+2,[])\\\psi_1(i)\text{ otherwise}\end{array}\right.\\
            \phi:& (POT,i,[]) \mapsto \left\{ \begin{array}{l}(TRA,|T|+1,q)\text{, if }\phi_1(POT,i,[])=(TRA,|T|+1,[0]@q)\\(TRA,|T|+2,q)\text{, if }\phi_1(POT,i,[])=(TRA,|T|+2,q)\\ \phi_1(POT,i,[])\text{ otherwise}\end{array}\right.
          \end{align*}
          The proof is quite similar to the proof done for the case of crossing a $\wn C$ link downwards.
          
        \item Let us suppose that $C= (g,Q,T.\oc_{\sign(u_1,u_2)},-) \mapsto_s (f,Q,T.\oc_{u_1}.\oc_{u_2},-)$, crossing a $\wn N$ link upwards. Then, we set $\psi=\psi_1$ and 
          \begin{equation*}
            \phi: i \mapsto \left\{ \begin{array}{l}(TRA,|T|+1,[0]@q)\text{, if }\phi_1(i)=(TRA,|T|+1,q)\\(TRA,|T|+1,[1]@q)\text{, if }\phi_1(i)=(TRA,|T|+2,q)\\ \phi_1(i)\text{ otherwise}\end{array}\right.
          \end{equation*}

        \item Let us suppose that $C=(f,Q,T,+) \mapsto_s (g,Q,T.\wn_{\sige},+)=C_1$, crossing a $\wn D$ link downwards. Then, we set $\psi=\psi_1$ and $\phi=\phi_1$. The important point to prove is that there is no $1 \leq i \leq \partial(e)$ such that $\phi_1(i)=(TRA,|T|+1,q)$. If such a $i$ existed, then either there exists $j$ such that $\forall t \in Sig,  C_1 \downarrow^{\sige}_{(TRA,|T|+1,[])}= (g,Q,[\oc_{\sige[q \mapsto t]}],-) \mapsto^*_s (\sigma_j(D_i),Q_i,[\oc_{t}],-)$ (which is impossible because $(g,Q,[\oc_{\sige[q \mapsto t]}],-) \not \mapsto_s$) or $\forall t \in Sig, (\sigma(D_i),Q_i,[\oc_{t}],+) \rightsquigarrow^* (g,Q,[\oc_{\sige[q \mapsto t]}],-)$ (which is impossible because $\not \rightsquigarrow (g,Q,[\oc_{\sige[q \mapsto t]}],-)$.

            When crossing a $\wn D$ link upwards, we set $\psi=\psi_1$ and $\phi=\phi_1$ and the proofs are the same as in the non-exponential cases.
        \end{itemize}
      \end{itemize}
    \end{proof}
    
    \begin{theorem}
      Let $(B,P) \in Pot(B_G)$ and $(e,Q_{s-1}) \in Can_{s-1}(E_G)$, and $I$ an itinerary. Then,
      \begin{equation*}
        \left|\Set*{t \in C_s(B,P)}{
          \begin{array}{c}
            \exists Q_t, Q_t^{/e,s-1}=Q_{s-1}\\
            (\sigma(B),P,[\oc_t],+) \mapsto^*_s (e,Q,[\oc_{\sige}],-)\\
            I((\sigma(B),P,[\oc_t],+),(\sigma(B),P,[\oc_t],+))=I
          \end{array}
        } \right| 
        \leq 
        \max_{\substack{(B',P') \in Pot(B_G)\\B \succcurlyeq^{S_G}_2 B'}}|C_s(B',P')|^{\partial(B)}
      \end{equation*}
      \end{theorem}
      \begin{proof}
        To prove this, we only need to exhibit an injection from the set on the left (which we will name $F$ in this proof) to $\mapsto_s$-copies of $\partial(B)$ fixed potential boxes. If $F$ is not empty, there exists $t_0\in Sig$, $n\in \mathbb{N}$ and $Q_0 \in Pot$ such that $C=(\sigma(B),P,[\oc_{t_0}],+) \mapsto^n_s (e,Q_{0},[\oc_{\sige}],-)=C_e$. Then, by Lemma~\ref{lemma_digging_depcont}, there exists an injection $\phi: \{1,\cdots,\partial(e)\} \mapsto Pos(C)$, a mapping $\psi$ and for each $1 \leq i \leq \partial(e)$ a potential box $(D_i,Q_i)$ such that:
        \begin{itemize}
          \item For every $Q'=[q_1;\cdots;q_{\partial(e)}] \in L_s(e)$ with $Q'^{s-1}=Q_{s-1}$, if we set $C'_e=(e,Q',[\oc_{\sige}],-)$ and suppose that $C' \mapsto^n_s C'_e$ and $I(C',C'_e)=I$ then $C[\{ \phi(i) \mapsto (q'_i)_{|\psi(i)}\}]=C'$.
          \item For every $1 \leq i \leq \partial(e)$ and $t \in Sig$,
            \begin{itemize}
            \item Either $\exists j, \forall t \in Sig, C\downarrow_{\phi(i)}^t \mapsto^* (\sigma_j(D_i),Q_i,[\oc_{t}],-)$.
            \item Or $\phi(i)=(\_,\_,[])$ does not correspond to a $\oc$ trace element and $\forall t \in Sig,(\sigma(D_i),Q_i,[\oc_t],+) \rightsquigarrow^* C\downarrow_{\phi(i)}^t$.
            \end{itemize}
        \end{itemize}
        We set $D=\Set*{1\leq i \leq \partial(e)}{\begin{array}{c}\exists q \in Sig, \phi(i)=(TRA,1,q)\\ |C_s(D_i,Q_i)| \geq 2 \end{array}}$. Then, to every $t \in F$ we associates the mapping $\nu_t:i \in D \mapsto (q_i)_{|\psi(i)}$ (where $[q_1;\cdots;q_{\partial(e)}]$ is the potential $Q_t$ in the definition of $f$). To finish the proof we have to prove that:
        \begin{enumerate}
        \item \label{stat1}For every $i \in D$, $t \in F$ and $Q=[q_1,\cdots,q_{\partial(e)}]$ the potential for $e$associated to $t$, we have $(q_i)_{|\psi(i)} \in C_s(D_i,Q_i)$ 
        \item \label{stat2}The mapping $t \mapsto \nu_t$ is an injection
        \item \label{stat3}For every $i \in D$, $(B,P) \succcurlyeq^s_2 (D_i,Q_i)$
        \end{enumerate}
        Then, we will have 
        \begin{align*}
          |F| &\leq |\Set{\nu_t}{t \in F}| \\
          |F| &\leq |\Set{\nu \in Sig^D}{\forall i \in D, \nu(i) \in C_s(D_i,Q_i)}\\
          |F| &\leq (\max_{i \in D}|C_s(D_i,Q_i)|)^{|D|}\\
          |F| &\leq (\max_{\substack{(B',P') \in Pot(B_G)\\B \succcurlyeq^{S_G}_2 B'}}|C_s(B',P')|)^{\partial(e)}\\
          |F| &\leq \max_{\substack{(B',P') \in Pot(B_G)\\B \succcurlyeq^{S_G}_2 B'}}|C_s(B',P')|^{\partial(e)}
        \end{align*}
        Which is the lemma stated. We will successively prove the three needed statements.
        \begin{enumerate}
        \item Let $i \in D$, $t \in F$ and $Q=[q_1,\cdots,q_{\partial(e)}]$ the potential for $e$ associated to $t$, let us show that $(q_i)_{|\psi(i)} \in C_s(D_i,Q_i)$. We know that $\phi(i)$ corresponds to a $\oc$ trace element. So there exists $j$ such that $\forall t \in Sig, C\downarrow^t_{\phi(i)} \mapsto^*_s (\sigma_j(D_i),Q_i,[\oc_t],-)$. In particular, $C\downarrow^{(q_i)_{|\psi(i)}}_{\phi(i)} \mapsto^*_s (\sigma_j(D_i),Q_i,[\oc_{(q_i)_{|\psi(i)}}],-)$. Moreover, we know that $C[\phi(i) \mapsto (q_i)_{|\psi(i)}]=(\sigma(B),P,[\oc_u],+)$ with $u \in C_s(B,P)$. So $C\downarrow^{(q_i)_{|\psi(i)}}_{\phi(i)}=(\sigma(B),P,[\oc_v],+)$ with $v \in S_s(B,P)$. So $C\downarrow^{(q_i)_{|\psi(i)}}_{\phi(i)}$ is a $\mapsto_s$-canonical context. By Lemma~\ref{canonical_context_stable}, $(\sigma_j(D_i),Q_i,[\oc_{(q_i)_{|\psi(i)}}],-)$ is $\mapsto_s$-canonical. So, in particular $(q_i)_{|\psi(i)} \in S_s(D_i,Q_i)$. Notice that $u$ is standard, so $(q_i)_{|\psi(i)}$ is standard, so $(q_i)_{|\psi(i)} \in C_s(D_i,Q_i)$.
        \item Let us suppose that there are $t,u \in F$ such that $\nu_t=\nu_u$. Let us prove that $t=u$. Let $Q=[q_1;\cdots;q_{\partial(e)}]$ (respectively $R=[r_1;\cdots;r_{\partial(e)}]$) the potential such that $(\sigma(B),P,[\oc_t],+) \mapsto^n_s (e,Q,[\oc_{\sige}],-)$ (respectively $(\sigma(B),P,[\oc_u],+) \mapsto^n_s (e,R,[\oc_{\sige}],-)$. Then, we have $(\sigma(B),P,[\oc_{t_0}],+)[\{\phi(i) \mapsto (q_i)_{|\psi(i)}\}]= (\sigma(B),P,[\oc_t],+)$ so $t_0[\Set{p \mapsto (q_i)_{|\psi(i)}}{\phi(i)=(TRA,1,p)}]=t$. Similarly $t_0[\Set{p \mapsto (r_i)_{|\psi(i)}}{\phi(i)=(TRA,1,p)}]=u$. To prove that $t=u$, it is enough to prove that for each $1 \leq i \leq \partial(e)$ such that $\phi(i)=(TRA,1,p)$, $(q_i)_{|\psi(i)}=(r_i)_{|\psi(i)}$. From statement \ref{stat1}, we know that $(q_i)_{|\psi(i)} \in C_s(D_i,Q_i)$ and $(r_i)_{|\psi(i)} \in C_s(D_i,Q_i)$. If $|C_s(D_i,Q_i)| < 2$, then there is only one element in the set, so $(q_i)_{|\psi(i)}=(r_i)_{|\psi(i)}$. If $|C_s(D_i,Q_i)| \geq 2$, then $i \in D$. So $\nu_t(i)=\nu_u(i)$, more explicitly $(q_i)_{|\psi(i)}=(r_i)_{|\psi(i)}$.
        \item Let $i \in D$, we will show that $(B,P) \succcurlyeq^s_2 (D_i,Q_i)$. Let $(TRA,1,p)$ be $\phi(i)$. Let us consider $u$ and $v$ two different $\mapsto_s$-copies of $(D_i,Q_i)$. Then $t_0[p \mapsto u_{|\psi(i)}]$ and $t_0[p \mapsto v_{|\psi(i)}]$ are different copies of $(B,P)$. So $t_0[p \mapsto u_{|\psi(i)}]\downarrow_{\phi(i)}$ and $t_0[p \mapsto v_{|\psi(i)}]\downarrow_{\phi(i)}$ are simplifications of two different copies. Moreover, $(\sigma(B),P,t_0[p \mapsto u_{|\psi(i)}]\downarrow_{\phi(i)},+) \mapsto^*_s (\sigma_j(D_i),Q_i,u_{|\psi(i)},-)$ and $(\sigma(B),P,t_0[p \mapsto v_{|\psi(i)}]\downarrow_{\phi(i)},+) \mapsto^*_s (\sigma_j(D_i),Q_i,v_{|\psi(i)},-)$. So $(B,P) \succcurlyeq^s_2 (D_i,Q_i)$.
        \end{enumerate}
      \end{proof}

    \subsection{Polynomial bound on stratified proof-nets controlling dependence}
    \label{subsection_depcontrol_final}

    \begin{theorem}\label{theo_polynomial_bound}
      The maximal reduction length of a stratified proof-net $G$ which controls dependence, with $x=|E_G|$, $N=N_G+1$, $S=S_G+1$ and $\partial=\partial_G+1$, is bounded by 
      \begin{equation*}
        x^{3+4\left(4 N\cdot \partial^{2\cdot N\cdot S}\right)}
      \end{equation*}
    \end{theorem}
    \begin{proof}
      Let us consider any potential box $(B,P)$, and any $s \in \mathbb{N}$, then we have
      \begin{align*}
        |C_s(B,P)| &\leq |Can_{s-1}(E_G)|\cdot I_s(B,P) \cdot \max_{\substack{(B',P') \in Pot(B_G)\\B \succcurlyeq^{S_G}_2 B'}}|C_s(B',P')|^{\partial_G}\\
        |C_s(B,P)| &\leq |Can_{s-1}(E_G)|\cdot (D_G\cdot |Can_{s-1}(E_G)|)^{2\cdot N_s(B)}. |Can_{s-1}(E_G)| \cdot \max_{\substack{(B',P') \in Pot(B_G)\\B \succcurlyeq^{S_G}_2 B'}}|C_s(B',P')|^{\partial_G}\\
        |C_s(B,P)| &\leq (D_G\cdot |Can_{s-1}(E_G)|)^{2\cdot N(B)+2} \cdot \max_{\substack{(B',P') \in Pot(B_G)\\N(B') < N(B')}}|C_s(B',P')|^{\partial_G}\\
      \end{align*}
        
      So, for any $s,n \in \mathbb{N}$,
      \begin{align*}
        \max_{\substack{(B,P)\in Pot(B_G) \\ N(B)=n}}|C_s(B,P)| &\leq (D_G\cdot |Can_{s-1}(E_G)|)^{2\cdot n+2)} \cdot \max_{\substack{(B',P') \in Pot(B_G)\\N(B')<n}}|C_s(B',P')|^{\partial_G}\\
        \max_{\substack{(B,P)\in Pot(B_G) \\ N(B)=n}}|C_s(B,P)| &\leq (D_G\cdot |Can_{s-1}(E_G)|)^{(2\cdot n+2)\cdot (1+\partial_G)^n} \\
      \end{align*}

      For any $s \in \mathbb{N}$,
      \begin{align*}
        \max_{(B,P)\in Pot(B_G)}|C_s(B,P)| &\leq \left(D_G\cdot |Can_{s-1}(E_G)|\right)^{(2\cdot N_G+2)\cdot (1+\partial_G)^{N_G}} \\
        \max_{(B,P)\in Pot(B_G)}|C_s(B,P)| &\leq \left(D_G\cdot |E_G|.\max_{(B,P)\in Pot(B_G)}|C_{s-1}(B,P)|^{\partial_G}\right)^{(2\cdot N_G+2)\cdot (1+\partial_G)^{N_G}} \\
        \max_{(B,P)\in Pot(B_G)}|C_s(B,P)| &\leq \left(D_G\cdot |E_G|.\max_{(B,P)\in Pot(B_G)}|C_{s-1}(B,P)|\right)^{(2\cdot N_G+2)\cdot (1+\partial_G)^{N_G+1}} \\
      \end{align*}

      If we consider the sequence $\left(\max_{(B,P)\in Pot(B_G)}|C_s(B,P)|\right)_{s \in \mathbb{N}}$, we have bounded it by an inequality of the shape $u_{s+1} \leq (a.u_s)^b$ and $u_0 \leq a$. In this case we have $\forall s, u_s \leq a^{b^{2\cdot s +1}}$. So, if we set $S=S_G+1$, $N=N_G+1$ and $\partial=\partial_G+1$, then
      \begin{align*}
        \max_{(B,P)\in Pot(B_G)}|C_s(B,P)| &\leq \left(D_G\cdot |E_G|\right)^{(2\cdot N_G+2)\cdot (1+\partial_G)^{(N_G+1)\cdot (2\cdot s +1)}}\\
        \max_{(B,P)\in Pot(B_G)}|C_{\mapsto}(B,P)| &\leq \left(D_G\cdot |E_G|\right)^{2 N \cdot \partial^{2\cdot N\cdot S}}\\
        \max_{(B,P)\in Pot(B_G)}|C_{\mapsto}(B,P)| &\leq |E_G|^{4 N \cdot \partial^{2\cdot N\cdot S-1}}\\
        \max_{e \in E_G}|L_{\mapsto}(e)| &\leq |E_G|^{4\partial \cdot N\cdot \partial^{2\cdot N\cdot S-1}}\\
        \max_{e \in E_G}|L_{\mapsto}(e)| &\leq |E_G|^{4 N\cdot \partial^{2\cdot N\cdot S}}\\
        \end{align*}
      \begin{align*}
        T_G & =  \sum_{e \in E_G}|L_{\mapsto}(e)|+\sum_{B \in B_G}\left( |D_G(B)|\cdot \sum_{P \in L_{\mapsto}(B)}\sum_{t \in C_{\mapsto}(B,P)}|t| \right)\\
        & \leq |E_G|.|E_G|^{4 N\cdot \partial^{2\cdot N\cdot S}}+ |B_G|\cdot |D_G(B)|\cdot |E_G|^{4 N\cdot \partial^{2\cdot N\cdot S}} |E_G|^{4 N\cdot \partial^{2\cdot N\cdot S}} |E_G|^{4 N\cdot \partial^{2\cdot N\cdot S}}\\
        & \leq |E_G|^{1+4 N\cdot \partial^{2\cdot N\cdot S}}+ |E_G|^{2+3\left(4 N\cdot \partial^{2\cdot N\cdot S}\right)}\\
        T_G & \leq |E_G|^{3+4\left(4 N\cdot \partial^{2\cdot N\cdot S}\right)}\\
      \end{align*}
    \end{proof}
    
    The degree of the polynomial in the bound only depends on the depth, maximal stratum and maximal nest of the proof-net. Those three parameters are bounded by the number of boxes. So a stratified proof-net controlling dependence normalizes in a time bounded by a polynomial on the size of the proof-net, the polynomial depending only on the number of boxes of the proof-net.
    
    In Church encoding, binary words correspond to the type 
    \begin{equation*}
      W=\forall X. !(X \multimap X) \multimap !(X \multimap X)\multimap \S(X \multimap X)
    \end{equation*}
    The normalized proof-nets whose only conclusion have type $W$ have at most 1 box.

    Thus, let $G$ be a proof-net representing a function on binary words (the only conclusion of $G$ has type $W \multimap Y$ for some $Y$).  Let us suppose that for all normal proof-net $H$ representing a binary word of length $n$, the application of $G$ to $H$ is stratified and controls dependence. Then, there exists a polynomial $P$ such that for all normal proof-net $H$ representing a binary word of length $n$, the application of $G$ to $H$ normalizes in at most $P(n)$ $cut$-elimination steps.
    
    We can notice that the degree of the polynomial rises extremely fast. In fact, during the proof we used rough bounds. Otherwise, the sheer statement of the bound would have been quite complex. We believe that any real-world application would not use our general bounds, but would instead infer tighter bounds by computing the exact $\hookrightarrow$ and $\succcurlyeq_2$ relations.

    \section{Applications}
    \label{section_applications}
    In this section we will consider several restrictions of linear logic and prove that all their proof-nets are stratified, and in some cases they also control dependence. Then we will deduce strong bounds on the $cut$-elimination of those systems. For some systems ($L^3$, $L^4$ and $L^{3a}$), only weak bounds were known, for farfetched strategies of reduction. Thus, the strong bounds we prove for those systems are important steps if one wants to transform those systems into type systems for functional languages. For other systems, strong bounds on $cut$-elimination were already known and the bounds we prove are higher. We nonetheless think it is important to include them in this paper, to show how it simple the proofs become.

    \subsection{Elementary bounds}
    \subsubsection{$ELL$, Elementary Linear Logic}
    \paragraph{}$ELL$ (Elementary Linear Logic) is the first defined subsystem of linear logic which characterizes elementary time (tower of exponential of fixed height). It was hinted by Girard~\cite{girard1995light} and made explicit by Danos and Joinet~\cite{danos2003linear}. The principle of $ELL$ is to forbid dereliction and digging. Therefore, the depth of any edge is not changed by cut elimination. We can reduce at depth 0, then at depth 1, then at depth 2,... During one of those round, the size of the proof-net will at most be exponentiated: there are at most $|G|$ $?C$ nodes at depth $i$ at the beginning of round $i$. In the ``worst'' case, each of these $?C$ nodes will double the size of the proof net. So at the end of the round, the size of the net is inferior to $2^{|G|}$. By induction on the depth, we have the elementary bound.

    \begin{definition}
      $ELL$ is the fragment of $LL_0$ proof-nets where:
      \begin{itemize}
        \item There is neither $\S$, $?D$ nor $?N$ link
        \item There is no $\S$ in the formulae labelling the edges
        \item All the indexes on the ground formulae are $0$
      \end{itemize}
    \end{definition}

    \paragraph{} The equivalent of the property ``the depth is not changed by cut elimination'' in context semantics is a property of conservation of some quantity (representing the depth of a context) through the $\mapsto$ relation.

    \begin{lemma}\cite{lago2006context} \label{ellConsDepthContext}
      If $G$ is a $ELL$ proof net and $(e,P,T,p) \mapsto^*_G (f,Q,U,q)$, then $\partial(e)+ |T|_{\{\oc,\wn \}} = \partial(f) + |U|_{\{\oc,\wn\}}$
    \end{lemma}
    \begin{proof}
      The only transitions which break this property are crossing $?D$ or $?N$. There can be no such transitions in a $ELL$ proof net
    \end{proof}

    \begin{lemma}
      If $G$ is a $ELL$ proof net and $B \twoheadrightarrow B'$ then $\partial(B) > \partial(B')$
    \end{lemma}
    \begin{proof}
      If $B \twoheadrightarrow B'$, there exists $P,P' \in Pot$, $T' \in Tra$ and $s \in Sig$ such that $(\sigma(B),P,[\oc_s],+) \twoheadrightarrow (\sigma(B'), P',T',-)$. According to lemma \ref{ellConsDepthContext}, $\partial(\sigma(B))+1 = \partial(\sigma(B'))+|T'|_{!}+|T'|_{?}$. 

      The initial $!$ in the trace can not disappear along the path, so $|T'|_{\{\oc\}} \geq 1$. Moreover, by lemma \ref{lemma_underlying_step}, $(\sigma(B'),P',T',-)$ has an underlying formula, so the rightmost trace element of $T'$ is a $?_{t}$ trace element, so $|T'|_{\{\wn\}} \geq 1$. So, 
      \begin{align*}
        \left(\partial(B)-1\right)+1 & \geq \left(\partial(B')-1\right)+1+1 \\
        \partial(B)     & > \partial(B')
      \end{align*}
    \end{proof}

    \begin{theorem}\label{lemma_ell_stratifiee}
      Let $G$ be an $ELL$ proof-net, the length of its longest path of reduction is bounded by  $2^{3.|E_G|}_{3.\partial_G+1}$
    \end{theorem}
    \begin{proof}
      Let $B$ be a box of $G$, a $ELL$ proof net. Then the depth of $B$ in terms of $\twoheadrightarrow$ is finite and inferior or equal to $\partial(B)$. Indeed, let us suppose that $B \twoheadrightarrow B_1 \twoheadrightarrow \cdots \twoheadrightarrow B_n$, then $\partial(B) > \partial(B_1) > \cdots > \partial(B_n) = 0$. So $n \leq \partial(B)$. So $G$ is stratified, with $S(B) \leq \partial(B)$ for every box $B$.
    
      Theorem \ref{theoStratElementaryBound} immediately give us the stated theorem. 
    \end{proof}
    
    This bound is not new. Dal Lago already proved a strong bound using context semantics~\cite{lago2006context}. Amadio and Madet also proved a strong bound for a modal $\lambda$-calculus inspired by $ELL$~\cite{madet2011elementary}, this proof could easily be adapted to $ELL$ itself. Both proofs use tighter bounds than us. However, none of them explicit the bound. The normalization sequences are proved to have length inferior to $P_{\partial(G)}(|E_G|)$ with $P_i$ being defined by induction on $i$.

    Notice that the strata only depends on the depth of the proof-net. So the height of the exponential tower depends only on the depth and maximal level of the proof-net. Moreover, the depth of a normalized proof-net representing a church numeral or a binary word are both bound by $1$. So, if $G$ is a proof-net whose only pending edge has type $!W_{LL} \multimap A$ (with $W_{LL}=\forall \alpha. \oc (\alpha \multimap \alpha) \multimap \oc (\alpha \multimap \alpha) \multimap \oc (\alpha \multimap \alpha)$ the type for binary words), \textbf{for all} normalized proof-net $H$ of type $!W_{LL}$, the maximum number of steps to eliminate the $cut$ in the application of $G$ to $H$ is bounded by $2^{3|E_G|+3|E_H|}_{3\partial_G+1}$.
    
    \subsubsection{$L^3$, Linear Logic by Levels}
    $L^3$ (Linear Logic by Levels) is a system introduced by Baillot and Mazza~\cite{baillot2010linear} which generalizes $ELL$. $L^3$ is defined as the fragment of linear logic containing exactly the proof-nets for which we can label each edge $e$ with an integer $l(e)$ verifying the rules of figure \ref{figureLabelL3}.
    
    Let $G$ be a $L^3$ proof-net, we define $l_G$ as $\max \Set{l(e)}{e \in E_G}$. There has been a weak elementary bound proved for $L^3$ proof-nets in~\cite{baillot2010linear} for a particular strategy, but no strong bound.
    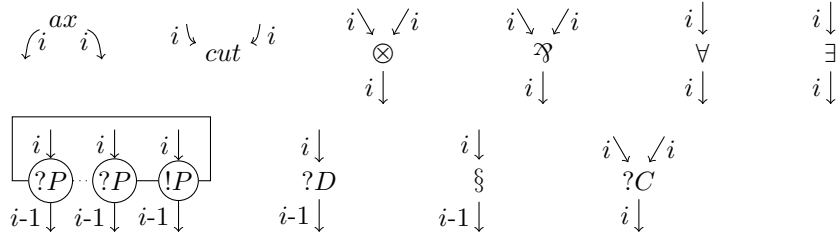
\begin{figure}\centering
      \begin{tikzpicture}
        \begin{scope}[scale = 0.85]
          \tikzstyle{type}=[opacity=0]
          \begin{scope}[shift={(0,0.5)}]
            \draw (0,0) node [ax] (ax) {};
            \draw[ar, out=-20,in=100] (ax) to node [type] {$A$} node [level] {$i$} (0.6,-0.6);
            \draw[ar, out=-160,in=80] (ax) to node [type] {$A^\perp$} node [level,right] {$i$} (-0.6,-0.6);
          \end{scope}
          
          \begin{scope}[shift={(2.5,0)}]
            \draw (0,0) node [cut] (cut) {};
            \draw (cut) ++ (-0.6,0.6) node (G) {};
            \draw (cut) ++ (0.6,0.6) node (H) {};
            \draw[ar,out=-80,in=160] (G) to node [type] {$A$} node [level] {$i$} (cut);
            \draw[ar,out=-100,in=20] (H) to node [type] {$A^\perp$} node [level,right] {$i$} (cut);
          \end{scope}
          
          \begin{scope}[shift={(5,0)}]
            \draw (0,0) node [tensor] (tens)  {};
            \draw [revar] (tens)--++(120:0.8) node [type] {$A$} node [level] {$i$};
            \draw [revar] (tens)--++(60:0.8) node [type] {$B$} node [level,right] {$i$};
            \draw [ar] (tens)--++(0,-0.8) node [type] {$A \otimes B$} node [level] {$i$};
          \end{scope}
          
          \begin{scope}[shift={(7.5,0)}]
            \draw (0,0) node [par] (par)  {};
            \draw [revar] (par)--++(120:0.8) node [type] {$A$} node [level] {$i$};
            \draw [revar] (par)--++(60:0.8) node [type] {$B$} node [level,right] {$i$};
            \draw [ar] (par)--++(0,-0.8) node [type] {$A \parr B$} node [level] {$i$};
          \end{scope}
          
          \begin{scope}[shift={(10,0)}]
            \draw (0,0) node [forall] (forall)  {};
            \draw [revar] (forall)--++(0,0.8) node [type] {$A$} node [level] {$i$};
            \draw [ar] (forall)--++(0,-0.8) node [type] {$\forall X.A$} node [level] {$i$};
          \end{scope}
          
          \begin{scope}[shift={(12,0)}]
            \draw (0,0) node [exists] (exists)  {};
            \draw [revar] (exists)--++(0,0.8) node [type] {$A$} node [level] {$i$};
            \draw [ar] (exists)--++(0,-0.8) node [type] {$\exists X.A$} node [level] {$i$};
          \end{scope}

          \begin{scope}[shift={(0,-2.)}]
            \begin{scope}[shift={(4,0)}]
              \draw (0,0) node [der] (der)  {};
              \draw [revar] (der)--++(0,0.8) node [type] {$A$} node [level] {$i$};
              \draw [ar] (der)--++(0,-0.8) node [type] {$?A$} node [level] {$i$-1};
            \end{scope}
            \begin{scope}[shift={(9,0)}]
              \draw (0,0) node [cont] (cont)  {};
              \draw [revar] (cont)--++(120:0.8) node [type] {$?A$} node [level] {$i$};
              \draw [revar] (cont)--++(60:0.8) node [type] {$?A$} node [level,right] {$i$};
              \draw [ar] (cont)--++(0,-0.8) node [type] {$?A$} node [level] {$i$};
            \end{scope}
            \begin{scope}[shift={(6.5,0)}]
              \draw (0,0) node [neut] (neut)  {};
              \draw [revar] (neut)--++(0,0.8) node [type] {$A$} node [level] {$i$};
              \draw [ar] (neut)--++(0,-0.8) node [type] {$\S A$} node [level] {$i$-1};
            \end{scope}
            
            \begin{scope}[shift={(1,0)}]
              \draw (0.8,0) node [princdoor] (bang) {};
              \draw (-0.2,0) node [auxdoor] (whyn2) {};
              \draw (-1.2,0) node [auxdoor] (whyn1) {};
              \draw[revar] (bang) --++ (0,0.8) node [type] {$A$} node [level] {$i$};
              \draw[ar] (bang) --++ (0,-0.8) node [type] {$! A$} node [level] {$i$-1};
              \draw[revar] (whyn1) --++(0,0.8) node [type] {$A$} node  [level] {$i$};
              \draw[ar] (whyn1) --++ (0,-0.8) node [type] {$?A$} node  [level] {$i$-1};
              \draw[revar] (whyn2) --++(0,0.8) node [type] {$A$} node  [level] {$i$};
              \draw[ar] (whyn2) --++ (0,-0.8) node [type] {$?A$} node  [level] {$i$-1};
              \draw (whyn2)--(bang) -| ++(0.5,1) -| ($(whyn1)+(-0.6,0)$) -- (whyn1);
              \draw [dotted] (whyn1)-- (whyn2);
            \end{scope}
          \end{scope}
        \end{scope}
      \end{tikzpicture}
      \caption{\label{figureLabelL3}Relations between levels of neighbour edges.}
    \end{figure}

    Baillot and Mazza noticed that the level of a box is not changed by $cut$-elimination~\cite{baillot2010linear}. This property also has an equivalent in the context semantics presentation: lemma \ref{L3ConsLevelContext}. Notice that the conservation of the quantity is only true for the $\rightsquigarrow$ relation, not the $\mapsto$ relation. This is why we used the $\rightsquigarrow$ relation in the definition of stratification. This makes the reasonings on $L^3$ more complex and partly explains why it was difficult to prove a strong bound for $L^3$.
    
    \begin{lemma} \label{L3ConsLevelContext}
      If $G$ is a $L^3$ proof-net and $(e,P,T,p) \rightsquigarrow^*_G (f,Q,U,q)$, then \begin{equation*}l(e)+ |T|_{\{!,?,\S\}} = l(f) + |U|_{\{!,?,\S\}} \end{equation*}
    \end{lemma}
    \begin{proof}
      We can examine each $\rightsquigarrow$ rule. It is straightforward to see that none of those rules break the property. Notice that the $\hookrightarrow$ would break the property because the level on two doors of the same box can be different.
    \end{proof}

    \begin{lemma}
      If $G$ is a $L^3$ proof-net and $B \twoheadrightarrow B'$ then $l(\sigma(B)) > l(\sigma(B'))$
    \end{lemma}
    \begin{proof}
      If $B \twoheadrightarrow B'$, there exists $P,P' \in Pot$, $T' \in Tra$ and $t \in Sig$ such that $(\sigma(B),P,[!_t],+) \twoheadrightarrow (\sigma(B'), P',T',-)$. According to lemma \ref{L3ConsLevelContext}, $l(\sigma(B))+1 = l(\sigma(B'))+|T'|_{!}+|T'|_{?}$. 
      
      The initial $!$ in the trace can not disappear along the path, so $|T'|_{\{\oc\}} \geq 1$. Moreover, the underlying formula of $(\sigma(B),P,[!_t],+)$ is well-defined, so the underlying formula of $(\sigma(B'),P',T',-)$ is well-defined (Lemma \ref{lemma_underlying_mapsto}). The rightmost trace element of $T'$ is a $?_{u}$ trace element, so $|T'|_{\{\wn\}} \geq 1$. So, 
  \begin{align*}
    l(\sigma(B))+1 & \geq l(\sigma(B'))+1+1 \\
    l(\sigma(B))   & > l(\sigma(B'))
  \end{align*}
\end{proof}

    \begin{theorem}\label{lemma_l3_stratifiee}
      Let $G$ be an $L^3$ proof-net, the length of its longest path of reduction is bounded by  $2^{3.|E_G|}_{3.l_G+1}$
    \end{theorem}
    \begin{proof}
      Let $B$ be a box of $G$, a $ELL$ proof net. Then the depth of $B$ in terms of $\twoheadrightarrow$ is finite and inferior or equal to $l(B)$. Indeed, let us suppose that $B \twoheadrightarrow B_1 \twoheadrightarrow \cdots \twoheadrightarrow B_n$, then $l(B) > l(B_1) > \cdots > l(B_n) = 0$. So $n \leq l(B)$. So $G$ is stratified, with $S(B) \leq l(B)$ for every box $B$.
    
      Theorem \ref{theoStratElementaryBound} immediately give us the stated theorem. 
    \end{proof}
    
    This bound is the first strong elementary bound for $L^3$. Nonetheless, we can compare it to the weak bound proved by Baillot and Mazza. Their bound was tighter. Indeed, they prove that their strategy reaches a normal form in at most $(l_G+1)\cdot2^{|E_G|}_{2\cdot l_G}$ steps. So their exponential tower has height $2 \cdot l_G$ while ours has height $3 \cdot l_G+1$. We do not think that the strategy used by Baillot and Mazza is particularly efficient. So, we believe that the gap between our two bounds is partly due to the fact that their proof is specialized for $L^3$ and mostly due to the rough bounds we used.

    \subsection{Polynomial bounds}
    \subsubsection{$LLL$: $ELL$ restricted to one auxiliary door boxes}
    \paragraph{}$LLL$~\cite{girard1995light} is a milestone in the implicit complexity field. Though $BLL$~\cite{girard1992bounded} was the first subsystem of linear logic characterizing polynomial time, $BLL$ was not totally implicit because it was based on polynomials indexing the formulae. So $LLL$ is the first totally implicit characterization of polynomial time based on linear logic. In order to obtain polynomial time soundness, Girard restricts linear logic in multiple ways: $?N$ and $?D$ are forbidden (so $LLL$ is a subsystem of $ELL$) and we allow at most one auxiliary door by box.

    \paragraph{}This system may be the most studied system among the linear logic based systems in implicit computational complexity. Therefore, there are already several proofs of the strong polynomial bound it entails. One of those is already done in the context semantics framework~\cite{lago2006context}. So, as in the $ELL$ case, the following result is not new. We include it, as an introduction to the next examples which will be slightly more complex.

    \begin{lemma}\label{lemma_lll_dependence_control}
      Let $G$ be a proof-net without digging, where all boxes have at most one auxiliary door. Then, $\forall B,C \in B_G, B \succcurlyeq_k C \Rightarrow k \leq 1$
    \end{lemma}
    \begin{proof}
      As a first step, we will prove the following property: $\forall (C,Q) \in Pot(B_G), \sum_{(B,P) \succcurlyeq_k (C,Q)} k  \leq 1$. Suppose there are two potential boxes $(B_1,P_1)$ and $(B_2,P_2)$ such that $(B_1,P_1) \succcurlyeq_{k_1} (C,Q)$ and $(B_2,P_2) \succcurlyeq_{k_2} (C,Q)$. Box $C$ has only one auxiliary door (by definition of $LLL$) so, $(\sigma(B_1),P_1,[!_{t_1}],+) \rightsquigarrow^* (\sigma_1(C),Q,[!_{\sige}],-)$ and $(\sigma(B_2),P_2,[!_{t_2}],+) \rightsquigarrow^* (\sigma_1(C),Q,[!_{\sige}],-)$. However, $\rightsquigarrow$ is bideterministic and $\not \rightsquigarrow (\sigma(B_i), P_i, [\oc_{t_i}],+)$ so $(B_1,P_1)=(B_2,P_2)$. To prove that $k$ is always $\leq 1$, we use the same argument.
      
      Then to prove the lemma, we suppose that there exists two sequels: $(B,P) \succcurlyeq_{k_1} \cdots \succcurlyeq_{k'_n} (B_n,P_n) \succcurlyeq (C,Q)$ and  $(B,P) \succcurlyeq_{k_1'} \cdots \succcurlyeq^1_{k'_{n'}} (B'_{n'},P'_{n'}) \succcurlyeq (C,Q)$. Using the injection lemma (lemma \ref{injection_lemma}), we can prove by induction on $i$ that $(B_{n-i},P_{n-i})=(B'_{n'-i},P'_{n'-i})$ and $k'_{n'-i}=k_{n-i}=1$. This proves that the sum in the definition of $B \succcurlyeq C$ has at most one term and this term is $1$.
    \end{proof}

    \begin{theorem}
      Let $G$ be a $LLL$ proof-net, $x=|E_G|$ $\partial=\partial_G+1$. The lenght of the longest path of reduction of $G$ is bounded by
      \begin{equation*}
        x^{3+4\left(4  (\partial_G +1)^{2+ 2 \cdot \partial_G }\right)}
      \end{equation*}
    \end{theorem}
    \begin{proof}
      $LLL$ is included in $ELL$ so, $G$ is stratified and $S_G \leq \partial_G$ (Theorem \ref{lemma_ell_stratifiee}). Moreover, from Lemma \ref{lemma_lll_dependence_control}, $LLL$ controls dependence and $N_G =0$. So, Theorem \ref{theo_polynomial_bound} gives us the expected bound.
    \end{proof}
    The polynomial only depends on the depth of the proof net. The depth of cut-free binary words (in the church encoding) is bounded by $1$, so for any proof-net $G$ of $LLL$ representing a function on binary words, there exists a polynomial $p$ such that the number of steps to reduce $G$ applied to $H$ ($H$ being a cut-free proof-net representing a word of size $n$) is inferior to $p(n)$.
    
    \subsubsection{$L^4$: $L^3$ restricted to one auxiliary door boxes}
    Similarly to $LLL$ which is obtained from $ELL$ by forbidding boxes with more than one auxiliary door, Baillot and Mazza restricted $L^3$ to capture polynomial time. However, forbidding boxes with more than one auxiliary door was not enough to ensure $Ptime$ soundness so they also forbid the digging. $L^4$ is defined as the fragment of $L^3$ containing exactly the proof-nets for without digging and such that all the $!$ boxes have at most one auxiliary door. 

We define $l_G$ as $\max \Set{l(e)}{e \in E_G}$. There has been a weak polynomial bound proved for $L^4$ proof-nets in~\cite{baillot2010linear} for a particular strategy, but no strong bound\footnote{In fact, a proof of a strong bound is claimed in~\cite{vercelli2010phd}, but it contains flaws which do not seem to be easily patchable (more details in appendix \ref{appendix_vercelli}).}. This is problematic for the goal of designing a type system for $\lambda$-calculus based in $L^4$, because it is unclear wether the particular strategy on proof-nets of~\cite{baillot2010linear} could be converted into a $\beta$-reduction strategy. Therefore we want to prove a strong polynomial bound.

    \begin{theorem}\label{theo_polynomial_bound_l4}
      Let $G$ be a $L^4$ proof-net, the maximal reduction length of $G$, with $x=|E_G|$, is bounded by 
      \begin{equation*}
        x^{3+4\left(4 (\partial_G+1)^{2+2\cdot S_G}\right)}
      \end{equation*}
    \end{theorem}
    \begin{proof}
      From Theorem \ref{lemma_l3_stratifiee}, we get that $G$ is stratified with $S_G \leq l_G$. From Lemma \ref{lemma_lll_dependence_control}, we get that $G$ controls dependence with $N_G =0$. So, we can use Theorem \ref{theo_polynomial_bound} to get the expected result.
    \end{proof}

    \subsubsection{$MS$, extending the ``one auxiliary door'' condition}
    In \cite{roversi2009some}, Roversi and Vercelli define $MS$, a framework of subsystems of $ELL$. First, for any $M \in \mathbb{N}$, they define a set of formulae $\mathcal{F}^M_{MS}$ defined as $\mathcal{F}_{LL}$ where we index $\oc$ and $\wn$ modalities by integers in $\{1,2,\cdots,M\}$. 

    Then, they pay special attention to the contraction and promotion schemes of Figure \ref{rules_labeling}. They name $Y_q(n,m)$ the scheme allowing a contraction link with $\wn_qA$ as a conclusion, $\wn_nA$ as a left premise and $\wn_mA$ as a right premise. They name $P_q(m_1,\cdots,m_k)$ the scheme allowing a box with conclusions $\wn_{m_1}A_1$, $\wn_{m_2}A_2$, …, $\wn_{m_k}A_k$, $\oc_q C$. 
    
    A subsystem of $MS$ is a set of contraction and promotion schemes instances. It represents the set of proof-nets which can be build using only those instances of contraction and promotion. For example $\{Y_6(1,2)\}$ represents a very restricted system with no box and with contraction allowed only with premises $\wn_1A$ and $\wn_2A$ and conclusion $\wn_6(A)$. We can get $ELL$ with the following $MS$ subsystem (compared to $ELL$ this system has indices on $\oc$ and $\wn$ modalities but do not use them):$\Set{P_q(m_1,\cdots,m_k)}{k,q,m_1,\cdots,m_k \in \mathbb{N}} \cup \Set{Y_q(n,m)}{q,n,m \in \mathbb{N}}$.

    In \cite{roversi2010local}, they prove that the subsystems containing (up to permutation of indices):
    \begin{itemize}
    \item All the rules $Y_q(m,n)$ for every $q < m,n$
    \item All the rules $Y_q(q,n)$ for every $q < n$
    \item All the rules $P_q(m_1,\cdots,m_k)$ for every $m_1,\cdots,m_k < q$
    \item All the rules $P_q(q,m_1,\cdots,m_k)$ for every $m_1,\cdots,m_k < q$
    \item Either $Y_q(q,q)$ or $P_q(m_1,\cdots,m_k)$ for every $m_1,\cdots,m_k \leq q$
    \end{itemize}
    are $PTIME$ sound. The ``or'' in the last line is exclusive. Those systemes are named the $PTIME$-maximal $MS$ systems (because they are the maximal $PTIME$ sound subsystems of $MS$ verifying some other conditions).

    For the following lemma, let us observe that we could extend the notion of underlying formula to the $MS$ proof-nets. Contrary to the indices on ground formulae which are deleted in the underlying formula, we will keep the indices on the $\wn$ and $\oc$. Then, Theorem \ref{theorem_underlying_formula} becomes:
    \begin{lemma}
      Let $C$ and $D$ be contexts of a proof-net $G$ of a $PTIME$-maximal $MS$ system. If $C \mapsto D$ and $\beta(C)=\oc_m A$ then $\beta(D)=\oc_nB$ with $m \leq n$
    \end{lemma}

    \begin{theorem}
      Let $G$ be a proof-net of a $PTIME$-maximal $MS$ system. The maximal reduction length of $G$, with $x=|E_G|$ and $m$ the maximal index on $\oc$ and $\wn$ modalities in $G$, is bounded by 
      \begin{equation*}
        x^{3+4\left(4 (m+1)\cdot (\partial_G+1)^{2\cdot (m+1) \cdot (\partial_G+1)}\right)}
      \end{equation*}
    \end{theorem}

    \subsubsection{$L^{3a}$, merging the ideas of $L^3$ and $MS$}
    The system $L^{3a}$ relies on the idea of relaxing the ``one auxiliary door'' condition~\cite{dorman2009linear}on $L^4$. It is defined as the subset of $L^3$ where : there is no digging and the level of the auxiliary doors of a box are greater or equal to the level of the principal door (with a maximum of one auxiliary door at the same level than the principal door).

    \begin{lemma}
      Let $G$ be a $L^3_a$ proof-net then $G$ control dependence and for every $B \in B_G$, $N(B) \leq l(B)$
    \end{lemma}
    \begin{proof}
      We will prove that if $B \succcurlyeq^s_2 C$, then $l(B)>l(C)$. The lemma will follow by induction.

      Let us suppose that $B \succcurlyeq^s_2 C$, then there exists $P \in L_{\mapsto}(B)$, $Q \in L_{\mapsto}(C)$ and two sequences of $\succcurlyeq^s$: $(B,P) \succcurlyeq^s (B_0,P_0) \succcurlyeq^s \cdots \succcurlyeq^s (B_n,P_n) \succcurlyeq^s (C,Q)$ and $(B,P) \succcurlyeq (B'_0,P'_0) \succcurlyeq^s \cdots \succcurlyeq^s (B_{n'},P_{n'}) \succcurlyeq^s (C,Q)$. Let $i= \max \Set{j \in \mathbb{N}}{ \forall k \leq j, (B_{n-k},P_{n-k})=(B'_{n'-k},P'_{n'-k})}$, then either $l(B_{n-k-1}) > l(B_{n-k})$ or $l(B_{n'-k-1}) > l(B_{n'-k})$. The level of boxes is decreasing or stable during the remaining of the sequence, so $l(B)>l(C)$.
    \end{proof}

\begin{theorem} \label{theorem_strong_bound_l3a}
  Let $G$ be a $L^{3a}$ proof-net, with $x=|E_G|$, then the length of the longest reduction path is inferior to:
    \begin{equation*}
      x^{3+4\left(4 (S_G+1)\cdot (\partial_G+1)^{2\cdot (S_G+1)^2}\right)}
    \end{equation*}
\end{theorem}
    
    \section{Strong polynomial bound for $L^4_0$}
    \label{section_l40}
    \paragraph{} When Baillot and Mazza created $L^4$, they noticed that $\S$ commuted with every other connective (for example $\forall X, \forall Y, \S (X \otimes Y) \multimap (\S X \otimes \S Y)$ and $\forall X, \forall Y, (\S X \otimes \S Y) \multimap \S (X \otimes Y)$ are provable in $L^4$. Therefore, it was a pity to differentiate proof-nets which had not any computationnal difference. They created $L^4_0$ which can be considered as a modification of $L^4$ where all the $\S$ connectives of the formulae are pushed down to the ground formulae and where all the $\S$ links of the proof-nets are pushed up to the axioms.
    
    \begin{definition}
      $L^4_0$ is the subset of $LL_0$ proof-nets where there is no $\S$ (neither as connectives in the formulae nor as links in the proof-nets), where all boxes have at most one auxiliary door and where we can label each edge with a level according to the rules of figure \ref{figureLabelL3}.
    \end{definition}

    \paragraph{} The commuting diagram of figure \ref{fig_commute_l4_l40} is used in~\cite{baillot2010linear} to prove that $L^4_0$ captures $Ptime$ in some sense. In the following, we will use it to prove a strong polynomial bound for $L^4_0$. Basically, it shows that a calculus in $L^4_0$ can be simulated in $\eta$-expansed $L^4_0$ and that a calculus in $\eta$-expansed $L^4_0$ is the same that a calculus in $L^4$. The rules of $\eta$-expansion can be seen in figure \ref{def_eta_expansion}. The strong polynomial bound for $L^4_0$ results from the following inequality:
    \begin{equation} \label{ineq_l40}
      T_{G_0} \leq T_{G_1} \leq T_{G} \leq poly(|G|) \leq poly(|G_0|)
    \end{equation}
    
    \begin{figure}\centering
      \begin{tikzpicture}
        \node [ax]     (ax)   at (0,0)  {};
        \draw [ar, out=-160, in=90] (ax) to node [type, left] {$A^\perp \parr B^\perp$} ($(ax)+(-1,-0.7)$);
        \draw [ar, out= -20, in=90] (ax) to node [type] {$p.(A    \otimes B)$}     ($(ax)+( 1,-0.7)$);

        \node [par]    (par)  at (7.5,0)            {};
        \node [tensor] (tens) at ($(par)+(3,0)$) {};
        \node [ax]     (axA)  at ($(tens)!0.5!(par)+(0,0.7)$) {};
        \node [ax]     (axB)  at ($(tens)!0.5!(par)+(0,1.5)$)   {};
        \draw [ar,dashed, out=-180, in=130] (axA) to node [type, above left]  {$A^\perp$} (par) ;
        \draw [ar,dashed, out=   0, in=130] (axA) to node [type, above]       {$p.A$}    (tens);
        \draw [ar,dashed, out=-180, in= 50] (axB) to node [type, above]       {$B^\perp$} (par);
        \draw [ar,dashed, out=   0, in= 50] (axB) to node [type, above right] {$p.B$}    (tens);
        \draw [ar] (par) -- ($(par) +(0,-0.8)$) node [type,left]      {$A^\perp \parr B^\perp$};
        \draw [ar] (tens)-- ($(tens)+(0,-0.8)$) node [type] {$p.A \otimes p.B = p. (A \otimes B)$};
        
        \draw [->]  ($(ax)!0.3!(par)$) -- ($(ax)!0.7!(par)$) node [below left] {$\eta$};

        \begin{scope}[shift={(0,-3)}]
          \node [ax]     (ax)   at (0,0)  {};
          \draw [ar, out=-160, in=90] (ax) to node [type, left] {$\wn A^\perp$} ($(ax)+(-1,-0.7)$);
          \draw [ar, out= -20, in=90] (ax) to node [type] {$p.(\oc A)$}     ($(ax)+( 1,-0.7)$);

          \node [auxdoor]    (aux)  at (7.5,0)             {};
          \node [princdoor]  (princ) at ($(aux)+(3,0)$) {};
          \node [ax]     (axA)  at ($(aux)!0.5!(princ)+(0,1)$) {};
          \draw [ar,dashed, out=-180, in=90] (axA) to node [type, left]  {$A^\perp$} (aux) ;
          \draw [ar,dashed, out=   0, in=90] (axA) to node [type]        {$p.A$}    (princ);
          \draw [ar] (aux)  -- ($(aux)  +(0,-0.8)$) node [type,left]      {$\wn (A^\perp)$};
          \draw [ar] (princ)-- ($(princ)+(0,-0.8)$) node [type] {$\oc (p.A) = p. (\oc A)$};
          \draw (princ) -| ++ (1,1.5) -| ($(aux)+(-1,0)$) -- (aux) -- (princ);

          \draw [->]  ($(ax)!0.3!(aux)$) -- ($(ax)!0.7!(aux)$) node [below left] {$\eta$};
        \end{scope}
      \end{tikzpicture}
      \caption{\label{def_eta_expansion}Rules of the $\eta$ relation}
    \end{figure}
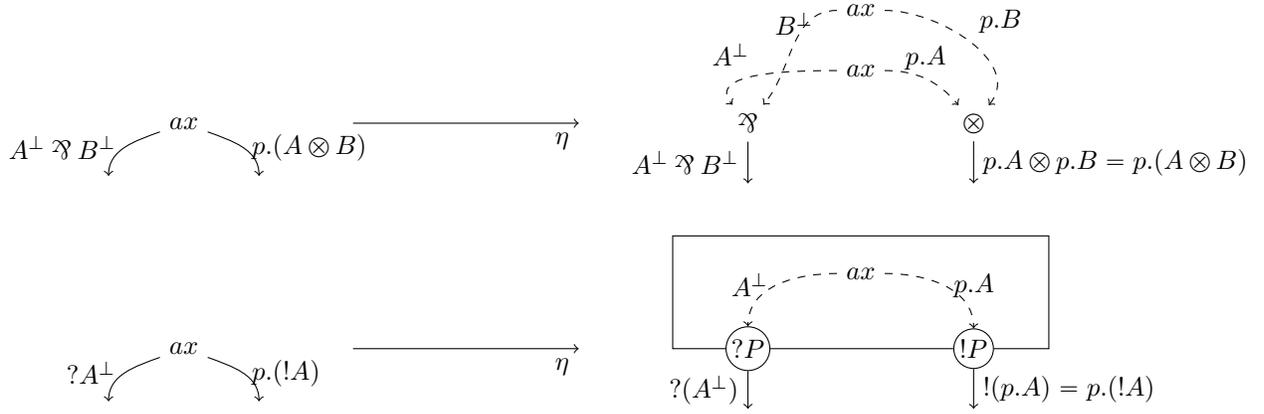

    The polynomial will only depend of the depth, the maximum level of $G_0$, and the size and depth of its formulae. Thus, when a $L^4_0$ proof-net is applied to binary words in normal form, the length of the reduction is bounded by a fixed polynomial on the size of the argument word.

    \paragraph{} The diagram of figure \ref{fig_commute_l4_l40} uses several functions $(\_)_0,(\_)_1,(\_)^-$ on proof-nets and a relation $\rightarrow_{\eta}$ on proof-nets, which we will define below. $(\_)_0$ and $(\_)_1$ were first defined in~\cite{baillot2010linear} (along with the definition of $L^4_0$). As in the original paper, we will first define $(\_)_0$ and $(\_)_1$ on formulae: We set $X_0 = X$, $(X^\perp)_0 = X^\perp$, $(\S A)_0 = 1.(A_0)$ and $(\_)_0$ commutes with all the other connectives. Similarly, $(i.X)_1= \S^i X$, $(i.X^\perp)_1 = \S^i X^\perp$ and $(\_)_1$ commutes with all the other connectives. We can notice that $(\_)_0 \circ (\_)_1$ is the identity on $L^4_0$ while $(\_)_1 \circ (\_)_0$ is not the identity on $L^4$ because it pushes all the $\S$ to the axioms. Now, we will define $(\_)_0$ and $(\_)_1$ on proof-nets. 
    
    $G_0$ is the proof-net obtained from $G$ in the following way. For every edge $e$ we replace the label $\beta(e)$ by $(\beta(e))_0$ and we delete all $\S$ link.
    
    $G_1$ is the proof-net obtained from $G$ by: 
    \begin{itemize}
      \item First, $\eta$-expansing each axiom link
      \item Then, replace the axiom of conclusions labelled by $p.X^\perp$ and $q.X$ (with $p > q$) by the axiom of conclusions $e$, labelled by $X^\perp$ and $f$, labelled by $X$. Finally we add $p$ links labelled by $\S$ on edge $e$ and $q$ links labelled by $\S$ on edge $f$.
    \end{itemize}

    \begin{lemma}
      If $G$ is a $L^4$ proof-net, $G_0$ is a $L^4_0$ proof-net. If $G$ is a $L^4_0$ proof-net, $G_1$ is a $L^4$ proof-net.
    \end{lemma}

    \begin{figure}
      \centering
      \begin{tikzpicture}
        \begin{scope}[scale=0.7]
          \tikzstyle{every node}=[minimum width=0.3]
          \draw [lightblue, fill] (-2.2,4.5) rectangle (4.5,1.4);
          \draw [lightgreen,fill] (-2.2,0.5) rectangle (4.5,-1);
          \draw [gray, fill, opacity=0.2] (5.6,-1.2) rectangle (6.8,4.5);
          \draw [gray] (7.2,1.7) node {\large {$LL$}};
          \draw [darkgreen] (-2.7,0) node {\large {$L^4$ }};
          \draw [blue] (-2.7,3) node {\large {$L^4_0$}};
          \draw (0,4) node (pi0) {$G_0$};
          \draw (0,-0.5) node (pi)  {$G$};
          \draw (4,-0.5) node (pi') {};
          \draw (4,4) node (pi0'){};
          \draw (0,2) node (pi1) {$G_1$};
          \draw (4,2) node (pi1'){};
          \draw (6.3,1) node (fin) {};
          
          \draw[->, thick] (pi0) to [bend right=45] node [midway, left] {$(.)_1$} (pi);
          \draw[->, thick, dashed] (pi0')-- (pi1') node [midway, right] {$\eta$} node [midway, left] {$*$};
          \draw[->, very thick] (pi0)--(pi0') node [midway, below, black] {$NF$} node [near end, above, red]{$n_0$};
          \draw[->, very thick](pi1)--(pi1') node [midway, below, black] {$NF$} node[near end, above, red]{$n_0^\eta$};
          \draw[->, very thick]  (pi) --(pi') node [midway, below, black] {$NF$} node [near end, above, red] {$n_1$};
          \draw[->, thick] (pi)--(pi1) node [midway, right] {$(.)_0$};
          \draw[->, thick] (pi0)--(pi1) node [midway, right] {$\eta$} node [left, midway]{*};
          \draw[->, thick] (pi1')--(fin);
          \draw[->, thick] (pi')--(fin);
        \end{scope}
      \end{tikzpicture}
      \caption{ \label{fig_commute_l4_l40} The computations of $\eta$-expansed $L^4$ and $\eta$-expansed $L^4_0$ are quite related.}
    \end{figure}
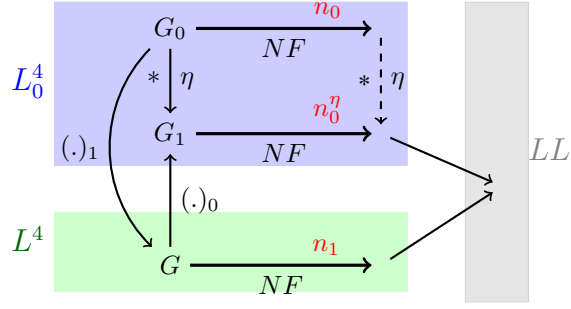

We will prove the inequalities of \ref{ineq_l40} one by one. Here are the ideas of the proof. The proofs will be detailed in the following subsections. 
\begin{itemize}
\item $T_{G_0} \leq T_{G_1}$: the $\eta$-expansion increases locally the net, leaving the rest unchanged. $T_G$ is a sum over edges. $\eta$-expansion will leave most terms unchanged ($L_{G_0}(e)=L_{G_1}(e)$) and add some terms in the sum.
\item $T_{G_1} \leq T_{G}$: $G_1$ is exactly $G$ without some $\S$ nodes. The reasoning is similar as above.
\item $T_G \leq poly(|G|)$: it is given by Theorem \ref{theo_polynomial_bound_l4}
\item $poly(|G|) \leq poly(|G_0|)$: $G$ is just $G_0$ $\eta$-expansed with $\S$ made explicit. We show that those operations do not increase the size of the net too much. The exact statement takes into account the level of $G_0$, its depth and the maximum size of its formula.
\end{itemize}

\subsection{$\eta$-expansion increases $T$:  $T_{G_0} \leq T_{G_1}$}
\begin{lemma}
  \label{eta_incr_t}
  Let $G_0$ be a proof net. Suppose $G_0 \rightarrow_\eta G_1$, then $T_{G_0} \leq T_{G_1}$.
\end{lemma}
\begin{proof}
  We will define a copymorphism $(E_{G_0},E'_{G_0},\phi,\psi)$ from $G_0$ to $G_1$. Were $E_{G_0}'$ represent the edges of $G_1$ which have a canonical corresponding edge in $E_{G_0}$ (the edges which are not dashed in figure \ref{def_eta_expansion}). $\phi$ represents the canonical correspondence, and $\psi$ is empty. We need to prove that this 4-uple is a copymorphism. All conclusions of $?D$ and $?W$ links of $G_1$ are in $E_{G_0}'$, as are the pending edges of $G_1$. So the only difficult point to prove is that the paths between canonical contexts are kept intact by $\phi$. We have to prove that whenever $C=(e,P,T,-) \mapsto_{G_0} (f,P,T,+)=D$ (with $e$ and $f$ the two conclusions of the $\eta$-expansed axiom and $C$ canonical) then $(\phi(e),P,T,-) \mapsto_{G_1}^* (\phi(f),P,T,+)$.

  If $C$ is canonical, then $\beta(e,P,T,-)$ is well-defined. We will make a disjunction over the size of $T$:
  \begin{itemize}
  \item If it has length $1$, then $T=!_{s}$. So, $\beta(e)=?A$ for some formula $A$ (otherwise, $\beta(e,P,T,-)$ would be undefined). The $\eta$-expansion creates a box, and $(\phi(e),P,T,-) \hookrightarrow (\phi(f),P,T,+)$.
  \item Else, let's suppose the left-most trace element of $T$ is $\parr_l$ (the other cases are similar), then $\beta(e)=A_l \otimes A_r$ (otherwise $\beta(e,P,T,-)$ would be undefined). So $(e,P,T,-)\mapsto_{G_1}^3 (f,P,T,+)$.
  \end{itemize}

  As proved in the subsection \ref{subsection_proof_dal_lago}, $T_{G_1}-T_{G_0}= \left( T_{G_1}^1-T_{G_0}^1 \right) + \left( T_{G_1}^2 - T_{G_0}^2 \right)$ with
  \begin{align*}
    T_{G_0}^1-T_{G_1}^1 & = \sum_{\substack{e \in E_{G_0} \cap D_{\phi} \\ P \in L_{G_0}(e)}}1-|\{ t / \phi(e,P,t)=(\_,\_, \sige) \}|+ \sum_{e\in E_{G_0} \cap \overline{D_{\phi}}}|L_{G_0}(e)| -\sum_{f \in E_{G_1} \cap \overline{D_{\phi}'}}|L_{G_0}(e)| \\
    T_{G_0}^1-T_{G_1}^1& = \sum_{\substack{e \in E_{G_0} \cap D_{\phi} \\ P \in L_{G_0}(e)}}1-1+ \sum_{e\in \varnothing}|L_{G_0}(e)| -\sum_{f \in E_{G_1} \cap \overline{D_{\phi}'}}|L_{G_0}(e)| \\
    T_{G_0}^1-T_{G_1}^1& = -\sum_{f \in E_{G_1} \cap \overline{D_{\phi}'}}|L_{G_0}(e)| \\
    T_{G_0}^1-T_{G_1}^1& < 0
  \end{align*}

  and 
  \begin{align*}
    T_{G_0}^2-T_{G_1}^2 & = 2 \left( \sum_{\substack{e \in P_{G_0}\cap D_{\phi} \\ P \in L_{G_0}(e) \\ s \in Si_{G_0}(e,P)\\ (f,Q,t)=\phi(e,P,s)}}D_{G_0}(e).|s|-D_{G_1}(f).|t| + \sum_{\substack{e \in P_{G_0}\cap \overline{D_{\phi}} \\ P \in L_{G_0}(e) \\ s \in Si_{G_0}(e,P)}}D_{G_0}(e).|s| \right)\\
    T_{G_0}^2-T_{G_1}^2 & = 2 \left( \sum_{\substack{e \in P_{G_0} \\ P \in L_{G_0}(e) \\ s \in Si_{G_0}(e,P)\\ (f,Q,t)=\phi(e,P,s)}}D_{G_0}(e).|s|-D_{G_0}(e).|s| + \sum_{\substack{e \in \varnothing \\ P \in L_{G_0}(e) \\ s \in Si_{G_0}(e,P)}}D_{G_0}(e).|s| \right)\\
    T_{G_0}^2-T_{G_1}^2 & = 0
  \end{align*}
  So we have, as expected, $T_{G_0} \leq T_{G_1}$
\end{proof}

\begin{lemma} \label{trans_incr_t}
If $G_1 = (G)_0 $, then $T_{G_1} \leq T_G$
\end{lemma}
\begin{proof}
  The proof is quite similar to the proof of lemma \ref{eta_incr_t}, but easier. We consider the canonical mapping from $G_1$ to $G$ and show that this is a copymorphism. The main point being the conservation of paths. And this conservation is shown using the $\beta(\_,\_,\_,\_)$ function on contexts.
\end{proof}

\subsection{$G$ is not much bigger than $G_0$}
For purposes of concision, we will define for any proof-net $G$, $maxF(G)= \max_{e \in E_G}|\beta(e)|$ and $maxP(G)=\max( \max_{\substack{e \in E_{G_1}\\ \beta_{G_1}(e)=C[p.X]}}p, \max_{\substack{e \in E_{G_1}\\ \beta_{G_1}(e)=C[p.X^\perp]}}p)$.

\begin{lemma} \label{comp_g1_g0}
  If $G_0 \rightarrow_\eta^* G_1$, then:
  \begin{itemize}
  \item $|E_{G_1} | \leq | E_{G_0} | + 3 * |E_{G_0}| * maxF(G_0)$
  \item $l(G_1)  \leq l(G)  + |G| * maxF(G_0)$ 
  \item $\partial_{G_1}) \leq \partial (G) + |G| * maxF(G_0)$
  \item $maxF(G_1) = maxF(G_0)$
  \end{itemize}
\end{lemma}
\begin{proof}
  We first prove the following properties:
  \begin{itemize}
  \item If $G_0 \rightarrow_\eta G_1$, then $| E_{G_1}| +\frac{3}{2} * \Sigma_{\alpha_{G_1}(l)=ax} |\beta_{G_1}(l,\_)| \leq | E_{G_0} | + \frac{3}{2} * \Sigma_{\alpha_{G_0}(l)=ax}  |\beta_{G_0}(e)|$
  \item If $G_0 \rightarrow_\eta G_1$, then $l(G_1)+\Sigma_{\alpha_{G_1}(l)=ax} \frac{|\beta_{G_1}(l,\_)|}{2} \leq l(G_0) + \Sigma_{\alpha_{G_0}(l)=ax} \frac{|\beta_{G_0}(l,\_)|}{2}$
  \item If $G_0 \rightarrow_\eta G_1$, then $\partial_{G_1}) + \Sigma_{\alpha_{G_1}(l)=ax} \frac{|\beta_{G_1}(l,\_)|}{2} \leq \partial_{G_0}) + \Sigma_{\alpha_{G_0}(l)=ax} \frac{|\beta_{G_0}(l,\_)|}{2}$ 
  \item If $G_0 \rightarrow_\eta G_1$, then $maxF(G_1)= maxF(G_0)$
  \end{itemize}
  The expected results are then obtained because each $\eta$-expansion step decreases by at least $1$, the quantity $\Sigma_{\alpha_{G_1}(l)=ax}|\beta_{G_1}(l,\_)|$.
\end{proof}

\begin{lemma} \label{comp_g_g1}
  If $(G)_0=G_1$, then:
  \begin{itemize}
  \item $|E_{G}| \leq |E_{G_1}|+ |E_{G_1}|.maxP(G_1)$
  \item $l(G) \leq l(G_1) + maxP(G_1)$
  \item $\partial_G \leq \partial_{G_1}) + maxP(G_1)$
  \item $maxF(G) \leq maxF(G_1).maxP(G_1)$
  \end{itemize}
\end{lemma}
\begin{proof}
  It is enough to notice that, along a directed path (from an axiom $l$ to a cut or pending edge) of $G$ the number of $\S$ link is exactly the index $p$ which is assigned to the variables of $l_1$ (the link of $G_1$ corresponding to $l$)
\end{proof}

\begin{theorem}
Let $G$ be a $L^4_0$ proof net, any reduction of $G$ terminates in at most $p_{\partial_G,l(G),f(G)}(|G|)$ steps, with $p_{a,b,c}$ being a fixed polynomial for any $a,b,c \in \mathbb{N}^3$ and $f(G)=\max_{F \in E_G}\partial(F)$
\end{theorem}
\begin{proof} ~
For matters of readability, we will write $\partial$ for $\partial_G$, $l$ for $l(G)$, $f$ for $\max_{F \in E_G} |F|$and $d$ for $\max_{F \in E_G}\partial(F)$. Adding $~_0$ or $~_1$ to those symbols designate the corresponding notions for $G_0$ or $G_1$.
\begin{align*}
  T_{G_0} & \leq T_{G_1} \text{(lemma~\ref{eta_incr_t})}\\
         & \leq T_{G}  \text{(lemma~\ref{trans_incr_t})}\\
         & \leq \left(\partial.|E_G| \right)^{(\partial . (6.l+11))^{l+1}} \text{(theorem~\ref{theorem_strong_bound_l4})}\\
         & \leq \left(l_1.\partial_1.|E_{G_1}| \right)^{(\partial . (6.l_1+f_1+11))^{l_1+f_1+1}} \\
         & \leq \left((l_0+f_0).(\partial_0+f_0).(1+3.f_0).|E_{G_0}| \right)^{(\partial . (6.l_0+7.f_0+11))^{l_0+2.f_0+1}} \\
\end{align*}
\end{proof}

\bibliographystyle{plain}
\bibliography{megabib}

\appendix
\section{Definition of Residues}\label{def_residues}
First, we will define positions of a term. A position of $t$ indicates a node in the syntactic tree of $t$. It will be useful to manage occurrences of subterms.

Let $t$ be a $\lambda$-term, the positions of $t$ is a set of lists of integers defined recursively by:
\begin{itemize}
\item If $t=x$, then $Pos(t)= \{[]\}$
\item If $t=(u) v$, then $Pos(t)= \{[]\} \cup \Set{p.0}{p \in Pos(u)} \cup \Set{p.1}{p \in Pos(v)}$
\item If $t=\lambda x.u$, then $Pos(t)= \{[]\} \cup \Set{p.0}{p \in Pos(u)}$
\end{itemize}

If $p \in Pos(t)$, the subterm of $t$ at position $p$, denoted $t|_p$ is defined inductively on $p$ by:
\begin{itemize}
\item If $p=[]$, then $t|_p=t$
\item If $p=q.0$ and $t=\lambda x.u$ then $t|_p=u|_q$
\item If $p=q.0$ and $t=(u)v$ then $t|_p=u|_q$
\item If $p=q.1$ and $t=(u)v$ then $t|_p=v|_q$
\end{itemize}

Then, during a $\beta$-reduction step $(\lambda x. t)u \rightarrow_\beta t[u/x]$, we want to know the positions of $t[u/x]$ at which $u$ has been copied. These correspond to the positions of $t$ where $x$ appear free. Thus, we need to define what ``the positions of $t$ where $x$ appears free'' means. 

If $x \in FV(t)$ and $x$ does not appear bound in $t$ (we can use $\alpha$-conversion to be in this case), we define the {\it positions of $x$ in $t$} as the set:
\begin{equation*}
  pos_t(x)= \Set{p \in Pos(t)}{t|_p=x}
\end{equation*}

Finally, we have all the tools to define the residue of a subterm through a $\beta$-reduction step. If $p \in Pos(t)$ and $t \rightarrow_{\beta} t'$ then we define the residues of $p$ in $t'$ by induction on $p$
\begin{itemize}
  \item If $p=[]$, then $R_{t,t'}([])=[]$
  \item If $p=q.0$ and $t=\lambda x.u$ then $t'=\lambda x.u'$ with $u \rightarrow_{\beta}u'$ and $R_{t,t'}(p)=\Set{r.0}{ r \in R_{u,u'}(q)}$
  \item If $t=(u) v$, $p=q.0$ and $t'=(u)v'$ with $v \rightarrow_{\beta}v'$ then $R_{t,t'}(p)=p$.
  \item If $t=(u) v$, $p=q.1$ and $t'=(u')v$ with $u \rightarrow_{\beta}u'$ then $R_{t,t'}(p)=p$.
  \item If $t=(u) v$, $p=q.1$ and $t'=(u)v'$ with $v \rightarrow_{\beta}v'$ then $R_{t,t'}(p)=\Set{r.1}{ r \in R_{v,v'}(q)}$.
  \item If $t=(u) v$, $p=q.0$ and $t'=(u')v$ with $u \rightarrow_{\beta}u'$ then $R_{t,t'}(p)=\Set{r.0}{ r \in R_{u,u'}(q)}$.
  \item If $t=(\lambda x.u) v$, $p=q.0$ and $t'=u[v/x]$ then $R_{t,t'}(p)=p$
  \item If $t=(\lambda x.u) v$, $p=q.1$ and $t'=u[v/x]$ then $R_{t,t'}(p)=\Set{q@r}{r \in pos_u(x)}$
\end{itemize}

\section{Discussion on previous work on $L^4$ strong bound}
\label{appendix_vercelli}
In Section 8.2.2 of \cite{vercelli2010phd}, Vercelli claims a proof of strong polynomial bounds for some subsystems of $MS^\dagger$. $L^4$ is one of these systems. However, the proof of the strong bound contains some flaws. Indeed, the $\mapsto$ relation used in this section has no rule to leave a box by its principal door. Moreover, the weight $T_G$ used differs from the weight used by Dal Lago~\cite{dal2005quantitative} and us. In the following, $T_G$ designs the weight defined by Vercelli and we will show that the lemma 8.2.15 - which corresponds to the ``Dal Lago's weight theorem'' - is false. Indeed, in Figure \ref{vercelli_fail}, $G \rightarrow_{cut} H$ but $T_G=0+2+2+10=14$: $0$ for box $B$ door because no maximal $CS$-path begin by $\sigma(B)$, $2$ for both boxes at depth $0$, $1$ for each node which is neither an axiom nor a door. And $T_H=2.4.2+2+2+10=22$: $2.4.2$ for $B$ door because each of the $4$ $B$ copies has length $2$, $2$ for the box at depth $0$, $1$ for each node which is neither an axiom nor a door. So $T_G < T_H$.

\begin{figure}
  \centering
  \begin{tikzpicture}
    \node [princdoor] (bi)  at (0,0) {};
    \node at ($(bi)+(0.4,-0.2)$) {$B$};
    \node [par]       (par) at ($(bi) +(0, 0.7)$) {};
    \node [ax]        (ax)  at ($(par)+(0, 0.7)$) {};
    \draw [ar,out=- 40,in= 40] (ax) to (par);
    \draw [ar,out=-140,in=140] (ax) to (par);
    \draw [ar] (par) -- (bi);
    \node [princdoor] (be)  at ($(bi) +(0,-0.8)$) {};
    \draw [ar] (bi)  -- (be);
    \draw (bi) -| ++ (0.5,1.5) -| ($(bi)+(-0.5,0)$) -- (bi);
    \draw (be) -| ++ (0.8,2.5) -| ($(be)+(-0.8,0)$) -- (be);
    \node [auxdoor]  (aux) at ($(be)+(2.5,0)$) {};
    \node [cut]      (cut) at ($(be)!0.5!(aux)+(0,-0.5)$) {};
    \draw [ar, out= -60, in=180] (be)  to (cut);
    \draw [ar, out=-120, in=  0] (aux) to (cut);
    \node [cont] (cont11) at ($(aux)+(0,0.8)$) {};
    \node [cont] (cont21) at ($(cont11)+( 40:1)$) {};
    \node [cont] (cont22) at ($(cont11)+(140:1)$) {};
    \node [weak] (weak1) at ($(cont21)+(60:0.7)$) {};
    \node [weak] (weak2) at ($(cont21)+(120:0.7)$) {};
    \node [weak] (weak3) at ($(cont22)+(60:0.7)$) {};
    \node [weak] (weak4) at ($(cont22)+(120:0.7)$) {};
    \draw [ar] (weak1) -- (cont21);
    \draw [ar] (weak2) -- (cont21);
    \draw [ar] (weak3) -- (cont22);
    \draw [ar] (weak4) -- (cont22);
    \draw [ar] (cont21) -- (cont11);
    \draw [ar] (cont22) -- (cont11);
    \draw [ar] (cont11) -- (aux);
    \node [princdoor] (pd)  at ($(aux)+(1.8,0)$) {};
    \node [par]       (par) at ($(pd) +(0, 0.7)$) {};
    \node [ax]        (ax)  at ($(par)+(0, 0.7)$) {};
    \draw [ar,out=- 40,in= 40] (ax) to (par);
    \draw [ar,out=-140,in=140] (ax) to (par);
    \draw [ar] (par) -- (pd);
    \draw (pd) -| ++ (0.7,2.5) -| ($(aux)+(-1.5,0)$) -- (aux) -- (pd);
    \draw [ar] (pd) -- ++ (0,-0.8);

    \draw [->,very thick] (5.5,0) --++ (1.5,0) node [below left] {$cut$};
    \begin{scope}[xshift=8cm]
    \node [princdoor] (bi)  at (0,0) {};
    \node at ($(bi)+(0.4,-0.2)$) {$B$};
    \node [par]       (par) at ($(bi) +(0, 0.7)$) {};
    \node [ax]        (ax)  at ($(par)+(0, 0.7)$) {};
    \draw [ar,out=- 40,in= 40] (ax) to (par);
    \draw [ar,out=-140,in=140] (ax) to (par);
    \draw [ar] (par) -- (bi);
    \coordinate (be)  at ($(bi) +(0,-0.8)$) ;
    \draw (bi) -| ++ (0.5,1.5) -| ($(bi)+(-0.5,0)$) -- (bi);
    \coordinate (aux) at ($(be)+(2.5,0)$) {};
    \node [cont] (cont11) at ($(aux)+(0,0.8)$) {};
    \node [cont] (cont21) at ($(cont11)+( 40:1)$) {};
    \node [cont] (cont22) at ($(cont11)+(140:1)$) {};
    \node [weak] (weak1) at ($(cont21)+(60:0.7)$) {};
    \node [weak] (weak2) at ($(cont21)+(120:0.7)$) {};
    \node [weak] (weak3) at ($(cont22)+(60:0.7)$) {};
    \node [weak] (weak4) at ($(cont22)+(120:0.7)$) {};
    \node [cut] (cut) at ($(bi)!0.5!(cont11)+(0,-0.5)$) {};
    \draw [ar, out= -60, in=180] (bi)  to (cut);
    \draw [ar, out=-120, in=  0] (cont11) to (cut);
    \draw [ar] (weak1) -- (cont21);
    \draw [ar] (weak2) -- (cont21);
    \draw [ar] (weak3) -- (cont22);
    \draw [ar] (weak4) -- (cont22);
    \draw [ar] (cont21) -- (cont11);
    \draw [ar] (cont22) -- (cont11);
    \node [princdoor] (pd)  at ($(aux)+(1.8,0)$) {};
    \node [par]       (par) at ($(pd) +(0, 0.7)$) {};
    \node [ax]        (ax)  at ($(par)+(0, 0.7)$) {};
    \draw [ar,out=- 40,in= 40] (ax) to (par);
    \draw [ar,out=-140,in=140] (ax) to (par);
    \draw [ar] (par) -- (pd);
    \draw (pd) -| ++ (0.7,2.5) -| ($(aux)+(-3.2,0)$) -- (aux) -- (pd);
    \draw [ar] (pd) -- ++ (0,-0.8);

    \end{scope}

  \end{tikzpicture}
  \caption{\label{vercelli_fail} The proof-net $G$ reduces to $H$, but $T_G < T_H$}
\end{figure}
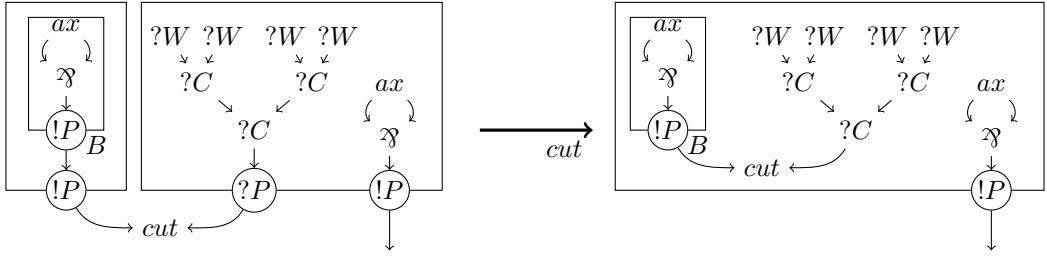

If we want the lemma 8.2.15 to hold, we could allow the contexts to leave boxes by their principal door (as in our $\mapsto$ relation). Then there would be a problem in the way the $\wn D$ is handled. Indeed, crossing a $\wn D$ node upwards adds a signature on the potential without entering a box. Thus, the lemma 8.2.15 would still fail, as shown on Figure \ref{vercelli_fail2}: $(\sigma(B),[\sigl(\sige);\sige],[\oc_t],+) \rightsquigarrow (a,[\sigl(\sige)],[\oc_t; \oc_{\sige}],+) \rightsquigarrow (b,[\sigl(\sige)],[\oc_t; \oc_{\sige}],-) \rightsquigarrow (c,[\sigl(\sige);\sige],[\oc_t],-) \rightsquigarrow^2 (d,[\sigl(\sige)],[\oc_t;\oc_{\sige}],+) \rightsquigarrow (e,[\sigl(\sige)],[\oc_t;\oc_{\sige}],-)  \not \rightsquigarrow$.

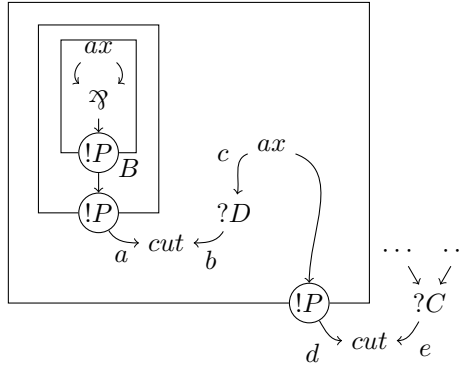
\begin{figure}
  \centering
  \begin{tikzpicture}
    \node [princdoor] (bi)  at (0,0) {};
    \node at ($(bi)+(0.4,-0.2)$) {$B$};
    \node [par]       (par) at ($(bi) +(0, 0.7)$) {};
    \node [ax]        (ax)  at ($(par)+(0, 0.7)$) {};
    \draw [ar,out=- 40,in= 40] (ax) to (par);
    \draw [ar,out=-140,in=140] (ax) to (par);
    \draw [ar] (par) -- (bi);
    \node [princdoor] (be)  at ($(bi) +(0,-0.8)$) {};
    \draw [ar] (bi)  -- (be);
    \draw (bi) -| ++ (0.5,1.5) -| ($(bi)+(-0.5,0)$) -- (bi);
    \draw (be) -| ++ (0.8,2.5) -| ($(be)+(-0.8,0)$) -- (be);
    \node [der] (der) at ($(be) +(1.8,0)$) {};
    \node [cut] (cut) at ($(be)!0.5!(der)+(0,-0.4)$) {};
    \draw [ar, out= -60, in=180] (be)  to node [edgename,below] {$a$} (cut);
    \draw [ar, out=-120, in=  0] (der) to node [edgename,below] {$b$} (cut);
    \node [princdoor] (bb) at ($(der)+(1,-1.2)$) {};
    \node [ax] (ax) at ($(der)!0.5!(bb)+(0,1.5)$) {};
    \draw [ar, out=-160, in=80] (ax) to node [edgename, above left] {$c$} (der);
    \draw [ar, out= -20, in=90] (ax) to (bb);
    \draw (bb) -| ++(0.8,4) -| ($(bb)+(-4,0)$) -- (bb);
    \node [cont] (cont) at ($(bb)+(1.6,0)$) {};
    \node [cut] (cut) at ($(bb)!0.5!(cont)+(0,-0.5)$) {};   
    \draw [ar, out= -60, in=180] (bb)   to node [edgename, below left]  {$d$} (cut);
    \draw [ar, out=-120, in=  0] (cont) to node [edgename, below right] {$e$} (cut);
    \node [etc] (etc1) at ($(cont)+(120:0.8)$) {};
    \node [etc] (etc2) at ($(cont)+( 60:0.8)$) {};
    \draw [ar] (etc1)--(cont);
    \draw [ar] (etc2)--(cont);
  \end{tikzpicture}
  \caption{\label{vercelli_fail2} The path beginning by $(\sigma(B),[\sigl(\sige);\sige],[\oc_t],+)$ can not cross the contraction}
\end{figure}

If we fix this problem by taking our $\mapsto$ relation, then lemma 8.2.17 would fail. Indeed, crossing a $\wn D$ node changes the number of exponential stack element in the stack without changing the length of the potential. If we fixed it by replacing the length of the potential by the level of the edge in the enunciation of lemma 8.2.17 then the lemma would fail on the $\hookrightarrow$ steps because the doors of a same box may have different levels. So the correct form of the lemma is:

If $G$ is a $L^4$ proof-net and $(e,P,T,p) \rightsquigarrow^*_G (f,Q,U,q)$, then 
\begin{equation*}l(e)+ |T|_{\{!,?,\S\}} = l(f) + |U|_{\{!,?,\S\}} \end{equation*}

This is exactly our lemma \ref{L3ConsLevelContext}. However, proving that this weaker lemma is enough is far from trivial.

\section{Notations}
\subsection{Subscripts and superscripts}
\begin{itemize}
\item $A_{/LL}$, with $A$ a $LL_0$ formula, designs the formula $A$ where the indices on the atomic formulae and the $\S$ connectives are deleted. It is defined in page \pageref{def_proj_fll}.
\item $A_{/L^4}$, with $A$ a $LL_0$ formula, designs the formula $A$ where the indices on the atomic formulae are deleted. It is defined in page \pageref{def_proj_fll}.
\item $A_{/L^4_0}$, with $A$ a $LL_0$ formula, designs the formula $A$ where the $\S$ connectives are deleted. It is defined in page \pageref{def_proj_fll}.
\item $a^c_b$ is equal to $c$ if $b=0$, otherwise it is defined as $a^{a^c_{b-1}}$. It is defined in page \pageref{def_exponentialtower}.
\item $b^\perp$, with $b$ a polarity designs the other polarity ($+^\perp=-$ and $-^\perp=+$). It is defined in page \pageref{def_perp}.
\item $t^\perp$ with $t$ a trace element is defined by ${\parr_l}^\perp=\otimes_l$, ${\parr_r}^\perp=\otimes_r$, ${\otimes_l}^\perp=\parr_l$, ${\otimes_r}^\perp=\parr_r$, $\forall^\perp=\exists$, $\exists^\perp=\forall$, $\S^\perp=\S$, ${\oc_t}^\perp=\wn_t$ and ${\wn_t}^\perp=\oc_t$. It is defined in page \pageref{def_perp}.
\item $A^\perp$, with $A$ a formula designs the dual formula $(X)^\perp=X^\perp$, $(X^\perp)^\perp$ and the connectives are changed by their duals. It is defined in page \pageref{def_perpformula}.
\item $t_{|p}$, with $t$ an exponential signature and $p$ a position of $t$ refers to the sub-exponential signature pointed by $p$. It is defined in page \pageref{def_tbarp}.
\item $C_{|p}$, with $C$ a context and $p$ a position of $C$ refers to the exponential signature pointed by $p$. It is defined in page \pageref{def_cbarp}.
\item $C\downarrow^t_p$, with $C$ a context, $t \in Sig$ and $p \in Pos(C)$, represents the context obtained by replacing in $C$ the exponential signature at position $p$ by $t$, replacing the $\sign(t_1,t_2)$ above it by $\sigp(t_2)$ (if it exists), and (if $p$ refers to a trace element) delete the trace elements on the left of $p$. It is defined in page \pageref{def_downarrowtp}.
\end{itemize}

\subsection{Arrows}
\begin{itemize}
\item $\rightsquigarrow$: local relation on contexts. It is defined in page \pageref{def_rightsquigarrow}.
\item $\hookrightarrow$: relation on contexts which makes a ``jump'' between an auxiliary door and a principal door of a box. It is defined in page \pageref{def_hookrightarrow}.
\item $\mapsto$: relation on contexts, it is the union of $\rightsquigarrow$ and $hookrightarrow$. It is defined in page \pageref{def_mapsto}.
\item $\twoheadrightarrow$ is a relation defined on the boxes of a proof-net. $B \twoheadrightarrow B'$ if there exists a path beginning by the principal door of $B$ (with trace $[\oc_t]$) and entering $B'$ by its principal door.
\end{itemize}

\subsection{Orders}
\begin{itemize}
\item $\sqsubseteq$: cf. the definition of ``simplification''.
\item $\sqsubset$: $t \sqsubset t'$ if $t'$ is a simplification of $t'$ ($t \sqsubseteq t'$) and $t \neq t'$. It is defined in page \pageref{def_sqsubset}.
\item $\Subset$: relation on $Pot(B_G) \times Sig$. $(B,P,t) \Subset (B',P',t')$ intuitively means ``In the normal forms of $G$, the reduct of $B$ corresponding to the copy $t$ of its potential $P$ is strictly included in the reduct of $B'$ corresponding to the copy $t'$ of its potential $P'$. It is defined in page \pageref{def_Subset}. 
\item $\preccurlyeq$ is a relation on exponential signatures: intuitively $t \preccurlyeq u$ if $t$ is ``shorter'' than ``u''. It is defined in page \pageref{def_preccurlyeqsig}.
\item $\succcurlyeq^s_k$, cf. the definition of ``$k$-joins''.
\item $\geq^s_k$, let $(B,P)$ and $(B',P')$ be potential boxes, $(B,P) \geq^s_k (B',P')$ iff there are $k$ auxiliary doors of $(B',P')$ reachable from contexts of the shape $(\sigma(B),P,[\oc_t],+)$. It is defined in page \pageref{def_geqsk}.
\item $\vartriangleleft$, cf. the definition of ``tree truncation''.
\item $\blacktriangleleft$, cf. the definition of ``subtree''.
\end{itemize}

\subsection{Others}
\begin{itemize}
\item $q.A$, with $q \in \mathbb{N}$ and $A$ a formula stands for the formula $A$ where we add $q$ to all indices on atomic formulae. It is defined in page \pageref{def_decindices}.
\item $l.x$, with $l$ a list, is the list obtained by adding the element $x$ on the right of $l$. It is defined in page \pageref{def_insertion}.
\item $\partial(~)$, cf. the definition of ``depth''.
\item $l_1@l_2$ is equal to the concatenation of the lists $l_1$ and $l_2$. It is defined in page \pageref{def_arobase}.
\item $X$, whenever $X$ is a set is the cardinal of $X$ (its number of elements). It may be infinite. It is defined in page \pageref{def_cardinal}. 
\item $|[a_1;\cdots;a_k]|$ is equal to $k$, the number of elements of the list. It is defined in page \pageref{def_listrestriction}.
\item $|[a_1;\cdots;a_k]|_X$ is the number of indices $i$ such that $a_i$ is in $X$. It is defined in page \pageref{def_listrestriction}.
\item $f(A)$, with $f$ is a mapping and $A$ a subset of the domain of $f$, refers to the set of images of elements of $A$ by $f$. It is defined in page \pageref{def_setimage}.
\item $(B,P) \curvearrowright_{I,(e,Q)}(B',P')$: intuitively means that there exists a copy of $(B,P)$ whose itinerary is $I$ and arrive at context $(\sigma(e),Q,[\oc_{\sige}],-)$.
  \item $\partial_G$, whenever $G$ is a proof-net, stands for the maximal depth of an edge of $G$. It is defined in page \pageref{def_maxdepth}.
  \item $A[\theta]$, is the formula obtained by applying the substitution $\theta$ to the formula $A$. It is defined in page \pageref{def_acrochettheta}.
  \item $C[f]$, with $C$ a context and $f$ a mapping from positions of $t$ to $Sig$. Then $C[f]$ refers to the context obtained from $C$ by replacing the exponential signature at position $p$ by $f(p)$ (if $p$ is in the domain of $f$). It is defined in page \pageref{def_ccrochetf}.
  \item $t[f]$, with $t$ an exponential signature and $f$ a mapping from positions of $t$ to $Sig$. Then $t[f]$ refers to the exponential signature obtained from $t$ by replacing the sub-exponential signature at position $p$ by $f(p)$ (if $p$ is in the domain of $f$). It is defined in page \pageref{def_tcrochetf}.
\end{itemize}

\subsection{Letters}
\begin{itemize}
\item $C_G$, cf. the definition of ``context''.
\item $C_\rightarrow(B,P)$, cf. the definition of ``$\rightarrow$-copy''.
\item $C_s(B,P)$, cf. the definition of ``$\rightarrow$-copy''.
\item $C_s(x,P)$, with $s$ an integer, is a (more readable) synonym for $C_{\mapsto_s}(x,P)$. It is defined in page \pageref{def_ls}.
\item $D_G(B)$, where $B$ is a box of $G$ refers to the the doors of $B$. It is defined in page \pageref{def_dgb}.
\item $D_G$ is defined as the maximal number of doors of boxes in $B$. It is defined in page \pageref{def_dg}.
\item $E_G$, where $G$ is a proof-net, designs the set of edges of $G$. It is defined in Definition \ref{def_proofnet} in page \pageref{def_proofnet}.
\item $F_G$, cf. the definition of ``final contexts''.
\item $\mathcal{F}_{LL}$: designs the formulae of linear logic. It is defined in page \pageref{def_fll}.
\item $\mathcal{F}_{LL_0}$: designs the formulae of the system $LL_0$. It corresponds to the formulae of linear logic extended with indices on atomic formulae and the $\S$ modality. It is defined in page \pageref{def_fll0}.
\item $I(C,C')$, $I_s((B,P),(e,Q))$: cf. the definition of itineraries. It is defined in page \pageref{def_itinerary}.
\item $L_{\rightarrow}(B)$, cf. the definition of ``$\rightarrow$-canonical potential''.
\item $L_s(x)$, with $s$ an integer, is a (more readable) synonym for $L_{\mapsto_s}(x)$. It is defined in page \pageref{def_ls}.
\item $N_s(B)$, with $s \in \mathbb{N}$ and $B$ a box, cf. the definition of ``nest''.
\item $N(B)$ is defined as $N_{S_G}(B)$, cf. the definition of ``nest''. It is defined in page \pageref{def_nest}.
\item $N_G$ is the maximum nest of boxes: $N_G= \max_{B \in B_G}N(B)$. It is defined in page \pageref{def_nest}.
\item $Si_\rightarrow(B,P)$, cf. the definition of ``$\rightarrow$-copy''.
\item $S(B)$, with $B$ a box, cf. the definition of ``strata of a box''.
\item $S(C)$, with $C$ a context, cf. the definition of ``strata of a context''.
\item $S_G$, with $G$ a proof-net, cf. the definition of ``stratified proof-net''.
\item $T_G$, where $G$ is a proof-net is a weight associated to this proof-net. This weight decreases along $cut$-eslimination. It is defined as $T_G = \sum_{e \in E_G} | L_{\mapsto}(e)| + 2. \sum_{B \in B_G} \left( D_G(B) \sum_{P \in L_{\mapsto}(B)} \sum_{t \in S_{\mapsto}(B,P)} |t| \right )$. It is defined in Definition \ref{def_tg} in page \pageref{def_tg}.
\end{itemize}

\subsection{Greek letters}
\begin{itemize}
\item $\alpha(l)$, where $l$ is a link of a proof-net, refers to the label of $l$ ($ax$, $cut$, $\parr$, $\otimes$, $\exists$, $\forall$, $\oc P$, $\wn P$, $\wn C$, $\wn W$, $\wn D$ or $\wn N$). It is defined in Definition \ref{def_proofnet} in page \pageref{def_proofnet}.
\item $\beta(e)$, where $e$ is an edge of a proof-net, refers to the formula labelling $e$. It is defined in Definition \ref{def_proofnet} in page \pageref{def_proofnet}.
\item $\beta_{\{\}}(A,e,P,T,T',p)$, where $A$ is a formula, $(e,P)$ a potential edge, $T$ and $T'$ traces and $p$ a polarity, is the set of underlying formulae of $(A,e,P,T,T',p)$. Its unique purpose is to be used to define the underlying formulae of a context. It is defined in page \pageref{def_betasetshit}.
\item $\beta_{\{\}}(C)$, where $C$ is a context, cf. the definition of the underlying formula{\bf e} of a context.
\item $\beta(C)$, where $C$ is a context, cf. the definition of the underlying formula of a context.
\item $\beta(C)$, where $C$ is a context, cf. the definition of the underlying formula of a context.
\item $\rho_G(e)$ is the deepest box of $G$ containing $e$. It is defined in page \pageref{def_rhoge}.
\item $\sigma(B)$, where $B$ is a box, is the edge going out of the principal door of $B$. It is defined in page \pageref{def_sigmab}.
\item $\sigma_i(B)$, where $B$ is a box, is the edge going out of the $i$-th auxiliary door of $B$. It is defined in page \pageref{def_sigmaib}.
\end{itemize}

\subsection{Words}
\begin{itemize}
\item acyclic: A proof net is said acyclic if there is no $\mapsto$-copy context $(e,P,[\oc_t],p)$, $(e,Q) \in Pot(e)$ and $u \in Sig$ such that $(e,P,[\oc_t],p) \mapsto^* (e,Q,[\oc_u],p)$. It is defined in page \pageref{def_acyclic}.
\item auxiliary doors: The auxiliary doors of box $B$ are the links, labelled by $\wn P$ on the bottom side of a box. It is defined in page \pageref{def_auxiliarydoor}.
\item box: set of links of a proof-net. The boxes are usually by rectangles. It is defined in page \pageref{def_box}.
\item $Can(x)$ is a shortcut for $Can_{\mapsto}(x)$
\item $Can_s(x)$ is a shortcut for $Can_{\mapsto_s}(x)$
\item $Can_\rightarrow(x)$, with $x$ an edge or a box, is defined as the set $\Set{(x,P)}{P \in L_\rightarrow(x)}$. It is defined in Definition \ref{def_canonicalpotential} in page \pageref{def_canonicalpotential}.
\item $\rightarrow$-canonical box: A potential box $(B,P)$ is called a canonical edge if $P \in C_{\rightarrow}(B,P)$. It is defined in page \pageref{def_canonicalbox}.
\item $\rightarrow$-canonical context. Intuitively, a context $C$ is $\rightarrow$-canonical if there exists some $(B,P) \in Can_{\rightarrow}(B_G)$ and $t \in S_s(B,P)$ such that $(\sigma(B),P,[\oc_t],+) \rightarrow^* C$. It is defined in Definition \ref{def_canonicalcontext} in page \pageref{def_canonicalcontext}.
\item $\rightarrow$-canonical edge: A potential edge $(e,P)$ is called a canonical edge if $P \in C_{\rightarrow}(e,P)$. It is defined in page \pageref{def_canonicaledge}.
\item $\rightarrow$-canonical potential. Let $\rightarrow$ be a cut simulation, and $e$ an edge such as $e \in B_{\partial(e)} \subset \cdots \subset B_1$, then a canonical potential for $e$ is a potential $P$ whose length is $\partial(e)$ and such as the $i$-th exponential signatures of $P$ is a copy for $B_i$. The set of $\rightarrow$canonical potentials of $e$ is written $L_\rightarrow(e)$, $L_\rightarrow(B)$ refers to $L_\rightarrow(\sigma(B)$. It is defined in Definition \ref{def_canonicalpotential} in page \pageref{def_canonicalpotential}.
\item $Col_s(B,P)$, cf. the definition of colonies.
\item colony: Let $(B,P)$ be a potential box. The colonies of $(B,P)$ at stratum $s$ are the first auxiliary doors that a $\mapsto_s$-path from $(\sigma(B),P,[\oc_t],+)$ can reach (with $t \in Sig$) which belong to a box $B'$ with $N(B)>N(B')$. The set of colonies of $(B,P)$ at stratum $s$ is written $Col_s(B,P)$. $Col(B,P)$ refers to $Col_{S_G}(B,P)$. It is defined in page \pageref{def_colony}.
\item $concl$: cf. the definition of ``conclusion''.
\item conclusion: the conclusions of the link $l$ refers to the outgoing edges of $l$. The set of conclusions of $l$ is written $concl(l)$. It is defined in page \pageref{def_conclusion}
\item context: A context of $G$ is an element $(e,P,T,p)$ with $e$ an edge of $G$, $P$ a potential, $T$ a trace and $p$ a polarity. The set of contexts is written $C_G$. It is defined in page \pageref{def_context}.
\item controls dependence: A principal door stratified proof net $G$ controls dependence if $\succcurlyeq_2^{S(G)}$ is irreflexive. It is defined in page \pageref{def_controlsdependence}.
\item $\rightarrow$-copy: a $\rightarrow$-copy of a potential box $(B,P)$ corresponds to duplicates of the box $B$, restricting the $cut$-elimination according to the cut simulation $\rightarrow$ and knowing in which duplicates of the box including $B$ we are. The set of $\rightarrow$-copy of potential box $(B,P)$ is written $C_\rightarrow(B,P)$. We also have special notations for specific cut simulations: $C_{s}(B,P)$ refers to $C_{\mapsto_s}(B,P)$. The simplification of $\rightarrow$-copies of $(B,P)$ is written $S_\rightarrow(B,P)$. It is defined in Definition \ref{def_copy} in page \pageref{def_copycontext}. More intuitions can be found in the beginning of Subsection \ref{timecomplexity}.
\item $\rightarrow$-copy context: informally, a context is a $\rightarrow$-copy context if the path beginning by it can not be extended by extending the exponential signature of its first trace element, and taking a shorter exponential signature for this first trace element will shorten the path. It is defined in page \pageref{def_copycontext}. More intuitions can be found in the beginning of Subsection \ref{timecomplexity}.
\item copymorphism: Let $G$ and $H$ be two proof-nets, a copymorphism from $G$ to $H$ is a tuple $(D_\phi,D'_\phi,\phi,\psi)$ with $D_\phi \subseteq E_G$, $D'_\phi \subseteq E_H$, $\phi: Pot(D_\phi) \times Sig \mapsto Pot(D'_\phi) \times Sig$ and $\psi: C_G \mapsto C_G$. These objects are required to satisfy many more properties. A copymorphism is meant to explicit the relations between a proof-net and its reduct by the $\mapsto$ relation. It is defined in Definition \ref{def_copymorphism} in page \pageref{def_copymorphism}.
\item $cut$-elimination: relation on proof-nets defined of figures \ref{mult_rules}, \ref{exp_rules} and \ref{quant_rules}.
\item $cut$-simulation: relation on contexts defined to simulate $cut$-elimination. $\mapsto$ simulates the full $cut$-elimination. $\rightsquigarrow$ and $\mapsto_s$ are other examples of $cut$ simulations. It is defined in page \pageref{def_cutsimulation}.
\item cyclic: A proof net is said cyclic if there is a $\mapsto$-copy context $(e,P,[\oc_t],p)$, $(e,Q) \in Pot(e)$ and $u \in Sig$ such that $(e,P,[\oc_t],p) \mapsto^* (e,Q,[\oc_u],p)$. It is defined in page \pageref{def_acyclic}.
\item depth: The depth of an $x$, if the relation is not precised, is its depth in terms of box inclusion. The depth of $x$, written $\partial(x)$ designs the number of boxes containing $x$. Formally, it only makes sense if $x$ is a link. We extend it to edges and boxes. $\partial(B)$ is the number of boxes in which $B$ is strictly included. $\partial(e)$ refers to the depth of the tail of $e$. It is defined in page \pageref{def_depth}.
\item eigen variables: the eigenvariables of a proof-net are the variables which are replaced in a $\forall$ link. It is defined in page \pageref{def_eigenvariable}.
\item exponential signature: objects used to represent sequences of choices during a path. They are defined by $Sig = \sige \mid \sigl(Sig) \mid \sigr(Sig) \mid \sigp(Sig) \mid \sign(Sig,Sig)$. It is defined in page \pageref{def_expsignature}.
\item final context: contexts which may correspond to the end of paths of copies. It is defined in Definition \ref{def_finalcontext} in page \pageref{def_finalcontext}.
\item head: The head of the edge $(l,m)$ of a proof-net refers to $l$. So, the head of an edge is a link. It is defined in page \pageref{def_head}.
\item itinerary: Let $C$ and $C'$ be two contexts. The itinerary between $C$ and $C'$ is the list of the indices of the auxiliary doors on which there is $\hookrightarrow$ steps in the $\mapsto$-path from $C$ to $C'$. It is denoted $I(C,C')$. We also write $I_s((B,P),(e,Q))$ for the set of itineraries of the shape $I((\sigma(B),P,[\oc_t],+),(e,Q,[\oc_{\sige}],-))$ with $t \in Si_s(B,P)$. It is defined in page \pageref{def_itinerary}.
\item $k$-joins: For any $k,s \in \mathbb{N}$, we first define a relation $\succcurlyeq_k^s$ on potential boxes by: $(B,P) \succcurlyeq_k^s (B',P')$ (we say that $(B,P)$ $k$-joins $(B',P')$) iff at least $k$ duplicates of $(B,P)$ join $(B',P')$ (while firing only cuts at level $\leq s$). Then, we define a relation $\succcurlyeq_k^s$ on boxes by: $B \succcurlyeq_k^s B'$ ($B$ $k$-joins $B'$) iff at least $k$ duplicates of $B$ join $B'$ (while firing only cuts at level $\leq s$). It is defined in page \pageref{def_kjoins}.
\item nest: Let $B$ be a box of a stratified proof-net $G$ which controls dependence. Let $s \in \mathbb{N}$, the nest of $B$ at stratum $s$ (denoted $N_s(B)$) is the depth of $B$ in terms of the $\succcurlyeq^s_2$ relation. $N(B)$ refers to $N_{S_G}(B)$. It is defined in page \pageref{def_nest}.
\item parallel position: Let $p$ and $q$ be positions. $p$ and $q$ are parallel iff there are no $r$ such that $p=q@r$ or $q=p@r$. It is defined in page \pageref{def_parallelpositions}.
\item pending edges: the pending edges of a proof-net are its edges which have no conclusions. By convention, we write that their conclusions are $\bullet$. It is defined in page \pageref{def_pending}.
\item $Pol$: cf. the definition of ``polarity''.
\item polarity: either $+$ or $-$. The set $\{+,-\}$ is written $Pol$. It is defined in page \pageref{def_polarity}.
\item $Pos(t)$, with $t$ an exponential signature: cf. the definition of ``position''.
\item $Pos(C)$, with $C$ a context: cf. the definition of ``position''.
\item Position: Let $t$ be an exponential signature. If we consider it as a tree, the positions of $t$ (denoted $Pos(t)$) refers to the node of the tree. It is defined in page \pageref{def_positions}.Let $C$ be a context. The positions of $C$ (denoted $Pos(C)$) refers to the positions of the exponential signatures of $C$ (in the potential of $C$ and in the trace of $C$). It is defined in page \pageref{def_contextpositions}.
\item positive weights: a proof-net $G$ has positive weights if for all potential boxes $(B,P) \in Pot(B_G)$, $C_{\mapsto}(G) > 0$. It is defined in page \pageref{def_positiveweights}.
\item $Pot$: cf. the definition of ``potential''.
\item $Pot(x)$: cf. the definition of ``potential box'', ``potential edge'' or ``potential link'' depending on the nature of $x$.
\item potential: list of exponential signatures. The set of potentials is written $Pot$. It is defined in page \pageref{def_potential}.
\item potential box: couple $(B,P)$ with $B$ a box and $|P|=\partial(B)$. $Pot(B)$ refers to the potential boxes which have $B$ as a first component. It is defined in page \pageref{def_potentialbox}.
\item potential edge: couple $(e,P)$ with $e$ an edge and $|P|=\partial(e)$. $Pot(e)$ refers to the potential edges which have $e$ as a first component. It is defined in page \pageref{def_potentialbox}.
\item potential link: couple $(l,P)$ with $l$ a link and $|P|=\partial(l)$. $Pot(l)$ refers to the potential boxes which have $l$ as a first component. It is defined in page \pageref{def_potentialbox}.
\item premise: the premises of the link $l$ refers to the incoming edges of $l$. It is defined in page \pageref{def_premise}.
\item principal door: The principal door of box $B$ is the link, labelled by $\oc P$ on the bottom side of a box. It is defined in page \pageref{def_principaldoor}.
\item proof-net: A proof-net is a graph-like structure representation of a proof. It is defined in Definition \ref{def_proofnet} in page \pageref{def_proofnet}.
\item quasi-standard: an exponential signature $t$ is said quasi-standard if for every subtree $n(t_1,t_2)$ of $t$, the exponential signature $t_2$ is standard. It is defined in page \pageref{def_standard}.
\item $Sig$: cf. the definition of ``exponential signatures''.
\item simplification: We say that $t'$ is a simplification of $t$ (written $t \sqsubseteq t'$) if we can transform $t$ into $t'$ by transforming some of the subtrees $n(t_1,t_2)$ of $t$ into $p(t_2)$. It is defined in page \pageref{def_sqsubseteq}.
\item skeleton: The skeleton of a trace is the trace where we drop the exponential signatures on $\oc_t$ and $\wn_t$ trace elements. It is defined in page \pageref{def_skeleton}.
\item spindle: A spindle is a couple of boxes $(B,C)$ such that the principal doors of two copies of $B$ are cut with auxiliary doors of $C$. It is defined in page \pageref{def_spindle}.
\item standard: an exponential signature is said standard if it does not contain the constructor $p$. It is defined in page \pageref{def_standard}.
\item strata of a box: the strata of a box $B$, written $S(B)$, is its depth in terms of $\twoheadrightarrow$. It is defined in page \pageref{def_strata}.
\item strata of a context: Let $C$ be a context such that $(\sigma(B),P,[\oc_t],+) \rightsquigarrow^* C$, the stratum of $C$ (written $S(C)$) is the stratum of box $B$. It is defined in page \pageref{def_strata}.
\item stratified proof-net: a proof-net $G$ is stratified if the relation $\twoheadrightarrow$ defined on its boxes is acyclic. The strata of $G$ refers to the maximum strata of its boxes: $S_G=\max_{B \in B_G}(B)$. It is defined in page \pageref{def_stratified}.
\item substitution tree: A tree with internal nodes labelled by substitutions on a single variable and leafs labelled by the void function. It is defined in page \pageref{def_substitutiontree}.
\item subtree: Let $T$ and $U$ be trees, we say that $T$ is a subtree of $U$ (denoted $T \blacktriangleleft U$) if $T$ corresponds to a branch of $U$: it contains exactly a node of $U$ and all its descendents. It is defined in page \pageref{def_subtree}.
\item tail: The tail of the edge $(l,m)$ of a proof-net refers to $m$. So, the tail of an edge is a link. It is defined in page \pageref{def_tail}.
\item $Tra$: cf. the definition of ``trace''.
\item trace: A trace is a non-empty list of trace element. The set of traces is written $Tra$. It is defined in page \pageref{def_trace}.
\item trace element: A trace element is one of the following: $\parr_l$, $\parr_r$, $\otimes_l$, $\otimes_r$, $\forall$, $\exists$, $\S$, $\oc_t$ and $\wn_t$ (with $t$ an exponential signature). The set of trace elements is written $TrEl$. It is defined in page \pageref{def_traceelement}.
\item $TrEl$: cf. the definition of ``trace element''.
\item truncation: Let $T$ and $U$ be trees. We say that $T$ is a truncation of $U$ (denoted $T \vartriangleleft U$) if $T$ can be obtained from $U$ by cutting some branches of $U$. It is defined in page \pageref{def_truncation}.
\item underlying formulae of a context: Let $C$ be a context, the underlying formulae of $C$, written $\beta_{\{\}}(C)$ represents the formula of the potential edge it ``comes from'' and its possible evolutions along the $cut$-elimination of the proof-net. It is defined in Definition \ref{def_underlyingformulae} in page \pageref{def_underlyingformulae}.
\item underlying formula of a context: Let $C$ be a context of the proof-net $G$, $e$ the edge $C$ ``comes from'' and $e'$ the reduct of $e$ in the normal form of $G$. The underlying formula of $C$, written $\beta(C)$, intuitively is the formula indexing $\beta(e')$. It is defined in page \pageref{def_underlyingformula}.
\item $ztree$: Let $(e,P)$ be a potential edge of $G$, then the complete substitution of $(e,P)$, written $ztree(e,P)$ is a substitution tree meant to represent the substitutions of eigenvariables that will occur on this edge during the $cut$ normalization of $G$. It is defined in page \pageref{def_ztree}.
\end{itemize}

\end{document}